\documentclass[a4paper,11pt]{article}
\usepackage[letterpaper,margin=1in]{geometry}
\usepackage{preamble-fv-simple}
\usepackage{macros}
\usepackage{multicol}
\usepackage{multirow}
\usepackage{todonotes}
\usepackage{textpos}

\usepackage{mathtools}
\usepackage{soul}
\usepackage{cite}
\usepackage{mathrsfs}
\usepackage{tikz}

\usepackage{phaistos}

\usepackage{twemojis}

\usepackage{fontawesome5}

\usepackage{figchild}

\usepackage{environ,graphicx}

\newcommand{\keepcomment}{1}%

\AtBeginDocument{\ifnum\keepcomment=1
  \excludecomment{comment}
\else
  \includecomment{comment}
\fi}

\usetikzlibrary{shapes}\usetikzlibrary{calc}
\tikzset{black node/.style={draw, circle, fill = black, minimum size = 4pt, inner sep = 0pt}}
\tikzset{white node/.style={draw, circle, fill = white, minimum size = 4pt, inner sep = 0pt}}
\tikzset{rect node/.style={draw, rectangle, fill = black, minimum size = 4pt, inner sep = 0pt}}

\author[1]{Benjamin Bergougnoux} %
\author[2]{Vera Chekan}
\author[3]{Giannos Stamoulis}

\affil[1]{Aix Marseille Université, CNRS, LIS, Marseille, France}%
\affil[ ]{\texttt{benjamin.bergougnoux@lis-lab.fr}} %

\affil[2]{Humboldt-Universität zu Berlin, Germany} 
\affil[ ]{\texttt{vera.chekan@hu-berlin.de}} 

\affil[3]{Université Paris Cité, CNRS, IRIF, F-75013, Paris, France} 
\affil[ ]{\texttt{stamoulis@irif.fr}} 

\title{A Logic-based Algorithmic Meta-Theorem for Treedepth: Single Exponential FPT Time and Polynomial Space\thanks{G.S.\ was supported by the project BOBR that is funded from the European Research Council (ERC) under the European Union’s Horizon 2020 research and innovation programme with grant agreement No.\ 948057. 
Most of the work on this manuscript was done while B.B.\ and G.S.\ were affiliated with University of Warsaw.
This project was initiated during the research stay of V.C.\ at the University of Warsaw.
A preliminary and short version of this manuscript appeared in the proceedings of SODA~2026~\cite{BergougnouxCS26}.
}}

\date{}

\begin{document}

\maketitle

\begin{abstract}
    For a graph $G$, the parameter treedepth measures the minimum depth among all forests~$F$, called elimination forests, such that $G$ is a subgraph of the ancestor-descendant closure of~$F$.
    We introduce a logic, called \emph{neighborhood operator logic with acyclicity, connectivity and clique constraints} ($\NEOtwo[\finrec]\ack$ for short), that captures all NP-hard problems---like \textsc{Independent Set} or \textsc{Hamiltonian Cycle}---that are known to be tractable in time $2^{\Oh(\td)}n^{\Oh(1)}$ and space~$n^{\mathcal{O}(1)}$ on $n$-vertex graphs provided with elimination forests of depth $\td$.
    We provide a model checking algorithm for $\NEOtwo[\finrec]\ack$ with such complexity that unifies and extends these results.
    For $\NEOtwo[\finrec]\plusk$, the fragment of the above logic that does not use acyclicity and connectivity constraints, we get a strengthening of this result, where the space complexity is reduced to $\Oh(\td\log(n))$.

    With a similar mechanism as
    the distance neighborhood logic introduced in [Bergougnoux, Dreier and Jaffke, \textit{SODA 2023}], the logic $\NEOtwo[\finrec]\ack$ is an extension of the fully-existential \MSOtwo with predicates for (1) querying generalizations of the neighborhoods of vertex sets, (2) verifying the connectivity and acyclicity of vertex and edge sets, and (3) verifying that a vertex set induces a clique.
    Interestingly, $\NEOtwo[\finrec]$, the fragment of $\NEOtwo[\finrec]\plusk$ that does not use clique constraints, is equivalent (up to minor features) to a variant of modal logic---introduced in [Pilipczuk, \textit{MFCS 2011}]---that captures many problems known to be tractable in single exponential time when parameterized by treewidth.
    Our results provide $2^{\Oh(\td)}n^{\Oh(1)}$ time and $n^{\mathcal{O}(1)}$ space algorithms for problems for which the existence of such algorithms was previously unknown.
    In particular, $\NEOtwo[\finrec]$ captures CNF-SAT via the incidence graphs associated to CNF formulas, and it also captures several modulo counting problems like \textsc{Odd Dominating Set}.

    To prove our results, we extend the applicability of  algebraic transforms such as the inclusion-exclusion principle and the discrete Fourier transform.
    To our knowledge, this is the first time, the discrete Fourier transform have been used to obtain space-efficient algorithms on graphs of bounded treedepth.
    To achieve the logspace complexity for $\NEOtwo[\finrec]\plusk$, we also use the technique from [Pilipczuk and Wrochna, \textit{ACM Trans.\ Comput.\ Theory 2018}] based on Chinese remainder theorem.
\end{abstract}

\begin{textblock}{20}(12.3,0.2)
\includegraphics[width=30px]{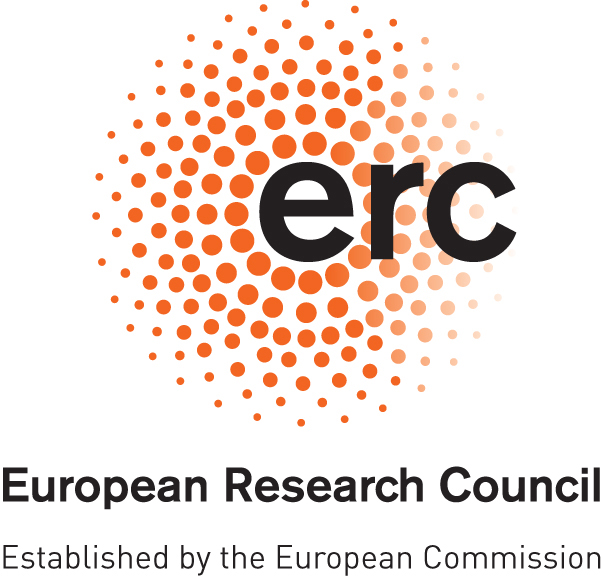}%
\end{textblock}
\begin{textblock}{4}(12.3,0.9)
\includegraphics[width=30px]{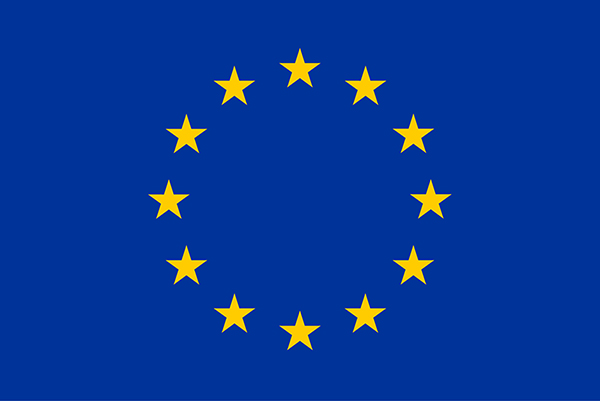}%
\end{textblock}

\thispagestyle{empty}
\newpage

\tableofcontents

\thispagestyle{empty}
\newpage

\setcounter{page}{1}

\section{Introduction}
\label{sec:intro}

Many NP-hard problems---like \textsc{Independent Set} and \textsc{Hamiltonian Cycle}---are tractable in time $2^{\Oh(\tw)} n^{\Oh(1)}$ with $\tw$ and $n$ being the treewidth and the number of vertices of the input graph, respectively (see~\cite{BodlaenderCKN15,RooijBR09,Rooij20,FlumG06,CyganNPPRW22,Pilipczuk11} for examples).
However, the space usage of these algorithms is exponential in the width of the decomposition.
In practical applications, this is often a prohibitive factor, because a machine
is likely to simply run out of space much before the elapsed time exceeds tolerable limits.
Therefore, a natural question arises: Can we reduce the space complexity of
these algorithms while keeping their time complexity?
Lokshtanov, Mnich, and Saurabh~\cite{LokshtanovMS11} 
sketched that many problems admitting a $2^{\Oh(\tw)} n^{\Oh(1)}$ time algorithm can be solved in polynomial space and time $2^{\Oh(\tw^2)} n^{\Oh(1)}$.
Unfortunately, some complexity-theoretical evidences~\cite{AllenderCLPT14,PilipczukW18} suggest that this increase in the running time might be inevitable, and problems like \textsc{Independent Set} are probably not solvable by any algorithm using $2^{\Oh(\tw)}n^{\Oh(1)}$ time and $2^{o(\tw)}n^{\Oh(1)}$ space.

This reduction of space complexity is possible when we consider \textit{treedepth}: a shallow restriction of treewidth. Formally, the treedepth of a graph~$G$ is the minimum depth of a rooted forest $F$---called an \emph{elimination forest} of $G$---in which $G$ can be embedded by bijectively mapping the vertices of $G$ to the nodes of $F$ such that the endpoints of every edge of $G$ are in an ancestor-descendant relation in $F$.
This parameter was introduced in~\cite{NesetrilO06treedepth} by Nešetřil and Ossona de Mendez; see e.g.~\cite[Chapter 6]{NestrilOdM12sparsity} for a general discussion on the notion of treedepth.
The treedepth of a graph is never smaller than its treewidth, but it is also never larger than its
treewidth times $\log n$~\cite{BodlaenderGHK95} where $n$ denotes the number of vertices in $G$.
Several NP-hard problems such as \textsc{Independent Set} or \textsc{Hamiltonian Cycle} can be solved in time $2^{\Oh(\td)}n^{\Oh(1)}$ and polynomial space (i.e., $n^{\Oh(1)}$) where $\td$ is the treedepth of a given elimination forest~\cite{FurerY17,HegerfeldK20,NederlofPSW23,PilipczukW18}.
In certain cases, the space complexity can even be as low as $\Oh(\td+\log n)$ or $\Oh(\td\log n)$~\cite{PilipczukW18}. 
Typically, the main idea is to reformulate the classic bottom-up dynamic programming approach so that it can be replaced by a simple top-down recursion. This reformulation is by no means easy---it often involves a highly non-trivial use of algebraic transforms or other tools of algebraic flavor, such as inclusion-exclusion~branching.
These algorithmic applications of treedepth lead naturally to the following question:	
\begin{quote}
    What families of problems can be solved in time $2^{\Oh(\td)}n^{\Oh(1)}$ and polynomial space, assuming that an elimination forest of depth $\td$ is provided with input?
\end{quote}
We provide a wide family of such problems 
by introducing a logic called \emph{neighborhood operator logic with acyclicity} (\logicfont{A}), \emph{connectivity} (\logicfont{C}), and \emph{clique} (\logicfont{K}) \emph{constraints} ($\NEOtwo[\finrec]\ack$ for short) and we show that the  model checking problem, for every fixed $\NEOtwo[\finrec]\ack$ formula, can be solved in time $2^{\Oh(\td)}n^{\Oh(1)}$ and polynomial space, provided an elimination forest of depth $\td$ is part of the input.
For $\NEOtwo[\finrec]\plusk$, i.e., the same logic without the acyclicity and connectivity constraints, we can further reduce the space complexity of our model checking algorithm to $\Oh(\td \log n)$.

Before proceeding to the description of our logic and a more detailed discussion about our results, let us mention here that the assumption of having an elimination forest as part of the input is an essential part of this question. 
In terms of running time, the best known algorithm for computing an optimal elimination forest is from Reidl et al.~\cite{ReidlRVS14} and it runs in time $2^{\Oh(\td^2)}n^{\Oh(1)}$.
Nadara, Pilipczuk, and Smulewicz~\cite{NadaraPS22} proved that there exists an algorithm with the same running time that uses polynomial space.
Recently, Bonnet, Neuen, and Sokolowski~\cite{abs-2507-13818} showed that the 1.0003-approximation of treedepth is NP-hard.
And it is a major open question whether one can obtain a constant-factor approximation of treedepth in time $2^{o(\td^2)}n^{\Oh(1)}$.

\paragraph{Neighborhood operator logics.} 
The logic $\NEOtwo[\finrec]$ extends fully-existential \MSOtwo with the additional ability to directly reason about the generalized neighborhoods of sets.
The central building block in our neighborhood logic is the following neighborhood operator:
\begin{definition}\label{def:neighborhood:operator}
	Given a graph $G$, $\sigma\subseteq \N$, $U\subseteq V(G)$ and $F\subseteq E(G)$, we define the \emph{$(\sigma,F)$-neighborhood} of $U$ in $G$ as
	$$N^\sigma(U,F)\defeq \bigl\{v\in V(G) \mid \abs{\{ u  \in U \mid uv\in F \}} \in \sigma\bigr\}.$$
    Moreover, we define the $\sigma$-neighborhood of $U$, denoted by $N^\sigma(U)$, as the set $N^\sigma(U,E(G))$.
\end{definition}
In this paper, we deal with neighborhood operators $N^\sigma(U,F)$ where $\sigma$ is \emph{finitely recognizable}:
\begin{definition}\label{def:prelim:finitely:recognizable}
    A subset $\sigma$ of $\bN$ is \emph{finitely recognizable} if there exist integers $d, p\in \bN$ with $p > 0$ such that for every $n\geq d$, we have 
    \begin{equation}\label{eq:def:finitely:recognizable}
        n \in \sigma \text{ if and only if } n+p\in \sigma.
    \end{equation} 
    In this case we also call it \emph{$(d, p)$-recognizable}.
    We denote by $p(\sigma)$ and $d(\sigma)$ the smallest $p$ and $d$ satisfying \eqref{eq:def:finitely:recognizable}, respectively. Moreover, we call $p(\sigma)$ the \emph{period} of $\sigma$.
\end{definition}
We use this neighborhood operator to define so-called \emph{vertex terms}.
These are built from vertex set variables, vertex set constants (i.e., unary relations) or other vertex terms by applying neighborhood operators of form $N^\sigma(\cdot)$ as well as standard set operators such as intersection ($\cap$), union ($\cup)$, subtraction ($\setminus$) or complementation (denoted by a bar on top of the term). 
The second arguments of our neighborhood operators are \emph{edge terms}; those are built from edge set variables and edge set constants (i.e., binary relations)
by applying the above-mentioned standard set operators.
Then $\NEOtwo[\finrec]$ is the extension of the fully-existential \MSOtwo by allowing the following:
\begin{itemize}
    \item \emph{Size measurement of terms}: We can write, for example, $|t| \le m$ 
		to express that a vertex term or an edge term $t$ should have size at most $m$.
    \item \emph{Comparison between terms}: We can write, for example, $t_1 = t_2$ or $t_1 \subseteq t_2$ to express that $t_1$ should be equal or contained, respectively, in $t_2$ where $t_1$ and $t_2$ are two vertex terms or two edge terms.
\end{itemize}
We define $\NEOtwo[\finrec]\plusk$ as the extension of $\NEOtwo[\finrec]$ with predicates $\clique(t)$ that verify that a vertex term $t$ induces a clique.
Furthermore, we define $\NEOtwo[\finrec]\ack$ as the extension of $\NEOtwo[\finrec]\plusk$ with predicates $\conn(t)$ and $\acy(t)$ that verify the connectivity and acyclicity, respectively, of a vertex term or an edge term $t$.
We refer to \Cref{sec:logic} for the formal definition of this logic.

Let us give some examples of $\NEOtwo[\finrec]$ formulas that capture some well-known problems.
The very simple formula
$\mathbf{independent}(X) \defeq N^{\bN^+}(X) \cap X = \emptyset$
ensures that 
none of the vertices adjacent to (some vertex in) the set $X$ belongs to $X$.
Thus, the $\NEOtwo[\finrec]$ formula $\exists X \ |X| \ge t \land \mathbf{independent}(X)$
expresses the existence of an independent set of size at least $t$.
Similarly, the formula
\[
    \exists Y \ |Y| \le t \land \mathbf{independent}(\constant V\setminus N^{\bN^+}(\constant V,Y))
\]
(where $\constant V$ stands for the whole vertex set) expresses the existence of an edge dominating set of size at most $t$, where $V\setminus N^{\bN^+}(\constant V,Y)$ is the set of vertices that are not incident to an edge of $Y$.
With $\NEOtwo[\finrec]\ack$, we can trivially express the \textsc{Clique} problem. Moreover, we can express connectivity problems, for example, for \textsc{Hamiltonian Cycle} we can use the formula
\[
    \exists Y  \ \conn(Y) \land N^{\{2\}}(\constant V,Y) = \constant V,
\]
where $N^{\{2\}}(\constant V,Y)=\constant V$ expresses the property that every vertex of the graph is incident to exactly two edges of $Y$.
For \textsc{Max Cut} to express the existence of a cut of size at least $k$, we can use the formula:
\[
    \exists Y \exists A \exists B \ A \cup B = \mathtt{V} \land A \cap B = \emptyset \land N^{\bN^+}(A, Y) \subseteq B \land N^{\bN^+}(B, Y) \subseteq A \land |Y| \geq k.
\]
Observe that finitely recognizable sets allow us to express problems based on modulo counting such as the ones defined in \cite{BelmonteS21,GreilhuberSW25} such as \textsc{Odd Dominating Set}, or more generally the existence of a $(\sigma,\rho)$-set \cite{TelleProskurowski1997} where $\sigma$ and $\rho$ are finitely recognizable: 
\[
    \exists X \  X \subseteq N^{\sigma}(X) \land \bar{X} \subseteq N^\rho(X).
\]
And even more generally, we can model a generalization of this problem called $(\sigma, \rho)$-\textsc{MinParDomSet} introduced by Greilhuber and Marx~\cite{abs-2506-01645} (they only consider it for (co)finite sets $\sigma$ and $\rho$) where one allows a set of at most $k$ vertices to violate the $(\sigma, \rho)$-constraints, i.e.,
\[
    \exists X \ \abs{(X \setminus N^\sigma(X)) \cup (\bar X \setminus N^\rho(X))} \leq k.
\]
In \Cref{ssec:logic:problems:captured}, we provide more examples of problems captured by our logics: \textit{locally checkable partitioning} problem, \textsc{CNF-SAT} and \textsc{Max-CNF-SAT} (via the incidence graph of a formula).

\paragraph{Our results.} In order to describe the running time of our model checking algorithms, we need to introduce the following notation. Given a $\NEOtwo[\finrec]\ack$ formula $\phi$, we define $\dphi$ as %
the maximum of %
$d(\sigma)$ 
over all sets $\sigma \subseteq \bN$ that occur in $\phi$, and we define the value $\pphi$ as the least common multiple of $p(\sigma)$ over the sets $\sigma \subseteq \bN$ that occur in $\phi$.
First, we show that given a $\NEOtwo[\finrec]\plusk$ formula, the model checking problem can be solved in single-exponential time and logarithmic space. 
Formally, we prove the following theorem:
\begin{restatable}{theorem}{thmMaink}
\label{thm:plus-cliques}
    There is an algorithm that, given a formula $\phi \in \NEOtwo[\finrec]\plusk$ with $\ell$ size measurements, a graph $G$, and an elimination tree $T$ of $G$ of depth at most $\td$, decides whether $G \vDash \phi$ holds in time $((\dphi+1) \cdot \pphi)^{\cO(|\phi|\cdot\td)} \cdot 2^{\cO(\pphi)} \cdot n^{\cO(\ell)}$ and space $\cO(\td \cdot |\phi|^2 \cdot (\log n + \log{|\phi|} + \log{\dphi} + \pphi))$.
\end{restatable}
When the formula contains acyclicity and connectivity constraints, we prove the existence of the following model checking algorithm:\footnote{In the preliminary version of this manuscript \cite{BergougnouxCS26}, the running time of this model checking algorithm has an extra factor of $n^{O(q)}$ where $q$ is the number of acyclicity constraints.}
\begin{restatable}{theorem}{thmMainack}
\label{thm:plus-ac}
    There is a Monte-Carlo algorithm that, given a formula $\phi \in \NEOtwo[\finrec]\ack$ with~$\ell$ size measurements, %
    a graph $G$, and an elimination tree $T$ of $G$ of depth at most $\td$, decides whether $G \vDash \phi$ holds in time $( (\dphi+1) \cdot \pphi)^{\cO(|\phi|\cdot\td)} \cdot 2^{\cO(\pphi)} \cdot n^{\cO(\ell)}$ and space polynomial in $n + \abs{\phi} + \dphi + \pphi$.
    The algorithm cannot give false positives and may give false negatives with probability at most $1/2$. 
\end{restatable}
Let us remark that to handle connectivity and acyclicity constraints we rely on the Cut\&Count technique by Cygan et al.~\cite{CyganNPPRW22}.
This technique works under the assumption that a solution (if it exists) is unique. 
The standard way of ensuring this is the classic Isolation lemma~\cite{MulmuleyVV87}: one samples the weights in a certain probabilstic way so that, with high probability, the solution becomes unique.
The application of Isolation lemma is the source of both (1) randomization and (2) the increase in space complexity as storing the sampled weights requires polynomial space.
As soon as the weights are fixed, our algorithm is deterministic and runs in logarithmic space similarly to \cref{thm:plus-cliques}.

Another aspect we would like to remark is the dependency on the input formula $\phi$ in the exponent of $n$.
As stated in both theorems, the exponents of $n$ depend on the number $\ell$ of size measurements. 
In all natural problems captured by our logic, the numbers of size measurements %
are bounded by very small constants, like 1 or 2, and therefore, the degree of $n$ is independent of the natural problem we solve.
Moreover, as will be discussed in the conclusion, we cannot get rid of the dependence on $\ell$ in the exponent of $n$ unless $\FPT = \W[1]$.

\paragraph{Algorithmic meta-theorems and related logics.}
We would like to point out that our main result reads as an \emph{algorithmic meta-theorem}, i.e., a general statement of the form ``All problems with property $P$, restricted to a class of inputs $I$, can be solved efficiently''.
The archetypal, and possibly most celebrated algorithmic meta-theorem is commonly known as \emph{Courcelle's theorem}~\cite{Courcelle90them,Courcelle97,Courcelle92} (see also \cite{BoriePT92auto,ArnborgLS91easy}) which
states that all problems expressible in \MSOtwo logic are linear time solvable
on graphs of bounded treewidth.

Designing and proving algorithmic meta-theorems is a very active area of research bringing close topics related to logic, combinatorics (especially graphs), and algorithms; see~\cite{Kreutzer11,Grohe08,GroheK09,Pilipczuk25graph,SiebertzV24advances} for related surveys. 
A remarkable feature of our meta-theorem is that it has a
quite moderate run time dependence on the length of the input formula (assuming a constant number of size measurements).
In stark contrast,
the run time of model checking algorithms for \FO and \MSO
often heavily depends on the formula length and the width parameter.
Even non-elementary tower functions $2^{2^{\cdot^{_{\cdot^{_\cdot}}}}}$, 
where the height depends on the length of the formula, 
are quite common and necessary~\cite{FrickG04}.
The efficiency of our model checking algorithm is shared with the model checking algorithm for the \textit{existential counting modal logic} (\ECML for short)---introduced by Pilipczuk~\cite{Pilipczuk11}---and the model checking algorithm for the \emph{distance neighborhood logic} (\DN for short)---introduced by Bergougnoux, Dreier and Jaffke~\cite{BergougnouxDJ23}---whose model checking problem admits a single-exponential algorithm for treewidth and mim-width, respectively (both model checking algorithms can handle connectivity and acyclicity constraints).

\medskip

In fact, $\NEOtwo[\finrec]$, \DN and \ECML are closely related, let us elaborate on that.
Our logic $\NEOtwo[\finrec]$ relies on the same neighborhood operators mechanism as \DN
introduced by Bergougnoux, Dreier and Jaffke \cite{BergougnouxDJ23}. 
Let $\NEOone[\fincofin]$ be the restriction of $\NEOtwo[\finrec]$ without edge terms and where every set $\sigma\subseteq \bN$ is finite or co-finite (i.e., they forbid modulo counting). Then, \DN is essentially equivalent to the extension of $\NEOone[\fincofin]$ with neighborhood operators $N^\sigma_r(U)$ with~$r\in \bN$ that evaluate as the $\sigma$-neighborhood of $U$ in the $r$-th power of $G$.

In \cite{BergougnouxDJ23}, the authors provide a model checking algorithm for \DN running in time $n^{\Oh(d \cdot w \cdot \abs{\phi}^2)}$ where $\abs{\phi}$ is the length of the input formula, $d$ is the maximum $d(\sigma)$ over the sets $\sigma$ that occur in $\phi$, and $w$ is the mim-width of a given decomposition. 
The running time of their model checking algorithm is also upper-bounded by efficient FPT functions
in terms of the treewidth, cliquewidth, and rankwidth of the given decomposition.\footnote{Assuming that the input formula has a constant number of size measurements (e.g., $\abs{t} \leq m$) involving a non-constant integer (e.g., $m$).}

On the other hand, \ECML is a variant of modal logic which captures a good number of problems known to be tractable in time $2^{\Oh(\tw)}n^{\Oh(1)}$ where $\tw$ is the treewidth of the input graph. 
The evaluation of an \ECML formula is tied to an active vertex that changes over time and is not explicitly represented by a variable of the formula.
To change this active vertex, \ECML has two central operators: $\square^\sigma$ and $\diamondsuit^\sigma$ where $\sigma$ is a finitely recognizable set. 
When the active vertex is $v$,
a formula $\square^\sigma \phi$ shall be read as ``the number of neighbors of $v$ that satisfy $\phi$ as an active vertex belongs to $\sigma$'',
while $\diamondsuit^{\sigma} \phi$ is equivalent to $\neg\square^\sigma\neg\phi$.

Somewhat surprisingly, the authors of \cite{BergougnouxDJ23} prove that \ECML and \DN are semantically close in the sense that the operators $\square^\sigma$ and $\diamondsuit^\sigma$ are equivalent to our neighborhood operator $N^\sigma(\cdot)$. 
In fact, we prove with similar arguments that \ECML is essentially equivalent to $\NEOtwo[\finrec]$ (see \Cref{sec:relation:other:logics}). 
Nonetheless, we introduce $\NEOtwo[\finrec]$ because the neighborhood operators are a simpler mechanism to deal with in both the design of the model checking algorithm and the expression of problems.

\paragraph{Comparing our logic to $\mathsf{CMSO}_2$ and the case of universal quantifiers.}
Let us briefly mention here how the expressive power of $\NEOtwo[\finrec]$ compares to the expressive power of $\FO$, $\MSOtwo$, and $\mathsf{CMSO_2}$, the latter being the enhancement of $\MSOtwo$ with modular counting predicates (i.e., predicates expressing that for a given set variable $X$ and a fixed integer $p>1$, $|X|\equiv0 \pmod p$).
Note that every formula $\varphi$ of $\NEOtwo[\finrec]$ can be rewritten to an equivalent formula~$\varphi'$ of $\mathsf{CMSO}_2$ but this transformation comes with a length blow-up that depends on the maximum of values $d(\sigma)$, over all $\sigma$ appearing in $\varphi$. 
Moreover, each size measurement $\abs{X}\leq m$ can also be expressed in $\mathsf{CMSO}_2$, but the length blow-up will depend on $m$ which can be prohibitively big, e.g., $\cO(n)$.
From another viewpoint, while in graph classes of bounded treedepth 
the expressive power of $\MSOtwo$ and $\FO$ coincide~\cite{ElberfeldGT16}, it is unclear to which variant of $\FO$ our logic $\NEOtwo[\finrec]$ maps to.
Also, we would like to stress that $\NEOtwo[\finrec]$ is an extension of fully-existential $\MSOtwo$ that allows limited use of universal quantifiers, ``hidden'' in the neighborhood operators, the size measurements, and the comparison between terms. The extension of the logic with acyclicity, connectivity and clique constraints, captures even more problems and all these extra features allow further use of universal quantifiers in restricted forms. 
This extension, while allowing to encode more problems and get further algorithmic results, also investigates the potential of our techniques to the regime of universal quantification.  

\paragraph{Organization of the paper.}
In \cref{sec:technical} we illustrate the main ideas behind our algorithm on a concrete example of a formula.
Then in \cref{sec:prelim} we introduce notions and present some algebraic results used in the paper.
In \cref{sec:logic} we formally define our logic,
provide an equivalent core fragment of this logic on which we focus in the remainder, and present some problems expressible in the logic.
In \cref{sec:partial} we define the partial solutions, indices, and polynomials on which our algorithm relies on.
With this in \cref{sec:recursive} we present the recursive equalities and in \cref{sec:modelcheck} we put them together to obtain our algorithm for $\NEOtwo[\finrec]$.
After that in \cref{sec:extensions} we explain how to extend the algorithm to handle clique, acyclicity, as well as connectivity constraints and briefly discuss the directed version of the logic.
In \cref{sec:relation:other:logics} we establish the connections of our logic to \DN and \ECML.
Finally, in \cref{sec:conclusions} we conclude with some open questions.

\section{Technical Overview}
\label{sec:technical}
In this section, we briefly present our logic as well as the main techniques behind our model checking algorithm.
We first explain how our algorithm works with a concrete formula, and then, we briefly discuss what changes in the general case.
So fix a graph $G$ on $n$ vertices, an elimination tree $T$ of $G$ of depth $\td$, and the following formula
\[
    \phi = \exists Y \ N^{\sigma}(\mathtt{V}, Y) = \mathtt{V} \land |Y| \leq k
\]
where $\sigma$ is the set 
\begin{equation*} %
    \sigma = \{3\} \cup \{ i \mid i\geq 4 \text{ and } i \text{ is even}\}.
\end{equation*}
By \cref{def:neighborhood:operator}, the set $N^{\sigma}(\mathtt{V}, Y)$ with $\mathtt{V} = V(G)$ is the set of all vertices $v$ in $G$ such that the number of edges in $Y$ incident to $v$ is at least $4$ and even, or it is precisely $3$.
So the property $N^{\sigma}(\mathtt{V}, Y) = \mathtt{V}$ states that for every vertex $v$ in $G$, the number of incident edges from $Y$ is at least $4$ and even, or it is precisely $3$.
Finally, $|Y| \leq k$ bounds the size of $Y$ by $k$.

According to \cref{def:prelim:finitely:recognizable} we have $d(\sigma) = 4$ and $p(\sigma) = 2$. Observe that, for any integer $i \geq 4$, whether $i$ belongs to $\sigma$ depends solely on its parity, i.e., rest modulo $2$.
The definition of $\sigma$ implies that for every vertex $v$, two values are relevant in order to know whether $v\in N^\sigma(\mathtt{V},Y)$.
First, we are interested in the number of incident edges in $Y$ \emph{counting up to $d(\sigma)$ = 4}, i.e., whether $v$ has \emph{precisely} $i$ incident edges for some value $0 \leq i \leq d(\sigma) - 1 = 3$, or it has \emph{at least} $d(\sigma) = 4$ such edges.
We will use the so-called \emph{neighborhood functions} for this: such a function maps every vertex to a value between $0$ and $d(\sigma) = 4$.
Second, it is useful to know the \emph{parity} of the number of edges in~$Y$ incident to $v$.
We will use the so-called \emph{mod-degrees} for this: such a function maps every vertex to a parity from~$\FF_{p(\sigma)} = \{0, 1\}$.

Our goal is to develop a space- and time-efficient top-down recursion along $T$ to check whether a set $Y$ with desired properties exists.
Each recursive call of our algorithm is associated with:
\begin{itemize}
    \item A node $u$ of $T$.
    \item A set of edges $\sheaf(u)$. The edge set of $G$ is partitioned into certain sets $\sheaf(w)$ among the leaves $w$ of $T$ in such a way that both end-points of every edge in $\sheaf(w)$ belong to the path from $w$ to the root (see \Cref{subsec:elimtrees} for a formal definition). Then, $\sheaf(u)$ is defined as the union of $\sheaf(w)$ over the leaves $w$ of the subtree of $T$ rooted at $u$.
    \item A set of vertices called the \textit{tail} which contains the \emph{active} vertices that may be incident with \emph{still unprocessed} edges outside $\sheaf(u)$.
    \item A set of vertices called the \textit{subtree} that represents the set of already \emph{processed} vertices.
\end{itemize}
With a node $u$, two types of recursive calls are associated, namely when $u$ is treated as processed and when it is not processed.
In the former case, the subtree of a recursive call associated with $u$ is $\subtree[u]$, namely the set of descendants of $u$ in $T$ including $u$, and it is $\subtree(u)\defeq \subtree[u] \setminus \{u\}$ in the latter case. %
Similarly, its tail is $\tail(u)\defeq \tail[u]\setminus \{u\}$, namely the set of proper ancestors of $u$ in $T$, or $\tail[u]$, respectively.

Every recursive call takes as input an \emph{index} that stores the information about the tail, and each index is \emph{compatible} with a set of \emph{partial solutions}.
We will gradually present the components of an index and a compatible partial solution in the context of a recursive call associated with a vertex $u$ of $T$, a tail $\tail$, and a subtree $\subtree$.
The first component of %
a \emph{partial solution} is a set of edges $\psg(Y) \subseteq \sheaf(u)$ corresponding to a partial interpretation of the edge set variable $Y$.

Let us explain why an index $I$ cannot simply use a function $h \colon \tail \to \{0,1,2,3,4\}$ 
to keep track of the number (up to $d(\sigma) = 4$) of edges in $\psg(Y)$ incident to each vertex of $\tail$.
If $u$ is an internal node of $T$ with children $v_1,\dots,v_t$, then $\sheaf(v_1), \dots, \sheaf(v_t)$ is a partition of $\sheaf(u)$, and for each child $v_i$ of $u$ and each $v\in\tail$, there are $d(\sigma)+1$ options for the number (up to $d(\sigma) = 4$) of edges in $\psg(Y)\cap \sheaf(v_i)$ incident to $v$.
In this case, the recursive call associated with $I$ would require to consider $(d(\sigma)+1)^{|\tail| \cdot t}$ combinations of indexes of the children of $u$.
Since  $t$ can be in $\Theta(n)$, this is prohibitively large for the running time we aim at.
One might, instead, process the children of $u$ one after another but then keeping track of all indexes of the children processed so far requires $\cO((d(\sigma)+1)^\td)$ space. %

So we have to proceed differently in order to control the number of edges in $Y$ incident to each vertex. %
For this, the second component of a partial solution is a neighborhood function $\pse \colon \tail\cup \subtree \to \{0,\dots,4\}$, and each value~$\pse(v)$ is an \emph{expectation} on the number of edges of $Y$ incident to $v$ when the whole graph is processed, and not among the edges processed so far.
For this reason, the value~$\pse(v)$ remains fixed during subsequent recursive calls.
And when the whole graph is processed, our algorithm ensures that we \emph{filtered out} (via counting tricks) all partial solutions where there exists at least one vertex  $v\in V(G)$ with $\pse(v)<d(\sigma)$ such that the expectation $\pse(v)$ is not equal to the number of edges in $\psg(Y)$ incident to $v$. 
Due to the special meaning of the value $d(\sigma) = 4$, for vertices $v$ with $\pse(v) = d(\sigma)$ the filtering only ensures that after the whole graph is processed, there are at least $\pse(v)$ edges in $\psg(Y)$ incident to $v$.
Thus each \emph{index} contains a neighborhood function $\inde \colon \tail \to \{0,\dots,4\}$ to keep track of the values $\pse(v)$ on the vertices $v\in\tail$.
Note that, in this concrete example, the domains of $\pse$ and $\inde$ could be restricted to $\{3,4\}$ because every vertex needs to be incident with exactly 3 or at least 4 edges in $Y$ but we stick to the set $\{0,\dots, 4\}$ to remain consistent with our algorithm. 

We first focus on ensuring that each $v\in\tail$ has at least $\pse(v)$ incident edges when $v$ is processed. %
For this, we use an \emph{auxiliary graph} in which every vertex $w$ of $G$ is replaced by $d(\sigma) = 4$ copies $w^1, \dots, w^4$ of itself.
Then if we choose, in a certain way, for every edge from $Y$ exactly one copy of this edge in the auxiliary graph (such choices will be parts of our \emph{partial solutions}), then the fact that $v$ has at least $\pse(v)$ incident edges is equivalent to the fact that each of the first $\pse(v)$ copies of $v$ is incident with at least one edge in this auxiliary graph.
The following idea is based on \emph{inclusion-exclusion} and generalizes a technique used in \cite{NederlofPSW23} for \textsc{Hamiltonian Cycle}.
To obtain the number of \emph{partial solutions} in which each of the first $\pse(v)$ copies of $v$ is incident with at least one edge:
\begin{itemize} 
\item We first take the number of partial solutions where all of the first $\pse(v)$ copies of $v$ are \emph{allowed} to have incident edges.
\item Then we subtract, for each $j \in \{1,\dots, \pse(v)\}$, the number of partial solutions where the $j$-th copy of $v$ is \emph{forbidden} to have incident edges. 
Observe that we subtracted the partial solutions that \emph{forbid} edges incident with, say, the first and the second copy of $v$ at least twice.
\item So now for each pair $j_1 \neq j_2 \in \{1,\dots, \pse(v)\}$, we add back the number of partial solutions where both the $j_1$-th and the $j_2$-th copy are \emph{forbidden} to have incident edges.
\item The process is repeated for every possible choice of \emph{not allowed} copies.
\end{itemize}
This motivates the second part of the \emph{index}, namely, a \emph{neighborhood function} $\allow \colon \tail \to \{0, \dots, 4\}$ which determines which copies of the vertices in $\tail$ are \emph{allowed} to have incident edges.
Finally, as already motivated above, the third component of the index is a \emph{mod-degree} $\pdeg \colon \tail \to \FF_{p(\sigma)} = \{0, 1\}$.
Unlike the \emph{expectation} $\inde$, the value $\pdeg(v)$ is the \emph{real} parity of the number of edges in $Y$ processed so far and incident to a vertex $v$. 
We will see below that an algebraic discrete Fourier transform allows us to handle mod-degrees efficiently.
Altogether, an \emph{index} for $\tail$ is a tuple $(\inde,\allow,\pdeg)$ where $\inde$ and $\allow$ are neighborhood functions on $\tail$ and $\pdeg$ is a mod-degree on $\tail$.

Now we are ready to formally define partial solutions as well as the notion of \emph{compatibility} between partial solutions and indexes.
A \emph{partial solution} for $\tail(u)$ %
is a triple $(\psg(Y), \pse, \psw)$ where 
\begin{itemize}
    \item $\psg(Y)$ is a subset of edges from $\sheaf(u)$. %
    \item $\pse$ is a neighborhood function on $\tail(u) \cup \subtree[u]$.
    \item $\psw$ is a set of edges of the auxiliary graph %
    such that for every pair $\{v, w\}$ of distinct vertices of $G$, there exists a pair $v^i w^j \in \psw$ if and only if we have $vw \in \psg(Y)$. 
    Furthermore, we require that for every pair $\{v, w\}$, there exists at most one pair $v^i w^j \in \psw$. 
    Informally speaking, $\psw$ \emph{assigns} every edge from $\psg(Y)$ to precisely one edge between the copies of its endpoints.
    \item Additionally, we require the following for every \emph{processed} vertex $v\in \subtree[u]$. 
    First, $v$ has at least~$\pse(v)$ incident edges in $\psg(Y)$---this ensures that the \emph{expectation} $\pse(v)$ is a lower bound on the number of edges incident to $v$ in $\psg(Y)$. Second, we either have $\pse(v) = 3$, or we have $\pse(v) = 4$ and $v$ has an even number of incident edges in~$\psg(Y)$, i.e., assuming that the expectation is correct, the number of edges incident to the vertex $v$ in $\psg(Y)$ belongs to $\sigma$ as desired.
\end{itemize}
Let now $I \defeq (\inde, \allow, \pdeg)$ be an index of $\tail(u)$. %
We say that a partial solution $P \defeq (\psg(Y), \pse, \psw)$ for $\tail(u)$ %
is compatible with $I$ if the following properties are satisfied:
\begin{itemize}
    \item $\pse$ and $\inde$ are equal on $\tail(u)$.
        \item For every $v\in \tail(u)$ and every edge $v^i w^j \in \psw$, we have $i \leq \allow(v)$, i.e., there are only edges incident to \emph{allowed} copies.
    \item For every vertex $v \in \tail(u)$, the number $\pdeg(v)$ is equal, modulo $p(\sigma) = 2$, to the total number of edges in $\psw$ incident to the copies of $v$. Observe that by definition of a partial solution, this is simply the parity of the number of edges in $\psg(Y)$ incident to $v$.
\end{itemize}
To keep track of the size of $\psw$, we define the value $c_u(P) = |\psw|$.
Furthermore, we define two following values: %
\begin{align*}
    a_u(P) = \sum_{\substack{v \in \subtree[u] \colon \\ \pse(v) < d(\sigma) = 4}} \pse(v) && b_u(P) = \sum_{\substack{v \in \tail(u) \cup \subtree[u] \colon \\ \pse(v) < d(\sigma) = 4}} \abs{\{v^i w^j \in \psw\}}.
\end{align*}
The definition of a partial solution ensures that every vertex in $\subtree[u]$ has at least $\pse(u)$ incident edges in $\psg(Y)$.
Also recall that the edges in $\psg(Y)$ are in bijection with the edges in $\psw$, and the bijection maps the endpoints of every edge to their copies.
Therefore, we always have $a_u(P) \leq b_u(P)$.
To explain the crux of these two values, for a moment imagine that $u$ is the root of the elimination tree $T$.
At the root-node all vertices of the graph are \emph{processed}, i.e., we have $\subtree[u] = V(G)$ and $\tail(u) = \emptyset$ in this case.
So the equality $a_u(P) = b_u(P)$ is achieved precisely if every vertex $v \in \subtree[u] = V(G)$ with $\pse(v) < d(\sigma) = 4$ has \emph{exactly}~$\pse(v)$ incident edges in $\psg(Y)$, i.e., the \emph{expectation} is \emph{correct}.
Therefore, for a general node $u$, keeping track of the values~$a_u(P)$ and $b_u(P)$ permits us to, in the very end (i.e., at the root-node), only keep the partial solutions for which $\pse$ correctly guesses the number of incident edges (except for the special value $d(\sigma) = 4$ for which we only ensure the existence of at least $d(\sigma) = 4$ incident edges) as desired.

We now define the polynomial
\[
    \cP_u(I) = \sum_P \alpha^{a_u(P)} \beta^{b_u(P)} \gamma^{c_u(P)}
\]
where the sum is over all partial solutions for $\tail(u)$ compatible with $I$.
Crucially, computing the \emph{root polynomial} $\cP_r(\emptyset)$ for the root $r$ of $T$ suffices to solve the problem.
Indeed, we have $G\models \phi$ if and only if there is a monomial $q \cdot \alpha^i \beta^i \gamma^j$ in $\cP_r(\emptyset)$ with $j \leq k$ and such that for the coefficient $q$ we have $q\neq 0$. 
Here choosing the same exponent of $\alpha$ and $\beta$ ensures, as argued above, that the \emph{expectations} made by $\pse$ reflect the correct number of incident edges in $\psg(Y)$ when it is strictly smaller than $d(\sigma)$.
And the exponent of $\gamma$ reflects the cardinality of $\psg(Y)$ so we only allow the sets of size $j \leq k$.

We can similarly define partial solutions and indexes for $\tail[u]$ by replacing $\tail(u)$ by $\tail[u]$ and $\subtree[u]$ by $\subtree(u)$, respectively, in the definitions above; we then denote the corresponding polynomials by $\cP_u[\cdot]$.

Let us now explain the recursive computation of the \emph{root polynomial}.
The following equality holds for every non-leaf node $u$
\[
    \cP_u[\inde, \allow, \pdeg] = 
    \sum_{\substack{\pdeg_1, \dots, \pdeg_t \colon \tail[u] \to \{0,1\} \colon \\ \pdeg_1 + \dots + \pdeg_t \equiv_{2} \pdeg}} \prod_{i=1}^t \cP_{v_i}
    (\inde, \allow, \pdeg_i)
\]
where $v_1, \dots, v_t$ denote the children of $u$ in the elimination tree $T$.
Vaguely speaking, this holds mainly for the following reasons:
\begin{enumerate}
    \item each of the sets $\tail(v_1), \dots, \tail(v_t)$ is equal to $\tail[u]$,
    \item the sets $\subtree[v_1], \dots, \subtree[v_t]$ partition the set $\subtree(u)$,
    \item the sets $\sheaf(v_1), \dots, \sheaf(v_t)$ partition the set $\sheaf(u)$, %
    \item to obtain the degree (modulo $2$) of a certain vertex in $\tail[u]$ in the edge set $\psg(Y)$, we can simply sum up its degrees restricted to the edge sets $\psg(Y)\cap\sheaf(v_1), \dots, \psg(Y)\cap\sheaf(v_t)$,
    \item and when multiplying the polynomials, their degrees are added so that the values $a_u[\cdot]$, $b_u[\cdot]$, and $c_u[\cdot]$ are computed correctly.
\end{enumerate}
The above formula has the problem that it requires to branch over $2^{|\tail[u]| \cdot t}$ options with $t$ possibly in $\Theta(n)$.
But we can show that for a prime number $p$ satisfying certain technical preconditions and the \emph{discrete Fourier transforms} $\cQ^p_u[\cdot]$ and $\cQ^p_{v_i}(\cdot)$ of $\cP_u[\cdot]$ and $\cP_{v_i}(\cdot)$, respectively, (see \Cref{subsec:dft-definitions} and \Cref{def:modelcheck:transform} for a formal definition) we have 
\[
    \cQ^p_u[\inde, \allow, \pdeg] \equiv_p \prod_{i=1}^t \cQ^p_{v_i}(\inde, \allow, \pdeg).
\]
For this reason, our algorithm works with discrete Fourier transforms and thus quickly combines the children of $u$ without any branching.

And to compute a polynomial of form $\cP_u(\cdot)$ from some polynomials of form $\cP_u[\cdot]$, we use the following equality:
\begin{align*}
    \cP_u(\inde, \allow, \pdeg) = &\sum_{0 \leq i \leq 3} \incexc(3,i) \cdot \cP_u[\inde[u \mapsto 3], \allow[u \mapsto i], \pdeg[u \mapsto 1]] \cdot \alpha^{3} \\
    +&\sum_{0 \leq i \leq 4} \incexc(4,i) \cdot \cP_u[\inde[u \mapsto 4], \allow[u \mapsto i], \pdeg[u \mapsto 0]].
\end{align*}
Let us sketch what all of those values are and why the formula is correct.
For an integer $i$ and a function $f \in \{\inde, \allow, \pdeg\}$, the function~$f[u \mapsto i]$ has the domain $\tail[u]$, agrees with $f$ on $\tail(u)$, and maps $u$ to $i$. %
Since the vertex $u$ moves to the \emph{processed} part, by definition of a partial solution we want to ensure that 
\begin{enumerate}
    \item either $u$ has precisely $j = 3$ incident edges in~$\psg(Y)$, in particular this number is odd---this corresponds to the first summand; 
    \item or $u$ has at least $j = 4$ incident edges in $\psg(Y)$ and the number of such edges is even---this corresponds to the second summand.
\end{enumerate}
And in both cases we need to ensure that $u$ has at least $j$ incident edges in $\psg(Y)$, i.e., each of the first $j$ copies of~$u$ is incident with at least one edge in $\psw$. 
The function $\incexc(j, i)$ is the coefficient used to apply the \emph{inclusion-exclusion} sketched earlier: to ensure that the first $j$ copies of $u$ are incident to an edge in $\psw$, we \emph{allow} edges to be incident to various subsets of the copies of $u$ and properly combine those options via the coefficients $\incexc(j, i)$.
Finally, the first summand is multiplied by~$\alpha^3$ to correctly update the value $a_u(\cdot)$, i.e., reflect that $u$ now belongs to the \emph{processed} part $\subtree[u]$ and we have $\pse(u) = 3 < 4 = d(\sigma)$.
This justifies the above equality for the polynomial $\cP_u(\bE,\allow,\pdeg)$.
From this equality and with some ``partial'' inversion of the discrete Fourier transform, we can deduce a similar equality for $\cQ_u(\bE,\allow,\pdeg)$.
In terms of the running time, the branching factor of this computation is only constant. 

The last missing piece for the recursive algorithm is the non-recursive computation of the polynomial $\cP_u[I]$ for some leaf $u$.
We omit the details here due to their technicality but we can show that for the \emph{discrete Fourier transform} $\cQ^p_u[I]$, the computation can be carried out in polynomial time by considering every edge of $\sheaf(u)$ independently and correctly initializing the exponents of $\alpha$, $\beta$, and $\gamma$ depending on whether the edge is taken into $\psg(Y)$.
Altogether, the ideas presented so far suffice to obtain a recursive algorithm with only a constant branching factor and the recursion depth bounded by the depth $\td$ of the elimination tree $T$, i.e., with the running time of $2^{\cO(\td)} n^{\cO(1)}$.
At each level of recursion, we keep track of at most $n^{\cO(1)}$ coefficients of the polynomials (this follows from the upper bounds on the degrees of the formal variables $\alpha$, $\beta$, and $\gamma$) each of length upper-bounded by $n^{\cO(1)}$ so the space complexity is polynomial.

\paragraph*{General case.}
We now briefly summarize which of the techniques from the paper were not covered by this example.
\begin{enumerate}
    \item Neighborhood operators with various sets $\sigma$ can occur in a formula $\phi$. We show that we can deal with all of them by counting the number of neighbors (1) up to $\dphi$ where $\dphi$ is defined as the maximum of all $d(\sigma)$ %
    and (2) modulo $\pphi$ where $\pphi$ is defined as the least common multiple of all $p(\sigma)$ such that $\sigma$ occurs in $\phi$.
    Further, we use \emph{private} functions $\inde$, $\pdeg$, $\allow$ for every pair $(X, Y)$ such that $N^{\sigma}(X, Y)$ occurs in the formula for some $\sigma \subseteq \bN$.
    \item As done in \cite{BergougnouxDJ23}, we define a simpler fragment $\core\NEOtwo[\finrec]\ack$ of $\NEOtwo[\finrec]\ack$ and show that it has the same expressive power. Moreover, we observe that the conversion of any $\NEOtwo[\finrec]\ack$ formula into a $\core\NEOtwo[\finrec]\ack$ formula 
    can be done with only a linear increase in the formula length, by increasing the number of size measurements by at most two, and without altering the values $\dphi$ and $\pphi$.  %
    In particular, $\core\NEOtwo[\finrec]\ack$ has no nested neighborhood operators; in fact, all of them are of the form $N^\sigma(X,Y)$ where $X$ and~$Y$ are a vertex set variable and an edge set variable, respectively.
    \item To keep track of vertex set variables we extend the index by a so-called \emph{vertex interpretation}, usually denoted by $\indf$, which stores, for every vertex of the $\tail$, to which vertex set variables it belongs. 
    \item Dealing with neighborhood operators in which the first operand is a vertex set variable creates certain \emph{asymmetries}: whether a vertex $v$ belongs to $N^{\sigma}(X, Y)$ depends only on the edges $vw \in Y$ with $w\in X$.
    To deal with these asymmetries, the auxiliary graph becomes an \emph{auxiliary digraph} and $\psw$ becomes a function mapping each pair $(X, Y)$, such that $N^{\sigma}(X, Y)$ occurs in the formula for some $\sigma \subseteq \bN$, to a subset of arcs of this digraph.
    
    \item A formula, in general, might be not in conjunctive normal form and it can also contain negations (unlike our example).
    Thus some of the vertex or edge set equalities might, in general, not hold in a solution. So we extend the index by a subset $\False$ of set equalities occurring in the formula: it then determines which equalities are \emph{allowed to be falsified}, and this value remains fixed during the recursive calls.
    We then use another \emph{inclusion-exclusion} to obtain partial solutions where \emph{precisely} the equalities from $\False$ are falsified.
    \item The formula might contain multiple size measurements, we introduce one formal variable $\gamma_i$ per measured variable $X_i$.
    \item To decrease the space complexity to logarithmic in %
    $n$, %
    we make use of the technique of \cite{PilipczukW18} based on the Chinese remainder theorem. This technique reconstructs a polynomial by evaluating it in several points modulo certain primes of length logarithmic in $n$. %
    \item To handle connectivity and acyclicity constraints, we employ Cut\&Count by Cygan et al.~\cite{CyganNPPRW22}. %
    \item To handle clique constraints, we rely on the fact that every clique is contained in some root-to-leaf path of the elimination tree.
    And then we employ fast subset cover by Björklund et al.~\cite{BjorklundHKK07} to efficiently combine multiple subtrees of children of some node.
\end{enumerate}

\section{Preliminaries}
\label{sec:prelim}

We follow the convention that $0$ belongs to $\bN$ and we use $\bN^+$ to denote the set $\bN \setminus \{0\}$.
For $i \in \bN$, the set $[i]$ is defined as $\{j \in \bN^+ \mid j \leq i\}$, in particular, it holds that $[0] = \emptyset$.
The set $[i]_0$ then denotes the set $[i] \cup \{0\}$.
We will also use the Iverson bracket notation: for Boolean predicate $p$, the value $[p]$ is equal to $1$ if $p$ is true and equal to zero otherwise.

\subsection{Graphs}
We consider finite, simple, vertex and edge colored graphs $G$ where $V(G)$ and $E(G)$ denote the vertex and edge set of $G$, respectively. We assume that each color is used either on vertices or on edges.
For a graph $G$ and color class $\constant{P}$, the set of elements in $G$ with color $\constant{P}$ is denoted by $\constant{P}(G)$.
Color classes may be empty, and a vertex or an edge may have zero, one, or multiple colors.
If the corresponding graph $G$ is clear from the context, we use $n$ to denote $|V(G)|$.
For a vertex set $A\subseteq V(G)$ we write $\comp{A}$ to denote the set $V(G)\setminus A$.
The subgraph of $G$ \emph{induced} by $A$, 
whose vertex set is $A$ and whose edge set is $\{uv \in E(G) \mid u, v \in A\}$,
is denoted by $G[A]$.
And for an edge set $F \subseteq E(G)$, the subgraph $G$ \emph{induced} by $F$, 
whose vertex set is $V(F)$ and whose edge set is $F$,
is denoted by $G[F]$.
Further, for a vertex $v \in V(G)$ we define the set $N(v, F) = \{w \in V(G) \mid vw \in F\}$.

\subsection{Elimination trees and treedepth}\label{subsec:elimtrees}

Let $G = (V, E)$ be a graph. A rooted forest $T$ is an \emph{elimination forest} of $G$ if the vertex set of $T$ coincides with $V$ and for every edge $uv \in E$, either $u$ is an ancestor of $v$, or $v$ is an ancestor of $u$ in $T$.
We refer to the vertices of $T$ as \emph{nodes}.
The \emph{depth} of $T$ is the maximum number of vertices on any root-to-leaf path in $T$.
And the \emph{treedepth} of $G$ is the minimum depth over all elimination forests of $G$.
Every node in a tree of depth at most $d$ has at most $d$ ancestors so the following well-known property holds.
\begin{observation}\label{obs:prelim:treedepth:nb:edges}
    If a graph on $n$ vertices admits an elimination forest of depth $d$, this graph has at most $d \cdot n$ edges.
\end{observation}
Along this work, we assume that the elimination forest consists of a single tree (and we then refer to it as \emph{elimination tree}) as we can choose one dedicated root in an elimination forest and attach the roots of the remaining trees as its new children: the arising rooted tree is then elimination tree and its depth is at most one larger than the depth of the forest we started with.
The elimination tree of a graph will be fixed and we will use $\td$ to denote its depth.

Let $T$ be a elimination tree of a graph $G$. Following the standard terminology (see \cite{HegerfeldK20,NederlofPSW23,PilipczukW18}) for algorithms on elimination trees, we define the following sets for every node $u$ of $T$:
\begin{align*}
    \subtree[u] & \defeq \{v\colon v \text{ is descendant of } u\}, & \qquad \subtree(u) & \defeq \subtree[u]\setminus \{u\}, \\
    \tail[u] & \defeq \{v\colon v\text{ is ancestor of } u\}, & \qquad \tail(u) & \defeq \tail[u]\setminus \{u\}, \\
 \broom[u] & \defeq \tail[u]\cup \subtree[u]. & &
\end{align*}
We remark that we follow the convention that $u$ is both an ancestor and a descendant of itself, i.e., both $u \in \tail[u]$ and $u \in \subtree[u]$ hold.
We fix an arbitrary left-to-right ordering of the children of every node. 
Similarly to \cite{NederlofPSW23} we define the following.
For a node $u$ by $\operatorname{left}(u)$ we denote the leftmost leaf descendant of $u$, that is, the leaf obtained by starting at $u$ and iteratively moving to the leftmost child of the current vertex until a leaf is reached.
For a leaf $w$ of $T$, we define 
\begin{equation}\label{eq:def-sheaf}
    \sheaf(w) \defeq \{ ab \mid ab \in E(G), a \text{ is an ancestor of } b, \operatorname{left}(b) = w\}
\end{equation}
and finally, for a node $u$ we define 
\[
    \sheaf(u) \defeq \bigcup_{w \in \subtree[u], w \text{ is a leaf of }T} \sheaf(w).
\]
See~\Cref{fig:sheaf} for an illustration.
\begin{figure}[ht]
    \centering
    \begin{tikzpicture}
        \node[black node,blue] (b0) at (0,0) {};
        \node[black node,blue] (b1) at (230:0.7) {};
        \node[black node,blue] (b2) at (230:1.4) {};
        \node[rect node,blue,label={below:$w$}] (w) at (230:2.1) {};

        \node[white node] (x) at (120:0.7) {};
        \node[white node] (z) at (110:1.7) {};
        \node[white node] (x') at ($(b0)+(-1.5,0)$) {};
        
        \node[white node] (y1) at ($(b1)+(0.8,0)$) {};
        \node[white node] (y1') at ($(b1)+(1.3,0)$) {};
        \node[rect node,magenta] (y2) at ($(b2)+(0.8,0)$) {};
        \node[label={below:$w'$}] at ($(y2)+(0.1,0.2)$) {};
        \node[white node] (y3) at ($(w)+(0.8,0)$) {};

        \draw[gray!20!white, ultra thick]
            (b0) -- (b1) -- (b2) -- (w)
            (b0) -- (x)
            (b0) -- (y1) (b0) -- (y1')
            (b1) -- (y2)
            (b2) -- (y3)
            (x) -- (x')
            (x) -- (z) ;

        \draw[black]
            (x) to [bend left=30] (y1')
            (b0) to [bend left=30] (y1)
            (b0) to [bend left=30] (y1')
            (b2) to [bend right=10] (y3)
            (y3) to [bend left=10] (b1)
            ;

        \draw[magenta, very thick]
            (x) to [bend right=10] (y2)
            (b0) to [bend right=10] (y2)
            (b1) to [bend right=10] (y2)
            ;
        
        \draw[blue, very thick]
            (z) to [bend left=30] (b0)
            (w) to [bend left=30] (b2)
            (b2) to [bend left=30] (b1)
            (b1) to [bend left=30] (b0)
            (b0) to [bend right=30] (x)
            (b1) to [bend left=30] (z)
            (b2) to [bend left=30] (x)
            (w) to [bend left=45] (b1)
            (b2) to [bend left=45] (b0)
            ;    
        
    \end{tikzpicture}
    \caption{Illustration of an elimination tree of a graph with two distinguished leaves $w$ and $w'$; the edges of the tree are depicted in light gray. The vertices $b$ with $\operatorname{left}(b)=w$ are depicted in blue. The (thick) blue edges are the edges of $\sheaf(w)$ and the (thick) magenta edges are the edges of $\sheaf(w')$.}
    \label{fig:sheaf}
\end{figure}

First, observe that $(\sheaf(w))_{w \text{ is a leaf of $T$}}$ is a partition of $E(G)$.
Second, for every non-leaf node $u$ with children $v_1, \dots, v_t$, the sets $\sheaf(v_1), \dots, \sheaf(v_t)$ partition $\sheaf(u)$ and the sets $\subtree[v_1]$, $\dots, \subtree[v_t]$ partition $\subtree(u)$.
And finally, the following holds:

\begin{observation}\label{obs:sheaf-end-vertices}
    For every node $u$ of $T$, we have $V(\sheaf(u)) \subseteq \tail[u] \cup \subtree[u]$.
\end{observation}   
The notion of partial solutions we use in our proofs is defined with respect to a triple $\tss = (\tail, \subtree, \sheaf)$---called a \emph{triad} (in $G$)---where $\tail$ and $\subtree$ are disjoint subsets of $V(G)$ and $\sheaf$ is a subset of $E(G[\tail\cup \subtree])$.
The elements of this triad are then referred to as the \emph{tail}, the \emph{subtree}, and the \emph{sheaf} of $\tss$, respectively.
We define the \emph{broom} of the triad $\tss$ as the set $\tail \cup \subtree$. 
In our model checking algorithm, we deal with the triads that are naturally induced by the elimination tree. For every node $u$ of $T$, we define the triads $\tss[u] = (\tail[u], \subtree(u), \sheaf(u))$ and $\tss(u) = (\tail(u), \subtree[u], \sheaf(u))$.
Observe that $\tss[u]$ and $\tss(u)$ have the same sheaf and the same broom which we refer to as $\broom[u]$.
Also notice that, for the root $r$ of $T$, we have $\tss(r) = (\emptyset, V(G), E(G))$.

\subsection{Discrete Fourier transforms}\label{subsec:dft-definitions}
For integers $a, b, c$ with $c > 0$ we write $a \equiv_c b$ to denote that $a$ and $b$ leave the same rest when divided by $c$, i.e., are equal modulo $c$.
Also if $a \in [c-1]_0$, we sometimes write $a = b \pmod{c}$ to denote that $a$ is the rest of $b$ modulo $c$.
We naturally extend this notation to polynomials and functions. 
For a positive integer $r$, we use $\bZ_r$ to denote the ring of integers modulo~$r$.
For a prime number $p$, we use $\FF_p$ to denote the field of integers modulo $p$, and by $\FF_p^*$ we denote the multiplicative group of this field.
If $\FF_p^*$ contains a number $\omega$ such that $\omega^r \equiv_p 1$ and $\omega^j \not\equiv_p 1$ for all $1 \leq j < r$, then $\omega$ is called the \emph{$r$-th root of unity}.
The $r$-th root of unity of $\FF_p^*$ does not necessarily exist and it is, in general, not unique.

For the remainder of this subsection we fix a finite set $D$, a positive integer $r$, and a prime number $p$ such that in $\FF^*_{p}$ the $r$-th root of unity, denoted by $\omega = \omega_{r, p}$, exists. 
If multiple such roots of unity exist, we fix one arbitrary (it will be the smallest one in our applications). %
For a mapping $h \colon (D \to \bZ_r) \to \FF_{p}$ the \emph{discrete Fourier transform} (also called \emph{DFT}) of $h$ $\DFT_{D, r, p}(h) \colon (D \to \bZ_r) \to \FF_{p}$ and the \emph{inverse discrete Fourier transform} (also called \emph{inverse DFT}) of $h$ $\DFT_{D, r, p}^{-1}(h) \colon (D \to \bZ_r) \to \FF_{p}$ are defined as
\[
    \DFT_{D, r, p}(h)(y) \equiv_p \sum_{q \colon D \to \bZ_r} \omega^{\sum_{z \in D} q(z) \cdot y(z)} \cdot h(q)
\]
and
\[
    \DFT^{-1}_{D, r, p}(h)(y) \equiv_p \frac{1}{r^{|D|}} \sum_{q \colon D \to \bZ_r}\omega^{\sum_{z \in D} -q(z) \cdot y(z)} \cdot h(q)
\]
for every $y \colon D \to \bZ_r$, here $\frac{1}{r^{|D|}}$ denotes the inverse of $r^{|D|}$ in the group $\FF^*_p$.
For $y, q \colon D \to \bZ_r$ we use the notation $q \cdot y$ for $\sum_{z \in D} q(z) \cdot y(z)$.
With these shortcuts we have 
\begin{align*}
    \DFT_{D, r, p}(h)(y) \equiv_p \sum_{q \colon D \to \bZ_r} \omega^{q \cdot y} \cdot h(q) && \DFT^{-1}_{D, r, p}(h)(y) \equiv_p \frac{1}{r^{|D|}} \sum_{q \colon D \to \bZ_r} \omega^{-q \cdot y} \cdot h(q).
\end{align*}
The following two classical results will be useful for us.
We refer, for example, to the survey by van Rooij~\cite{Rooij20} for the proofs and applications in the area of fast algorithms on tree decompositions.
First, the inverse DFT is indeed the inverse of the DFT:
\begin{lemma}[\!\!\cite{Rooij20}]\label{thm:inverse-dft}
    It holds that $\DFT^{-1}_{D, r, p} (\DFT_{D, r, p}(h)) \equiv_p h$ and $\DFT_{D, r, p} (\DFT^{-1}_{D, r, p}(h)) \equiv_p h$. 
\end{lemma}

For $y, q \colon D \to \bZ_r$ we define  $y+q \colon D \to \bZ_r$ as the pointwise sum of the functions $y$ and $q$, i.e., $(y+q)(d) \equiv_r y(d) + q(d)$ for every $d \in D$.
Further, for $h, g \colon (D \to \bZ_r) \to \FF_{p}$ we define $h \cdot g \colon (D \to \bZ_r) \to \FF_{p}$ 
as the pointwise product of $h$ and $g$, i.e., $(h \cdot g)(y) \equiv_p h(y) \cdot g(y)$ for every $y \colon D \to \bZ_r$.
Then the following holds:
\begin{lemma}[Lemma 2 in~\cite{Rooij20}\!\! ]\label{thm:fourier-convolution-result}
    Let $t \in \bN$, let $h^1, \dots, h^t \colon (D \to \bZ_r) \to \FF_{p}$, and let $y \colon D \to \bZ_r$. Then it holds that 
    \[
        \sum_{\substack{y^1, \dots, y^t \colon D \to \bZ_r \\ y^1 + \dots + y^t \equiv_r y}} \prod_{i = 1}^t h^i(y^i) \equiv_p \Bigl(\DFT^{-1}_{D, r, p} \bigl(\prod_{i=1}^t \DFT_{D, r, p}(h^i)\bigr)\Bigr)(y).
    \]
\end{lemma}

Now we extend the notions and results to the case where the codomain of $h$ is not $\FF_{p}$ but rather the set of polynomials with coefficients from $\FF_{p}$.
Let $\ell \in \bN$ and let $\gamma_1, \dots, \gamma_\ell$ be pairwise distinct formal variables.
The set $\FF_{p}[\gamma_1, \dots, \gamma_\ell]$ consists of all polynomials in formal variables $\gamma_1, \dots, \gamma_\ell$ with coefficients from $\FF_{p}$.
For a mapping $h \colon (D \to \bZ_r) \to \FF_{p}[\gamma_1, \dots, \gamma_\ell]$, we define the \emph{DFT} of $h$ $\DFT_{D,r,p}(h) \colon (D \to \bZ_r) \to \FF_{p}[\gamma_1, \dots, \gamma_\ell]$ and the \emph{inverse DFT} of $h$ $\DFT_{D,r,p}^{-1}(h) \colon (D \to \bZ_r) \to \FF_{p}[\gamma_1, \dots, \gamma_\ell]$ as
\begin{align*}
    \DFT_{D, r, p}(h)(y) \equiv_p \sum_{q \colon D \to \bZ_r} \omega^{q \cdot y} \cdot h(q) && \DFT^{-1}_{D, r, p}(h)(y) \equiv_p \frac{1}{r^{|D|}} \sum_{q \colon D \to \FF_{p}} \omega^{-q \cdot y} \cdot h(q)
\end{align*}
for all $y \colon D \to \bZ_r$.
Observe that in the case $\ell = 0$ the codomain of $h$ is simply $\FF_{p}$ so we indeed extend the above case.
Thanks to the linearity of DFT and its inverse, \Cref{thm:inverse-dft,thm:fourier-convolution-result} hold for polynomials as well. A formal proof of this claim is available in \Cref{sec:appendix:dft}.

\subsection{Inclusion-exclusion}
We will also make use of the folklore inclusion-exclusion principle (we refer to \cite{cygan2015parameterized} for a proof and applications). In particular, we will rely on its intersection version:
\begin{lemma}[Theorem 10.2 in~\cite{cygan2015parameterized}]\label{thm:inclusion-exclusion-intersection}
    Let $U$ and $Y$ be finite sets and let $A_y \subseteq U$ for every $y \in Y$. 
    Then it holds that:
    \[
        \abs{\bigcap_{y \in Y} A_y} = \sum_{X \subseteq Y} (-1)^{|X|} \cdot \abs{ U \setminus \bigcup_{x \in X}  A_x}.
    \]
\end{lemma}

\section{Neighborhood Operator Logics}
\label{sec:logic}

In the following, we introduce the syntax and the semantics of \NEO logics. 
These logics are based on neighborhood operators similarly to the \DN logic introduced in~\cite{BergougnouxDJ23}, moreover, the logic $\NEOtwo[\finrec]$ is essentially equivalent to the \ECML logic introduced in~\cite{Pilipczuk11}.
We discuss the relation between \NEO logics, \DN and \ECML in \Cref{sec:relation:other:logics}.

We define \NEO logics as an extension of fully-existential monadic second-order (\MSOtwo) logic.
Recall that \MSOtwo allows quantification over vertices, edges, sets of vertices and sets of edges
together with a vertex-edge incidence relation $\inc(u,e)$, equality relation $=$ between variables of the same type, unary vertex relations (i.e., colors)
as well as the containment relation $\in$ of vertices (resp. edges) in vertex (resp. edge) sets.
Fully-existential \MSOtwo is the restriction of \MSOtwo to existential (first-order and second-order) quantifiers only, while furthermore requiring that only quantifier-free formulas may be negated.
In the following, set variables are denoted by upper-case letters ($X,Y,Z,\dots$) while unary and binary relations (also called \emph{colors}) are denoted by typewriter letters ($\constant{P},\constant{Q},\constant{R},\dots$).
The special letters $\constant V$ (unary relation) and $\constant{E}$ (binary relation) are reserved for the set of all vertices and all edges, respectively.

\paragraph*{Syntax.}
We first define \emph{edge terms} using the following rules:
\begin{enumerate}
    \item Every edge set variable $Y$ is an edge term.
    \item Every binary relational symbol $\constant{P}$ is an edge term.
    \item $\emptyset$ and $\constant{E}$ are edge terms.
    \item If $t_1$ and $t_2$ are edge terms, then $\comp{t_1}$, $t_1 \cap t_2$, $t_1 \cup t_2$ and $t_1 \setminus t_2$ are also edge terms.
\end{enumerate}

\noindent Second, we define \emph{vertex terms} using the following rules:
\begin{enumerate}
	\item Every vertex set variable $X$ is a vertex term.
    \item Every unary relational symbol $\constant{U}$ is a vertex term.
	\item $\emptyset$ and $\mathtt V$ are vertex terms.
    \item If $t_1$ and $t_2$ are vertex terms, then $\comp{t_1}$, $t_1 \cap t_2$, $t_1 \cup t_2$ and $t_1 \setminus t_2$ are also vertex terms.
	\item If $t_1$ is a vertex term, $t_2$ is an edge term, and $\sigma \subseteq \bN$ is a set of integers, then 
    $N^\sigma(t_1)$ and $N^\sigma(t_1,t_2)$ are vertex terms. %
\end{enumerate}

\noindent Then $\NEOtwo$ is the extension of fully-existential \MSOtwo by the following two rules:
\begin{enumerate}
	\setcounter{enumi}{5}
	\item If $t$ is a vertex term or an edge term and $m \in \N$ is an integer, then $|t| = m$, $|t| \le m$, and $|t| \ge m$ are formulas called \emph{size measurements}.
	\item If $t_1$ and $t_2$ are both vertex terms or both edge terms, then $t_1 = t_2$, $t_1 \subseteq t_2$, and $t_1 \supseteq t_2$ are formulas.
\end{enumerate}

We define $\NEOone$ as the restriction of $\NEOtwo$ where we forbid variables for edges and edge terms.
For $i\in\{1,2\}$ and a family $\mathsf{F}$ of subsets of $\bN$, we use $\NEO_i[\mathsf{F}]$ to denote the restriction of $\NEO_i$ to the formulas $\phi$ such that for all neighborhood operators $N^\sigma(\cdot)$ occurring in $\phi$, we have $\sigma\in \mathsf{F}$.
In this paper, we consider two families of subsets of $\bN$, namely:
\begin{itemize}
    \item $\fincofin$ consists of all subsets of $\bN$ that are finite or co-finite.
    \item  $\finrec$ consists of all subsets of $\bN$ that are \emph{finitely recognizable} (see \Cref{def:prelim:finitely:recognizable}).
\end{itemize}
 \noindent On top of this, given a fragment of $\logicfont{L}$ of $\NEOtwo$, we define $\logicfont{L\ac}$ as the extension of $\logicfont{L}$ with the following additional rule:
\begin{enumerate}
	\setcounter{enumi}{7}
	\item If $t$ is a vertex term or an edge term, then $\conn(t)$ and $\acy(t)$ are formulas.
\end{enumerate}
Similarly, we define
$\logicfont{L\plusk}$ as the extension of $\logicfont{L}$ with the following additional rule:
\begin{enumerate}
	\setcounter{enumi}{8}
	\item If $t$ is a vertex term, then $\clique(t)$ is a formula.
\end{enumerate}
And we define $\logicfont{L\ack}$ as the extension of $\logicfont{L}$ with the rules 8. and 9.

Let $\xi$ be a formula or a term of $\NEOtwo \ack$.
We use the ordered tuple $\var(\xi)$ to denote the set of variables that occur in $\xi$.%
We denote by $\var^V(\xi)$ the set of variables for vertices and vertex sets, and we denote by $\var^E(\xi)$ the set of variables for edges and edge sets.
Moreover, we denote by $\const^V(\xi)$ and $\const^E(\xi)$ the sets of unary and binary, respectively, relational symbols that occur in $\xi$.
We define the \emph{length} $|\phi|$ of a $\NEOtwo \ack$ formula $\phi$ to be the number of symbols of $\phi$.
We emphasize that every number and every subset of integers (for example, as occurring in a size measurement $t \le m$ or in a superscript of a vertex term $N^\sigma(\cdot)$) is \emph{one} symbol.

\paragraph*{Semantics.}
Next, we define the semantics of the $\NEO$ logics. 
We consider \emph{vertex and edge colored} graphs.
This means, each vertex (resp.\ edge) of a graph may be in zero, one, or more unary (resp.\ binary) relations ($\constant{P},\constant{Q},\constant{R},\dots$). Our notion of the interpretation of a logical formula uses two functions: one is for interpreting the variables for vertices and vertex sets, and the other is for the variables for edges and edge sets. We do so to improve readability:
\begin{definition}\label{def:logic:interpretation}
    For a graph $G$ and $A\subseteq V(G)$, a vertex interpretation of a formula $\phi$ on $A$ is a function mapping each vertex variable of $\phi$ to a vertex in $A$ and each vertex set variable of $\phi$ to a subset of $A$.
    And for $F\subseteq E(G)$, an edge interpretation of $\phi$ on $F$ is a function mapping each edge variable of $\phi$ to an edge in $F$ and each edge set variable of $\phi$ to a subset of $F$.

    An \emph{interpretation} of a formula $\phi$ is a tuple $(G,\interp{f},\interp{g})$ consisting of a graph $G$, a vertex interpretation $\interp f$ of $\phi$ on $V(G)$,and an edge interpretation $\interp g$ of $\phi$ on $E(G)$.
\end{definition}

Given an interpretation $(G,\interp{f},\interp{g})$, we define the semantics of neighborhood and edge terms as follows:

\begin{enumerate}
	\item $\ip{X}^{(G,\interp{f},\interp{g})} = \interp{f}(X)$ if $X\in \var^V(\phi)$, and $\ip{X}^{(G,\interp{f},\interp{g})} = \interp{g}(X)$ if $X\in \var^E(\phi)$,
     \item $\ip{\constant{P}}^{(G,\interp{f},\interp{g})} = \constant{P}(G)$ if $\constant{P}$ is a unary or a binary relational symbol,
	\item $\ip{\emptyset}^{(G,\interp{f},\interp{g})} = \emptyset$, $\ip{\constant V}^{(G,\interp{f},\interp{g})} = V(G)$, and $\ip{\constant E}^{(G,\interp{f},\interp{g})} = E(G)$,
    \item 
    $\ip{N^\sigma(t,\star)}^{(G,\interp{f},\interp{g})} = N^\sigma(\ip{t}^{(G,\interp{f},\interp{g})}, \ip{\star}^{(G,\interp{f},\interp{g})})$, and\\
   $\ip{N^\sigma(t)}^{(G,\interp{f},\interp{g})} = N^\sigma(\ip{t}^{(G,\interp{f},\interp{g})})$, where the neighborhood operators on the right are evaluated in $G$ (see \Cref{def:neighborhood:operator}),
	\item 
	$\ip{t_1 \star t_2}^{(G,\interp{f},\interp{g})} = \ip{t_1}^{(G,\interp{f},\interp{g})} \star \ip{t_2}^{(G,\interp{f},\interp{g})}$ for $\star \in \{\cap,\cup, \setminus\}$,
    \item $\ip{\ \bar{t}\ }^{(G,\interp{f},\interp{g})} = V(G) \setminus \ip{t}^{(G,\interp{f},\interp{g})}$ if $t$ is a vertex term and $\ip{\ \bar{t}\ }^{(G,\interp{f},\interp{g})} = E(G) \setminus \ip{t}^{(G,\interp{f},\interp{g})}$ if $t$ is an edge term.
\end{enumerate}
$\NEOtwo$ inherits the semantics from \MSOtwo with the following semantics of the additional rules:
\begin{enumerate}
	\setcounter{enumi}{6}
	\item $\ip{|t| \prec m}^{(G,\interp{f},\interp{g})} = 1$ if $|\ip{t}^{(G,\interp{f},\interp{g})}| \prec m$ and $\ip{|t| \prec m}^{(G,\interp{f},\interp{g})} = 0$ otherwise for $\prec \in \{=,\le,\ge\}$,
\item $\ip{t_1 \prec t_2}^{(G,\interp{f},\interp{g})} = 1$ if $\ip{t_1}^{(G,\interp{f},\interp{g})} \prec \ip{t_1}^{(G,\interp{f},\interp{g})}$ and $\ip{t_1 \prec t_2}^{(G,\interp{f},\interp{g})} = 0$ otherwise, for $\prec \in {\{=,\subseteq,\supseteq\}}$.
\end{enumerate}
For $\NEOtwo\ack$, the semantics of the additional rules is defined as follows:
\begin{enumerate}
	\setcounter{enumi}{8}
	\item 
    \begin{itemize}
        \item $\ip{ \conn(t)}^{(G,\interp{f},\interp{g})}= 1$ if the subgraph of $G$ induced by the set $\ip{t}^{(G,\interp{f},\interp{g})}$ is connected and $\ip{ \conn(t)}^{(G,\interp{f},\interp{g})}= 0$ otherwise,
        
        \item $\ip{ \acy(t)}^{(G,\interp{f},\interp{g})}= 1$ if the subgraph of $G$ induced by the set $\ip{t}^{(G,\interp{f},\interp{g})}$ is acyclic and $\ip{ \acy(t)}^{(G,\interp{f},\interp{g})}= 0$ otherwise,
        
        \item $\ip{ \clique(t)}^{(G,\interp{f},\interp{g})}= 1$ if the subgraph of $G$ induced by the set $\ip{t}^{(G,\interp{f},\interp{g})}$ forms a clique and $\ip{ \clique(t)}^{(G,\interp{f},\interp{g})}= 0$ otherwise.
    \end{itemize}
\end{enumerate}

\subsection{Problems captured by our logics}
\label{ssec:logic:problems:captured}
Now we provide some examples of problems that are captured by $\NEOtwo[\finrec]\ack$ and not yet mentioned in the introduction including \textsc{SAT} and \textsc{Max-SAT} on CNF formulas.

As demonstrated in the introduction, we can express any $(\sigma,\rho)$-set property---as introduced by Telle and Proskurowski \cite{TelleProskurowski1997}---where $\sigma$ and $\rho$ are finitely recognizable sets.
With the same arguments as used in \cite{BergougnouxDJ23} for the \DN logic, we can easily adapt such formulas to express the more general \textit{locally checkable vertex partitioning} problems from \cite{TelleProskurowski1997}.
Such a problem is defined by a $q\times q$ matrix $D$ whose entries are finitely recognizable sets, and one seeks a partition $(X_1,\dots,X_q)$ of the vertex set $V(G)$ such that for every~$i,j\in [q]$, the inclusion $X_i\subseteq N^{D[i,j]}(X_j)$ holds.

Similarly to \textsc{Hamiltonian Cycle} mentioned in the introduction, we can express many standard connectivity and acyclicity problems.
For \textsc{Steiner Tree}, if the terminals are colored with the color $\constant{T}$, then we can express it as
\[
    \exists X \ \constant{T} \subseteq X \land \conn(X) \land |X| \leq k
\]
where $k$ is a bound on the size of the solution.
For \textsc{Feedback Vertex Set} we simply use
\[
    \exists X \ \acy(\overline{X}) \land |X| \leq k.
\]
And for \textsc{Max Leaf Spanning Tree} we use
\[
    \exists Y \exists L \ \acy(Y) \land \conn(Y) \land \constant{V} = N^{\bN^+}(\constant{V}, Y) \land L = N^{\{1\}}(\constant{V}, Y) \land |L| \geq k.
\]
Here $Y$ represents the set of edges of the sought tree, $L$ is the set of its leaves, and $k$ is the desired number of those leaves.
The first two conditions ensure that $G[Y]$ is a tree, the third one forces it to be spanning, the fourth requires the vertices in $L$ to be precisely the leaves of this tree, and the last one lower-bounds the number of leaves.

Now we show that we can express \textsc{SAT} via the \emph{signed incidence graph} of a CNF formula.
Given a CNF formula $\xi$, we define its signed incidence graph $G_\xi$ as follows.
The vertices of $G_\xi$ are the variables and the clauses of $\xi$, and the each clause is colored with the color $\constant{C}$.
The edges of $G_\xi$ are defined as follows.
For every variable $x$ and every clause $C$ of $\xi$:
\begin{itemize}
    \item If $x$ appears positively in $C$, then $xC$ is an edge of $G_\xi$ colored with $\constant P$.
    \item If $x$ appears negatively in $C$, then $xC$ is an edge of $G_\xi$ colored with $\constant N$.
\end{itemize}
Then we can express \textsc{SAT}, i.e., the existence of a model for $\xi$ with the following formula
\[
    \exists X \ X \subseteq \overline{\constant{C}} \land \constant{C} = N^{\bN^+}(X,\constant P) \cup N^{\bN^+}(\overline{X},\constant N).
\]
Basically, in this formula $X$ is the set of variables that a model of $\xi$ sets to true, and every clause needs to be adjacent to either a variable in $X$ via an edge labeled $\constant{P}$ or to a variable not in $X$ via an edge labeled $\constant{N}$.
Moreover, we can express \textsc{Max-SAT}, i.e., the existence of an interpretation of $\xi$ that satisfies at least $k$ clauses, with the following formula
\[
    \exists X \subseteq \overline{\constant{C}} \land \abs{ N^{\bN^+}(X,\constant P) \cup N^{\bN^+}(\overline{X},\constant N) } \geq k.
\]

\subsection{Core logic}
Some of the operations in our logic can be understood as ``syntactic sugar'' that does not increase the expressive power but merely reduces friction when expressing problems.
To facilitate our proofs, we consider a smaller fragment $\core\NEOtwo$ of \NEOtwo, which we call \emph{core neighborhood operator logic}, that has the same expressive power as \NEOtwo,
and a similar equivalent fragment $\core\NEOtwo\ac$ for $\NEOtwo\ac$.
Similarly, we also consider the fragments $\core\NEOtwo\plusk$ and $\core\NEOtwo\ack$. 

\paragraph*{Definition of core logics.}
We first define \emph{primitive formulas}.
\begin{enumerate}
    \item \label{item:core-vertex-color-equalities} If $X$ is a vertex set variable and $\constant{U}$ is a unary relational symbol, $\emptyset$, or $\mathtt V$, then $\constant{U}=X$ is a \emph{primitive formula}.
    \item \label{item:core-edge-color-equalities} If $Y$ is an edge set variable and $\constant{B}$ is a binary relational symbol, $\emptyset$, or $\mathtt E$, then $\constant{B}=Y$ is a \emph{primitive formula}.
    \item \label{item:core-vertex-equalities} If $X$, $Y$ and $Z$ are either all vertex set variables or all edge set variables, then $X = Y$, $\comp X = Y$, $X \cup Y = Z$, and $X \cap Y = Z$ are \emph{primitive formulas}.
    \item \label{item:core-neighborhood-equalities} If $X$ and $W$ are vertex set variables, $Y$ an edge set variable and, $\sigma\subseteq \N$, %
    then %
    $N^\sigma(X,Y) = W$ %
    is a \emph{primitive formula}. 
    \item\label{item:core-AC} If $X$ is a vertex or edge set variable and $m \in \N$, then $|X| \le m$,  $|X| \ge m$, and $|X| = m$ are  \emph{primitive formulas}.
    \item If $X$ is a vertex or edge set variable, then $\acy(X)$ and $\conn(X)$ are \emph{primitive formulas}.
    \item If $X$ is a vertex set variable, then $\clique(X)$ is a \emph{primitive formula}.
\end{enumerate}

Let \core\NEOtwo be the fragment of \NEOtwo containing all formulas of the form
$\exists X_1 \, \dots \exists X_k \, \psi$, where $\psi$ is a Boolean combination
of primitive formulas as described by the items 1.\ to 5.\ and such that $\var(\psi)=\{X_1,\dots,X_k\}$.
We define $\core\NEOtwo\ac$ (resp.\ $\core\NEOtwo\plusk$) to be the fragment of $\NEOtwo\ac$ containing all formulas of the form
$\exists X_1 \, \dots \exists X_k \, \psi$ where $\psi$ is a Boolean combination
of primitive formulas as described by the items 1.\ to 6.\ (resp.\ 1. to 5. as well as 7.) and such that $\var(\psi)=\{X_1,\dots,X_k\}$.
Finally, we define $\core\NEOtwo\ack$ to be the fragment of $\NEOtwo\ack$ containing all formulas of the form
$\exists X_1 \, \dots \exists X_k \, \psi$ where $\psi$ is a Boolean combination
of primitive formulas as described by the items 1.\ to 7.\ and such that $\var(\psi)=\{X_1,\dots,X_k\}$.
The following lemma proves that $\core\NEOtwo\ack$ is equivalent to $\NEOtwo\ack$, a similar statement for the \DN logic can be found in \cite{BergougnouxDJ23}.
We only provide a sketch of proof as most of the arguments can be found in \cite[Subsection 3.1]{BergougnouxDJ23}.

\begin{lemma}\label{obs:core}
    Given a formula $\phi \in \NEOtwo\ack$, in time $O(|\phi|^2)$ one can compute an equivalent formula $\phi' \in \core\NEOtwo\ack$ such that
    \begin{itemize}
        \item $|\phi'| \le 10 \cdot |\phi|$,
        \item the sets $\sigma\subseteq \bN$ occurring in $\phi'$ are exactly the ones occurring in $\phi$ together with possibly the set $\{1\}$, and
        \item $\phi'$ has at most two size measurements more than $\phi$.
    \end{itemize}
    Moreover, if $\phi \in \NEOtwo$ then $\phi' \in \core\NEOtwo$; and if $\phi \in \NEOtwo\plusk$ then $\phi' \in \core\NEOtwo\plusk$.
\end{lemma}

\begin{proof}
    We start by converting $\phi$ into prenex-normal form, i.e., of the form $\exists \overline X \, \xi$, where $\exists \overline X$ is a sequence of existentially quantified variables and  $\xi$ is quantifier-free. 

    Next, we get rid of every vertex and edge variable.
    Let $x_1,\dots,x_h$ be the vertex and $x_{h+1},\dots,x_r$ the edge variables of $\phi$, we transform $\exists \overline X \, \xi$ by replacing each $x_i$ by a new set variable $X_i$, and we make the following translations:
    $x_i \in Y$ translates to $X_i \subseteq Y$,
    $x_i=x_j$ translates to $X_i=X_j$, and
    $\inc(x_i,x_j)$---where $x_i$ is a vertex variable and $x_j$ is an edge variable---translates to $N^{\{1\}}(X_i, X_j) \neq \emptyset$.
    After that we replace $\exists \overline X \ \xi$ by
    \[
        \exists \overline X \ \xi \land X_1 \neq \emptyset \land \dots \land X_r \neq \emptyset \land \abs{X_1\cup \dots \cup X_t} \leq h \land \abs{X_{t+1}\cup \dots \cup X_r} \leq r - h.
    \]
    This ensures that every model of the new formula interprets the new variables $X_1,\dots,X_r$ as singletons.    

    Now we replace every $t_1 \supseteq t_2$ by $t_2 \subseteq t_1$.
    After that, replace every $t_1 \subseteq t_2$ with $t_1 \cap t_2 = t_1$.
    Then we replace every vertex term $N^\sigma(t)$ by $N^\sigma(t,\constant{E})$.
    Further, we replace every $t' \setminus t$ with $t' \cap \comp{t}$.
    The above operations conserve a formula in prenex-normal form that we denote by $\exists \overline X \, \psi$ where $\psi$ is quantifier-free.
    Then we can exhaustively apply the following simplification:
    If $\psi$ contains a term $N^{\sigma}(t, t')$, $N^{\sigma}(t', t)$, $t \cup t'$, $t' \cup t$,
    $t \cap t'$, $t' \cap t$, $\comp t$, $\conn(t)$, $\acy(t)$, $\clique(t)$, $t' = t$, or $|t| \prec m$ for $\prec \in \{\leq, =, \geq\}$ where $t$ is not a variable,
    then we replace $\psi$ by 
    $\exists Y \, t=Y \land \psi'$ where $Y$ is a new unused variable and $\psi'$ is obtained from $\psi$ by replacing $t$ with $Y$.
    For example, the formula $N^{\sigma_1}(N^\sigma(Y, \constant{B}), W) = \comp{Z}$ is replaced with 
    \[
        \exists X_1 \exists X_2 \exists X_3 \ \exists X_4 \ N^{\sigma_1}(X_1, W) = X_2 \land N^\sigma(Y, X_3) = X_1 \land X_2 = X_4 \land \constant{B} = X_3 \land \comp{Z} = X_4.
    \]
    This transformation iteratively removes nested neighborhood operators and equalities that are not primitive formulas. 
    Let $\phi'$ be the resulting formula.
    We deduce that $\phi'$ is a $\core\NEOtwo\ack$ formula.
    It can be checked that this procedure can be carried out in time quadratic in $|\phi|$, and the length $\phi'$ is at most a factor of $10$ longer than the original formula $\phi$.
    Moreover, we did not change the set of $\sigma$-s occurring in the formula (except for possibly adding $\{1\}$ when eliminating vertex and edge set variables), and we introduce at most two new size measurements (for the vertex variables and edge variables). 
    Finally, since we did not introduce new $\conn,\acy,\clique$ constraints, the last statement of the lemma is satisfied.
\end{proof}

\section{Partial Solutions and Related Notions}
\label{sec:partial}

Let $\phi$ be a $\core\NEOtwo[\finrec]$
formula, we fix it for \cref{sec:partial,sec:recursive,sec:modelcheck}.
We also fix  a vertex and edge colored graph $G$. We will show how given an elimination tree of depth $\td$ of $G$, to check whether $G\models \phi$ in time $2^{\Oh(\td)}n^{\Oh(1)}$ and space $\cO(\td\log(n))$. 
Here $\cO$ hides some factors depending on the formula $\phi$.
In \cref{thm:neo-theorem} we will provide the dependence of time and space complexity on $\phi$ explicitly.

\begin{definition}
    We define $\Cat$ 
    as the set of all pairs $(X,Y)$ with $X\in\var^V(\phi)$ and
    $Y \in \var^E(\phi)$
    such that $N^\sigma(X,Y)$ occurs in $\phi$ for some $\sigma \subseteq \bN$.\footnote{This notation is explained by the lack of imagination of the authors during the last day of a research visit. It stuck. We adore cats!} %
    
    We denote by $\EQ^V(\phi)$ the set of equalities in $\phi$ between  vertex terms and by $\EQ^E(\phi)$ the set of equalities between edge terms.
    And we use $\EQ(\phi)$ to denote the set $\EQ^V(\phi) \cup \EQ^E(\phi)$.
    For an interpretation $(G,\interp f, \interp g)$ of $\phi$, we denote by $\FalseFunc(G,\interp f,\interp g)$ the set of equalities $t_1 =t_2$ in $\EQ(\phi)$ such that $\ip{t_1 = t_2}^{(G,\interp f, \interp g)} =0$.
    
    And we let $X_1, \dots, X_\ell$ be all (vertex set and edge set) variables of $\phi$ that occur in at least one size measurement of $\phi$.
\end{definition}

\subsection{Formula parameters and the auxiliary digraph}

In this subsection, based on $G$ and $\phi$ we define two crucial integers $\dphi$ and $\pphi$ as well as the auxiliary digraph $\vec G_\phi$ relying on which we will later define our notion of partial solutions.
We start with the following simple observation:
\begin{observation}
    Let $d,d',p$, and $p'$ be integers and let $\sigma \subseteq \bN$ be a $(d,p)$-recognizable set.
    If $d \leq d'$ and $p$ divides $p'$, then $\sigma$ is also $(d',p')$-recognizable.
\end{observation}
Based on this observation, we define $\dphi$ as the maximum over all $d(\sigma)$ for the sets $\sigma \subseteq \bN$ occurring in $\phi$; and we define $\pphi$ as the least common multiple of all $p(\sigma)$ for all sets $\sigma \subseteq \bN$ occurring in $\phi$ (see \Cref{def:prelim:finitely:recognizable} for the definitions of $d(\sigma)$ and $p(\sigma)$).
To simplify our arguments, we assume that $\dphi\geq 2$. This is safe as we can always increase $\dphi$ artificially without changing the correctness of any of the following statements including the following key observation.
\begin{observation}\label{obs:modelcheck:dphi:pphi}
    Let $\sigma \subseteq \bN$ be a set that
    occurs in $\phi$. And let $U\subseteq V(G)$, $F\subseteq E(G)$, $v\in V(G)$, and $i\in [0,\pphi-1]$ be such that $\abs{N(v,F)\cap U} \geq \dphi$ and $i \equiv_{\pphi} \abs{N(v,F)\cap U}$.
    Then we have $\abs{N(v,F)\cap U} \equiv_{p(\sigma)} i + (\dphi\cdot \pphi)$.
    And therefore, $v\in N^\sigma(U,F)$ if and only if $i + (\dphi\cdot \pphi) \in \sigma$.
\end{observation}
We will later use this observation to replace the counting neighbors modulo $p(\sigma)$ by modulo $\pphi$.
Our algorithm also makes use of the following auxiliary digraph associated with $G$ and $\phi$.

\begin{definition}\label{def:modelcheck:digraph}
We define $\vec G_\phi$ as the directed graph with vertex set $V(\vec G_\phi)=\{v^1,\dots,v^{\dphi}  \mid v\in V(G) \}$ and arc set $\vec A(\vec G_\phi)=\{ (u^1, v^i) \mid uv\in E(G), i\in[\dphi]\}$ (see~\Cref{fig:graph-phi} for an illustration).
For a set $W\subseteq E(G)$ of edges, we denote by $\vec W$ the set of arcs $\{(u^1,v^i) \mid uv\in W, i\in [\dphi]\}$.

We say that a set $S \subseteq \vec A(\vec{G_{\phi}})$ of arcs is \emph{simple} if for every edge $uv\in E(G)$, there is at most one $i\in [\dphi]$ such that $(u^1,v^i)\in  S$ and at most one $j\in [\dphi]$ such that $(v^1,u^j)\in  S$.
\end{definition}
\begin{figure}[ht]
    \centering
    \begin{tikzpicture}
        \node[black node, label={left:$u$}] (u) at (0,1) {};
        \node[black node, label={left:$v$}] (v) at (0,0) {};
        \draw[] (u) -- (v);

        \draw[gray!30!white,line width=3pt,->] (2,0.5) -- (3,0.5);

        \begin{scope}[xshift=5cm]
            \node[black node, label={above:$u^1$}] (u1) at (0,1) {};
            \node[black node, label={above:$u^2$}] (u2) at (0.7,1) {};
            \node at (1.4,1) {$\dots$};
            \node[black node, label={above:$u^{d_\phi}$}] (uf) at (2.1,1) {};

            \node[black node, label={below:$v^1$}] (v1) at (0,0) {};
            \node[black node, label={below:$v^2$}] (v2) at (0.7,0) {};
            \node at (1.4,0) {$\dots$};
            \node[black node, label={below:$v^{d_\phi}$}] (vf) at (2.1,0) {};

            \draw[-stealth,thick] (v1) to [bend right=30] (u1);
            \draw[-stealth,thick] (v1) to [bend right=10] (u2);
            \draw[-stealth,thick] (v1) to [bend right=0] (uf);
            \draw[-stealth,thick] (u1) to [bend right=30] (v1);
            \draw[-stealth,thick] (u1) to [bend left=10] (v2);
            \draw[-stealth,thick] (u1) to [bend right=0] (vf);
            
        \end{scope}
    \end{tikzpicture}
    \caption{Illustration of the transformation from a graph $G$ to the graph $\vec G_\phi$. The edge $uv\in E(G)$ on the left corresponds to the arcs $(u^1,v^1),\ldots,(u^1,v^{d_\varphi})$ and $(v^1,u^1),\ldots,(v^1,u^{d_\varphi})$ on the right.}
    \label{fig:graph-phi}
\end{figure}
\subsection{Neighborhood functions, mod-degrees, and partial interpretations}

The following definition introduces two important notions---neighborhood functions  and mod-degrees---that play a key role in our algorithm. Indeed, they are important parts of our partial solutions but also of the indexes of our recursive algorithm.

\begin{definition}\label{def:modelcheck:special:functions}
    Let $\dom \subseteq V(G)$ and\footnote{Here, we use $\sheaf$ to denote an edge subset of $G$ to prepare the reader that this definition will be applied for the sheaf of some node of the given elimination tree of $G$.} $\sheaf\subseteq E(G)$.
    A mod-degree on $\dom$ is a function $\pdeg \colon \dom\times \Cat \to \Fp$. We denote by $\pdegdom(\dom)$ the set of all mod-degrees on $\dom$. 
    A \emph{neighborhood function} on $\dom$ is a function $\bE \colon \dom \times \Cat \to [\dphi]_0$.
\end{definition}

Let us give some intuitions about how we use these two kinds of functions. Roughly, mod-degrees are used to handle the periodic part of the sets $\sigma$ occurring in $\phi$ while neighborhood functions handle the part between 0 and $\dphi$.
We use mod-degree to control, for each $v\in \dom$ and $(X,Y)\in\Cat$, the number of edges in $Y$ between $v$  and $X$ modulo $\pphi$. We will apply the discrete Fourier transform to mod-degrees to accelerate the so-called join-computation combining the subtrees rooted at the children of some node of an elimination forest (see \cref{subsec:join}). 
We use neighborhood functions for two purposes:
\begin{itemize}
    \item The \emph{guessing} neighborhood functions denoted by $\bE$ or $\hat\bE$ are used to guess, for each $v\in\dom$ and  $(X,Y)\in \Cat$, the number of expected edges in $Y$ between $v$ and $X$.  
    In fact, when $\bE(v,(X,Y)) < \dphi$, we expect that there will be \textbf{exactly} $\bE(v,(X,Y))$ edges in $Y$ between $v$ and $X$.
    However, when $\bE(v,(X,Y)) = \dphi$, we expect \textbf{at least} $\dphi$ such edges---we emphasize this special meaning of the value $\dphi$ as it will be used several times in definitions and intuitions.
    In both cases, this number is an expectation and we will first include some partial solutions that do not respect this at all but then we will filter out these renegade partial solutions with inclusion-exclusion principle and counting.
    \item The \emph{restricting} neighborhood functions denoted by $\allow$ are used to forbid a partial solution to contain certain arcs in $A(\vec{G_{\phi}})$. 
    We use this kind of neighborhood functions to apply the inclusion-exclusion principle in order to kill some of the renegade partial solutions mentioned above.
\end{itemize} 
Given a guessing neighborhood function $\bE$ and a mod-degree $\pdeg$, we know the set of vertices from $\dom$ that is expected to belong in a vertex term $N^\sigma(X, Y)$ occurring in $\phi$.
Together with a vertex interpretation and an edge interpretation of $\phi$, we can interpret all the terms and equalities of $\phi$ as follows.
\begin{definition}\label{def:modelcheck:partial:interpretation}
    Given a vertex interpretation $\interp{ f }$ of $\phi$ on $\dom$, an edge interpretation $\interp g$ of $\phi$ on $\sheaf$, a neighborhood function $\bE$ on $\dom$, and a mod-degree $\pdeg$ on $\dom$, we define 
    \begin{itemize}
        \item $\ip{t}^{(\interp f, \interp{g}, \bE, \pdeg)} \defeq \ip{t}^{(G, \interp{f}, \interp{g})} \cap \sheaf$ for every edge term $t$,
        \item $\ip{t}^{(\interp f, \interp{g}, \bE, \pdeg)} \defeq \ip{t}^{(G, \interp{f}, \interp{g})} \cap \dom$ for every vertex term $t$ that is not a neighborhood operator $N^\sigma(C)$ with $C\in \Cat$, and
        \item for every $C\in \Cat$ and every $\sigma \subseteq \bN$ that occurs in $\phi$,\begin{align*}
            \ip{N^\sigma(C)}^{(\interp f, \interp{g}, \bE, \pdeg)} \defeq & \Bigl\{v\in \dom \mid \bE(v, C) < \dphi \text{ and } \bE(v,C)\in \sigma\Bigr\} \\
            \cup&\Bigl\{v\in \dom \mid \bE(v, C) = \dphi \text{ and } \pdeg(v, C) + (\dphi\cdot \pphi) \in \sigma \Bigr\}.
        \end{align*}%
     \end{itemize}
     Also, we set
     \begin{itemize} 
        \item $\FalseFunc(\interp f, \interp{g}, \bE, \pdeg)\defeq\{ t_1 = t_2 \in \EQ(\phi) \mid \ip{t_1}^{(\interp f, \interp{g}, \bE, \pdeg)} \neq \ip{t_2}^{(\interp f, \interp{g}, \bE, \pdeg)}\}$.
        \item $\FalseFunc^E(\interp f, \interp{g}, \bE, \pdeg) \defeq \FalseFunc(\interp f, \interp{g}, \bE, \pdeg) \cap \EQ^E(\phi)$ and $\FalseFunc^V(\interp f, \interp{g}, \bE, \pdeg) = \FalseFunc(\interp f, \interp{g}, \bE, \pdeg) \cap \EQ^V(\phi)$.
    \end{itemize}
\end{definition}
Let us remark that in the first two bullet points, the intersection with $\sheaf$ and $\dom$ ensures that all terms are interpreted as subsets of $\sheaf$ and $\dom$, respectively.
For instance, this intersection is crucial to ensure this property for neighborhood or edge terms of form $\overline{t}$.
Also observe that since $\phi$ is a \core\NEOtwo formula, \Cref{def:modelcheck:partial:interpretation} provides an exhaustive way of interpreting the edge terms, vertex terms, and the equalities of $\phi$.

\begin{remark}
    Observe that $\ip{t}^{(\interp f, \interp{g}, \bE, \pdeg)}$ does not depend on $\interp g$ when $t$ is a vertex term, and it is also the case for $\FalseFunc^V(\interp f, \interp{g}, \bE, \pdeg)$. Thus, we sometimes simply write $\ip{t}^{(\interp f, \bE, \pdeg)} $ and $\FalseFunc^V(\interp f, \bE, \pdeg)$. For the same reasons, we sometimes write $\ip{X}^{(\interp f)} $ and $\ip{t}^{(\interp g)} $ when $X$ is a vertex set variable and $t$ is an edge term---sometimes, we add the superscript $G$ to emphasize what the underlying graph is. 
    Moreover, we sometimes use $\FalseFunc^E(\interp{g})$ as a shortcut for $\FalseFunc^E(\interp f, \interp{g}, \bE, \pdeg)$.
\end{remark}

The rest of this subsection is dedicated to introducing some notations and lemmas that will be used later on.
Since some of the equalities in $\EQ(\phi)$ might be negated or joined by disjunction in $\phi$, a model of $\phi$ might, in general, falsify some of those equalities.
For this reason, we will be interested in the set of equalities from $\EQ(\phi)$ falsified by a certain partial solution (which will be formally defined in the next subsection).
The following definition and lemma below will later be useful to argue about the equalities falsified by a restriction of a partial solution:

\begin{definition}\label{def:logic:restriction:func}
    Let $\dom' \subseteq \dom \subseteq V(G)$. %
    Given a vertex interpretation %
    $\interp f$ on $\dom$, we define the \emph{restriction} $\interp{f}_{|_{\dom'}}$ of $\interp f$ on $\dom'$ as the vertex interpretation on $\dom'$ such that $\interp{f}_{|_{\dom'}}(X) = \interp{f}(X)\cap \dom'$ for every $X \in \var^V(\phi)$. 
    Moreover, for every neighborhood function or mod-degree $\interp h$ on $\dom$, we define the \emph{restriction} $\interp{h}_{|_{\dom'}}$ of $\interp{h}$ on $\dom'$ as the neighborhood function or mod-degree, respectively, on $\dom'$ such that $\interp{h}_{|_{\dom'}}(u,C) = \interp{h}(u,C)$ for every $u\in \dom'$ and every $C\in \Cat$.

    Similarly, let $\sheaf' \subseteq \sheaf \subseteq E(G)$.
    Given an edge interpretation $\interp g$ on $\sheaf$, we define the \emph{restriction} $\interp{g}_{|_{\sheaf'}}$ of $\interp g$ on $\sheaf'$ as the edge interpretation on $\sheaf'$ such that $\interp{g}_{|_{\sheaf'}}(X) = \interp{g}(X)\cap \sheaf'$ for every $X \in \var^E(\phi)$. 
    And for a mapping $\psw \colon \Cat \to 2^{\vec \sheaf}$, we define the \emph{restriction} $\psw_{|_{\sheaf'}} \colon \Cat \to 2^{\vec{\sheaf'}}$ on $\sheaf'$ such that $\psw_{|_{\sheaf'}}(C) = \psw(C) \cap \vec{\sheaf'}$ for all $C \in \Cat$.
\end{definition}

\begin{lemma}\label{claim:restriction-vs-false-set}
    Let $\dom' \subseteq \dom \subseteq V(G)$ and let $\sheaf' \subseteq \sheaf \subseteq E(G)$.
    Further let $\interp{f}$ be a vertex interpretation of $\phi$ on $\dom$, let $\interp g$ be an edge interpretation of $\phi$ on $\sheaf$, let $\bE$ be a neighborhood function on $\dom$, and let $\pdeg$ be a mod-degree on $\dom$.
    Then it holds that 
    \[
        \FalseFunc(\interp{f}_{|_{\dom'}}, \interp{g}_{|_{\sheaf'}}, \bE_{|_{\dom'}}, \pdeg_{|_{\dom'}}) \subseteq \FalseFunc(\interp{f}, \interp{g}, \bE, \pdeg).
    \]
\end{lemma}

\begin{proof}
It suffices to observe that, by \Cref{def:logic:interpretation}, for every vertex term $t$, we have
\[
        \ip{t}^{(\interp{f}_{|_{\dom'}}, \interp{g}_{|_{\sheaf'}}, \bE_{|_{\dom'}}, \pdeg_{|_{\dom'}})} = \ip{t}^{(\interp f, \interp{g}, \bE, \pdeg)} \cap \dom' 
\]
and that for every edge term $t$, we have 
\[
     \ip{t}^{(\interp{f}_{|_{\dom'}}, \interp{g}_{|_{\sheaf'}}, \bE_{|_{\dom'}}, \pdeg_{|_{\dom'}})} = \ip{t}^{(\interp f, \interp{g}, \bE, \pdeg)} \cap \sheaf'.
\]

\end{proof}

To argue about the equalities falsified by a union of partial solutions, we will also use the following definition and the lemma after it.

\begin{definition}\label{def:union-of-functions}
    Let $t \in \bN$, let $\dom \subseteq V(G)$, and let $\sheaf \subseteq E(G)$.
    Further, let $\dom_1, \dots, \dom_t$ be a partition of $\dom$, let $\sheaf_1, \dots, \sheaf_t$ be a partition of $\sheaf$, and for every $i \in [t]$, let $(\indf_i, \indg_i, \inde_i, \pdeg_i)$ be such that $\interp{f}_i$ is a vertex interpretation of $\phi$ on $\dom_i$, $\interp g_i$ is an edge interpretation of $\phi$ on $\sheaf_i$, $\bE_i$ be a neighborhood function on $\dom_i$, and $\pdeg_i$ be a mod-degree on $\dom_i$.
    Then we define 
    \begin{itemize}
        \item the vertex interpretation $\indf_1 \cup \dots \cup \indf_t$, also denoted by $\indf_{[t]}$, on $\dom$ via $\indf_{[t]}(X) = \bigcup_{i \in [t]} \indf_i(X)$ for all $X \in \var^V(\phi)$,
        \item the edge interpretation $\indg_1 \cup \dots \cup \indg_t$, also denoted by $\indg_{[t]}$, on $\sheaf$ via $\indg_{[t]}(Y) = \bigcup_{i \in [t]} \indg_i(Y)$ for all $Y \in \var^E(\phi)$,
        \item the neighborhood function $\inde_1 \cup \dots \cup \inde_t$, also denoted by $\inde_{[t]}$, on $\dom$ via $\inde_{[t]}(v, C) = \inde_i(v, C)$ for every $v \in \dom$ and $C \in \Cat$ where $i \in [t]$ is the unique value such that $v \in \dom_i$,
        \item the mod-degree $\pdeg_1 \cup \dots \cup \pdeg_t$, also denoted by $\pdeg_{[t]}$, on $\dom$ via $\pdeg_{[t]}(v, C) = \pdeg_i(v, C)$ for every $v \in \dom$ and $C \in \Cat$ where $i \in [t]$ is the unique value such that $v \in \dom_i$.
    \end{itemize}
\end{definition}

\begin{lemma}\label{claim:union-false-set}
    For the objects from the above definition we have
    \[
        \FalseFunc(\indf_{[t]}, \indg_{[t]}, \inde_{[t]}, \pdeg_{[t]}) = \bigcup_{i \in [t]} \FalseFunc(\indf_i, \indg_i, \inde_i, \pdeg_i).
    \]
\end{lemma}
\begin{proof}
    It suffices to observe that, by \Cref{def:modelcheck:partial:interpretation}, for every vertex term or edge term $t$, we have
    \[
            \ip{t}^{(\indf_{[t]}, \indg_{[t]}, \inde_{[t]}, \pdeg_{[t]})} = \bigcup_{i\in[t]} \ip{t}^{(\interp f_i, \interp{g}_i, \bE_i, \pdeg_i)}.
    \]

\end{proof}

\subsection{Partial solutions, indexes, and compatibility}

To simplify the definitions and statements in this subsection, we fix a triad $\tss = (\tail, \subtree, \sheaf)$ of $G$. 
To get some intuition, one should think of $\subtree$ as the set of already ``processed'' vertices, while the vertices in $\tail$ are still ``active'' and may, for example, get further neighbors.
Moreover, we denote by $\broom$ the \emph{broom} of $\tss$ defined as $\broom \defeq \tail\cup \subtree$.
Now we are ready to present the notion of a partial solution used by our algorithm.

\begin{definition}\label{def:modelcheck:partial}
    A \emph{partial solution} for the triad $\tss$ is a quadruple $(\interp{\hat f }, \interp{\hat g}, \hat\bE, \vec\cW%
    )$ such that the following conditions are satisfied:
        \begin{description}
            \item[(Domain-respecting)\label{item:partial:respect-domain}] We have:
            \begin{itemize}
                \item $\interp{\hat f }$ is a vertex interpretation of $\phi$ on $\broom$,
                \item $\interp{\hat g}$ is an edge interpretation of $\phi$ on $\sheaf$,
                \item $\hat\bE$ is a neighborhood function on $\broom$,
                \item $\vec\cW$ maps each element $(X,Y)\in \Cat$ to a \textbf{simple} subset of arcs from $\vec\sheaf$.   
            \end{itemize}
            \item[(Edge-consistent)\label{item:partial:edge:consistent}] For every $(X, Y) \in \Cat$ and every edge $uv \in \sheaf$ the following equivalence holds: there exists $j \in [\dphi]$ with $(u^1, v^j) \in \vec\cW(X, Y)$ if and only if $u \in \interp{\hat f}(X)$ and $uv \in \psg(Y)$.
            \item[(Respect-subtree-$\hat{\bE}$)\label{item:partial:respect:subtree:hatbE}] For every $(X, Y) \in \Cat$, every $v\in \subtree$ and every $j \in [\hat\bE(v,(X,Y))]$, the set $\vec \cW(X,Y)$ contains at least one arc incoming into $v^j$.
            \item[(Respect-broom-$\hat{\bE}$)\label{item:partial:respect:broom:hatbE}]
            For every $(X, Y) \in \Cat$, every $v \in \broom$ and every $j \in [\dphi]$ with $j > \hat\bE(v,(X,Y))$, the set $\vec \cW(X,Y)$ contains no arcs incoming into $v^j$.
        \end{description}
    We denote by $\Partial(\tss)$ the set of all partial solutions for $\tss$.
\end{definition}
It is not difficult to see (we will formally prove it later though) that the properties \ref{item:partial:edge:consistent} and \ref{item:partial:respect:subtree:hatbE} imply that for every vertex $v \in \subtree$ and every pair~$(X, Y) \in \Cat$, the vertex $v$ has at least $\pse(v, (X, Y))$ neighbors in $w \in \psf(X)$ such that $vw \in \psg(Y)$.
Later in the proofs of \cref{lemma:modelcheck:realizable:forward} and \cref{lemma:modelcheck:realizable:backward} we will see how how to ensure (unless the special case $\pse(v, (X, Y)) = \dphi$ occurs) that $\pse(v, (X, Y))$ is also an upper bound on the number of such neighbors, i.e., this is precisely this number.

\begin{definition}
     Given $(\interp{\hat f}, \interp{\hat g}, \hat\bE, \vec\cW)$ satisfying \ref{item:partial:respect-domain} for $\tss$, we define the mod-degree $\pdeg_{\vec \cW}$ on $\broom$ such that for every $w \in \broom$ and every $C \in \Cat$, we have:
    \[
        \pdeg_{\vec \cW}(w, C) \defeq |\{(z^1, w^i) \in \vec\cW(C)\}| \pmod{\pphi}.
    \]
\end{definition}
For a partial solution $(\interp{\hat f }, \interp{\hat g}, \hat\bE, \vec\cW)$, in particular satisfying \ref{item:partial:edge:consistent}, we will be able to use, for a vertex $v \in \subtree$ and $(X, Y) \in \Cat$, the value $\pdeg_{\psw}(v, (X, Y))$ as the number, modulo $\pphi$, of neighbors $w \in \psf(X)$ of $v$ such that $vw \in \psg(Y)$.
Altogether, assuming that $\pse(v, (X, Y))$ is %
the number of such neighbors ``counting up to $\dphi$'' (which will be ensured later), we will be able to decide whether $v \in \subtree$ belongs to $N^{\sigma}(\psf(X), \psg(Y))$ based on the values $\pse(v, (X, Y))$ and $\pdeg_{\psw}(v, (X, Y))$.

This motivates the following definition of the indexes used by our recursive algorithm. %
Given two neighborhood functions $\allow$ and $\bE$ on the same domain $\dom$, we write $\allow \leq \bE$ if for every $v\in \dom$ and $C\in\Cat$, we have $\allow(v,C)\leq \bE(v,C)$.

\begin{definition}\label{def:modelcheck:index:compatibility}
    An \emph{index} for $\tss$ is a tuple $(\interp f, \bE, \allow, \pdeg, \False)$ where:
    \begin{itemize}
        \item $\interp f$ a vertex interpretation of $\phi$ on $\tail$.
        \item $\bE$ and $\allow$ are neighborhood functions on $\tail$ such that $\allow\leq \bE$.
        \item $\pdeg$ is a mod-degree on $\tail$,
        \item $\False \subseteq \EQ(\phi)$.
    \end{itemize}
    We denote by $\Index(\tss)$ the set of all indexes for $\tss$.
    We say that a partial solution $(\interp{\hat f }, \interp{\hat g }, \hat\bE, \vec \cW)\in \Partial(\tss)$ is \emph{compatible} with an index $(\interp f, \inde, \allow, \pdeg, \False)\in\Index(\tss)$ if the following conditions are satisfied:
    \begin{description}
        \item[($\interp f$-consistent)\label{item:partial:fconsistent}] For every $X\in \var^V(\phi)$, we have $\interp f(X) = \interp{\hat f}(X) \cap \tail$.
        \item[($\bE$-consistent)\label{item:partial:econsistent}] For every $v \in \tail$ and every $C \in \Cat$, we have $\bE(v, C) = \hat\bE(v, C)$.
        \item[($\pdeg$-consistent)\label{item:partial:pdegconsistent}]For every $v\in \tail$ and every $C\in \Cat$, we have $\pdeg(v,C)=\pdeg_{\vec \cW}(v,C)$.
        \item[(Respect-$\allow$)\label{item:partial:respect:allow}] For every $v\in \tail$, every $C\in \Cat$, and every arc $(u^1,v^j) \in \vec\cW(C)$, we have $j\leq \allow(v,C)$.
        \item[(Respect-$\False$)\label{item:partial:respect:False}]
        We have $\FalseFunc(\interp{ \hat f}_{|_{\subtree}}, \interp{\hat g }, \hat \bE_{|_{\subtree}}, (\pdeg_{\vec \cW})_{|_{\subtree}}) \subseteq \False$.
    \end{description}
    For every index $I\in \Index(\tss)$, we denote the set of all partial solutions in $\Partial(\tss)$ compatible with $I$ by $\Partial(\tss,I)$.
    Moreover, we define $\Partial^=(\tss, I)$ as the subset of $\Partial(\tss,I)$ that contains all the partial solutions $(\interp{\hat f }, \interp{\hat g }, \hat\bE, \vec \cW)$ such that 
    $\FalseFunc(\interp{ \hat f}_{|_{\subtree}}, \interp{\hat g }, \hat \bE_{|_{\subtree}}, (\pdeg_{\vec \cW})_{|_{\subtree}}) = \False$. 
\end{definition}
The neighborhood function $\allow$ will be used in the so-called forget-computation (see \cref{subsec:forget}) to ensure that when a vertex is moved to the ``processed'' set $\subtree$, it satisfies \ref{item:partial:respect:subtree:hatbE}, i.e., certain copies of this vertex in $\vec G_{\phi}$ have incoming arcs.
This, in turn, ensures that the function $\pse$ indeed provides lower bounds on the cardinalities of certain relevant neighborhoods of this vertex. 
The inclusion-exclusion (see \cref{thm:inclusion-exclusion-intersection}) allows to replace the task of ensuring the existence of incoming arcs by the task of forbidding such arcs. The neighborhood function $\allow$ is then used to ensure that some copies of the vertex are forbidden to have incoming arcs.
We will provide all technical details in \cref{subsec:forget}.

Since $\allow$ and $\False$ are involved only in \ref{item:partial:respect:allow} and \ref{item:partial:respect:False}, respectively, we deduce the following inclusions.

\begin{observation}\label{obs:modelcheck:inclusions}
    Let $(\interp f, \bE, \allow, \pdeg, \False)\in\Index(\tss)$.
    For every neighborhood function $\allow'$ on $\tail$ such that $\allow' \leq \allow$, we have:
    \[ \Partial(\tss, (\interp f, \bE, \allow', \pdeg, \False) ) \subseteq \Partial(\tss, (\interp f, \bE, \allow, \pdeg, \False)).\]
    Moreover, for every $\False'\subseteq \False$, we have:
    \[ \Partial(\tss, (\interp f, \bE, \allow, \pdeg, \False') ) \subseteq \Partial(\tss, (\interp f, \bE, \allow, \pdeg, \False)).\]
\end{observation}
In particular, this emphasizes that a partial solution of $\tss$ may, in general, be compatible with multiple indices.
On the one hand, this notion of compatibility will allow us to efficiently combine partial solutions ``corresponding to the children'' of a node of the elimination tree in \cref{subsec:join}---we formalize partial solutions of a node later.
On the other hand, we will still be able to employ inclusion-exclusion to avoid double-counting and furthermore, discard the partial solutions which violate certain properties desired for a model of $\phi$.

We further stratify the sets $\Partial(\tss,I)$ as follows.
The values $c_1, \dots, c_\ell$ simply keep track of the cardinalities of sets involved in size measurements in $\phi$, while $a$ and $b$ will later be used to ensure that~$\pse$ provides correct estimates on the cardinalities of certain neighborhoods.

\begin{definition}\label{def:modelcheck:constants}
    Given a partial solution $P=(\interp{\hat f },\interp{\hat g }, \hat\bE, \vec \cW)\in \Partial(\tss)$, we define the following constants:
        \begin{align}\label{eq:def-a-b}
        a_\tss(P)\defeq \sum_{C\in \Cat}\sum_{\substack{ v\in \subtree \colon \\ \hat \bE(v,C) < \dphi}} \hat\bE(v,C), &&
        b_\tss(P)\defeq 
        \sum_{C\in \Cat}  
                \sum_{\substack{ v\in \broom \colon \\ \hat\bE(v,C) < \dphi }} 
                    \abs{\{(u^1, v^i) \in \vec\cW(C)\}},
        \end{align}
    and for every $i\in [\ell]$,
    \begin{equation}\label{eq:def-of-c}
        c^i_\tss(P) \defeq \abs{\ip{X_i}^{(G,\interp{\hat f}_{|\subtree},\interp{\hat g})}}.
    \end{equation}
    Moreover, we use $\bar{c}_\tss(P)$ to denote the tuple $(c^1_\tss(P), \dots, c^\ell_\tss(P))$.
    For an index $I\in\Index(\tss)$, $a,b\in \bN$, and $\barc\in \bN^\ell$, we define $\Partial(\tss,I,a,b,\barc)$ as the set of all partial solutions $P\in \Partial(\tss,I)$ such that $a_\tss(P)=a$, $b_\tss(P)=b$, and $\barc_\tss(P)=\barc$.
    We analogously define $\Partial^=(\tss,I,a,b,\barc)$ from $\Partial^=(\tss,I)$.
\end{definition}

The following lemma will be useful to bound the range of the values we are interested in.
One should think of $\td$ as the depth of an elimination tree of $G$ which we will fix soon.
\begin{lemma}\label{lem:modelcheck:upperbounds:abc}
    If the graph $G$ has at most $n \cdot \td$ edges for some value $\td$, 
    then for every partial solution $P\in \Partial(\tss)$, the value $a_\tss(P)$ is at most $|\phi| \cdot \dphi \cdot n$ and the value $b_\tss(P)$ is at most $2 \cdot |\phi| \cdot n \cdot \td$.
    Further, for every $i\in [\ell]$, we have $c^i_\tss(P) \leq n \cdot \td$.
\end{lemma}
\begin{proof}
    The size of $\Cat$ is upper-bounded by $|\phi|$, the sizes of $\subtree$ and $\broom$ are upper-bounded by $n$, and the size of $\sheaf$ is upper-bounded by $n \cdot \td$.
    From this the bound on $a_\tss(P)$ follows immediately.
    For every $C \in \Cat$, every edge $uv$ of $G$ contributes at most two arcs, namely one of form $(u^1, v^i)$ and the other of form $(v^1, u^j)$ (for some integers $i, j$), to the simple set $\psw(C)$ so the bound on $b_\tss(P)$ follows as well.
    And for every $i \in [\ell]$, the value $c^i_\tss(P)$ is the size of a set of vertices or edges, and thus it is upper-bounded by $n \cdot \td$.
\end{proof}

\subsection{Relations between partial solutions and the formula satisfiability}

\def\root{\mathsf{root}}

Let $\root$ denote the triad defined as $\root\defeq (\emptyset,V(G),E(G))$. Informally speaking, this corresponds to the case where all vertices and edges are ``processed'' already.
In this subsection, we will see how we can determine whether $G\models \phi$ holds from the cardinality $|\Partial^=(\root,I,a,b,\barc)|$ for some values of $I\in \Index(\root)$, $a,b\in \bN$, and $\barc\in \bN^\ell$.

 Given $\False\subseteq \EQ(\phi)$ and $\barc=(c^1,\dots,c^\ell)\in\bN^{\ell}$, we denote by $\ip{\phi}^{(\False,\barc)}$ the logical value of the formula obtained from $\phi$ by (1) replacing every subformula of form $t_1=t_2\in \EQ(\phi)$ by 0 if it belongs to $\False$, and by~1 otherwise; and (2) replacing every subformula of form $\abs{X_i} \prec m$ with $\prec\in\{\leq,\geq,=\}$ by 1 if $c^i \prec m$, and by 0 otherwise.

This subsection is dedicated to prove the following characterization for $G\models \phi$ based on our notion of partial solutions. 
For simplicity, if the domain of some function is empty, we denote this function by $\emptyset$.
Similarly, we denote the unique vertex interpretation on the empty set by $\emptyset$. 

\begin{lemma}\label{lem:modelcheck:equivalence:G:models:phi}
    We have $G\models \phi$ if and only if there exist $\False\subseteq \EQ(\phi)$, $i \in \bN$, and $\barc\in\bN^{\ell}$ such that $\ip{\phi}^{(\False,\barc)}=1$ and $\Partial^=(\root, (\emptyset,\emptyset,\emptyset,\emptyset,\False), i,i, \barc)\neq \emptyset$.
\end{lemma}
Observe that in this lemma we restrict ourselves to partial solutions $P$ satisfying $a_{\root}(P) = b_{\root}(P) = i$ for some $i \in \bN$ (cf.\ \cref{def:modelcheck:constants}).
As we will see in the proof of this lemma, this restriction is a crucial trick to ensure that the neighborhood function $\pse$ used in $P$ provides correct estimates on the cardinalities of relevant neighborhoods.
In order to prove~\Cref{lem:modelcheck:equivalence:G:models:phi},
we first give an additional definition.

\begin{definition}\label{def:model:checking:realize}
    Given $\False\subseteq \EQ(\phi)$ and $\barc=(c^1,\dots,c^\ell)\in\bN^{\ell}$, we say that an interpretation $(G, \interp{\hat f}, \interp{\hat g})$ of $\phi$ \textit{realizes} $(\False,\barc)$ if the following conditions are satisfied:
    \begin{description}
        \item[(Realize-$\FalseFunc$)\label{item:realize:False}] $\FalseFunc(G,\psf, \psg) = \False$,
        \item[(Realize-$\barc$)\label{item:realize:c}] $\ip{ \abs{X_i} = c^i }^{(G,\interp{\hat f}, \interp{\hat g})}=1$ for every $i\in [\ell]$.
    \end{description}
\end{definition}
Since in the above definition $\False$ is the set of equalities falsified by the interpretation $(G, \interp{\hat f}, \interp{\hat g})$, and~$\barc$ determines the cardinalities of all sets involved in size measurements, the following observation immediately follows from the semantics of the logic:

\begin{observation}\label{obs:model-realization}
    We have $G \vDash \phi$ if and only if there exist $\False \subseteq \EQ(\phi)$ and $\barc \in \bN^\ell$ with $\ip{\phi}^{(\False, \barc)}=1$ such that there exists an interpretation $(G, \psf, \psg)$ of $\phi$ realizing $(\False, \barc)$.
\end{observation}

With this observation in hand, we can split the proof of the equivalence from \cref{lem:modelcheck:equivalence:G:models:phi} into two lemmas below.
\begin{lemma}\label{lemma:modelcheck:realizable:forward}
    Let $I=(\emptyset,\emptyset,\emptyset,\emptyset,\False)\in\Index(\root)$ and $\barc\in\bN^{\ell}$.
    If an interpretation $(G, \interp{\hat f}, \interp{\hat g})$ of $\phi$ realizes $(\False, \barc)$, then there exist functions $\hat \bE$ and $\vec \cW$ with the following properties.
    Let $P \defeq (\interp{\hat f}, \interp{\hat g}, \hat \bE, \vec \cW)$.
    Then we have $P \in \Partial^=(\root, I, a_\root(P),b_\root(P),\barc)$ and $a_\root(P) = b_\root(P)$.
\end{lemma}
\begin{proof}
    Let $(G, \interp{\hat f}, \interp{\hat g})$ be an interpretation of $\phi$ that realizes the pair $(\False,\barc)$. Also, let $c^1,\ldots,c^\ell\in\mathbb{N}$ be such that $\barc=(c^1,\dots,c^\ell)$.
    In the following, we construct a partial solution $P=(\interp {\hat f}, \interp{\hat g}, \hat \bE, \vec\cW)$ that belongs to $\Partial^=(\root, I,a_\root(P),b_\root(P),\barc)$ with $a_\root(P)=b_\root(P)$.
    We construct $\hat\bE$ and $\vec\cW $ by doing the following for every $v\in V(G)$ and $C=(X,Y)\in\Cat$.
    Let $U(v,C)$ be the set of all the vertices $u$ in $\interp {\hat f}(X)$ such that $uv$ is an edge in $\psg(Y)$. %
    We set $\hat\bE(v,C)= \min (\dphi, \abs{U(v,C)})$.
    Then for each $u\in U(v,C)$, we add exactly one arc of form $(u^1,v^j)$ to $\vec\cW(C)$ in such a way that the following properties are satisfied after processing the whole set $U(v,C)$: 
    \begin{itemize}
        \item for every $j \in [\hat\bE(v,C)]$, there is at least one arc of form $(w^1,v^j)$ in $\vec\cW(C)$, and
        \item there is no such arc in $\vec\cW(C)$ for any $j > \hat\bE(v,C)$.
    \end{itemize} 
    Recall that $\dphi \neq 0$ by assumption so this construction is indeed possible.
    Since $\pdeg_{\vec \cW}(v, C) = |\{(u^1, v^i) \in \vec\cW(C)\}| \bmod \pphi$, we deduce from \Cref{obs:modelcheck:dphi:pphi} and the construction of $\hat\bE$ and $\vec\cW$ the following.
    \begin{observation}\label{obs:modelcheck:root}
        For every vertex term $N^\sigma(X,Y)$ occurring in $\phi$ and $v\in V(G)$, we have $v\in  \ip{N^\sigma(X,Y)}^{(G,\interp{\hat f}, \interp{\hat g})}$ if and only if either:
        \begin{itemize}
            \item $\bE(v,(X,Y)) < \dphi$ and $\bE(v,(X,Y))\in \sigma$, or
            \item $\bE(v,(X,Y)) = \dphi$ and $\pdeg_{\vec \cW}(v, C)  + (\dphi\cdot\dphi) \in \sigma$.
        \end{itemize}
    \end{observation}
    Now we prove that $P\in \Partial(\root)$, that is, $P$ satisfies all the conditions of \Cref{def:modelcheck:partial}.
    By definition, $\interp{\hat f}$ is a vertex interpretation on $V(G)$ and $\interp{\hat g}$ an edge interpretation on $E(G)$.
    By construction, $\hat \bE$ is a neighborhood function on $V(G)$ and $\vec\cW$ maps each $(X,Y)\in\Cat$ to a simple subset of arcs from $\vec{E(G)}$ so \ref{item:partial:respect-domain} is satisfied.
    Now, for each $C=(X,Y)\in\Cat$, observe that:
    \begin{itemize}
        \item By construction, for every $v\in V(G)$ and $j\in [\hat\bE(v,(X,Y))]$, there is at least one arc in $\vec\cW(X,Y)$ incoming into $v^j$; and for every $j\in[d_\phi]$ with $j > \hat\bE(v,(X,Y))$, there is no arc in $\vec\cW(X,Y)$ incoming into $v^j$. Hence, \ref{item:partial:respect:subtree:hatbE} and \ref{item:partial:respect:broom:hatbE} are satisfied.

        \item For every $uv\in E(G)$ and every $C = (X,Y) \in \Cat$, there exists $j\in [\dphi]$ such that $(u^1,v^j)\in \vec\cW(C)$ if and only if $u\in U(v,C)$. Since $U(v,C)$ is the set of all vertices $u\in \interp{\hat f}(X)$ such that 
        $uv \in \psg(Y)$,
        it follows that \ref{item:partial:edge:consistent} is satisfied.
    \end{itemize}
    
    It follows that $P\in \Partial(\root)$.
    We need to prove that $P$ belongs to $\Partial(\root,I)$.
    Since the tail of $\root$ is empty, \ref{item:partial:fconsistent}, \ref{item:partial:econsistent}, \ref{item:partial:pdegconsistent} and \ref{item:partial:respect:allow} are trivially satisfied by $P$.
    We need now to prove that \ref{item:partial:respect:False} is satisfied.
    From \Cref{obs:modelcheck:root} and \Cref{def:modelcheck:partial:interpretation} we deduce that $\ip{N^\sigma(X,Y)}^{(G,\interp{\hat f}, \interp{\hat g})} = \ip{N^\sigma(X,Y)}^{(\interp {\hat f}, \interp{\hat g}, \hat\bE, \pdeg_{\vec\cW})}$ for every neighborhood operator $N^\sigma(X,Y)$ occurring in $\phi$.
    And for any vertex term $t$ that is not a neighborhood operator, we have $\ip{t}^{(G,\interp{\hat f}, \interp{\hat g})} = \ip{t}^{(\interp {\hat f}, \interp{\hat g}, \hat\bE, \pdeg_{\vec\cW})}$ simply by definition, and the same holds for every edge term $t$.
    Consequently, we have $\FalseFunc(G,\interp{\hat f}, \interp{\hat g}) = \FalseFunc(\interp {\hat f}, \interp{\hat g}, \hat\bE, \pdeg_{\vec\cW})$.
    From \ref{item:realize:False}, it follows that 
    \begin{equation}\label{eq:modelcheck:realize}
        \FalseFunc(\interp {\hat f}, \interp{\hat g}, \hat\bE, \pdeg_{\vec\cW}) = \False,
    \end{equation}
    and thus $P$ satisfies \ref{item:partial:respect:False}.
    This proves that $P$ belongs to $\Partial(\root,I)$. 
    Thanks to \eqref{eq:modelcheck:realize}, we also know that $P$ satisfies the equality requirement to be in $\Partial^=(\root,I)$.    
    Since $(G,\interp{\hat f}, \interp{\hat g})$ satisfies \ref{item:realize:c}, for each $i\in[\ell]$, we have that
    $\ip{ \abs{X_i} = c^i }^{(G,\interp{\hat f}, \interp{\hat g})}=1$, so consequently, we have $c_\root^i(P)=c^i$.
    It follows that $P$ belongs to $\Partial^=(\root,I,a_\root(P), b_\root(P), \barc)$.
    
    It remains to prove $a_\root(P)=b_\root(P)$.
    For every vertex $v\in V(G)$ and every $C\in\Cat$ with $\hat\bE(v,C) < \dphi$, we have $\hat\bE(v,C) = \abs{U(v,C)}$ and  by construction, there are exactly $\abs{U(v,C)}$ arcs in $\vec\cW(C)$ incoming into the vertices $v^1,\dots,v^{\dphi}$. We conclude that 
    \begin{align*}
        a_\root(P) = \sum_{C\in \Cat}\sum_{\substack{ v\in V(G) \\ \hat \bE(v,C) < \dphi}} \hat\bE(v,(C)) = 
        \sum_{C\in \Cat}\sum_{\substack{ v\in V(G) \\ \hat \bE(v,C) < \dphi}} \abs{ \{ (u^1,v^i) \in \vec\cW(C)\}} = b_\root(P).
    \end{align*}
\end{proof}

We now show that the converse of the above lemma holds as well, namely:
\begin{lemma}\label{lemma:modelcheck:realizable:backward}
    Let $I=(\emptyset,\emptyset,\emptyset,\emptyset,\False)\in\Index(\root)$ and let $\barc\in\bN^{\ell}$.
    Then, for every $P = (\interp{ \hat f}, \interp{\hat g}, \hat \bE, \vec \cW)\in \Partial^=(\tss(\root),I,a_\root(P),b_\root(P),\barc)$ with $a_\root(P)=b_\root(P)$, the interpretation $(G, \interp{\hat f}, \interp{\hat g})$ of $\phi$ realizes $(\False, \barc)$.
\end{lemma}
\begin{proof}
    Let $c^1,\ldots,c^\ell\in\mathbb{N}$ such that $\barc=(c^1,\dots,c^\ell)$. 
    Consider a partial solution $P=(\interp{\hat f}, \interp{\hat g}, \hat\bE, \vec\cW)$ in $\Partial^=(\root,I,a_\root(P),b_\root(P),\barc)$ with $a_\root(P)=b_\root(P)$. 
    For every $i\in [\ell]$, we have $c^i_\tss(P) = \abs{\ip{X_i}^{(G,\interp{\hat f}_{|_{V(G)}},\interp{\hat g})}}$.
    As $\interp{\hat f}_{|_{V(G)}} = \interp{\hat f}$, it follows that $(G,\interp{\hat f}, \interp{\hat g})$ satisfies \ref{item:realize:c}.
    
    It remains to prove that $(G,\interp{\hat f}, \interp{\hat g})$ satisfies \ref{item:realize:False}, i.e., $\FalseFunc(G,\interp{\hat f}, \interp{\hat g}) = \False$.
    As $P\in \Partial^=(\root,I)$, we have $\FalseFunc(\interp{ \hat f}, \interp{\hat g }, \hat \bE, \pdeg_{\vec \cW}) = \False$, thus it is enough to prove that $\FalseFunc(G,\interp{\hat f}, \interp{\hat g}) = \FalseFunc(\interp{ \hat f}, \interp{\hat g }, \hat \bE, \pdeg_{\vec \cW})$.
    To prove this, we now check that $\ip{t}^{(G,\interp{ \hat f}, \interp{ \hat g})} = \ip{t}^{(\interp{ \hat f}, \interp{\hat g }, \hat \bE, \pdeg_{\vec \cW})} $ holds for every neighborhood or edge term $t$ occurring in $\phi$. 
    By \Cref{def:modelcheck:partial:interpretation} this is trivially satisfied when $t$ is not a neighborhood operator.
    So we remain with the case that
    $t =  N^\sigma(X,Y)$ for some set $\sigma$ occurring in $\phi$ and some $C = (X, Y) \in \Cat$, for this we need the following claim.
    
    \begin{claim}\label{claim:modelcheck:realizable}
        For each $v\in V(G)$ and each $C = (X, Y) \in\Cat$, the number of edges 
        $uv \in \psg(Y)$
        with $u\in \interp{\hat f}(X)$ is:
        \begin{itemize}
            \item exactly $\abs{\{(u^1, v^i) \in \vec\cW(C)\} }$,
            \item exactly $\hat\bE(v,C)$ when $\hat\bE(v,C) < \dphi$, and
            \item at least $\hat\bE(v,C)$ when $\hat\bE(v,C) = \dphi$.
        \end{itemize}
    \end{claim}
    \begin{claimproof}
        Let $v\in V(G)$ and $C = (X,Y)\in\Cat$.
        By \ref{item:partial:edge:consistent}, for every edge $uv\in E(G)$, there exists an arc of form $(u^1,v^j)$ in $\vec \cW(C)$ if and only if $u\in \interp{\hat f}(X)$ and $uv \in \psg(Y)$ hold. %
        As $\psw(X)$ is simple, we deduce that $\abs{\{(u^1, v^j) \in \vec\cW(C) \}}$ is exactly the number of edges $uv \in \psg(Y)$ %
        with $u\in \interp{\hat f}(X)$ proving the first bullet point.

        Thanks to \ref{item:partial:respect:subtree:hatbE}, we know that for each vertex $v^j$ with $j \in [\hat\bE(v,C]$, there exists an arc of form $(u^1,v^j)\in \vec \cW(C)$. 
        It follows that there are at least $\hat\bE(v,C)$ edges 
        $uv \in \psg(Y)$
        with $u\in \interp{\hat f}(X)$.
        
        It remains to prove that there are exactly $\hat\bE(v,C)$ such edges when $\hat\bE(v,C) < \dphi$.
        This is where we use the fact that
        \begin{align*}
            a_\root(P) = \sum_{C\in \Cat}\sum_{\substack{ v\in V(G) \\ \hat \bE(v,C) < \dphi}} \hat\bE(v,C) = 
            \sum_{C\in \Cat}\sum_{\substack{ v\in V(G) \\ \hat \bE(v,C) < \dphi}} \abs{ \{ (u^1,v^i) \in \vec\cW(C)\}} = b_\root(P)
        \end{align*}
        holds.
        We already know that $\abs{\{(u^1, v^i) \in \vec\cW(C)\}} \geq \hat\bE(v,C)$ holds for every $v \in V(G)$ and $C\in\Cat$.
        So the above equality can only hold when $\abs{ \{ (u^1,v^i) \in \vec\cW(C)\}} =  \hat\bE(v,C)$ for every $v \in V(G)$ and every $C\in\Cat$ satisfying $\hat\bE(v,(X,Y)) < \dphi$.
        And this concludes the proof of the claim.
    \end{claimproof}
    Let again $v \in V(G)$ and $C = (X, Y) \in \Cat$.
    By \Cref{def:modelcheck:partial:interpretation}, the vertex $v$ belongs to $\ip{N^\sigma(C)}^{(\interp{ \hat f}, \interp{\hat g }, \hat \bE, \pdeg_{\vec \cW})}$ if and only if one of the following conditions applies:
    \begin{itemize}
        \item $\hat\bE(v,C)< \dphi$ and $\hat\bE(v,(X,Y))\in \sigma$, or
        \item $\hat\bE(v,C)= \dphi$ and $\pdeg_{\vec \cW}(v, C) + (\dphi\cdot \pphi) \in \sigma$.
    \end{itemize}
    When $\hat\bE(v,C)< \dphi$, by \Cref{claim:modelcheck:realizable}, we get that $\hat\bE(v,C)$ is exactly the number of edges $uv \in \psg(Y)$ %
    with $u\in \interp{\hat f}(X)$.
    Consequently, when $\hat\bE(v,C)< \dphi$, we have $v$ belongs to $\ip{N^\sigma(C)}^{(\interp{ \hat f}, \interp{\hat g }, \hat \bE, \pdeg_{\vec \cW})}$ if and only if $v$ belongs to $\ip{N^\sigma(C)}^{(G,\interp{ \hat f}, \interp{\hat g })}$.
    
    Now recall that by definition we have that $\pdeg_{\vec \cW}(v, C) \equiv_{ \pphi} \abs{\{(u^1, v^i) \in \vec\cW(C)\}}$.
    By \Cref{claim:modelcheck:realizable}, $\abs{\{(u^1, v^i) \in \vec\cW(C)\}}$ is exactly the number of edges $uv \in \psg(Y)$ 
    with $u\in \interp{\hat f}(X)$. 
    Moreover, when $\hat\bE(v,C)= \dphi$, this number is at least $\dphi$.
    From \Cref{obs:modelcheck:dphi:pphi}, we deduce that when $\hat\bE(v,C)= \dphi$, then $v$ belongs to $\ip{N^\sigma(C)}^{(G,\interp{ \hat f}, \interp{\hat g })}$ if and only if $\pdeg_{\vec \cW}(v, C) + (\dphi\cdot \pphi) \in \sigma$.
    We conclude that $\ip{N^\sigma(C)}^{(\interp{ \hat f}, \interp{\hat g }, \hat \bE, \pdeg_{\vec \cW})} = \ip{N^\sigma(C)}^{(G,\interp{ \hat f}, \interp{\hat g })}$.
    As discussed above, this implies that $(G,\interp{ \hat f}, \interp{\hat g })$ realizes $(\False, \barc)$ and concludes this proof.
\end{proof}

We conclude this subsection by observing that  \Cref{obs:model-realization,lemma:modelcheck:realizable:forward,lemma:modelcheck:realizable:backward} imply \Cref{lem:modelcheck:equivalence:G:models:phi}.

\section{Recursive Equalities}
\label{sec:recursive}
In this section, we formally define the polynomials that our algorithm will evaluate as well as the recursive equalities on which our algorithm relies on: the \emph{root, leaf, forget, and join equalities}.
This is the technical core of our work.

\subsection{The definitions of polynomials}

Similarly as in \cite{HegerfeldK20,NederlofPSW23,PilipczukW18}, we encode the different choices of $a,b\in \bN$ and $\barc \in \bN^{\ell}$ as the degrees of formal variables $\alpha,\beta,\gamma_1,\dots,\gamma_\ell$ so that all the relevant values, namely the values of form $\abs{\Partial(\tss,I,a,b,\barc)}$, can be stored as the coefficients of one polynomial from $\bZ[\alpha, \beta, \gamma_1,\dots,\gamma_\ell]$.

\begin{definition}\label{def:modelcheck:polynomial}
    For every triad $\tss$ of $G$ and every index $I = (\interp f, \bE, \allow, \pdeg, \False) \in \Index(\tss)$, we define the polynomial $\cP_\tss(I)$ as follows:
    \begin{equation} \label{eq:def-of-cP}
    \cP_\tss(I)\defeq\sum_{a,b,c^1,\dots,c^\ell\in\bN} \abs{\Partial(\tss,I,a,b,(c^1,\dots,c^\ell))} \cdot \alpha^a \cdot \beta^b \cdot \gamma_1^{c^1} \cdots\gamma_\ell^{c^\ell}.
    \end{equation}
    Moreover, we define its equality variant $\cP^=_\tss(I)$ as follows:
    \[
        \cP^=_\tss(I)\defeq\sum_{a,b,c^1,\dots,c^\ell\in\bN} \abs{\Partial^=(\tss,I,a,b,(c^1,\dots,c^\ell))} \cdot \alpha^a \cdot \beta^b \cdot \gamma_1^{c^1}\cdots\gamma_\ell^{c^\ell}.
    \]
    And for a vector $\barc=(c^1,\dots,c^\ell)\in \bN^{\ell}$, we denote by $\gamma^{\barc}$ the monomial $\prod_{i\in [\ell]} \gamma_i^{c^i}$.
\end{definition}

The following lemma bounds the coefficients of these polynomials.
Recall that we assumed that $\dphi\geq 2$ holds, this was done in order to simplify the following upper bounds.

\begin{lemma}\label{lem:bound-on-coefficients-of-p}
    Let $G$ be a graph with at most $n \cdot \td$ edges for some value $\td$. %
    For every triad $\tss = (\tail, \subtree, \sheaf)$ of $G$ and every index $I \in \Index(\tss)$, every coefficient of the polynomial $\cP_\tss(I)$ is non-negative and upper-bounded by $(\dphi + 1)^{5 n \cdot \td \cdot |\phi|} = 2^{5n \cdot \td \cdot |\phi| \cdot \log(\dphi + 1)}$.
\end{lemma}
\begin{proof}
    Every partial solution is a quadruple so we bound the number of options for each of its four elements.
    For this note that the cardinalities of the sets $\var^V(\phi)$, $\var^E(\phi)$, and $\Cat$ are upper-bounded by $|\phi|$.
    The size of thee broom of $\tss$ is upper-bounded by $n$ and the size of $\sheaf$ is upper-bounded by $\abs{E(G)}\leq n \cdot \td$ by \Cref{obs:prelim:treedepth:nb:edges}.
    The first element is a vertex interpretation on $\broom$ so there are at most $(2^n)^{|\phi|} = 2^{n \cdot |\phi|}$ options for it.
    Similarly, the second element is an edge interpretation on $\sheaf$ so there are at most $(2^{n \cdot \td})^{|\phi|} = 2^{n \cdot \td \cdot |\phi|}$ options.
    The third component is a neighborhood function on $\broom$ so there are at most $((\dphi+1)^n)^{|\phi|} = (\dphi+1)^{n \cdot |\phi|}$ possibilities.
    Finally, the fourth component maps every element of $\Cat$ to a simple subset of $\vec\sheaf$, i.e., for every edge $uv \in \sheaf$, at most one arc of form $u^1 v^i$ and at most one arc of form $v^1 u^j$ belongs to this simple subset.
    Because for every vertex $v$ of $G$, there are $\dphi$ copies of $v$ in $\vec G_{\phi}$,
    there are at most $\Big(\big((\dphi + 1)^2\big)^{n \cdot \td}\Big)^{|\phi|} = (\dphi + 1)^{2n \cdot \td \cdot |\phi|}$ options for the fourth element and the claim follows.
\end{proof}

We will strongly rely on the following polynomials to efficiently handle the \emph{join equalities} (cf.\ \cref{subsec:join}):
\begin{definition}\label{def:modelcheck:transform}
    Let $p$ be a prime number such that $\FF^*_p$ admits the $\pphi$-th root of unity and let $\omega(p)$ denote the smallest (with respect to the natural ordering of integers) such root of unity.
    For every triad $\tss$ of $G$ and every index $I = (\interp f, \bE, \allow, \pdeg, \False) \in \Index(\tss)$, we define the following polynomial
    \begin{equation}\label{eq:def-q}
        \cQ^p_\tss(I) \defeq \sum_{q \in \pdegdom(\tail)} \omega(p)^{q \cdot \pdeg} \cdot \cP_{\tss}(\interp f, \bE, \allow, q, \False) \bmod{p}
    \end{equation}
    from $\FF_p[\alpha, \beta, \gamma_1, \dots, \gamma_\ell]$.
\end{definition}
Note that if we fix all elements of the index except for $\pdeg$, then this definition captures the discrete Fourier transform with the domain $\pdegdom(\tail) = \tail \times \Cat$ (cf.\ \cref{subsec:dft-definitions}). 
Thus we will sometimes refer to polynomials of form $\cQ$ as \emph{DFT-transformed polynomials}.

In the remainder of this section we provide the recursive equalities on which our algorithm will rely on.
Let us emphasize that most of the time our arising algorithm in \cref{sec:modelcheck} will work with DFT-transformed polynomials of form $\cQ^p$ for certain prime numbers $p$ as this is crucial to achieve time and space efficiency simultaneously.
Nevertheless, in some lemmas we will explicitly provide the equalities for non-transformed polynomials $\cP$ to simplify the proofs and intuitions.
To this end, we fix an elimination tree $T$ of $G$ and let $\td$ denote its depth.
Assuming that one sufficiently large prime $p$ with certain properties is provided, the direct recursion relying on these equalities would result in an algorithm whose space complexity is polynomial in $n$.
In the next section we will show how, relying on the Chinese remainder theorem, to improve this dependence to be only logarithmic in $n$.
In particular, we will describe which (sufficiently short) prime numbers $p$ are used and how they are computed.
Let us recall that for a node $u$ of $T$ the triads $\tss[u]$ and $\tss(u)$ are $(\tail[u],\subtree(u),\sheaf(u))$ and  $(\tail(u),\subtree[u],\sheaf(u))$, respectively.
We define the following shortcuts: 

\begin{definition}\label{def:modelcheck:tables}
    Let $u$ be a node of $T$.
    
    Every notation in the form $\spadesuit_{\tss[u]}(\cdots)$ or $\spadesuit(\tss[u], \cdots)$ is denoted by $\spadesuit_u[\cdots]$. For instance, we have $\cP_u[I]\defeq \cP_{\tss[u]}(I)$, $\cQ^p_u[I]\defeq \cQ^p_{\tss[u]}(I)$, $a_u[P] \defeq a_{\tss[u]}(P)$, $\Partial_u[I] \defeq \Partial(\tss[u], I)$, and $\Index[u] \defeq \Index(\tss[u])$.
    
    Moreover, every notation in the form $\spadesuit_{\tss(u)}(\cdots)$ or $\spadesuit(\tss(u), \cdots)$ is denoted by $\spadesuit_u(\cdots)$. 
    For example, we have $\cP^=_u(I)\defeq \cP^=_{\tss(u)}(I)$, $\cQ^p_u(I)\defeq \cQ^p_{\tss(u)}(I)$, $b_u(P)\defeq b_{\tss(u)}(P)$, $\Partial_u(I)\defeq \Partial(\tss(u), I)$, and $\Index(u) \defeq \Index(\tss(u))$.
\end{definition}

\subsection{Root equalities}

In this subsection, $r$ denotes the root of the elimination tree $T$.
Recall that if the domain of some function is empty, we denote this function by $\emptyset$.
Similarly, we denote the unique vertex interpretation on the empty set by $\emptyset$. 
Now for simplicity by $\emptyset^4$ we denote the tuple $(\emptyset, \emptyset, \emptyset, \emptyset)$.
By \cref{lem:modelcheck:equivalence:G:models:phi}, to solve the problem, the values of form $|\Partial^=_r(\emptyset^4,\False, i,i, \barc)|$ are relevant where $\False\subseteq \EQ(\phi)$, $i \in \bN$, and $\barc\in\bN^{\ell}$.
So we start by providing %
a relation between the partial solutions where \emph{precisely} the equalities from $\False$ are falsified on the one hand, and the partial solutions where the equalities from $\False$ are \emph{possibly} falsified on the other hand.
In fact, the analogue of this lemma holds for any node of the elimination tree $T$ but in our applications we will only use it for the root node.
\begin{lemma}
    For every $\False \subseteq \EQ(\phi)$, it holds that
    \[
        \cP^=_r(\emptyset^4, \False) = \sum_{\False' \subseteq \False} (-1)^{|\False'|} \cdot \cP_r(\emptyset^4, \False \setminus \False').
    \]
    \label{lem:root:inclusion:exclusion}
\end{lemma}
\begin{proof}
    The proof is almost a direct application of the inclusion-exclusion principle (cf.\ \cref{thm:inclusion-exclusion-intersection}).
    To see this formally, consider arbitrary but fixed $a, b \in \bN$ and $\barc \in \bN^\ell$.
    For an equality $x \in \False$ we define the set 
    \[
        \Partial_r(\emptyset^4, \False, a, b, \barc, x) = \{(\psf, \psg, \pse, \psw) \in \Partial_r(\emptyset^4, \False, a, b, \barc) \mid x \in \FalseFunc(\psf, \psg, \pse, \pdeg_{\psw})\}
    \]
    of partial solution where the equality $x$ is falsified.
    First, it holds that 
    \begin{align*}
        &\bigcap_{x \in \False} \Partial_r(\emptyset^4, \False, a, b, \barc, x) =
        \{(\psf, \psg, \pse, \psw) \in \Partial_r(\emptyset^4, \False, a, b, \barc) \mid \False \subseteq\FalseFunc(\psf, \psg, \pse, \pdeg_{\psw})\} \\
        =&\{(\psf, \psg, \pse, \psw) \in \Partial_r(\emptyset^4, \False, a, b, \barc) \mid \False =\FalseFunc(\psf, \psg, \pse, \pdeg_{\psw})\} = \Partial_r^=(\emptyset^4, \False, a, b, \barc)
    \end{align*}
    where the second equality holds since we have $\FalseFunc(\psf, \psg, \pse, \pdeg_{\psw}) \subseteq \False$ for every $(\psf, \psg, \pse, \psw)$ from $\Partial_r(\emptyset^4, \False, a, b, \barc)$ due to \ref{item:partial:respect:False} together with the fact that both broom and subtree of $\tss(r)$ are equal to $V(G)$.
    Then by inclusion-exclusion principle (\cref{thm:inclusion-exclusion-intersection}) we get
    \begin{align}
        &|\Partial_r^=(\emptyset^4, \False, a, b, \barc)| = |\bigcap_{x \in \False} \Partial_r(\emptyset^4, \False, a, b, \barc, x)| \nonumber \\
        =&\sum_{\False' \subseteq \False} (-1)^{|\False'|} \cdot \abs{\Partial_r(\emptyset^4, \False, a, b, \barc) \setminus \Bigl(\bigcup_{x \in \False'} \Partial_r(\emptyset^4, \False, a, b, \barc, x)\Bigr)} \nonumber\\
        =&\sum_{\False' \subseteq \False} (-1)^{|\False'|} \cdot \abs{\Partial_r(\emptyset^4, \False \setminus \False', a, b, \barc)} \label{eq:root-incl-excl-proof}
    \end{align}
    where the last equality holds due to 
    \begin{align*}
        &\Partial_r(\emptyset^4, \False, a, b, \barc) \setminus \Bigl(\bigcup_{x \in \False'} \Partial_r(\emptyset^4, \False, a, b, \barc, x)\Bigr) \\
        &=\{(\psf, \psg, \pse, \psw) \in \Partial_r(\emptyset^4, \False, a, b, \barc) \mid x \not\in \FalseFunc(\psf, \psg, \pse, \pdeg_{\psw}) \text{ for all } x \in \False'\} \\
        &=\{(\psf, \psg, \pse, \psw) \in \Partial_r(\emptyset^4, \False, a, b, \barc) \mid \FalseFunc(\psf, \psg, \pse, \pdeg_{\psw}) \subseteq \False \setminus \False'\} =\Partial_r(\emptyset^4, \False \setminus \False', a, b, \barc)
    \end{align*}
    for every $\False' \subseteq \False$.
    Since \eqref{eq:root-incl-excl-proof} holds for arbitrary $a, b, \barc$, the lemma follows.
\end{proof}
So with \Cref{lem:modelcheck:equivalence:G:models:phi,lem:root:inclusion:exclusion} the task boils down to the computation of polynomials of form $\cP_r(\cdot)$.
In the remainder of the subsection, we prove that $\cQ^p_r(I) \equiv_p \cP_r(I)$ for every index $I \in \Index(r)$. %
Essentially, this follows from the fact that $\tail(r)$ is empty and thus the DFT-transformation is the identity function (modulo $p$).

\begin{lemma}\label{lem:root-p-equal-q-lemma}
    For every prime number $p$ such that $\FF^*_p$ admits the $\pphi$-th root of unity and every index $I = (\interp f, \bE, \allow, \pdeg, \False) \in \Index(r)$, it holds that $\cQ^p_r(I) \equiv_p \cP_r(I)$.
\end{lemma}

\begin{proof}
    We have $\tail(r) = \emptyset$ and thus also $\tail(r) \times \Cat$ is empty.
    So the set $\pdegdom(\tail(r)) = (\Fp)^{\tail(r) \times \Cat} = \Fp^\emptyset$ consists of precisely one mapping whose domain is empty.
    Since $\pdeg$ is such a mapping, we have $\pdegdom(\tail(r)) = \{\pdeg\}$.
    From this we get 
    \begin{align*}
        \cQ^p_r(I) \equiv_p &\cQ^p_{\tss(r)}(I) \stackrel{\eqref{eq:def-q}}{\equiv_p}
        \sum_{q \in \pdegdom(\tail(r))} \omega(p)^{q \cdot \pdeg} \cdot \cP_{\tss(r)}(\interp f, \bE, \allow, q, \False) \\       
        \equiv_p&\sum_{q \in \pdegdom(\tail(r))} \Bigl(\prod_{(v, C) \in \tail(r) \times \Cat} \omega(p)^{q(v,C) \cdot \pdeg(v, C)}\Bigr) \cdot \cP_{\tss(r)}(\interp f, \bE, \allow, q, \False) \\
        \equiv_p&\sum_{q \in \pdegdom(\tail(r))} 1 \cdot \cP_{\tss(r)}(\interp f, \bE, \allow, q, \False) \equiv_p \sum_{q \in \{\pdeg\}} 1 \cdot \cP_{\tss(r)}(\interp f, \bE, \allow, q, \False) \\
        \equiv_p& \cP_{\tss(r)}(\interp f, \bE, \allow, \pdeg, \False) \equiv_p \cP_{\tss(r)}(I) \equiv_p \cP_r(I).
    \end{align*}
\end{proof}

So now we may restrict our considerations to the computation of DFT-transformed polynomials of form~$\cQ_u(\cdot)$ and $\cQ_u[\cdot]$ for all the nodes $u$ of $T$.
To this end, we will sometimes prove some equalities about the polynomials of form $\cP_u[\cdot]$ and $\cP_u(\cdot)$ as intermediate steps in our proofs.

\subsection{Leaf equality}
We start with the equality for the non-recursive computation at an arbitrary leaf $u$ of $T$.
In particular, the lemma below states that for the DFT-transformed polynomials $\cQ$, instead of trying out all subsets of $\sheaf(u)$ for every edge set variable, we can consider every edge of $\sheaf(u)$ independently.
Before providing a formal procedure, we sketch the main ideas behind it.

Because $u$ is a leaf, we have $\tail[u] = \broom[u]$ and $\subtree(u) = \emptyset$.
So for a partial solution $P = (\psf, \psg, \pse, \psw)$ in $\Partial_r[\indf, \inde, \allow, \pdeg, \False]$, the vertex interpretation $\psf = \indf$ is fixed due to \ref{item:partial:fconsistent}, the neighborhood function $\pse = \inde$ is fixed due to \ref{item:partial:econsistent}, and only $\psg$ and $\psw$ may vary.
Because $\subtree(u)$ is empty, first, no vertex set equality is falsified by $P$ and second, the exponent of the formal variable $\alpha$ is equal to zero in every monomial.

To correctly handle an edge $e = vw \in \sheaf(u)$, we proceed as follows for every edge interpretation $\psg_e$ on $\{e\}$.
First of all, if there is an equality $Y_1 = Y_2 \in \EQ^E(\phi) \setminus \False$, i.e., an edge set equality forbidden to be falsified, such that under $\psg_e$ precisely one of the variables $Y_1$ and $Y_2$ contains the edge $e$, we discard $\psg_e$.
Second, for every edge set variable $X_i$ involved in a size measurement in $\phi$ and such that $e \in \psg_e(X_i)$, we multiply the current polynomial by $\gamma_i$ to reflect the contribution of $e$ to the cardinality of $\psg(X_i)$.
Similarly, but with more technical reasoning we multiply by $\omega(p)^{\pdeg(w, (X, Y))}$ to reflect the contribution of $e$ to the DFT-transformation (cf.\ \cref{def:modelcheck:transform})---here $(X, Y)$ is an element of $\Cat$ such that $v \in \indf(X)$ and $e \in \psg(Y)$, i.e., the edge $e= vw$ contributes to the %
neighborhood $N(w, \psg(Y)) \cap \psf(X)$ of $w$.
Finally, for every such $(X, Y) \in \Cat$, the set $\psw(X, Y)$ has to contain precisely one arc of form $(v^1, w^j)$.
The neighborhood function $\allow$ on $\tail[u]$ determines that precisely the values $1 \leq j \leq \allow(w, (X,Y))$ are allowed for this, i.e., there are precisely $\allow(w, (X,Y))$ options for this arc.
In the end, we also correctly initialize the exponent of $\beta$.
To reflect that the edge $e$ also contributes to certain neighborhoods of its other endpoint $v$, we essentially repeat the process by swapping the roles of $v$ and $w$.
Now we formalize this approach.

\begin{lemma}\label{lem:modelcheck:leaf}
    Let $u$ be a leaf of $T$ and let $I=(\interp f, \bE, \allow, \pdeg, \False)\in \Index[u]$.
    Further, let $p$ be a prime number such that $\FF_p^*$ admits the $\pphi$-th root of unity.
    Then it holds that:
    \[
        \cQ_u^p[I] \equiv_p \prod_{e = vw \in \sheaf(u)} \sum_{\psg_{e} \in \possible(e, \False)} \gamma^{\barc_{\psg_{e}}} \cdot \prod_{(X, Y) \in \Cat} L(v, w, \psg_{e}, (X, Y)) \cdot L(w, v, \psg_{e}, (X, Y))
    \]
    where
    \[
        \possible(e, \False) = \{\psg_e \mid \psg_e \text{ is an edge interpretation on } \{e\} \text{ and } \FalseFunc^E(\psg_e) \subseteq \False\},
    \]
    $\barc_{\psg_{e}} = (c_{\psg_{e}}^1, \dots, c_{\psg_{e}}^\ell) \in \{0,1\}^\ell$ with $c_{\psg_{e}}^i = |\psg_{e}(X_i)|$ if $X_i \in \var^E(\phi)$ and $c_{\psg_{e}}^i = 0$ if $X_i \in \var^V(\phi)$ for all $i \in [\ell]$, and
    \[
        L(v, w, \psg_{e}, (X, Y)) = 
        \begin{cases}
        \allow(w, (X, Y)) \cdot \omega(p)^{\pdeg(w, (X, Y))} \cdot \beta^{[\bE(w, (X, Y)) < \dphi]} & \!\!\text{if } v \in \indf(X) \text{\! and } e \in %
        \psg_e(Y)\\
        1 & \!\!\text{otherwise.}
        \end{cases}
    \]
\end{lemma}

\begin{proof} 
    Since $u$ is a leaf, we have $\tss[u] = (\tail[u], \emptyset, \sheaf(u))$ and the broom of $\tss[u]$ is equal to $\tail[u]$; we will sometimes implicitly use these facts in this proof.
	Before we start rewriting $\cQ_u^p[I]$, we need the following notations and a claim.
	
	For an edge $e\in\sheaf(u)$ and an edge interpretation $\psg$ of $\phi$ on $\sheaf(u)$, we use $\psg_{e}$ as a shortcut for the edge interpretation $\psg_{|_{\{e\}}}$ on $\{e\}$.
	We define $\possible(\False)$ via
    \begin{align*}
        \possible(\False) = \{\psg \mid &\psg \text{ is an edge interpretation on $\sheaf(u)$ and}\\
        &\psg_e \in \possible(e, \False) \text{ for every } e \in \sheaf(u) \},
    \end{align*}
    i.e., the set of all edge interpretations $\psg$ on $\sheaf(u)$ such that every edge set equality falsified by $\psg$ belongs to $\False$ (cf.\ \cref{claim:union-false-set}).
	Further, we denote by $\vec\bW$ the set of all functions $\psw$ that map each element from $\Cat$ to a simple subset of arcs from $\vec\sheaf(u)$.
	For $\psw \in \vec\bW$, $C\in\Cat$, and $v,w \in V(G)$, we define $\psw(C,v,w)$ as the intersection of $\psw(C)$ and the set of arcs $\{ (v^1, w^i) \mid i\in [\dphi]\}$.
	Since $\psw(C)$ is a simple set of arcs, it follows that $\psw(C,v,w)$ contains at most one arc.
	Furthermore for an edge interpretation $\psg_e$ on $\{e\}$ and $(X, Y) \in \Cat$, we define 
	\[
	\vec\bW(\psg_{e}, (X,Y), v, w) \defeq 
	\begin{cases}
		\Big\{ \big\{(v^1, w^i)\big\} \mid i\leq \allow(w,(X,Y)) \Big\} & \text{ if } v\in \indf(X) \text{ and } e\in %
        \psg_e(Y) \\
		\{ \emptyset \} & \text{ otherwise. }
	\end{cases}
	\]
    Informally speaking, this set reflects the ``allowed'' options for the (at most one) arc in $\psw(X,Y)$ with a tail in $\{v^1, \dots, v^{\dphi}\}$ and a head in $\{w^1, \dots, w^{\dphi}\}$ corresponding to the edge $e$ depending on $\psg_e$ and the vertex interpretation $\psf$.
    We emphasize the asymmetry of $v$ and $w$ in this definition.
	And for an edge interpretation $\psg$ on $\sheaf(u)$ we define $\vec\bW(\psg)$ as the set of all $\psw \in \vec\bW$ such that $\psw(C,v,w) \in \vec\bW(\psg_{e}, C, v, w)$ for every $C\in\Cat$ and pair $(v,w)$ such that $e=vw\in \sheaf(u)$.
    The definition of $\vec\bW(\psg)$ essentially combines the ``allowed'' for $\psg$ options in $\psw$ for all elements of $\Cat$, all edges $e$ of $\sheaf(u)$, and both (ordered) pairs $(v,w)$ and $(w,v)$ of endpoints of $e$.
    Finally, for a mod-degree $y \in \pdegdom(\tail[u])$, by $I[\pdeg \mapsto y]$ we denote the index $(\indf, \inde, \allow, y, \False) \in \Index[u]$.
    \begin{claim}\label{claim:leaf:characterization}
		For every $y\in  \pdegdom(\tail[u])$, we have $(\psf,\psg,\pse,\psw) \in \Partial_u[I[\pdeg \mapsto y]]$ if and only if the following conditions are satisfied:
		\begin{multicols}{5}
			\begin{itemize}
				\item $\psf = \interp f$, 
				\item $\pse = \bE$,
				\item $y = \pdeg_{\psw}$,
				\item $\psg \in \possible(\False)$,
				\item $\psw\in \vec\bW(\psg)$.
			\end{itemize}
		\end{multicols}
	\end{claim}
	\begin{claimproof}
		For the one direction, let $(\psf,\psg,\pse,\psw) \in \Partial_u[I[\pdeg \mapsto y]]$ for some $y\in \pdegdom(\tail[u])$.
		First, it holds that $\psf = \indf$, $\pse = \inde$, and $y = \pdeg_{\psw}$ by \ref{item:partial:fconsistent}, \ref{item:partial:econsistent}, and \ref{item:partial:pdegconsistent}, respectively.
		Consider an arbitrary edge $e = vw \in \sheaf(u)$.
		Then by \cref{claim:restriction-vs-false-set}, we have $\FalseFunc^E(\psg_e) \subseteq \FalseFunc(\indf_{|_{\emptyset}}, \psg, \inde_{|_{\emptyset}}, (\pdeg_{\psw})_{|_\emptyset}) \subseteq \False$ where the last inclusion holds due to \ref{item:partial:respect:False}.
		So we have $\psg_e \in \possible(e, \False)$.
		As it holds for any $e\in \sheaf(u)$, we deduce that $\psg\in \possible(\False)$.
		Now let $C=(X,Y)\in\Cat$. 
		Due to \ref{item:partial:edge:consistent}, if $v \notin \indf(X)$ or $e \notin \psg(Y)$, %
        then we have $\psw(C, v,w) = \emptyset$, and otherwise, there exists an arc of form $(v^1, w^j) \in \psw(C)$.
		In the latter case, since $\psw(C)$ is simple and due to \ref{item:partial:respect:allow}, such a value $j$ is unique and satisfies $1 \leq j \leq \allow(w, (X, Y))$.
		Finally, $e \notin \psg(Y)$ %
        is equivalent to $e \notin \psg_e(Y)$. %
		Altogether we deduce that $\psw(C, v,w) \in \vec\bW(\psg_{e}, C, v, w)$ and we conclude that $\psw\in \vec\bW(\psg)$.
		
		\medskip
		
		For the other direction let $\psg \in \possible(\False)$ and $\psw\in \vec\bW(\psg)$ with $y = \pdeg_{\psw}$.
		First, we argue that $P = (\indf, \psg, \inde, \psw)$ is a partial solution for $\tss[u]$.
		Observe that $P$ trivially satisfies \ref{item:partial:respect-domain}.
		Further, by definition of $\vec\bW(\psg)$, for every $C = (X,Y)\in\Cat$ and edge $e = vw \in \sheaf(u)$, we have $\psw(C,v,w) %
        \neq \emptyset$
        if and only if $v\in \indf(X)$ and $e \in \psg_e(Y)$. %
		Since $e \in \psg_e(Y)$ %
        is equivalent to $e \in \psg(Y)$, %
        we deduce that $P$ satisfies \ref{item:partial:edge:consistent}.
		Furthermore, for every $w\in \tail[u]$ and $C\in \Cat$, if $\psw(C)$ contains an arc of form $(v^1,w^i)$ for some $v \in V(G)$, then by definition of $\vec\bW(\psg_{e}, C, v, w)$, we have $i\leq \allow(w,C)$.
		Since $I$ is an index, we have $\allow\leq \bE$, and thus we have $i\leq \bE(w,C)$ as required by \ref{item:partial:respect:broom:hatbE}.
		Finally, $P$ trivially satisfies \ref{item:partial:respect:subtree:hatbE} because $\subtree(u)$ is empty. It follows that $P$ is indeed a partial solution for $\tss[u]$.
		
		Now it remains to prove that $P$ is compatible with $I[\pdeg\mapsto\pdeg_{\psw}]$, i.e., we have $P\in \Partial_u[I[\pdeg\mapsto\pdeg_{\psw}]]$.
		The conditions \ref{item:partial:fconsistent}, \ref{item:partial:econsistent}, and \ref{item:partial:pdegconsistent} are trivially satisfied.
		As mentioned above, for every $C\in\Cat$ and every arc $(v^1, w^j) \in \psw(C)$, we have $j \leq \allow(w, C)$ so \ref{item:partial:respect:allow} is satisfied.
		The set $\subtree(u)$ is empty so $\FalseFunc^V(\indf_{|_{\subtree(u)}}, \psg, \inde_{|_{\subtree(u)}}, (\pdeg_{\psw})_{|_{\subtree(u)}}) = \emptyset \subseteq \False$.
		And finally we have
		\[
		\FalseFunc^E(\indf_{|_{\emptyset}}, \psg, \inde_{|_{\emptyset}}, (\pdeg_{\psw})_{|_{\emptyset}}) = \FalseFunc^E(\psg) \stackrel{\cref{claim:union-false-set}}{=} \bigcup_{e \in \sheaf(u)} \FalseFunc(\psg_e) \subseteq \False
		\]
		where the last inclusion holds due to $\psg_e \in \possible(e, \False)$ for all $e \in \sheaf(u)$ so \ref{item:partial:respect:False} holds.
		Hence, $P$ is indeed compatible with $I[\pdeg\mapsto\pdeg_{\psw}]$ concluding the proof of the claim.
	\end{claimproof}
	 
	We are now ready to start rewriting the polynomial $\cQ_u^p[I]$.
 	In the rest of the proof $\omega$ denotes $\omega(p)$. 
    We have the following:
	\begin{align}
		\cQ_u^p[I] \stackrel{\eqref{eq:def-q}}{\equiv_p} &\sum_{y \in \pdegdom(\tail[u])} \omega^{\pdeg \cdot y} \cdot \cP_u[I[\pdeg \mapsto y]]  \nonumber \\
		\stackrel{\eqref{eq:def-a-b},\eqref{eq:def-of-c},\eqref{eq:def-of-cP}}{\equiv_p}&\sum_{\substack{y \in \pdegdom(\tail[u]) \\ P\in \Partial_u[I[\pdeg \mapsto y]]}} \omega^{\pdeg \cdot y} \cdot  \alpha^{a_u[P]} \cdot \beta^{b_u[P]} \cdot \gamma^{\barc_u[P]}   \nonumber \\
        \stackrel{\ref{claim:leaf:characterization}}{\equiv_p}&\sum_{\psg \in \possible(\False)} \sum_{\psw \in \vec\bW(\psg)} \omega^{\pdeg \cdot \pdeg_{\psw}} \cdot  \alpha^{a_u[(\indf, \psg, \inde, \psw)]} \cdot \beta^{b_u[(\indf, \psg, \inde, \psw)]} \cdot \gamma^{\barc_u[(\indf, \psg, \inde, \psw)]}.
		\label{eq:leaf:characterization} 
	\end{align}
	The following claim will help us to reformulate the inner product in \eqref{eq:leaf:characterization}.
	\begin{claim}\label{claim:leaf:abcg}
		Let $P=(\indf,\psg,\inde,\psw)$ with $\psg\in\possible(\False)$ and $\psw \in \vec\bW(\psg)$. 
		Then we have %
		\begin{align*}
            a_u[P] & = 0, \\
			b_u[P] & = \sum_{\substack{(v,w) \colon \\ vw \in \sheaf(u)}} \sum_{(X,Y)\in \Cat} [v\in \indf(X)]\cdot [e\in \psg_e(Y)] \cdot [\inde(w,(X,Y)) < \dphi], \\
            \barc_u[P] & = \sum_{e\in \sheaf(u)} \barc_{\psg_e}, \\ 
			\pdeg \cdot \pdeg_{\psw} & \equiv_{\pphi} \sum_{\substack{(v,w) \colon \\ vw \in \sheaf(u)}} \sum_{(X,Y)\in \Cat} [v\in \indf(X)]\cdot [e\in \psg_e(Y)] \cdot \pdeg(w,(X,Y)).
		\end{align*}
	\end{claim}
	\begin{claimproof}
        Before proving the claim, let us emphasize that in the second and the fourth equality, every edge $e = vw \in \sheaf(u)$ is considered twice in the sum: for the pair $(v,w)$ and for the pair $(w,v)$ of endpoints.
        
		Since $\subtree(u) = \emptyset$, by definition, we have $a_u[P] = 0$ (cf.\ \eqref{eq:def-a-b}).
		Let $i\in[\ell]$. 
		If $X_i \in \var^V(\phi)$, then we have $c^i_u[P] = \abs{\indf(X_i)\cap \subtree(u)}=0$ (cf.\ \eqref{eq:def-of-c}).
		And if $X_i \in \var^E(\phi)$, then $c^i_u[P] = \abs{\psg(X_i)}$. Since $\abs{\psg(X_i)} = \sum_{e \in \sheaf(u)} \abs{\psg_e(X_i)}$, we deduce that $\barc_u[P] = \sum_{e\in \sheaf(u)} \barc_{\psg_e}$.
		Now notice that $b_u[P]$ can be reformulated as follows (cf.\ \eqref{eq:def-a-b}):
		\begin{align*}
			b_u[P]& \defeq \sum_{(X,Y) \in \Cat} \sum_{\substack{w \in \broom[u] \colon \\ \inde(w, (X,Y)) < \dphi}} \abs{\{(v^1, w^i) \in \psw(X,Y)\}} 
			  = \sum_{(X,Y) \in \Cat} \sum_{(v^1, w^i) \in \psw(X,Y)} [\inde(w, (X,Y)) < \dphi],
		\end{align*}
        for this we rely on the fact that $V(\sheaf(u)) \subseteq \broom[u] = \tail[u]$ holds.
		Next observe that by definition, for every $(X,Y)\in \Cat$, we have 
		\[ \psw(X,Y) = \bigcup_{{\substack{(v,w) \colon \\ vw \in \sheaf(u)}}} \psw((X,Y),v,w), \]
        and this is in fact a partition.
		Moreover, by assumption $\psw \in \vec\bW(\psg)$, we have $\psw((X,Y),v,w) \in \vec\bW(\psg_{e}, (X,Y), v, w)$ for every edge $e = vw \in \sheaf(u)$.
        Then the definition of $\vec\bW(\psg_{e}, (X,Y), v, w)$ implies that
		\[ 
			\abs{ \psw((X,Y),v,w)} = [v\in \indf(X)]\cdot [e\in \psg_e(Y)].
		\]
		We deduce that the equality with $b_u[P]$ holds. 
		It remains to prove the equality with $\pdeg \cdot \pdeg_{\psw}$.
		Remember that by definition, we have
		\begin{align*}
		  \pdeg \cdot \pdeg_{\psw} \defeq &\sum_{(w,(X,Y)) \in \tail[u] \times \Cat} \pdeg(w,(X,Y)) \cdot \pdeg_{\psw}(w,(X,Y)) \pmod{\pphi} \\
          =&\sum_{(X,Y) \in \Cat} \sum_{w \in \tail[u]} \pdeg(w,(X,Y)) \cdot \pdeg_{\psw}(w,(X,Y)) 
          \pmod{\pphi}.
		\end{align*}
		Moreover, we have $\pdeg_{\psw}(w,(X,Y)) = \abs{\{(v^1,w^j) \in \psw(X,Y)\}} \pmod\pphi$ for every $w \in \tail[u]$.
		We deduce that 
		\begin{align*}
			\pdeg \cdot \pdeg_{\psw}  \equiv_{\pphi} &\sum_{(X,Y)\in \Cat} \sum_{w \in \tail[u]} \sum_{(v^1,w^j)\in \psw(X,Y)} \pdeg(w,(X,Y)) \\
            \equiv_{\pphi} &\sum_{(X,Y)\in \Cat} \sum_{(v^1,w^j )\in \psw(X,Y)} \pdeg(w,(X,Y)).
		\end{align*}
		With the same arguments as used for the equality with $b_u[P]$, we conclude that the equality with $\pdeg \cdot \pdeg_{\psw}$ holds as well.
	\end{claimproof}

    Consider again $P=(\indf,\psg,\inde,\psw)$ with $\psg\in\possible(\False)$ and $\psw \in \vec\bW(\psg)$.
    Recall that $\omega$ is the $\pphi$-th root of unity in $\FF^*_p$ so in particular, we have $\omega^{\pphi} \equiv_p 1$. 
    We deduce from \Cref{claim:leaf:abcg} that
    \begin{equation*}
        \omega^{\pdeg \cdot \pdeg_{\psw}} \equiv_p \prod_{\substack{(v,w) \colon \\ vw \in \sheaf(u)}} \prod_{(X,Y)\in \Cat} \omega^{[v\in \indf(X)]\cdot [e\in \psg_e(Y)] \cdot \pdeg(w,(X,Y))}.
    \end{equation*}
    Again from \Cref{claim:leaf:abcg}, it follows that
    \[
        \beta^{b_u[P]} = \prod_{\substack{(v,w) \colon \\ vw \in \sheaf(u)}} \prod_{(X,Y)\in \Cat} \beta^{[v\in \indf(X)]\cdot [e\in \psg_e(Y)] \cdot [\inde(w,(X,Y)) < \dphi]} \text{ and } \gamma^{\barc_u[P]} = \prod_{e \in \sheaf(u)} \gamma^{\barc_{\psg_e}}.
    \]
    Hence, we have
   	\begin{align*}\label{eq:leaf:one:ps}
   		\omega^{\pdeg \cdot \pdeg_{\psw}} \cdot \alpha^{a_u[P]} \cdot \beta^{b_u[P]} \cdot \gamma^{\barc_u[P]} & \equiv_p \prod_{e=vw \in \sheaf(u)} \gamma^{\barc_{\psg_e}} \prod_{(X,Y)\in \Cat} L'(v,w,\psg_e,(X,Y)) \cdot L'(w,v,\psg_e,(X,Y)),  \\
   	 \text{where } L'(v,w,\psg_e,(X,Y))& = 
   		\begin{cases}
			\omega^{\pdeg(w,(X,Y))} \cdot \beta^{[\bE(w,(X,Y)) < \dphi]} & \text{ if } v\in \indf(X) \text{ and } e \in \psg_e(Y)\\
			1 & \text{ otherwise.}
   		\end{cases}
   	\end{align*}
	By plugging the above in \eqref{eq:leaf:characterization}, we obtain:
	\begin{align*}
		\cQ_u^p[I]  \equiv_p &\sum_{\psg \in \possible(\False)} \sum_{\psw \in \vec\bW(\psg)} \prod_{e=vw \in \sheaf(u)} \gamma^{\barc_{\psg_e}} \cdot \prod_{(X,Y)\in \Cat} L'(v,w,\psg_e,(X,Y)) \cdot L'(w,v,\psg_e,(X,Y)) \\
		 \equiv_p& \sum_{\psg \in \possible(\False)} \abs{\vec\bW(\psg)} \cdot \prod_{e=vw \in \sheaf(u)} \gamma^{\barc_{\psg_e}} \cdot \prod_{(X,Y)\in \Cat} L'(v,w,\psg_e,(X,Y)) \cdot L'(w,v,\psg_e,(X,Y)).\\
	\end{align*}
	By definition of $\vec\bW(\psg)$, we have
	\[
		\abs{\vec\bW(\psg)} = \prod_{\substack{e=vw\in\sheaf(u) \\ (X,Y)\in \Cat}} \abs{\vec\bW(\psg_e, (X,Y),v,w)} \cdot \abs{\vec\bW(\psg_e, (X,Y),w,v)}.
	\]
	Finally, observe that $\abs{\vec\bW(\psg_e, (X,Y),v,w)} \cdot L'(v,w,\psg_e,(X,Y)) = L(v,w,\psg_e,(X,Y))$ holds for every $(X, Y) \in \Cat$ and every edge $e = vw \in \sheaf(u)$.
    It follows that
	\[
		\cQ_u^p[I]  \equiv_p \sum_{\psg \in \possible(\False)} \prod_{e=vw \in \sheaf(u)} \gamma^{\barc_{\psg_e}} \cdot \prod_{(X,Y)\in \Cat} L(v,w,\psg_e,(X,Y)) \cdot L(w,v,\psg_e,(X,Y)).\\
	\]
	And from the definition of $\possible(\False)$, we conclude that 
	\[
	\cQ_u^p[I]  \equiv_p \prod_{e=vw \in \sheaf(u)} \sum_{\psg_e \in \possible(e,\False)}  \gamma^{\barc_{\psg_e}} \cdot \prod_{(X,Y)\in \Cat} L(v,w,\psg_e,(X,Y)) \cdot L(w,v,\psg_e,(X,Y))\\
	\]
    holds as claimed in the lemma.
\end{proof}

\subsection{Forget equality}\label{subsec:forget}
Next we prove the equalities reflecting the relation between the partial solutions of $\tss(u)$ and of $\tss[u]$ for a node $u$ of the elimination tree $T$.
In this case, the vertex $u$ moves from the tail to the processed subtree.
It cannot get further incident edges and therefore, we sometimes say that $u$ is \emph{forgotten}.
If one recalls the definitions of a partial solution $(\psf, \psg, \pse, \psw)$ as well as the compatibility with an index $(\indf, \inde, \allow, \pdeg, \False)$, essentially three aspects are to take care of.
First, for every $C \in \Cat$ and every index $j \in [\pse(u,c)]$, the vertex $u^j$ has to have an incoming arc in $\psw(C)$.
We employ inclusion-exclusion to ensure this.
Second, the vertex set equalities falsified by $u$ need to belong to $\False$.
We achieve this by only allowing the restrictions of vertex interpretations, neighborhood functions, and mod-degrees to $\{u\}$ that respect $\False$.
And third, the exponents of formal variables have to be updated according to restriction of the partial solution to the vertex $u$.
This is rather straightforward.
Formally, we have the following equality.
\begin{lemma}\label{lem:modelcheck:forget-non-optimized}
    Let $u$ be a node of $T$ and let $I=(\interp f, \bE, \allow, \pdeg, \False)\in \Index(u)$.
    Let $\cJ$ denote the set of all tuples $(\interp f_u, \bE_u, \pdeg_u)$ such that the following hold:
    \begin{multicols}{2}
        \begin{itemize}
            \item $\interp f_u$ is a vertex interpretation on $\{u\}$,
            \item $\bE_u$ is a neighborhood function on $\{u\}$,
            \item $\pdeg_u$ is a mod-degree on $\{u\}$, and
            \item $\FalseFunc^V(\interp f_u, \bE_u, \pdeg_u) \subseteq \False$. 
        \end{itemize}
    \end{multicols}
    And let $\cJ'$ be the set of all tuples $(\interp f_u, \bE_u, \allow_u, \pdeg_u)$ such that $(\interp f_u, \bE_u, \pdeg_u) \in \cJ$ and $\allow_u $ is a neighborhood function on $\{u\}$ such that $\allow_u\leq \bE_u$.
    For a tuple $(\interp f_u, \bE_u, \allow_u, \pdeg_u)\in\cJ'$, we denote by $\interp f_u^+, \bE_u^+, \allow_u^+$, and $\pdeg_u^+$ the vertex interpretation $\interp f \cup \interp f_u$, the neighborhood functions $\bE \cup \bE_u$ and $\allow \cup \allow_u$, and the mod-degree $\pdeg \cup \pdeg_u$ on $\tail[u]$, respectively (cf.\ \cref{def:union-of-functions}).
    Then we have:
    \begin{align}
        &\cP_{u}(I) = \sum_{(\interp f_u,\bE_u,\allow_u,\pdeg_u)\in \cJ'} \incexc(\bE_u,\allow_u) \cdot  \cP_{u}[\interp f_u^+, \bE_u^+, \allow_u^+,\pdeg_u^+, \False] \cdot \alpha^{a_{\bE_u}} \cdot \gamma^{\barc_{\interp f_u}} \label{eq:forget-node-non-transformed}, 
    \end{align}
    with
    \[
        \incexc(\bE_u,\allow_u) = \prod_{C\in\Cat} \left( (-1)^{\bE_u(u,C)-\allow_u(u,C)} \cdot \binom{\bE_u(u,C)}{\allow_u(u,C)}\right), 
        \quad
        a_{\bE_u} = \sum_{\substack{C\in \Cat \colon \\ \bE_u(u,C) < d_\phi}} \bE_u(u,C), \text{ and }
    \]
    \[
        \barc_{\interp f_u} =  (c^1_{\interp f_u},\dots,c^\ell_{\interp f_u}) \in \{0,1\}^\ell \text{ such that } c^i_{\interp f_u} = \abs{\interp{f}_u(X_i)} \text{ if } X_i \in \var^V(\phi), \text{ and } c^i_{\interp f_u} = 0 \text{ if } X_i \in \var^E(\phi)
    \]
    for every $i \in [\ell]$.
    Further, for every prime number $p$ such that $\FF^*_p$ admits a $\pphi$-th root of unity it holds that:
    \begin{align}
        \cQ^p_{u}(I) \equiv_p & \sum_{(\interp f_u,\bE_u,\allow_u,\pdeg_u)\in \cJ'} \incexc(\bE_u,\allow_u) \cdot \alpha^{a_{\bE_u}} \cdot \gamma^{\barc_{\interp f_u}} \nonumber \\
        \cdot & \frac{1}{\pphi^{|\Cat|}}\sum_{\hat q \in \pdegdom(\{u\})} \omega(p)^{-\hat q \cdot \pdeg_u} \cdot \cQ^p_{u}[\interp f_u^+, \bE_u^+, \allow_u^+, \pdeg \cup \hat q, \False] \label{eq:forget-node-transformed-non-optimized} 
    \end{align}
    where $\frac{1}{\pphi^{\abs{\Cat}}}$ is the multiplicative inverse of $\pphi^{\abs{\Cat}}$ in $\FF_{p}$.
\end{lemma}

\begin{proof}
    We start by proving the non-DFT-transformed equality \eqref{eq:forget-node-non-transformed}.
    Before we proceed with a formal proof, we provide an intuition behind it.
    In essence, this equality is the application of the inclusion-exclusion principle. 
    Recall that we have $\subtree[u] = \subtree(u) \dot\cup \{u\}$ whilst the triads $\tss(u)$ and $\tss[u]$ share the same broom and the same sheaf.
    Now the following observations are crucial.
    Roughly, the partial solutions of $\tss[u]$ and $\tss(u)$ differ ``only'' on two aspects: in the partial solutions of $\tss(u)$, %
    the equalities ``falsified by $u$'' have to belong to $\False$ as well (due to \ref{item:partial:respect:False}), and second, due to \ref{item:partial:respect:subtree:hatbE}, the sets in $\vec\cW$ have to contain arcs incoming into certain vertices from $u^1, \dots, u^{\dphi}$.
    The definition of $\cJ$ ensures that \ref{item:partial:respect:False} is satisfied on the vertex $u$.
    Thus, when we fix the tuple $(\interp f_u,\bE_u,\pdeg_u) \in \cJ$ the computation of the desired cardinality of the set of partial solutions of $\tss(u)$ boils down to ensuring that \ref{item:partial:respect:subtree:hatbE} holds for $u$. 
   
    To simplify the explanation in this sketch, assume that $\Cat$ contains only one element $C$.
    Thus, we have to ensure that for every $j \in [\bE_u(u, C)]$, the set $\vec\cW(C)$ contains an arc incoming into $u^j$. 
    So if we denote by $P_j$ the set of ``relevant'' partial solutions of $\tss[u]$ where $\vec\cW(C)$ contains an arc incoming into $u^j$, then what we seek is the value $|\bigcap_{j \in [\bE_u(u, C)]} P_j|$.
    By inclusion-exclusion principle (see \cref{thm:inclusion-exclusion-intersection}), we can compute this cardinality by considering all possible $S \subseteq [\bE_u(u, C)]$ and determining the cardinality of the set of partial solutions where we \emph{forbid} the edges incoming into $u^j$ for all $j \in S$.
    Finally, we observe that for fixed $S$, due to symmetries, in terms of counting, this is the same as to forbid the edges incoming into $u^{\bE_u(u, C)-|S|+1}, \dots, u^{\bE_u(u, C)}$, i.e., we can use the function $\allow_u$ with $\allow_u(u,C) = \bE_u(u, C)-|S|$.
    Now we formalize this idea.

    \medskip

    For the rest of the proof, we fix an index $I = (\indf, \inde, \allow, \pdeg, \False)$.
    For a neighborhood function $\bE_u$ on $\{u\}$, we define the set of pairs 
    \[
        \cT_{\bE_u} \defeq \{(C, j) \mid C \in \Cat, j \in [\bE_u(u, C)]\}.
    \]
    For every $\cT\subseteq \cT_{\bE_u}$ and every $J=(\interp f_u,\bE_u,\pdeg_u)\in \cJ$, we define the set $\Partial_u[J,\cT]$ as
    \begin{align*}
        \Partial_u[J,\cT] \defeq \bigl\{&(\interp{\hat f }, \interp{\hat g}, \hat\bE, \vec\cW) \in \Partial_u[\interp f_u^+,\bE_u^+,\allow\cup\bE_u,\pdeg_u^+,\False] \mid \\
        &\text{for every $(C,j)\in \cT$, there is at least one arc in $\vec\cW(C)$ incoming into $u^j$}\bigr\}.
    \end{align*}
    Notice that we use $\allow\cup\bE_u$ here because it gives us the least restrictive index of $\tss[u]$ in the sense that %
    for every neighborhood function $\allow_u$ on $\{u\}$ with $\allow_u\leq \bE_u$, we have by \Cref{obs:modelcheck:inclusions} the following inclusion
    \begin{equation*}
        \Partial_u[\interp f_u^+,\bE_u^+,\allow_u^+,\pdeg_u^+,\False] \subseteq \Partial_u[\interp f_u^+,\bE_u^+,\allow\cup\bE_u,\pdeg_u^+,\False].
    \end{equation*}
    For $a,b\in\bN$ and $\barc\in \bN^\ell$, we also define the set $\Partial_u[J,\cT,a,b,\barc]$ of all $P\in \Partial_u[J,\cT]$ such that $a_u[P]=a$, $b_u[P]=b$, and $\barc_u[P]=\barc$.
    We aim to prove with the following claims that we have
    \begin{equation}\label{eq:forget:cT}
        \cP_{u}(I) = \sum_{\substack{J = (\interp f_u,\bE_u,\pdeg_u)\in \cJ\\ a,b\in\bN, \barc\in \bN^\ell}} \abs{\Partial_u[J, \cT_{\bE_u},a,b,\barc]} \cdot \alpha^{a + a_{\bE_u}} \cdot \beta^{b} \cdot \gamma^{\barc + \barc_{\interp f_u}}.
    \end{equation}
    We will see later how to prove \eqref{eq:forget-node-non-transformed} from \eqref{eq:forget:cT} with the inclusion-exclusion principle.
    
    \begin{claim}\label{claim:modelcheck:forget:(u)to[u]}
        For every $P\in \Partial_u(I)$, there is a unique $J\in \cJ$ such that $P\in \Partial_u[J,\cT_{\bE_u}]$.
    \end{claim}
    \begin{claimproof}
        Let $P=(\interp{\hat f }, \interp{\hat g}, \hat\bE, \vec\cW) \in \Partial_u(I)$.
        Let $\interp f_u, \bE_u$ and $\pdeg_u$ be the restrictions of $\interp{\hat f }, \hat\bE$ and $\pdeg_{\vec\cW}$ on $\{u\}$, respectively.
        First, observe that $P\in \Partial_u(I)$ and \ref{item:partial:respect:False} imply the inclusion
        \begin{equation}\label{eq:forget:(u)to[u]:false}
            \FalseFunc(\interp{ \hat f}_{|_{\subtree[u]}}, \interp{\hat g }, \hat \bE_{|_{\subtree[u]}}, (\pdeg_{\vec \cW})_{|_{\subtree[u]}}) \subseteq \False.
        \end{equation}
        By applying \Cref{claim:restriction-vs-false-set} on \eqref{eq:forget:(u)to[u]:false} with $\dom' = \{u\}$ and $\dom = \subtree[u]$, we deduce the inclusion
        \[
            \FalseFunc^V(\interp f_u, \bE_u, \pdeg_u) \subseteq \FalseFunc^V(\interp{ \hat f}_{|_{\subtree[u]}}, \hat \bE_{|_{\subtree[u]}}, (\pdeg_{\vec \cW})_{|_{\subtree[u]}}) \subseteq \False.
        \]
        Thus, we have $J=(\interp f_u, \bE_u ,\pdeg_u)\in \cJ$.
        Next, we argue that $P\in \Partial(\tss[u])$ holds because we have $P\in \Partial(\tss(u))$.
        Indeed, since each of the broom and sheaf is the same for $\tss(u)$ and $\tss[u]$, the quadruple $P$ satisfies \ref{item:partial:respect-domain}, \ref{item:partial:edge:consistent}, and \ref{item:partial:respect:broom:hatbE} for $\tss[u]$ too.
        Moreover, the subtree of $\tss[u]$ (i.e., $\subtree(u)$) is a subset of the subtree of $\tss(u)$ (i.e., $\subtree[u]$), consequently $P$ satisfies \ref{item:partial:respect:subtree:hatbE} for $\tss[u]$. Thus, we have $P\in \Partial(\tss[u])$.
        
        Now, we show that $P$ is compatible with $J^+ \defeq (\interp f_u^+,\bE_u^+,\allow\cup\bE_u,\pdeg_u^+,\False)$.
        \begin{itemize}
            \item Recall that $\interp f_u^+ = \interp f \cup \interp f_u$ and $\interp f_u$ is the restriction of $\hat{\interp{f}}$ on $\{u\}$. Moreover, $\interp{f}(X) = \hat{\interp{f}}(X)\cap \tail(u)$ for every $X\in \var^V(\phi)$ because $P$ is compatible with $I$.
            We deduce that $P$ satisfies \ref{item:partial:fconsistent} for $J^+$ because for every $X \in \var^V(\phi)$, we have:
            \[
                \interp{\hat f}(X) \cap \tail[u] = (\interp{\hat f}(X) \cap \tail(u)) \cup (\interp{\hat f}(X) \cap \{u\}) = \interp{f}(X) \cup \interp{f}_u(X) = \interp{f}^+_u(X).
            \]
            \item Recall that $\bE_u^+ = \bE \cup \bE_u$ and $\inde_u$ is the restriction of $\pse$ on $\{u\}$.
            Since $P\in\Partial_u(I)$, we have $\bE_u^+(v,C) = \bE(v,C) = \hat\bE(v,C)$ for every $v\in \tail(u)$ and $C\in\Cat$.
            As $\bE_u$ is the restriction of $\hat\bE$ on $\{u\}$, we have $\bE_u^+(u,C)= \hat\bE(u,C)$ for every $C\in\Cat$. Hence, $P$ satisfies \ref{item:partial:econsistent} for $J^+$. With the same arguments,  we can prove that $P$ satisfies \ref{item:partial:pdegconsistent} for $J^+$.
            
            \item Since $P\in \Partial_u(I)$, we know that it satisfies \ref{item:partial:respect:allow} for $J^+$ and every vertex in $\tail(u)$. Moreover, $P$ satisfying \ref{item:partial:respect:broom:hatbE} (for both $\tss(u)$ and $\tss[u]$) implies that for every $C\in\Cat$ and every arc $(v^1, u^j)\in \vec\cW(C)$, we have $j \leq \hat\bE(u,C)$. As $\hat\bE(u,C) = \bE_u(u, C) = (\allow\cup \bE_u)(u,C)$, we deduce that $P$ satisfies \ref{item:partial:respect:allow} for $J^+$ and the vertex $u$ as well.
            
            \item  By applying \Cref{claim:restriction-vs-false-set} on \eqref{eq:forget:(u)to[u]:false} with $\dom' = \subtree(u)$ and $\dom = \subtree[u]$, we obtain
            \[
                \FalseFunc(\interp{ \hat f}_{|_{\subtree(u)}}, \interp{\hat g }, \hat \bE_{|_{\subtree(u)}}, (\pdeg_{\vec \cW})_{|_{\subtree(u)}}) \subseteq  \False.
            \]
            It follows that $P$ satisfies \ref{item:partial:respect:False} for $J^+$.
        \end{itemize}
        Hence, we indeed have $P\in \Partial_u[J^+]$.

        Now we need to prove that $P\in \Partial_u[J,\cT_{\bE_u}]$, i.e., for each $(C,j)\in \cT_{\bE_u}$, there is at least one arc in $\vec\cW(C)$ incoming into $u^j$.
        Since $P\in \Partial(\tss(u))$, from \ref{item:partial:respect:subtree:hatbE}, it follows that for every $w\in \subtree[u]$, every $C\in \Cat$, and every $j\in [\hat\bE(w,C)]$, there is at least one arc in $\vec\cW(C)$ incoming into $w^j$. As $u\in \subtree[u]$ and $\inde_u$ is the restriction of $\pse$ to $\{u\}$, we thus have $P\in \Partial_u[J,\cT_{\bE_u}]$ by definition of $\cT_{\bE_u}$.
        
        Finally, it is easy to see that $J$ is the unique element of $\cJ$ such that $P\in \Partial_u[J,\cT_{\bE_u}]$ because of the conditions \ref{item:partial:fconsistent}, \ref{item:partial:econsistent}, and \ref{item:partial:pdegconsistent}.
    \end{claimproof}
    With this claim in hand, we can show the following stronger property.
    \begin{claim}\label{claim:modelcheck:forget:disjointunion}
        The collection $(\Partial_u[J,\cT_{\bE_u}])_{J\in \cJ}$ is a partition of $\Partial_u(I)$.
    \end{claim}
    \begin{claimproof}
        Thanks to \Cref{claim:modelcheck:forget:(u)to[u]}, to prove this statement, it is sufficient to prove that for every $J\in \cJ$ and every $P\in \Partial_u[J,\cT_{\bE_u}]$, we have $P \in \Partial_u(I)$.
        So consider $J=(\interp f_u, \bE_u ,\pdeg_u)\in \cJ$ and $P=(\interp{\hat f }, \interp{\hat g}, \hat\bE, \vec\cW)\in \Partial_u[J,\cT_{\bE_u}]$.
        
        We start by proving that $P\in \Partial(\tss(u))$.
        Observe that $P$ satisfies \ref{item:partial:respect-domain} and \ref{item:partial:edge:consistent}, \ref{item:partial:respect:broom:hatbE} for $\tss(u)$ because $P$ satisfies these properties for $\tss[u]$ and $\tss(u)$ and $\tss[u]$ have the same broom and sheaf.
        It remains to check that $P$ satisfies \ref{item:partial:respect:subtree:hatbE} for $\tss(u)$.
        Observe that $P$ satisfies this condition for $\tss[u]$ and we have $\subtree(u) = \subtree[u]\setminus \{u\}$.
        So it remains to prove that, for every $C\in \Cat$ and every $j\in [\hat\bE(u,C)]$, there is at least one arc incoming into $u^j$ in $\psw(C)$.
        Since $\pse(u, C) = \inde_u(u,C)$ by \ref{item:partial:econsistent}, this follows from the definitions of $\Partial_u[J,\cT_{\bE_u}]$ and $\cT_{\bE_u}$.

        It remains to prove that $P$ is compatible with $I$. Let $J^+ \defeq (\interp f_u^+,\bE_u^+,\allow\cup\bE_u,\pdeg_u^+,\False)$.
        By assumption on $P$ we then have $P \in \Partial_u[J^+]$.
        \begin{itemize}
            \item Observe that $\interp f, \bE, \allow$, and $\pdeg$ are the restrictions of $\interp f_u^+,\bE_u^+,\allow\cup\bE_u$, and $\pdeg_u^+$ on $\tail(u)$, respectively.
            Since $\tail(u)\subseteq \tail[u]$, we deduce that $P$ satisfies \ref{item:partial:fconsistent}, \ref{item:partial:econsistent}, \ref{item:partial:pdegconsistent}, and \ref{item:partial:respect:allow} for $I$ because it satisfies these properties for $J^+$.
            
            \item It remains to deal with \ref{item:partial:respect:False}. For a subtree $S\in \{\subtree(u),\subtree[u]\}$, we denote by $P_{|_{S}}$ the tuple $(\interp{ \hat f}_{|_{S}}, \interp{\hat g }, \hat \bE_{|_{S}}, (\pdeg_{\vec \cW})_{|_{S}})$. 
            As $P$ satisfies this condition for $J^+$, we know that
            \begin{equation}\label{eq:forget:disjointunion:subtree(u)}
                \FalseFunc(P_{|_{\subtree(u)}}) \subseteq \False.
            \end{equation}
            We want to prove that $\FalseFunc(P_{|_{\subtree[u]}}) \subseteq \False.$
            Observe that, for each $S\in \{\subtree(u),\subtree[u]\}$, the set 
            $\FalseFunc(P_{|_{S}})$ is a disjoint union of $\FalseFunc^V(P_{|_{S}})$ and $\FalseFunc^E(\interp{\hat g })$.
            From \eqref{eq:forget:disjointunion:subtree(u)}, we know that $\FalseFunc^E(\interp{\hat g }) \subseteq \False$. Thus, it remains to prove that $\FalseFunc^V(P_{|_{\subtree[u]}}) \subseteq \False$.
            As $\subtree[u] = \subtree(u) \cup \{u\}$, by \Cref{claim:union-false-set}, we have
            \begin{equation}\label{eq:forget:disjointunion:V}
                \FalseFunc^V(P_{|_{\subtree[u]}}) = \FalseFunc^V(P_{|_{\subtree(u)}}) \cup \FalseFunc^V(\interp{f}_u, \bE_u, \pdeg_u)
            \end{equation}
            because $\interp{f}_u, \bE_u$ and $\pdeg_u$ are the restrictions of $\interp{\hat f}, \hat\bE$ and $\pdeg_{\vec\cW}$ on $\{u\}$.
            The property $J\in \cJ$ implies the inclusion $\FalseFunc^V(\interp{f}_u, \bE_u, \pdeg_u) \subseteq \False$ by definition of $\cJ$.
            Together with \eqref{eq:forget:disjointunion:subtree(u)} %
            and \eqref{eq:forget:disjointunion:V} we deduce that $\FalseFunc^V(P_{|_{\subtree[u]}}) \subseteq \False$. Thus, $P$ satisfies \ref{item:partial:respect:False} for~$I$.
        \end{itemize}
        We conclude that $P\in \Partial_u(I)$ and this proves the present claim.
    \end{claimproof}

    Next we argue about the contribution of the vertex $u \in \subtree[u]$ to the exponents of the formal variables of our polynomials depending on $J \in \cJ$. 
    Formally, the following relations apply.
    \begin{claim}\label{claim:modelcheck:forget:coef}
        For every $J=(\interp{f}_u,\bE_u,\pdeg_u)\in \cJ$ and every $P\in \Partial_u[J,\cT_{\bE_u}]$, we have $a_{u}(P) = a_{u}[P] + a_{\bE_u}$, $b_{u}(P) = b_{u}[P]$, and $\barc_{u}(P) = \barc_{u}[P] + \barc_{\interp{f}_u}$.
    \end{claim}
    \begin{claimproof}
        Let $J=(\interp{f}_u,\bE_u,\pdeg_u)\in \cJ$ and let $P\in \Partial_u[J,\cT_{\bE_u}]$.
        We start by proving $a_{u}(P) = a_{u}[P] + a_{\bE_u}$. 
        By definition, we have
            \[
                a_{u}(P) = \sum_{C\in \Cat}\sum_{\substack{ v\in \subtree[u] \\ \hat \bE(v,C) < \dphi}} \hat\bE(v,C).
            \] 
            Recall that $a_{\bE_u}$ is the sum of $\inde_u(u,C)$ over all $C\in \Cat$ with $\bE_u(u,C)<\dphi$.
            Moreover, $a_{u}[P]$ is the sum of $\pse(v,C)$ over all $v\in \subtree(u)$ and $C\in \Cat$ with $\pse(v,C)<\dphi$.
            As $\subtree[u] = \subtree(u) \cup \{u\}$ and $\inde_u$ is the restriction of $\pse$ on $\{u\}$, we indeed have $a_{u}(P) = a_{u}[P] + a_{\bE_u}$.
            
            Now observe that $b_{u}(P) = b_{u}[P]$ because $\tss(u)$ and $\tss[u]$ have the same broom.
            
            It remains to prove that $\barc_{u}(P) = \barc_{u}[P] + \barc_{\interp{f}_u}$, i.e.,, for every $i\in [\ell]$, we have $c_{u}^i(P) = c^i_{u}[P] + c^i_{\interp{f}_u}$. Let $i\in [\ell]$.   
            From the definitions of $c^i_{u}(P)$ and $c^i_u[P]$, we have
            \begin{align*}
                c^i_{u}(P) = \abs{\ip{X_i}^{(G,\interp{\hat f}_{|\subtree[u]},\interp{\hat g})}} && c^i_{u}[P] = \abs{\ip{X_i}^{(G,\interp{\hat f}_{|\subtree(u)},\interp{\hat g})}}.
            \end{align*}
            If $X_i$ is an edge set variable, then $c^i_u(P) = c^i_u[P] = \abs{ \interp{\hat g}(X_i)}$ while $c^i_{\interp{f}_u} = 0$ holds so the claim follows.
            Assume now that $X_i$ is a vertex set variable.
            Then we have
            \[ c^i_u(P) = \abs{\interp{\hat{f}}(X_i)\cap \subtree[u]} = \abs{\interp{\hat{f}}(X_i) \cap \subtree(u)} + \abs{\interp{\hat{f}}(X_i) \cap \{u\}} = c^i_u[P] + \abs{\interp{\hat{f}}(X_i) \cap \{u\}}. \]
            Since $P\in \Partial[J,\cT_{\bE_u}]$, we know that $\interp{f}_u$ is the restriction of $\interp{\hat f}$ on $\{u\}$ by \ref{item:partial:fconsistent}.
            Thus we have $\interp{\hat{f}}(X_i) \cap \{u\} = \interp f_u(X_i)$.
            By definition, we have $c^i_{\interp{f}_u} = \abs{\interp f_u(X_i)}$.
            We conclude that $c^i_u(P) = c^i_u[P] + c^i_{\interp{f}_u}$ holds as claimed.
    \end{claimproof}

    We are ready to prove \eqref{eq:forget:cT}.
    By definition, we have 
    \begin{align*}
        \cP_u(I) = & \sum_{a,b\in \bN, \barc \in \bN^\ell} \abs{\Partial_u(I,a,b,\barc)} \cdot \alpha^a \cdot \beta^b \cdot \gamma^{\barc} = \sum_{P \in \Partial_u(I)} \alpha^{a_u(P)}\cdot \beta^{b_u(P)} \cdot \gamma^{\barc_u(P)}\\
        \stackrel{\ref{claim:modelcheck:forget:disjointunion}}{=} & \sum_{\substack{J\in \cJ \\ P \in \Partial_u[J,\cT_{\bE_u}]}} \alpha^{a_u(P)}\cdot \beta^{b_u(P)} \cdot \gamma^{\barc_u(P)}
        \stackrel{\ref{claim:modelcheck:forget:coef}}{=} \sum_{\substack{J=(\interp{f}_u,\bE_u,\pdeg_u)\in \cJ \\ P \in \Partial_u[J,\cT_{\bE_u}]}} \alpha^{a_u[P] + a_{\bE_u}}\cdot \beta^{b_u[P]} \cdot \gamma^{\barc_u[P] + \barc_{\interp{f}_u}}\\
        = & \sum_{\substack{J= (\interp{f}_u,\bE_u,\pdeg_u)\in \cJ\\ a,b\in \bN, \barc\in \bN^\ell}} \abs{\Partial_u[J,\cT_{\bE_u},a,b,\barc]} \cdot \alpha^{a + a_{\bE_u}}\cdot \beta^{b} \cdot \gamma^{\barc + \barc_{\interp{f}_u}}.
    \end{align*}
    Observe that for $J= (\interp{f}_u,\bE_u,\pdeg_u)\in \cJ$, $a,b \in \bN$, and $\barc \in \bN^\ell$, from the definition of $\Partial_u[J,\cT_{\bE_u},a,b,\barc]$, we have:
    \[
        \Partial_u[J,\cT_{\bE_u},a,b,\barc] = \bigcap_{T\in\cT_{\bE_u}} \Partial_u[J,\{T\},a,b,\barc].
    \]
    For every $\cT\subseteq \cT_{\bE_u}$, we define the set 
    \begin{align*}
        \barpart_u[J,\cT,a,b,\barc] \defeq \{&(\interp{\hat f}, \interp{\hat g}, \hat\bE, \vec\cW) \in \Partial_u[J,\emptyset,a,b,\barc] \mid \\ 
        &\text{for every $(C,j)\in\cT$, there is no arc in $\vec\cW(C)$ incoming into $u^j$}\}
    \end{align*}
    Observe that
    \[
        \barpart_u[J,\cT,a,b,\barc] = \Partial_u[J,\emptyset,a,b,\barc] \setminus \left(\bigcup_{T\in \cT} \Partial_u[J,\{T\},a,b,\barc]\right).
    \]
    Furthermore, we have $\Partial_u[J,\cT,a,b,\barc] \subseteq \Partial_u[J,\emptyset,a,b,\barc]$ for every $\cT \subseteq \cT_{\inde_u}$.
    By the the inclusion-exclusion principle (\cref{thm:inclusion-exclusion-intersection}), we then deduce that
    \begin{equation}\label{eq:forget:inclusion:exclusion}
    \abs{\Partial_u[J,\cT_{\bE_u},a,b,\barc]} = \sum_{\cT\subseteq \cT_{\bE_u}} (-1)^{\abs{\cT}} \abs{\barpart_u[J,\cT,a,b,\barc]}.
    \end{equation}
    We will prove the non-DFT-transformed equality \eqref{eq:forget-node-non-transformed} with \eqref{eq:forget:cT}, \eqref{eq:forget:inclusion:exclusion}, and the following claim.
    
    \begin{claim}\label{claim:modelcheck:forget:cT}
        Let $J=(\interp{f}_u,\bE_u, \pdeg_u)\in \cJ$, let $a,b\in\bN$, and let $\barc\in\bN^\ell$. For every $\cT\subseteq \cT_{\bE_u}$, there exists a neighborhood function $\allow_{\cT}$ on $\{u\}$ with $(\interp{f}_u,\bE_u, \allow_{\cT},\pdeg_u)\in \cJ'$ such that
        \begin{align}
            \abs{\barpart_u[J,\cT,a,b,\barc]} & = \abs{\Partial_u[(\interp{f}_u^+, \bE_u^+,\allow\cup \allow_{\cT}, \pdeg_u^+,\False),a,b,\barc]} \text{ and } \label{eq:forget:barpart}\\
            (-1)^{\abs{\cT}} & = \prod_{C\in\Cat} (-1)^{\bE_u(u,C) - \allow_{\cT}(u,C)}. \label{eq:forget:-1}
        \end{align}
        Furthermore, for every $(\interp{f}_u,\bE_u, \allow_{u},\pdeg_u)\in \cJ'$, we have
        \begin{equation}\label{eq:number-of-cT'-with-the-same-allow}
            |\{\cT\subseteq \cT_{\bE_u} \mid \allow_u = \allow_{\cT}\}| = \prod_{C \in \Cat} \binom{\bE_u(u, C)}{\allow_u(u, C)}.
        \end{equation}
    \end{claim}
    \def\rmOK{\mathsf{Alw}}
    \begin{claimproof}
        First, we fix an arbitrary $\cT \subseteq \cT_{\inde_u}$ and construct $\allow_{\cT}$ as follows.
        For every $C\in \Cat$, we define the set
        \[
            \rmOK(C)\defeq \{i\in [\bE_u(u,C)] \mid (C,i)\notin \cT\}.
        \]
        And we fix a bijection $h_C$ between $\rmOK(C)$ and 
        $[\abs{\rmOK(C)}]$.
        We define $\allow_{\cT}$ as the neighborhood function on $\{u\}$ such that for every $C\in\Cat$, we have 
        \[
            \allow_{\cT}(u,C) \defeq \abs{\rmOK(C)}.
        \] 
        Observe that $\allow_{\cT}$ is indeed a neighborhood function because by construction we have $\allow_{\cT}(u,C)\in %
        [\bE_u(u,C)] \subseteq [\dphi]$ for every $C \in \Cat$.
        
        Let $J_{\cT}^+ \defeq (\interp{f}_u^+, \bE_u^+,\allow\cup\allow_{\cT}, \pdeg_u^+,\False)$. It is easy to see that $(\interp{f}_u,\bE_u, \allow_{\cT},\pdeg_u) \in \cJ'$ because by construction we have $\allow_{\cT}\leq \bE_u$.
        This property also implies that we have $J_{\cT}^+ \in \Index[u]$.
        Next we provide a bijection $h$ between $\barpart_u[J,\cT,a,b,\barc]$ and $\Partial_u[J_{\cT}^+,a,b,\barc]$---from this \eqref{eq:forget:barpart} follows immediately.
        For this, for every $C \in \Cat$ we define the bijection $g_C$ between %
        the set
        \begin{equation}\label{eq:bijection-g-c-left}
            \{(z^1, w^j) \in \vec A(\vec G_\phi) \mid z \neq w\} \cup \{(z^1, u^j) \in \vec A(\vec G_\phi) \mid z \neq u, j \in \rmOK(C)\}
        \end{equation}
        and the set
        \begin{equation}\label{eq:bijection-g-c-right}
            \{(z^1, w^j) \in \vec A(\vec G_\phi) \mid z \neq w\} \cup \{(z^1, u^j) \in \vec A(\vec G_\phi) \mid z \neq u, j \in [|\rmOK(C)|]\}
        \end{equation}
        via 
        \[
            g_C((z^1, w^j)) =
            \begin{cases}
                (z^1, w^j) & \text{if } w \neq u \\
                (z^1,u^{h_C(j)}) & \text{otherwise}.
            \end{cases}
        \]
        Let $h$ be the function that maps every $P=(\interp{\hat f},\interp{\hat g},\hat\bE, \vec\cW)\in \barpart_u[J,\cT,a,b,\barc]$ to $h(P)\defeq(\interp{\hat f},\interp{\hat g},\hat\bE, \vec\cZ)$ where
        \[
            \vec\cZ(C) \defeq \{ g_C(a) \mid a \in \psw(C)\}
        \]
        for every $C \in \Cat$.
        This is well-defined because for every $C \in \Cat$ and every arc $(v^1, u^j) \in \psw(C)$, we have $j\in \rmOK(C)$: indeed, by definition of $\barpart_u[J,\cT,a,b,\barc]$, we know that $(C,j)\notin \cT$ and by \ref{item:partial:respect:broom:hatbE} and \ref{item:partial:econsistent} we have $j\in [\pse(u,C)] = [\inde_u(u,C)]$.
        In simple words, for every $C \in \Cat$, the function $h$ reorders the arcs in $\vec\cW(C)$ incoming into $\{u^j \mid j \in \rmOK(C)\}$
        in such a way that they keep their tails and their heads are now in $\{u^j \mid 1 \leq j \leq \abs{\rmOK(C)} = \allow_\cT(u,C)\}$.
        In particular, $g_C$ is a bijection between $\psw(C)$ and $\vec\cZ(C)$.
        
        We first prove %
        that (1) $h$ maps every element of $\barpart_u[J,\cT,a,b,\barc]$ to an element in $\Partial_u[J_{\cT}^+,a,b,\barc]$, and (2) every element in $\Partial_u[J_{\cT}^+,a,b,\barc]$ has a preimage  under $h$ in $\barpart_u[J,\cT,a,b,\barc]$. 
        Later we show that the bijectivity easily follows from these two properties.
        
        For (1), let $P=(\interp{\hat f},\interp{\hat g},\hat\bE, \vec\cW)\in \barpart_u[J,\cT,a,b,\barc]$ and let $h(P)=(\interp{\hat f},\interp{\hat g},\hat\bE, \vec\cZ)$.
        We need to check that $h(P) \in \Partial_u[J_{\cT}^+,a,b,\barc]$.
        First observe that $h(P)$ is a partial solution for $\tss[u]$. This is clear from the construction of $\vec\cZ$ and the fact that $g_C$ is a bijection between $\psw(C)$ and $\vec\cZ(C)$ for every $C\in \Cat$.
        In particular, \ref{item:partial:respect:broom:hatbE} holds because $[\abs{\rmOK(C)}]\subseteq [\inde_u(u,C)] = [\pse(u,C)]$ (due to \ref{item:partial:econsistent}) for every $C\in \Cat$.
        Now we show that $h(P)$ is compatible with $J_{\cT}^+$.
        Since $P\in \barpart_u[J,\cT,a,b,\barc] \subseteq \Partial_u[\interp{f}_u^+, \bE_u^+,\allow\cup\bE_u, \pdeg_u^+,\False,a,b,\barc]$, we only need to check the conditions involving $\vec\cZ$ and $\allow_\cT$, that is: \ref{item:partial:pdegconsistent}, \ref{item:partial:respect:allow}, and \ref{item:partial:respect:False}.
        By construction of $\vec\cZ$, for every $v\in \broom[u]$ and $C\in \Cat$, we have 
        \[
            \pdeg_u^+(v,C) = \pdeg_{\psw}(v,C) = \abs{ \{ (w^1, v^j) \in \psw(C) \} } = \abs{ \{ (w^1, v^j) \in \vec\cZ(C) \} } = \pdeg_{\vec\cZ}(v,C),
        \]
        where the second equality holds because $g_C$ is a bijection between $\psw(C)$ and $\vec\cZ(C)$ that maps, for every vertex $w$ of $G$, all arcs incoming into $w^1, \dots, w^{\dphi}$ to arcs incoming into vertices $w^1, \dots, w^{\dphi}$. %
        Thus, $h(P)$ satisfies \ref{item:partial:pdegconsistent} and \ref{item:partial:respect:False}.
        By construction, for every $C\in\Cat$ and every arc $(v^1,w^j)\in \vec\cZ(C)$ with $w\neq u$, we have $(v^1,w^j)\in\psw(C)$. The tuple $P$ satisfies \ref{item:partial:respect:allow} on $\tail(u)$ with the function $\allow\cup\bE_u$ so $h(P)$ satisfies \ref{item:partial:respect:allow} on $\tail(u)$ with the function $\allow\cup\allow_{\cT}$.
        Moreover, for every $C\in \Cat$ and every arc $(v^1, u^j) \in \vec\cZ(C)$, we have $j\leq |\rmOK(C)|= \allow_\cT(u,C)$ by construction. %
        Clearly, $a_u[P]=a_u[h(P)]$ and $\barc_u[P] = \barc_u[h(P)]$. Moreover, since the number of arcs incoming into $v^1,\dots,v^{\dphi}$ is the same in $\psw(C)$ and $\vec\cZ(C)$ for every $C\in\Cat$ and every $v \in V(G)$, we have $b_u[P]=b_u[h(P)]$.
        We deduce that $h(P) \in \Partial_u[J_{\cT}^+,a,b,\barc]$.

        \medskip
        
        For (2), let $P'=(\psf,\psg,\pse,\vec\cZ)\in \Partial_u[J_{\cT}^+,a,b,\barc]$. Let $\psw$ be the function that maps every $C\in\Cat$ to the set $\{g_C^{-1}(a) \mid a \in \vec\cZ(C)\}$.
        This is well-defined: because for every $(v^1,u^j)\in \vec\cZ(C)$, we have $j\leq \allow_\cT(u,C) = [\rmOK(C)]$ due to \ref{item:partial:respect:allow}, the vertex $g_C^{-1}((v^1,u^j)) = (v^1, u^{h_C^{-1}(j)})$ indeed exists. 
        Furthermore, $g_C$ is a bijection between $\psw(C)$ and $\vec\cZ(C)$.
        We claim that $P\defeq(\psf,\psg,\pse,\psw) \in \barpart_u[J,\cT,a,b,\barc]$. 
        With the same arguments used above, we can prove that $P$ is a partial solution for $\tss[u]$, we have $a_u[P']=a_u[P]$, $b_u[P']=b_u[P]$, $\barc_u[P']=\barc_u[P]$, and $P$ satisfies all the conditions to be compatible with $(\interp{f}_u^+, \bE_u^+,\allow\cup\bE_u, \pdeg_u^+,\False)$ except for \ref{item:partial:respect:allow}.
        This latter condition is satisfied by $P$ for $\tail(u)$ because %
        $P'$ satisfies \ref{item:partial:respect:allow} for $\tail(u)$ and the function $\allow\cup \allow_\cT$, and for $\{u\}$ it is satisfied because $P'$ satisfies \ref{item:partial:respect:broom:hatbE} and it satisfies \ref{item:partial:econsistent} for $\bE_u$. %
        We deduce that $P\in \Partial_u[\interp{f}_u^+, \bE_u^+,\allow\cup\bE_u, \pdeg_u^+,\False,a,b,\barc]$.
        Furthermore, observe that by construction, for every $C\in\Cat$ and every arc of form $(v^1,u^i)\in\psw(C)$, we have $i\in \rmOK(C)$, i.e., $(C, i) \notin \cT$. It follows that 
        $P\in \barpart_u[J,\cT,a,b,\barc]$, and, by the construction of $\psw$, it holds that $h(P)=P'$.

        Finally, recall that $g_C$ is a bijection between the sets \eqref{eq:bijection-g-c-left} and \eqref{eq:bijection-g-c-right} for every $C \in \Cat$.
        Thus $h$ is injective, i.e., every partial solution in $\Partial_u[J_{\cT}^+,a,b,\barc]$ has at most one preimage in $h$.
        Altogether, we obtain that $h$ is a bijection between $\barpart_u[J,\cT,a,b,\barc]$ and $\Partial_u[J_{\cT}^+,a,b,\barc]$.
        We conclude that \eqref{eq:forget:barpart} is true.

        \medskip

        Observe that \eqref{eq:forget:-1} simply follows from the following equation:
        \[
            \abs{\cT} = \sum_{C\in\Cat}\abs{\bigl\{(C,i) \in \cT \mid i \in [\bE_u(u, C)]\bigr\}} = \sum_{C\in\Cat} \bE_u(u,C) - \abs{\rmOK(C)} = \sum_{C\in\Cat} \bE_u(u,C) - \allow_{\cT}(u,C).
        \] 
        
        Finally, to prove \eqref{eq:number-of-cT'-with-the-same-allow}, 
        for every $C \in \Cat$ we define the set $\cT(C) = \{(C, i) \mid (C, i) \in \cT\}$.
        Now observe that for every neighborhood function $\allow_u$ on $\{u\}$ with $\allow_u\leq \bE_u$ and every $\cT\subseteq \cT_{\bE_u}$, we have $\allow_u = \allow_{\cT}$ if and only if $\abs{\cT(C)} = \bE_u(u,C) - \allow_u(u,C)$ for every $C\in \Cat$.
        Thus for every neighborhood function $\allow_u$ on $\{u\}$ satisfying $\allow_u\leq \bE_u$, there are $\binom{\bE_u(u, C)}{ \bE_u(u,C) - \allow_u(u, C)} = \binom{\bE_u(u, C)}{ \allow_u(u, C)}$ possibilities for each $\cT(C)$ so that $\allow_u = \allow_{\cT}$ holds.
    \end{claimproof}

    In the following, for every $J'=(\interp{f}_u,\bE_u, \allow_u, \pdeg_u)\in \cJ'$, we define the tuple $$J'^+ \defeq (\interp{f}_u^+,\bE_u^+, \allow_u^+, \pdeg_u^+, \False).$$
    Observe that we then have $J'^+ \in \Index[u]$.

     Let again $J= (\interp{f}_u,\bE_u,\pdeg_u)\in \cJ$, $a,b \in \bN$, and $\barc \in \bN^\ell$.
     By \Cref{claim:modelcheck:forget:cT}, we can rewrite $\abs{\Partial_u[J,\cT_{\bE_u},a,b,\barc]}$ as follows.
    \begin{align}
          &\abs{\Partial_u[J,\cT_{\bE_u},a,b,\barc]} \stackrel{\eqref{eq:forget:inclusion:exclusion}}{=}
           \sum_{\cT\subseteq \cT_{\bE_u}} (-1)^{\abs{\cT}} \abs{\barpart_u[J,\cT,a,b,\barc]} \\
         &\stackrel{\eqref{eq:forget:barpart}}{=} \sum_{\substack{\allow_u : \\
         J'=(\interp{f}_u,\bE_u, \allow_u, \pdeg_u)\in \cJ'}} \sum_{\substack{\cT\subseteq \cT_{\bE_u} \\ \allow_u = \allow_\cT}} (-1)^{\abs{\cT}} \abs{\Partial_u[J'^+,a,b,\barc]} \nonumber\\
         & \stackrel{\eqref{eq:forget:-1}}{=}  \sum_{\substack{\allow_u : \\J'=(\interp{f}_u,\bE_u, \allow_u, \pdeg_u)\in \cJ'}} \sum_{\substack{\cT\subseteq \cT_{\bE_u} \\ \allow_u = \allow_\cT}} \left( \prod_{C\in\Cat}(-1)^{\bE_u(u,C)-\allow_u(u,C)} \right) \cdot \abs{\Partial_u[J'^+,a,b,\barc]} \nonumber\\
         & \stackrel{\eqref{eq:number-of-cT'-with-the-same-allow}}{=} \sum_{\substack{\allow_u : \\J'=(\interp{f}_u,\bE_u, \allow_u, \pdeg_u)\in \cJ'}} \left( \prod_{C\in\Cat}(-1)^{\bE_u(u,C)-\allow_u(u,C)} \cdot \binom{\bE_u(u, C)}{ \allow_u(u, C)} \right) \cdot  \abs{\Partial_u[J'^+,a,b,\barc]} \nonumber\\
         & = \sum_{\substack{\allow_u : \\J'=(\interp{f}_u,\bE_u, \allow_u, \pdeg_u)\in \cJ'}} \incexc(\bE_u,\allow_u)  \cdot  \abs{\Partial_u[J'^+,a,b,\barc]} \label{eq:forget:alwu}
    \end{align}
    
    We are now ready to prove \eqref{eq:forget-node-non-transformed}:
    \begin{align*}
        \cP_u(I)  \stackrel{\eqref{eq:forget:cT}}{=} &\sum_{\substack{J =   (\interp{f}_u,\bE_u,\pdeg_u)\in \cJ\\ a,b\in \bN, \barc\in \bN^\ell}} \abs{\Partial_u[J,\cT_{\bE_u},a,b,\barc]} \cdot \alpha^{a + a_{\bE_u}}\cdot \beta^{b} \cdot \gamma^{\barc + \barc_{\interp{f}_u}} \\
          \stackrel{\eqref{eq:forget:alwu}}{=} 
        &\sum_{\substack{(\interp{f}_u,\bE_u,\pdeg_u)\in \cJ\\ a,b\in \bN, \barc\in \bN^\ell}} 
        \sum_{\substack{\allow_u : \\J'=(\interp{f}_u,\bE_u, \allow_u, \pdeg_u)\in \cJ'}} \incexc(\bE_u,\allow_u)  \cdot  \abs{\Partial_u[J'^+,a,b,\barc]} 
        \cdot \alpha^{a + a_{\bE_u}}\cdot \beta^{b} \cdot \gamma^{\barc + \barc_{\interp{f}_u}}\\
         =& \sum_{(\interp f_u,\bE_u,\allow_u,\pdeg_u)\in \cJ'} \incexc(\bE_u,\allow_u) \cdot  \cP_{u}[\interp f_u^+, \bE_u^+, \allow_u^+,\pdeg_u^+, \False] \cdot \alpha^{a_{\bE_u}} \cdot \gamma^{\barc_{\interp f_u}}.
    \end{align*}
    
    Now we will use this result to prove the ``DFT-transformed'' equality~\eqref{eq:forget-node-transformed}.
    In the following, for $J'=(\interp f_u,\bE_u,\allow_u,\pdeg_u)\in \cJ'$, we denote by $M(J')$ the product $\incexc(\bE_u,\allow_u) \cdot \alpha^{a_{\bE_u}} \cdot \gamma^{\barc_{\interp f_u}}$.
    We remind the reader that we aim at proving that
    \begin{align*}
        &\cQ^p_{u}(I) \equiv_p \sum_{J'=(\interp f_u,\bE_u,\allow_u,\pdeg_u)\in \cJ'}  M(J') \cdot  
        \frac{1}{\pphi^{|\Cat|}}\sum_{\hat q \in \pdegdom(\{u\})} \omega(p)^{-\hat q \cdot \pdeg_u} \cdot \cQ^p_{u}[\interp f_u^+, \bE_u^+, \allow_u^+, \pdeg \cup \hat q, \False] 
    \end{align*}
    for every prime number $p$ such that $\FF^*_p$ admits a $\pphi$-th root of unity. For the remainder of the proof, we fix such an arbitrary prime number $p$.
    In the following $\omega$ denotes $\omega(p)$.
    We have:
    \begin{align*}
        \cQ^p_{u}(I) & \stackrel{\eqref{eq:def-q}}{\equiv_p} \sum_{q \in \pdegdom(\tail(u))} \omega^{q \cdot \pdeg} \cdot \cP_u(\interp f,\bE,\allow,q,\False)\\
        & \stackrel{\eqref{eq:forget-node-non-transformed}}{\equiv_p} 
        \sum_{q \in \pdegdom(\tail(u))} \omega^{q \cdot \pdeg} \cdot \Bigl(\sum_{J'=(\interp f_u,\bE_u,\allow_u,\pdeg_u)\in \cJ'} M(J') \cdot 
        \cP_{u}[\interp f_u^+, \bE_u^+, \allow_u^+, q \cup \pdeg_u, \False] \Bigr) \\
        &\equiv_p \sum_{J'=(\interp f_u,\bE_u,\allow_u,\pdeg_u)\in \cJ'} M(J') \cdot \sum_{q \in \pdegdom(\tail(u))} \Bigl(\omega^{q \cdot \pdeg} \cdot \cP_{u}[\interp f_u^+, \bE_u^+, \allow_u^+, q \cup \pdeg_u, \False] \Bigr)
    \end{align*}

    Now to prove \eqref{eq:forget-node-transformed}, it suffices to show that for each $(\interp f_u,\bE_u,\allow_u,\pdeg_u)\in \cJ'$, we have
    \begin{equation}
        \sum_{q \in \pdegdom(\tail(u))} \omega^{q \cdot \pdeg} \cdot \cP_{u}[\interp f_u^+, \bE_u^+, \allow_u^+, q \cup \pdeg_u, \False] \equiv_p \frac{1}{\pphi^{|\Cat|}}\sum_{\hat q \in \pdegdom(\{u\})} \omega^{-\hat q \cdot \pdeg_u} \cdot \cQ^p_{u}[\interp f_u^+, \bE_u^+, \allow_u^+, \pdeg \cup \hat q, \False]. \label{eq:helpful-eq-for-dft}
    \end{equation}
    So fix an arbitrary $(\interp f_u,\bE_u,\allow_u,\pdeg_u)\in \cJ'$.
    To show \eqref{eq:helpful-eq-for-dft}, we start by rewriting the right-hand side of this equality which we denote by $\RHS$.
    We rely on the fact that $\tail[u] = \tail(u) \cup \{u\}$ implies that 
    \begin{equation}
        \pdegdom(\tail[u]) = \{\tilde q \mid \tilde q \colon \tail[u] \times \Cat \to \Fp\} 
        = \{q \cup \tilde q_u \mid q \in \pdegdom(\tail(u)), \tilde q_u \in \pdegdom(\{u\})\}. \label{eq:rewriting-pdegdom}
    \end{equation}
    Then we have:
    \begin{align*}
        \mathsf{RHS} \equiv_p &\frac{1}{\pphi^{|\Cat|}}\sum_{\hat q \in \pdegdom(\{u\})} \omega^{-\hat q \cdot \pdeg_u} \cdot \cQ^p_{u}[\interp f_u^+, \bE_u^+, \allow_u^+, \pdeg \cup \hat q, \False] \\
        \stackrel{\eqref{eq:def-q}}{\equiv_p}  &\frac{1}{\pphi^{|\Cat|}}\sum_{\hat q \in \pdegdom(\{u\})} \omega^{-\hat q \cdot \pdeg_u} \sum_{\tilde q \in \pdegdom(\tail[u])} \omega^{\tilde q \cdot (\pdeg \cup \hat q)} \cdot \cP_u[\interp f_u^+, \bE_u^+, \allow_u^+, \tilde q, \False] \\
        \stackrel{\eqref{eq:rewriting-pdegdom}}{\equiv_p} 
        &\frac{1}{\pphi^{|\Cat|}}\sum_{\hat q \in \pdegdom(\{u\})} \omega^{-\hat q \cdot \pdeg_u} \sum_{\substack{\tilde q_u \in \pdegdom(\{u\}), \\ q \in \pdegdom(\tail(u))}} \omega^{(q \cup \tilde q_u) \cdot (\pdeg \cup \hat q)} \cdot \cP_u[\interp f_u^+, \bE_u^+, \allow_u^+, q \cup \tilde q_u, \False] \\\equiv_p 
        &\frac{1}{\pphi^{|\Cat|}}\sum_{\hat q \in \pdegdom(\{u\})} \omega^{-\hat q \cdot \pdeg_u} \sum_{\substack{\tilde q_u \in \pdegdom(\{u\}), \\ q \in \pdegdom(\tail(u))}} \omega^{q \cdot \pdeg} \cdot \omega^{\tilde q_u \cdot \hat q} \cdot \cP_u[\interp f_u^+, \bE_u^+, \allow_u^+, q \cup \tilde q_u, \False]\\ \equiv_p 
        &\frac{1}{\pphi^{|\Cat|}}\sum_{\hat q \in \pdegdom(\{u\})} \omega^{-\hat q \cdot \pdeg_u} \sum_{\tilde q_u \in \pdegdom(\{u\})} \omega^{\tilde q_u \cdot \hat q} \sum_{q \in \pdegdom(\tail(u))} \omega^{q \cdot \pdeg} \cdot \cP_u[\interp f_u^+, \bE_u^+, \allow_u^+, q \cup \tilde q_u, \False]
    \end{align*}
    Now we define the function $\Delta \colon \pdegdom(\{u\}) \to \FF_{p}[\alpha, \beta, \gamma_1, \dots, \gamma_\ell]$
    via
    \[
        \Delta(\tilde q_u) \defeq \sum_{q \in \pdegdom(\tail(u))} \omega^{q \cdot \pdeg} \cdot \cP_u[\interp f_u^+, \bE_u^+, \allow_u^+, q \cup \tilde q_u, \False] \pmod{p}.
    \]
    
    By inserting this definition in the above equality we obtain
    \begin{align*}
        \mathsf{RHS} \equiv_p &
        \frac{1}{\pphi^{|\Cat|}}\sum_{\hat q \in \pdegdom(\{u\})} \omega^{-\hat q \cdot \pdeg_u} \sum_{\tilde q_u \in \pdegdom(\{u\})} \omega^{\tilde q_u \cdot \hat q} \cdot \Delta(\tilde q_u) \\
        \equiv_p 
        &\frac{1}{\pphi^{|\Cat|}}\sum_{\hat q \in \pdegdom(\{u\})} \omega^{-\hat q \cdot \pdeg_u} \cdot \bigl(\DFT_{\{u\} \times \Cat, \pphi, p} (\Delta)\bigr)(\hat q)\\
        \equiv_p 
        &\frac{1}{\pphi^{|\{u\} \times \Cat|}}\sum_{\hat q \in \pdegdom(\{u\})} \omega^{-\hat q \cdot \pdeg_u} \cdot \bigl(\DFT_{\{u\} \times \Cat, \pphi, p} (\Delta)\bigr)(\hat q) \\
        \equiv_p 
        &\Bigl(\DFT^{-1}_{\{u\} \times \Cat, \pphi, p}\bigl(\DFT_{\{u\} \times \Cat, \pphi, p}(\Delta)\bigr)\Bigr)(\pdeg_u)\\
        \stackrel{\ref{thm:inverse-dft-polynomials}}{\equiv_p} & \Delta(\pdeg_u)
        \equiv_p 
        \sum_{q \in \pdegdom(\tail(u))} \omega^{q \cdot \pdeg} \cdot \cP_u[\interp f_u^+, \bE_u^+, \allow_u^+, q \cup \pdeg_u, \False]
    \end{align*}
    proving the desired equality \eqref{eq:helpful-eq-for-dft} and thus concluding the proof.
\end{proof}

To simplify the analysis of the running time later in \cref{sec:modelcheck}, we rewrite the transformed equality \eqref{eq:forget-node-transformed-non-optimized}  from \cref{lem:modelcheck:forget-non-optimized} as follows.
\begin{lemma}\label{lem:modelcheck:forget}
    Let $u$ be a node of $T$ and let the set $\cS$ be defined as follows:
    \begin{align*}
        \cS = \bigl\{(\indf_u, \inde_u, \allow_u, \hat{q}) \mid & \indf_u \text{ is a vertex interpretation on } \{u\}, \\
        &\inde_u \text{ and } \allow_u \text{ are neighborhood functions on $\{u\}$,} \\
        &\hat{q} \text{ is a mod-degree on } \{u\}\bigr\}.
    \end{align*}
    For the notation from \cref{lem:modelcheck:forget-non-optimized}, it holds that:
    \begin{align}
        \cQ^p_{u}(I) \equiv_p & \frac{1}{\pphi^{|\Cat|}} \sum_{(\interp f_u,\bE_u,\allow_u,\hat{q})\in \cS} \incexc(\bE_u,\allow_u) \cdot \cQ^p_{u}[\interp f_u^+, \bE_u^+, \allow_u^+, \pdeg \cup \hat q, \False] \cdot \alpha^{a_{\bE_u}} \cdot \gamma^{\barc_{\interp f_u}} \nonumber \\
        & \cdot \Bigl(\sum_{\substack{\pdeg_u \in \pdegdom(\{u\}) \colon \\(\indf_u,\bE_u,\allow_u,\pdeg_u)\in \cJ'}} \omega(p)^{-\hat{q} \cdot \pdeg_u}\Bigr). \label{eq:forget-node-transformed} 
    \end{align}
\end{lemma}
\begin{proof}
    By \eqref{eq:forget-node-transformed-non-optimized} we have 
    \begin{align*}
        \cQ^p_{u}(I) \equiv_p & \sum_{(\interp f_u,\bE_u,\allow_u,\pdeg_u)\in \cJ'} \incexc(\bE_u,\allow_u) \cdot \alpha^{a_{\bE_u}} \cdot \gamma^{\barc_{\interp f_u}} \nonumber \\
        \cdot & \frac{1}{\pphi^{|\Cat|}}\sum_{\hat q \in \pdegdom(\{u\})} \omega(p)^{-\hat q \cdot \pdeg_u} \cdot \cQ^p_{u}[\interp f_u^+, \bE_u^+, \allow_u^+, \pdeg \cup \hat q, \False] \\
    \end{align*}
    By merely rewriting this equality we get
    \begin{align*}
        \cQ^p_{u}(I) \equiv_p & \frac{1}{\pphi^{|\Cat|}} \cdot \sum_{(\interp f_u,\bE_u,\allow_u,\pdeg_u)\in \cJ'} \sum_{\hat q \in \pdegdom(\{u\})} \omega(p)^{-\hat q \cdot \pdeg_u} \cdot \incexc(\bE_u,\allow_u)  \\ 
        \cdot & \cQ^p_{u}[\interp f_u^+, \bE_u^+, \allow_u^+, \pdeg \cup \hat q, \False] \cdot \alpha^{a_{\bE_u}} \cdot \gamma^{\barc_{\interp f_u}} \\
        \equiv_p &\frac{1}{\pphi^{|\Cat|}} \sum_{(\interp f_u,\bE_u,\allow_u,\hat{q})\in \cS} \sum_{\substack{\pdeg_u \in \pdegdom(\{u\}) \colon \\(\indf_u,\bE_u,\allow_u,\pdeg_u)\in \cJ'}} \omega(p)^{-\hat q \cdot \pdeg_u} \cdot \incexc(\bE_u,\allow_u) \\
        \cdot & \cQ^p_{u}[\interp f_u^+, \bE_u^+, \allow_u^+, \pdeg \cup \hat q, \False] \cdot \alpha^{a_{\bE_u}} \cdot \gamma^{\barc_{\interp f_u}} \\
        \equiv_p &\frac{1}{\pphi^{|\Cat|}} \sum_{(\interp f_u,\bE_u,\allow_u,\hat{q})\in \cS} \incexc(\bE_u,\allow_u) \cdot \cQ^p_{u}[\interp f_u^+, \bE_u^+, \allow_u^+, \pdeg \cup \hat q, \False] \cdot \alpha^{a_{\bE_u}} \cdot \gamma^{\barc_{\interp f_u}} \\ \cdot &\Bigl( \sum_{\substack{\pdeg_u \in \pdegdom(\{u\}) \colon \\(\indf_u,\bE_u,\allow_u,\pdeg_u)\in \cJ'}} \omega(p)^{-\hat q \cdot \pdeg_u} \Bigr).
    \end{align*}
    This proves the lemma.
\end{proof}

\subsection{Join equality}\label{subsec:join}
Finally, we prove the equalities relating the partial solutions of $\tss[u]$ for an internal node $u$ and the partial solutions of $\tss(v_1), \dots, \tss(v_t)$ for the children $v_1, \dots, v_t$ of $u$.
The sets $\subtree[v_1]$, $\dots$, $\subtree[v_t]$ partition $\subtree(u)$.
So essentially, we will show that the partial solutions of $\tss[u]$ compatible with some index are given by all combinations of the partial solutions of $\tss(v_1), \dots, \tss(v_t)$ compatible with almost the same index.
In fact, the only thing that varies in the indexes
of $\tss(v_1), \dots, \tss(v_t)$ is the mod-degree.
Because a mod-degree in an index reflects the cardinalities (modulo $p$) of certain neighborhoods, we, of course, only consider the combinations of mod-degrees that sum up to the mod-degree in the target index (cf.\ \eqref{eq:p-join}).
Crucially, to avoid the time-consuming iteration of all such combinations, we will rely on the DFT-transformed polynomials $\cQ$ for which, as we will show using \cref{thm:fourier-convolution-result}, the mod-degree remains the same for each child $v_i$ (cf.\ \eqref{eq:q-join}).

To formalize this, for a node $u$, an index $I=(\interp f, \bE, \allow, \pdeg, \False)\in \Index[u]$, and a mod-degree $y \in \pdegdom(\tail[u])$ on $\tail[u]$, we define the index $I[\pdeg \mapsto y]\defeq (\interp f, \bE, \allow, y, \False) \in \Index[u]$.
Observe that due to $\tail[u] = \tail(v_i)$ we then also have $I[\pdeg \mapsto y] \in \Index(v_i)$ for every $i \in [t]$.
Then the following join-equalities apply.
\begin{lemma}\label{lem:modelcheck:join}
    Let $u$ be an internal node of $T$ with children $v_1,\dots,v_t$.
    Then we have 
    \begin{equation}\label{eq:p-join}
        \cP_u[I] = 
        \sum_{\substack{\pdeg_1, \dots, \pdeg_t \in \pdegdom(\tail[u]) \colon\\ \pdeg_1 + \dots + \pdeg_t \equiv_{\pphi} \pdeg}} \prod_{i=1}^t \cP_{v_i}%
        (I[\pdeg \mapsto \pdeg_i])
    \end{equation}
    and
    \begin{equation}\label{eq:q-join}
        \cQ^p_u[I] \equiv_p \prod_{i=1}^t \cQ^p_{v_i}(I) 
    \end{equation}
    for every $I=(\interp f, \bE, \allow, \pdeg, \False)\in \Index[u]$.
\end{lemma}

\begin{proof} 
    We start by proving the non-transformed equality \eqref{eq:p-join}.
    So fix an arbitrary index $I=(\interp f, \bE, \allow, \pdeg, \False)\in \Index[u]$.
    We start by rewriting both sides of the desired equality.
    For the right-hand side we have:
    \begin{align*}
        &\sum_{\substack{\pdeg_1, \dots, \pdeg_t \in \pdegdom(\tail[u]) \colon \\ \pdeg_1 + \dots + \pdeg_t \equiv_{\phi} \pdeg}} \prod_{i=1}^t \cP_{v_i}(I[\pdeg \mapsto \pdeg_i]) \\        
        = &\sum_{\substack{\pdeg_1, \dots, \pdeg_t \in \pdegdom(\tail[u]) \colon \\ \pdeg_1 + \dots + \pdeg_t \equiv_{\pphi} \pdeg}} \prod_{i=1}^t \sum_{a,b \in \bN, \barc \in \bN^\ell} |\Partial_{v_i}(I[\pdeg \mapsto \pdeg_i], a, b, \barc)| \cdot \alpha^a \cdot \beta^b \cdot \gamma^{\barc} \\
        = & \sum_{a,b \in \bN, \barc \in \bN^\ell} \alpha^{a} \cdot \beta^{b} \cdot \gamma^{\barc} \cdot \sum_{\substack{a_1 + \dots + a_t = a, \\ b_1 + \dots + b_t = b, \\ \barc_1 + \dots + \barc_t = \barc, \\ \pdeg_1, \dots, \pdeg_t \in \pdegdom(\tail[u]) \colon \\ \pdeg_1 + \dots + \pdeg_t \equiv_{\pphi} \pdeg}} \prod_{i=1}^t |\Partial_{v_i}(I[\pdeg \mapsto \pdeg_i], a_i, b_i, \barc_i)| \\
    \end{align*}
    where %
    $a_1 + \dots + a_t = a$ is a shortcut for $a_1, \dots, a_t \in \bN \colon \sum_{i=1}^t a_i = a$, also $b_1 + \dots + b_t = b$ is a shortcut for $b_1, \dots, b_t \in \bN \colon \sum_{i=1}^t b_i = b$, and $\barc_1 + \dots + \barc_t = \barc$ is a shortcut for $\barc_1, \dots, \barc_t \in \bN^\ell \colon \sum_{i=1}^t \barc_i = \barc$.
    And the left-hand side of \eqref{eq:p-join} can be rewritten as follows:
    \begin{align*}
        \cP_u[I] = \sum_{a,b\in\bN, \barc \in \bN^\ell}  \abs{\Partial_u[I,a,b,\barc]} \cdot \alpha^a \cdot \beta^b \cdot \gamma^{\barc}.
    \end{align*} 
    So to prove \eqref{eq:p-join}, it suffices to show that we have:
    \begin{equation}\label{eq:join-proof-rewritten}
    \abs{\Partial_u[I,a,b,\barc]} = \sum_{\substack{a_1 + \dots + a_t = a, \\ b_1 + \dots + b_t = b, \\ \barc_1 + \dots + \barc_t = \barc, \\ \pdeg_1, \dots, \pdeg_t \in \pdegdom(\tail[u]) \colon \\ \pdeg_1 + \dots + \pdeg_t \equiv_{\pphi} \pdeg}} \prod_{i=1}^t |\Partial_{v_i}(I[\pdeg \mapsto \pdeg_i], a_i, b_i, \barc_i)|
    \end{equation}
    for all $a,  b \in \bN$ and ${\barc} \in \bN^\ell$.
    To obtain the desired equality, we fix arbitrary $a,  b \in \bN, {\barc} \in \bN^\ell$ and define the set
    \begin{align*}
        \RHS \defeq \Bigl\{&(P_1, \dots, P_t) \mid \\
        &\exists \pdeg_1, \dots, \pdeg_t \in \pdegdom(\tail[u]), a_1, \dots, a_t, b_1, \dots, b_t \in \bN, \barc_1, \dots, \barc_t \in \bN^\ell \colon \\
        &\sum_{i=1}^{t} \pdeg_i \equiv_{\pphi} \pdeg, \sum_{i=1}^{t} a_i = a, \sum_{i=1}^{t} b_i = b, \sum_{i=1}^{t} \barc_i = \barc, \\
        &P_i \in \Partial_{v_i}( I[\pdeg \mapsto \pdeg_i], a_i, b_i, \barc) \text{ for every } i \in [t]\Bigr\}.\\
    \end{align*}
    Observe that for every $i \in [t]$, by \cref{def:modelcheck:index:compatibility,def:modelcheck:constants}, the sets $\Partial_{v_i}( I[\pdeg \mapsto \pdeg'], a', b', \barc')$ and $\Partial_{v_i}( I[\pdeg \mapsto \pdeg''], a'', b'', \barc'')$ are disjoint for every pair $(\pdeg', a', b', \barc') \neq (\pdeg'', a'', b'', \barc'') \in \pdegdom(\tail[u]) \times \bN^{\ell+2}$.
    Therefore, the right-hand side of \eqref{eq:join-proof-rewritten} is equal to the cardinality of $\RHS$.
    So to prove \eqref{eq:join-proof-rewritten} it remains to provide an injective mapping from $\Partial_u[I,a,b,\barc]$ to $\RHS$ as well as an injective mapping from $\RHS$ to $\Partial_u[I,a,b,\barc]$.

    Observe that the tail of each triad in $\tss[u], \tss(v_1), \dots, \tss(v_t)$ is equal to $\tail[u]$, the sets $\sheaf(v_1), \dots, \sheaf(v_t)$ partition the set $\sheaf(u)$, the sets $\subtree(v_1), \dots, \subtree(v_t)$ partition the set $\subtree(u)$, and we have $\broom[u] = \cup_{i \in [t]} \broom[v_i]$.
    We will rely on these properties in the following, sometimes implicitly.
    \begin{claim}
        There exists an injective mapping $h_1$ from $\Partial_u[I,a,b,\barc]$ to $\RHS$.
    \end{claim}
    \begin{claimproof}
    We map every partial solution $(\interp{\hat f}, \interp{\hat g}, \hat \bE, \vec\cW)$ in $\Partial_u[I,a,b,\barc]$ to the tuple
    $h_1(\interp{\hat f}, \interp{\hat g}, \hat \bE, \vec\cW) \defeq (P_1,\dots,P_t)$
    where for every $i \in [t]$, the tuple $P_i\defeq (\interp{\hat f}_i, \interp{\hat g}_i, \hat \bE_i, \vec\cW_i)$ is such that
    \begin{itemize}
        \item $\interp{\hat f}_i$ is the restriction $\interp{\hat f}_{|_{\broom[v_i]}}$ of $\psf$ on $\broom[v_i]$, %
        \item $\interp{\hat g}_i$ is the restriction $\interp{\hat g}_{|_{\sheaf(v_i)}}$ of $\psg$ on $\sheaf(v_i)$, %
        \item $\pse_i$ is the restriction $\pse_{|_{\broom[v_i]}}$ of $\pse$ on $\broom[v_i]$, %
        \item and $\vec\cW_i$ is the restriction $\vec\cW_{|_{\sheaf(v_i)}}$ of $\psw$ on $\sheaf(v_i)$. %
    \end{itemize}
    Before showing that this is indeed a mapping into $\RHS$, let us start with injectivity since it is simpler.
    For this we consider two elements $P = (\interp{\hat f}, \interp{\hat g}, \hat \bE, \vec\cW)$ and $P' = (\interp{\hat f}', \interp{\hat g}', \hat \bE', \vec\cW')$ in $\Partial_u[I,a,b,\barc]$ such that $h_1(P) = h_1(P')$.
    This implies the following:
    \begin{itemize}
        \item For every $X \in \var^V(\phi)$ we have 
        \begin{align*}
            \interp{\hat f}(X) = &\interp{\hat f}(X) \cap \broom[u] = \interp{\hat f}(X) \cap (\broom[v_1] \cup \dots \cup \broom[v_t]) \\
            =&(\interp{\hat f}(X) \cap \broom[v_1]) \cup \dots \cup (\interp{\hat f}(X) \cap \broom[v_t]) = \interp{\hat f}_1(X) \cup \dots \cup \interp{\hat f}_t(X) \\
            \stackrel{h_1(P) = h_1(P')}{=} &\interp{\hat f}'_1(X) \cup \dots \cup \interp{\hat f}'_t(X) = (\interp{\hat f}'(X) \cap \broom[v_1]) \cup \dots \cup (\interp{\hat f}'(X) \cap \broom[v_t]) \\
            =&\interp{\hat f}'(X) \cap (\broom[v_1] \cup \dots \cup \broom[v_t]) =\interp{\hat f}'(X) \cap \broom[u] =\interp{\hat f}'(X).
        \end{align*}
        Hence $\psf = \psf'$.
        \item For every $Y \in \var^E(\phi)$ we have
        \begin{align*}
            \interp{\hat g}(Y) = &\interp{\hat g}(Y) \cap \sheaf(u) =\interp{\hat g}(Y) \cap (\sheaf(v_1) \cup \dots \cup \sheaf(v_t)) \\
            =&(\interp{\hat g}(Y) \cap \sheaf(v_1)) \cup \dots \cup (\interp{\hat g}(Y) \cap \sheaf(v_t)) =\interp{\hat g}_1(Y) \cup \dots \cup \interp{\hat g}_t(Y)\\
            \stackrel{h_1(P) = h_1(P')}{=} &\interp{\hat g}_1'(Y) \cup \dots \cup \interp{\hat g}_t'(Y) = (\interp{\hat g}'(Y) \cap \sheaf(v_1)) \cup \dots \cup (\interp{\hat g}'(Y) \cap \sheaf(v_t)) \\
            =&\interp{\hat g}'(Y) \cap (\sheaf(v_1) \cup \dots \cup \sheaf(v_t)) =\interp{\hat g}'(Y) \cap \sheaf(u) =\interp{\hat g}'(Y).
        \end{align*}
        Hence $\interp{\hat g} = \interp{\hat g}'$.
        \item %
        For every vertex $z \in \subtree(u)$ there exists a (unique) index $i \in [t]$ with $z \in \subtree[v_i]$. Then for every $C \in \Cat$ we have
        \[
            \pse(z, C) = \pse_i(z, C) \stackrel{h_1(P) = h_1(P')}{=} \pse'_i(z, C) = \pse'(z, C).
        \]
        And for every vertex $z \in \tail[u]$ and every $C \in \Cat$ we have $\pse(z, C) = \bE(z, C) = \hat\bE'(z, C)$ since both $P_1$ and $P_2$ satisfy \ref{item:partial:econsistent} for $I$.
        So we get $\pse_1 = \pse_2$.
        \item And for every $C \in \Cat$ we have 
        \begin{align*}
            \vec\cW(C) = &\vec\cW(C) \cap \vec\sheaf(u) = \vec\cW(C) \cap (\vec\sheaf(v_1) \cup \dots \cup \vec\sheaf(v_t)) \\
            =&(\vec\cW(C) \cap \vec\sheaf(v_1)) \cup \dots \cup (\vec\cW(C) \cap \vec\sheaf(v_t)) =\vec\cW_1(C) \cup \dots \cup \vec\cW_t(C) \\
            \stackrel{h_1(P) = h_1(P')}{=}&\vec\cW'_1(C) \cup \dots \cup \vec\cW'_t(C) =(\vec\cW'(C) \cap \vec\sheaf(v_1)) \cup \dots \cup (\vec\cW'(C) \cap \vec\sheaf(v_t)) \\
            =&\vec\cW'(C) \cap (\vec\sheaf(v_1) \cup \dots \cup \vec\sheaf(v_t)) 
            =\vec\cW'(C) \cap \vec\sheaf(u) 
            =\vec\cW'(C)
        \end{align*}
        Hence, $\vec\cW = \vec\cW'$.
    \end{itemize}
    Altogether, we obtain $P = P'$ implying the injectivity of $h_1$.
    
    It remains to show that for every $P = (\psf, \psg, \pse, \psw) \in \Partial_u[I,a,b,\barc]$, its image $h_1(P)$ belongs to $\RHS$.
    For this, we first define the following values for every $i \in [t]$ and every $j \in [\ell]$: %
    \[
        a_i = \sum_{C \in \Cat} \sum_{\substack{v \in \subtree[v_i] \colon \\ \pse(v, C) < \dphi}} \hat \bE(v, C),
    \]
    \[
        b_i = \sum_{C \in \Cat} \sum_{\substack{ w\in \broom[v_i] \\ \hat \bE(w,C) < \dphi}}  \abs{ \{(z^1,w^j)\in \vec\cW_i(C) \}}
    \]
    \[
        c^j_i =
        \begin{cases}
            \abs{\interp{\hat f}(X_j)\cap \subtree[v_i]} & \text{if } X_j \in \var^V(\phi) \\
            \abs{\interp{\hat g}(X_{j})\cap \sheaf(v_i)} & \text{if } X_j \in \var^E(\phi) \\
        \end{cases}
        .
    \]
    For $i \in [t]$, we then define the vector $\barc_i = (c_i^1, \dots, c_i^\ell) \in \bN^\ell$. 
    We then also have the following for every $i \in [t]$ and every $j \in [\ell]$: %
    \begin{equation}\label{eq:a-i}
        a_i = \sum_{C \in \Cat} \sum_{\substack{v \in \subtree[v_i] \colon \\ \pse_i(v, C) < \dphi}} \hat \bE_i(v, C),
    \end{equation}
    \begin{equation}\label{eq:b-i}
        b_i = \sum_{C \in \Cat} \sum_{\substack{ w\in \broom[v_i] \\ \hat \bE_i(w,C) < \dphi}}  \abs{ \{(z^1,w^j)\in \vec\cW_i(C) \}}
    \end{equation}
    \begin{equation}\label{eq:c-i}
        c^j_i =
        \begin{cases}
            \abs{\interp{\hat f}_i(X_j)\cap \subtree[v_i]} & \text{if } X_j \in \var^V(\phi) \\
            \abs{\interp{\hat g}_i(X_{j})} & \text{if } X_j \in \var^E(\phi) \\
        \end{cases}
        .
    \end{equation}
    We then have
    \begin{align*}
        \sum_{i = 1}^{t} a_i = &\sum_{i = 1}^{t} \sum_{C \in \Cat} \sum_{\substack{v \in \subtree[v_i] \colon \\ \pse(v, C) < \dphi}} \hat \bE(v, C)=\sum_{C \in \Cat} \sum_{i = 1}^{t} \sum_{\substack{v \in \subtree[v_i] \colon \\ \pse(v, C) < \dphi}} \hat \bE(v, C)  \\
        \stackrel{\subtree(u) = \subtree[v_1]\dot\cup \dots \dot\cup \subtree[v_t]}{=}&\sum_{C \in \Cat} \sum_{\substack{v \in \subtree(u) \colon \\ \pse(v, C) < \dphi}} \hat \bE(v, C) = a_u[P] \stackrel{P \in \Partial_u[I,a,b,\barc]}{=} a \\
    \end{align*}
    and
    \begin{align*}
        \sum_{i=1}^t b_i =
        &\sum_{i=1}^t \sum_{C \in \Cat} \sum_{\substack{ w\in \broom[v_i] \\ \hat \bE(w,C) < \dphi}}  \abs{ \{(z^1,w^j)\in \vec\cW_i(C) \}} \\
        \stackrel{\substack{\cW_i(C) \subseteq \vec\sheaf(v_i), \\ V(\sheaf(v_i)) \subseteq \broom[v_i] \subseteq \broom[u]}}{=}&\sum_{i=1}^t \sum_{C \in \Cat} \sum_{\substack{ w\in \broom[u] \\ \hat \bE(w,C) < \dphi}}  \abs{ \{(z^1,w^j)\in \vec\cW_i(C) \}}  \\
        =&\sum_{C \in \Cat}  \sum_{\substack{ w\in \broom[u] \\ \hat \bE(w,C) < \dphi}} \sum_{i=1}^t  \abs{ \{(z^1,w^j)\in \vec\cW_i(C) \}} \\
        \stackrel{\sheaf(u) = \sheaf(v_1) \dot\cup \dots \dot\cup \sheaf(v_t)}{=}&\sum_{C \in \Cat} \sum_{\substack{ w\in \broom[u] \\ \hat \bE(w,C) < \dphi}} \abs{ \{(z^1,w^j)\in \vec\cW(C) \}} = b_u[P] \stackrel{P \in \Partial_u[I,a,b,\barc]}{=} b.
    \end{align*}
    Similarly, for every $j \in [\ell]$, if $X_j \in \var^V(\phi)$, we have:
    \[
        \sum_{i=1}^t c_i^j = \sum_{i=1}^t \abs{\interp{\hat f}(X_j)\cap \subtree[v_i]} = \abs{\interp{\hat f}(X_j)\cap \subtree(u)} = c^j_u[P] \stackrel{P \in \Partial_u[I,a,b,\barc]}{=} c^j
    \]
    and if $X_j \in \var^E(\phi)$, we have
    \[
        \sum_{i=1}^t c_i^j = \sum_{i=1}^t \abs{\interp{\hat g}(X_{j})\cap \sheaf(v_i)} = \abs{\interp{\hat g}(X_j)} = c^j_u[P] \stackrel{P \in \Partial_u[I,a,b,\barc]}{=} c^j.
    \]
    
    We recall that we have $\tail[u] = \tail(v_i)$ for every $i \in [t]$.
    For every $i \in [t]$, we define the mapping $\pdeg_i \colon \tail(v_i) \times \Cat \to \FF_{\pphi}$ to be equal to $\pdeg_{\vec\cW_i}$.
    Then the following holds for every $w \in \tail[u]$ and every $C \in \Cat$: 
    \begin{align*}
        \sum_{i = 1}^t \pdeg_i(w, C) = &\sum_{i = 1}^t \pdeg_{\vec\cW_i}(w, C) \equiv_{\pphi} \sum_{i = 1}^t |\{(z^1, w^j) \in \vec\cW_i(C)\}| \\
        \stackrel{\sheaf(u) = \sheaf(v_1) \dot\cup \dots \dot\cup \sheaf(v_t)}{=} &|\{(z^1, w^j) \in \vec\cW(C)\}| \stackrel{P \in \Partial_u[I]}{\equiv_{\pphi}} \pdeg(w, C).
    \end{align*}
    Thus, we have $\pdeg \equiv_{\pphi} \pdeg_1 + \dots + \pdeg_t$.

    It thus remains to show that for every $i \in [t]$, the tuple $P_i = (\interp{\hat f}_i, \interp{\hat g}_i, \hat \bE_i, \vec\cW_i)$ belongs to the set $\Partial_{v_i}(I[\pdeg \mapsto \pdeg_i], a_i, b_i, \barc_i)$.
    So fix an arbitrary index $i \in [t]$. 
    First, we show that $P_i$ is a partial solution for $\tss(v_i) = (\tail(v_i), \subtree[v_i], \sheaf(v_i))$.
    For this we will actively rely on the fact that $P$ is a partial solution for $\tss[u]$.
    By definition, $P_i$ is \ref{item:partial:respect-domain}.
    Now consider an arbitrary element $C = (X, Y) \in \Cat$:
    \begin{itemize}
        \item \ref{item:partial:edge:consistent}: Consider an arbitrary edge $zw \in \sheaf(v_i)$. %
        Note that we then have $w \in \broom[v_i]$.
        Then there exists $j \in [d_\phi]$ with $(z^1, w^j) \in \vec \cW_i(X, Y)$ if and only if there exists $j \in [d_\phi]$ with $(z^1, w^j) \in \vec \cW(X, Y)$---this is because $\psw_i$ is the restriction of $\psw$ to $\sheaf(v_i) 
        $.
        Since $P$ is a partial solution for $\tss[u]$ this is equivalent to $z \in \interp{\hat f}(X)$ and $zw \in \psg(Y)$. %
        We have $z \in \broom[v_i]$ and $zw \in \sheaf(v_i)$ so this is equivalent to $z \in \interp{\hat f}_i(X)$ and $zw \in \psg_i(Y)$. %
        \item \ref{item:partial:respect:subtree:hatbE}: Consider an arbitrary vertex $z \in \subtree[v_i] \subseteq \subtree(u)$ and arbitrary $j \in [\hat\bE_i(v,C)] = [\hat\bE(v,C)]$. 
        Partial solution $P$ satisfies \ref{item:partial:respect:subtree:hatbE} for $\tss[u]$ so $\vec\cW(C)$ contains an arc $(w^1, z^j)$ for some vertex $w \in V(G)$. %
        By construction of the digraph $\vec{G}_{\phi}$, $wz$ is then an edge of $G$.
        Every edge of $G$ incident to $z$ in $G$ belongs to $\sheaf(z) \subseteq \sheaf(v_i)$. 
        Therefore, the arc $(w^1, z^j)$ is also contained in $\vec\cW(C) \cap \vec\sheaf(v_i) = \vec\cW_i(C)$.
        \item \ref{item:partial:respect:broom:hatbE}: Because $P$ satisfies \ref{item:partial:respect:broom:hatbE} for $\tss[u]$, for every $v \in \broom[u]$, the set $\vec\cW(C)$ does not contain any arc incoming into $v^j$ for any $j > \hat\bE(v, C)$.
        Then this property holds for the subset $\broom[v_i]$ of $\broom[u]$, the subset $\psw_i(C)$ of $\psw(C)$, and the restriction $\pse_i$ of $\bE$ to $\broom[v_i]$ for every $i \in [t]$.
    \end{itemize}
    So $P_i$ is indeed a partial solution for $\tss(v_i)$. 
    
    Next, we show that $P_i$ is compatible with the index $I[\pdeg \mapsto \pdeg_i] = (\interp f, \inde, \allow, \pdeg_i, \False)$ for $\tss(v_i)$. We will rely on the fact that $P$ is compatible with $I$ for $\tss[u]$.
    The properties \ref{item:partial:fconsistent} and \ref{item:partial:econsistent} follow directly from $\tail[u] = \tail(v_i)$.
    The property \ref{item:partial:pdegconsistent} holds by definition of $\pdeg_i$.
    The property \ref{item:partial:respect:allow} follows from $\tail[u] = \tail(v_i)$ as well as $\vec\cW_i(C) \subseteq \vec\cW(C)$ for every $C \in \Cat$.
    It remains to show that \ref{item:partial:respect:False} is satisfied. 
    The partial solution $P$ satisfies \ref{item:partial:respect:False} for $I$ so we have:
    $\FalseFunc(\interp{ \hat f}_{|_{\subtree(u)}}, \interp{\hat g}, \hat \bE_{|_{\subtree(u)}}, (\pdeg_{\vec \cW})_{|_{\subtree(u)}}) \subseteq \False$.
    Since $\subtree[v_i] \subseteq \subtree(u)$ and $\sheaf(v_i) \subseteq \sheaf[u]$ by \cref{claim:restriction-vs-false-set} we get 
    \begin{align*}
        &\FalseFunc((\interp{ \hat f}_{|_{\subtree(u)}})_{|_{\subtree[v_i]}}, \interp{\hat g}_{|_{\sheaf(v_i)}}, (\hat \bE_{|_{\subtree(u)}})_{|_{\subtree[v_i]}},((\pdeg_{\vec \cW})_{|_{\subtree(u)}})_{|_{\subtree[v_i]}}) \\
        \subseteq &\FalseFunc(\interp{ \hat f}_{|_{\subtree(u)}}, \interp{\hat g}, \hat \bE_{|_{\subtree(u)}}, (\pdeg_{\vec \cW})_{|_{\subtree(u)}}) \subseteq \False.
    \end{align*}
    Since $\subtree[v_i] \subseteq \subtree(u)$ and $\subtree[v_i] \subseteq \broom[v_i]$ by definition of $P_i$ this implies 
    \[
        \FalseFunc((\interp{ \hat f}_i)_{|_{\subtree[v_i]}}, \interp{\hat g}_i, (\hat \bE_i)_{|_{\subtree[v_i]}},(\pdeg_{\vec \cW})_{|_{\subtree[v_i]}}) \subseteq \False.  
    \]
    Now we show that $(\pdeg_{\vec \cW})_{|_{\subtree[v_i]}} \equiv_{\pphi} (\pdeg_{\vec \cW_i})_{|_{\subtree[v_i]}}$ holds.
    For this consider an arbitrary vertex $w \in \subtree[v_i]$ and arbitrary $C \in \Cat$.
    By \cref{obs:sheaf-end-vertices}, $w$ has no incident edges in $\sheaf(v_j)$ for any $j \neq i \in [t]$.
    Then we have
    \begin{align*}
        (\pdeg_{\vec \cW})_{|_{\subtree[v_i]}}(w, C) = &\pdeg_{\vec \cW}(w, C) \equiv_{\pphi} |\{(z^1, w^j) \in \vec\cW(C)\}| = |\{(z^1, w^j) \in \vec\cW(C) \cap \vec\sheaf(v_i)\}| \\
        = &|\{(z^1, w^j) \in \vec\cW_i(C)\}|\equiv_{\pphi} \pdeg_{\vec \cW_i}(w, C) = (\pdeg_{\vec \cW_i})_{|_{\subtree[v_i]}}(w, C) 
    \end{align*}
    So indeed $(\pdeg_{\vec \cW})_{|_{\subtree[v_i]}} \equiv_{\pphi} (\pdeg_{\vec \cW_i})_{|_{\subtree[v_i]}}$ holds and we get 
    \[
        \FalseFunc((\interp{ \hat f}_i)_{|_{\subtree[v_i]}}, \interp{\hat g}_i, (\hat \bE_i)_{|_{\subtree[v_i]}},(\pdeg_{\vec \cW_i})_{|_{\subtree[v_i]}}) \subseteq \False.  
    \]
    Altogether we obtain that $P_i$ is compatible with the index $I[\pdeg \mapsto \pdeg_i]$ for $\tss(v_i)$.
    Finally, by \eqref{eq:a-i}, \eqref{eq:b-i}, and \eqref{eq:c-i}, we have $a_{v_i}(P_i) = a_i$, $b_{v_i}(P_i) = b_i$, and $\barc_{v_i}(P_i) = \barc_i$, respectively.    
    So the partial solution $P_i$ is contained in $\Partial_{v_i}(I[\pdeg \mapsto \pdeg_i], a_i, b_i, \barc_i)$.
    Therefore, $h_1(P)$ belongs to $\RHS$, and $\phi$ is indeed a well-defined injective mapping from $\Partial_u[I,a,b,\barc]$ to $\RHS$. 
    \end{claimproof}

    \begin{claim}
        There exists an injective mapping $h_2$ from $\RHS$ to $\Partial_u[I,a,b,\barc]$.
    \end{claim}
    \begin{claimproof}
    Let $(P_1, \dots, P_t) \in \RHS$ with $P_i = (\interp{\hat f}_i, \interp{\hat g}_i, \hat \bE_i, \vec\cW_i)$ for every $i\in[t]$.
    We define $h_2(P_1, \dots, P_t)\defeq P \defeq (\interp{\hat f}, \interp{\hat g}, \hat \bE, \vec\cW)$ such that
    \begin{itemize}
        \item For every $X \in \var^V(\phi)$, we have $\interp{\hat f}(X) = \bigcup_{i \in [t]} \interp{\hat f}_i(X)$.
        \item For every $Y \in \var^E(\phi)$, we have $\interp{\hat g}(Y) = \bigcup_{i \in [t]} \interp{\hat g}_i(Y)$.
        \item For every $C \in \Cat$ and every $z \in \subtree(u)$, let $i \in [t]$ be such that $z \in \subtree[v_i]$.
        Then we have $\hat \bE(z, C) = \hat \bE_i(z, C)$. And for every $C \in \Cat$ and every $z \in \tail[u]$ we let $\hat \bE(z, C) = \bE(z, C)$. 
        \item For every $C \in \Cat$, we have $\vec \cW(C) = \bigcup_{i \in [t]} \vec\cW_i(C)$.
    \end{itemize}
    By definition of $\RHS$ there exist $\pdeg_i \in \pdegdom(\tail[u])$, %
    $a_i, b_i \in \bN$, and $\barc_i \in \bN^\ell$ with $P_i \in \Partial_{v_i}( I[\pdeg \mapsto \pdeg_i], a_i, b_i, \barc_i)$ for every $i \in [t]$ and $\sum_{i \in [t]} \pdeg_i \equiv_{\pphi} \pdeg$, $\sum_{i \in [t]} a_i = a$, $\sum_{i \in [t]} b_i = b$, and $\sum_{i \in [t]} \barc_i = \barc$.
    We denote the index $I[\pdeg \mapsto \pdeg_i]$ by $I_i$.
    
    First we show that $P$ is a partial solution for $\tss[u]$.
    Recall that $\sheaf(v_1), \dots, \sheaf(v_t)$ partition $\sheaf(u)$ while $\subtree[v_1], \dots, \subtree[v_t]$ partition $\subtree(u)$.
    And we also have $\broom[u] = \bigcup_{i \in [t]} \broom[v_i]$.
    We will rely on the fact that $P_i$ is a partial solution for $\tss(v_i)$ for every $i \in [t]$.
    First, this implies that $\psf$ is a vertex interpretation on $\broom[u]$ and $\psg$ is an edge interpretation on $\sheaf(u)$.
    Further, $\pse$ is a well-defined neighborhood function on $\broom[u]$.
    We now argue that for every $C \in \Cat$, the set $\psw(C)$ is a simple subset of $\vec\sheaf(u)$.
    For every edge $zw \in E(G)$ there exists at most one index $i \in [t]$ with $zw \in \sheaf(v_i)$.
    Since $P_i$ satisfies \ref{item:partial:respect-domain}, the set $\psw_i(C)$ is simple.
    Thus, $\psw(C)$ is also simple and $P$ satisfies \ref{item:partial:respect-domain}.
    Now we show that the remaining three conditions are satisfied as well:
    \begin{itemize}
        \item \ref{item:partial:edge:consistent}: Consider an arbitrary edge $zw \in \sheaf(u)$ and let $i \in [t]$ be unique such that $zw \in \sheaf(v_i)$. 
        Then we have $z \in \broom[v_i]$.
        Further, let $(X, Y) \in \Cat$ be arbitrary.
        Then there exists $j \in [d_\phi]$ with $(z^1, w^j) \in \vec \cW(X, Y)$ if and only if there exists $j \in [d_\phi]$ with $(z^1, w^j) \in \vec \cW_i(X, Y)$. %
        Since $P_i$ is a partial solution of $\tss(v_i)$ this is equivalent to $z \in \psf_i(X)$ and $zw \in \psg_i(Y)$. %
        
        First, we show that $z \in \psf_i(X)$ is equivalent to $z \in \psf(X)$. 
        The property $z \in \psf_i(X)$ trivially implies $z \in \psf(X)$.
        On the other hand, let $z \in \psf(X)$.
        Then by definition there exists an index $j \in [t]$ with $z \in \psf_j(X)$.
        Now if $z \in \tail(v_j)$, since $P_j$ satisfies \ref{item:partial:fconsistent} for $\indf$, we have $z \in \indf(X)$, and since $P_i$ satisfies \ref{item:partial:fconsistent} for $\indf$, we, in turn, get $z \in \psf_i(X)$.
        And if $z \in \subtree[v_j]$, we have $i = j$ as otherwise $\subtree[v_j]$ and $\broom[v_i]$ are disjoint.
    
        Since we have $zw \in \sheaf(v_i)$ and the sets $\sheaf(v_1), \dots, \sheaf(v_t)$ are pairwise disjoint, the property $zw \in \psg_i(Y)$ is equivalent to $\psg(Y)$.
        So \ref{item:partial:edge:consistent} is satisfied.
        \item \ref{item:partial:respect:subtree:hatbE}:
        Consider an arbitrary vertex $z \in \subtree(u)$, an arbitrary $C \in \Cat$, and an arbitrary index $j \in [\pse(z, C)]$.
        Let $i \in [t]$ be the unique index such that $z \in \subtree[v_i]$.
        By definition of $\pse$ we then have $\pse_i(z, C) = \pse(z, C)$ and therefore, $j \in [\pse_i(z, C)]$.
        Since $P_i$ is a partial solution of $\tss(v_i)$, the set $\psw_i(C)$ contains an arc incoming into $z^j$ so $\psw(C)$ does as well.
        \item \ref{item:partial:respect:broom:hatbE} 
        Consider a vertex $w \in \broom[u]$, element $C \in \Cat$, and an index $j \in [\dphi]$ with $j > \pse(w, C)$.
        If $w \in \tail[u]$, then for every $i \in [t]$, we have $\tss[u] = \tail(v_i)$ and thus $\pse(w, C) = \inde(w, C) = \pse_i(w)$ (since $P_i$ satisfies \ref{item:partial:econsistent} for $I_i$). 
        As $P_i$ also satisfies \ref{item:partial:respect:broom:hatbE}, the set $\psw_i(C)$ does not contain any arc incoming into $w^j$.
        And if $w \in \subtree[v_i]$, then first, we have $w \notin \broom[v_q]$ for every $q \neq i \in [t]$, and thus $\psw_q(C) \subseteq \vec\sheaf(v_q)$ contains no arcs incoming into $w^j$; and second, we have $\pse(w, C) = \pse_i(w,C)$ and since $P_i$ satisfies \ref{item:partial:respect:broom:hatbE} for $I_i$, the set $\psw_i(C)$ also contains no such arcs.
        So in both cases $\psw(C) = \bigcup_{i \in [t]} \psw_i(C)$ contains no arcs incoming into $w^j$.
    \end{itemize}
    Altogether, $P$ is a partial solution for $\tss[u]$. %

    Next we show that $P$ is compatible with the index $I = (\interp f, \bE, \allow, \pdeg, \False)$ for $\tss[u]$.
    First, for \ref{item:partial:fconsistent}, for every $X \in \var^V(\phi)$ we have 
    \[
        \psf(X) \cap \tail[u] = 
        (\bigcup_{i \in [t]} \interp{\hat f}_i(X)) \cap \tail[u] = 
        \bigcup_{i \in [t]} (\interp{\hat f}_i(X) \cap \tail[u]) = 
        \bigcup_{i \in [t]} (\interp{\hat f}_i(X) \cap \tail(v_i)) =      \bigcup_{i \in [t]} \interp{f}_i(X)
    \]
    where the last equality holds since each $P_i$ %
    satisfies \ref{item:partial:fconsistent} for $I_i$.
    \ref{item:partial:econsistent} is satisfied by $P$ by construction.
    Next, for every $z \in \tail[u]$ and every $C \in \Cat$ we have
    \begin{align*}
        \pdeg(z, C) \equiv_{\pphi} &\sum_{i \in [t]} \pdeg_i(z, C) \equiv_{\pphi} \sum_{i \in [t]} \pdeg_{\psw_i}(z, C) \equiv_{\pphi} \pdeg_{\psw}(z, C) 
    \end{align*}
    where the second equality since $P_i$ satisfies \ref{item:partial:pdegconsistent} for $I_i$ and the last equality holds since $\psw_1(C), \dots, \psw_t(C)$ partition $\psw(C)$.
    So \ref{item:partial:pdegconsistent} is applies to $P$ and $I$.
    Further, for every $z \in \tail[u]$, every $C \in \Cat$, and every edge $(w^1, z^j) \in \psw(C)$ the following holds.
    Let $i \in [t]$ be such that $(w^1, z^j) \in \psw_i(C)$.
    Since $P_i$ is compatible with $I_i$ and satisfies \ref{item:partial:respect:allow}, we have $j \leq \allow(z, C)$---so this property also applies to $P$.
    
    Finally, we show that %
    the inclusion $\FalseFunc(\psf_{|_{\subtree(u)}}, \psg, \pse_{|_{\subtree(u)}}, (\pdeg_{\vec\cW})_{|_{\subtree(u)}}) \subseteq \False$ holds.
    For every $i \in [t]$, since $P_i$ satisfies \ref{item:partial:respect:False} %
    for %
    the index $I_i$, %
    we have 
    \begin{equation}\label{eq:false-i}
        \FalseFunc((\psf_i)_{|_{\subtree[v_i]}}, \psg_i, (\pse_i)_{|_{\subtree[v_i]}}, (\pdeg_{\vec\cW_i})_{|_{\subtree[v_i]}}) \subseteq \False.
    \end{equation}
    The following properties apply because $\subtree[v_1], \dots, \subtree[v_t]$ partition $\subtree(u)$ and $\sheaf(v_1)$, $\dots$, $\sheaf(v_t)$ partition $\sheaf(u)$:
        \begin{align*}
            \psf_{|_{\subtree(u)}} = \bigcup_{i \in [t]} (\psf_i)_{|_{\subtree[v_i]}}%
            &&
            \psg = \bigcup_{i \in [t]} \psg_i%
            &&
            \pse_{|_{\subtree(u)}} = \bigcup_{i \in [t]} (\pse_i)_{|_{\subtree[v_i]}}. %
        \end{align*}
        \item 
    And for every $i \neq q \in [t]$, every vertex $z \in \subtree[v_i]$ %
    has no incident edges in $\sheaf(v_q)$ so for any $C \in \Cat$ we have $\pdeg_{\vec\cW}(z, C) = \pdeg_{\vec\cW_i}(z, C)$.
        Hence, $(\pdeg_{\vec\cW})_{|_{\subtree(u)}} = \bigcup_{i \in [t]} (\pdeg_{\vec\cW_i})_{|_{\subtree[v_i]}}$.
    So we get 
    \begin{align*}
        &\FalseFunc(\psf_{|_{\subtree(u)}}, \psg, \pse_{|_{\subtree(u)}}, (\pdeg_{\vec\cW})_{|_{\subtree(u)}}) \\
        = &\FalseFunc(\bigcup_{i \in [t]} (\psf_i)_{|_{\subtree[v_i]}}, \bigcup_{i \in [t]} \psg_i, \bigcup_{i \in [t]} (\pse_i)_{|_{\subtree[v_i]}}, \bigcup_{i \in [t]} (\pdeg_{\vec\cW_i})_{|_{\subtree[v_i]}}) \\
        \stackrel{\cref{claim:union-false-set}}{=}&\bigcup_{i \in [t]} \FalseFunc((\psf_i)_{|_{\subtree[v_i]}}, \psg_i, (\pse_i)_{|_{\subtree[v_i]}}, (\pdeg_{\vec\cW_i})_{|_{\subtree[v_i]}}) \stackrel{\eqref{eq:false-i}}{\subseteq} \False.
    \end{align*}
    Thus $P$ satisfies \ref{item:partial:respect:False} and therefore, we have $P \in \Partial_u[I]$.
     
    Next, we show that we have $P \in \Partial_u[I, a, b, \barc]$.
    First, it holds that
    \begin{align*}
        a =  &\sum_{i = 1}^t a_i \stackrel{P_i \in \Partial_{v_i}( I_i, a_i, b_i, \barc_i)}{=} \sum_{i = 1}^t a_{v_i}(P_i) = \sum_{i = 1}^t \sum_{C \in \Cat} \sum_{\substack{v \in \subtree[v_i] \colon \\ \hat\bE_i(v, C) < \dphi}} \pse_i(v, C) \\
        = &\sum_{C \in \Cat} \sum_{i = 1}^t \sum_{\substack{v \in \subtree[v_i] \colon \\ \hat\bE(v, C) < \dphi}} \pse(v, C) = \sum_{C \in \Cat} \sum_{\substack{v \in \subtree(u) \colon \\ \hat \bE(v, C) < \dphi}} \pse(v, C) = a_u[P].
    \end{align*}
    Similarly, we have
    \begin{align*}
        b = &\sum_{i = 1}^t b_i \stackrel{P_i \in \Partial_{v_i}( I_i, a_i, b_i, \barc_i)}{=} \sum_{i = 1}^t b_{v_i}(P_i) = \sum_{i = 1}^t \sum_{C\in \Cat} \sum_{\substack{w \in \broom[v_i] \\ \hat\bE_i(w,C) < \dphi }} \abs{\{(z^1, w^j) \in \vec\cW_i(C)\}} \\
        \stackrel{\substack{V(\sheaf(v_i)) \subseteq \broom[v_i]}}{=}&\sum_{i = 1}^t \sum_{C\in \Cat} \sum_{\substack{w \in \broom[u] \\ \hat\bE_i(w,C) < \dphi }} \abs{\{(z^1, w^j) \in \vec\cW_i(C)\}} \\
        \stackrel{\substack{\text{Def.\ of $\pse$, } \\ \text{$P_i$ compatible with $I_i$}}}{=} &\sum_{i = 1}^t \sum_{C\in \Cat} \sum_{\substack{w \in \broom[u] \\ \hat\bE(w,C) < \dphi }} \abs{\{(z^1, w^j) \in \vec\cW_i(C)\}} \\
        =&\sum_{C\in \Cat} \sum_{\substack{w \in \broom[u] \\ \hat\bE(w,C) < \dphi }} \sum_{i = 1}^t \abs{\{(z^1, w^j) \in \vec\cW_i(C)\}} \\
        =&\sum_{C\in \Cat} \sum_{\substack{w \in \broom[u] \\ \hat\bE(w,C) < \dphi }} \abs{\{(z^1, w^j) \in \vec\cW(C)\}} = b_u[P].
    \end{align*}
    Also for every $j \in [\ell]$, if $X_j \in \var^V(\phi)$, we have
    \begin{align*}
        c^j = &\sum_{i=1}^t c^j_i \stackrel{P_i \in \Partial_{v_i}( I_i, a_i, b_i, \barc_i)}{=} \sum_{i=1}^t c^j_{v_i}(P_i) = \sum_{i=1}^t |\psf_i(X_j) \cap \subtree[v_i]| \\
        \stackrel{\psf_q(X_j) \cap \subtree[v_i] = \emptyset \text{ for } q \neq i \in [t]}{=} &\sum_{i=1}^t |\psf(X_j) \cap \subtree[v_i]|
        = |\psf(X_j) \cap \subtree(u)| = c^j_u[P],\\
    \end{align*}
    and if $X_j \in \var^E(\phi)$, then we have
    \begin{align*}
        c^j = &\sum_{i=1}^t c^j_i = \stackrel{P_i \in \Partial_{v_i}( I_i, a_i, b_i, \barc_i)}{=} \sum_{i=1}^t c^j_{v_i}(P_i) = \sum_{i=1}^t \psg_i(X_j) = \psg(X_j) = c^j_u[P].
    \end{align*}
    Altogether, we obtain that $P \in \Partial_u[I, a, b, \barc]$ holds and therefore, $h_2$ is a well-defined mapping from $\RHS$ to $\Partial_u[I,a,b,\barc]$.
    
    It remains to show that $h_2$ is injective.
    For this consider $(P^1_1, \dots, P^1_t), (P^2_1, \dots, P^2_t) \in \RHS$ with $h_2(P^1_1, \dots, P^1_t) = h_2(P^2_1, \dots, P^2_t)$.
    And for $i \in [t]$ and $j \in [2]$, let $P^j_i = (\psf_i^j, \psg_i^j, \pse_i^j, \psw_i^j)$.
    Also let $(\psf, \psg, \pse, \psw) = h_2(P^1_1, \dots, P^1_t) = h_2(P^2_1, \dots, P^2_t)$.
    Now fix an arbitrary $i \in [t]$.
    First, for every vertex $z \in \tail[u]$, %
    every $j \in [2]$, and every $X \in \var^{V}(\phi)$, since $P_i^j$ is compatible with the index $I[\pdeg \mapsto \pdeg']$ for $\tss(v_i)$ for some $\pdeg' \in \pdegdom(\tail[u])$, %
    we have $u \in \psf^j_i(X)$ if and only if $u \in \interp{f}(X)$.
    Further, %
    and every $X \in \var^{V}(\phi)$, since $\subtree[v_1], \dots, \subtree[v_t]$ are pairwise disjoint, 
    we have 
    \[
        \psf^1_i(X) \cap \subtree[v_i] = \psf(X) \cap \subtree[v_i] = \psf^2_i(X) \cap \subtree[v_i].
    \]
    Hence, $\psf_i^1 = \psf_i^2$.
    Further, for every $Y \in \var^E(\phi)$, since $\sheaf(v_1), \dots, \sheaf(v_t)$ are pairwise disjoint, we also have
    \[
         \psg^1_i(Y) = \psg(Y) \cap \sheaf(v_i) = \psg^2_i(Y),
    \]
    i.e., $\psg_i^1 = \psg_i^2$.
    Similarly, for every vertex $z \in \tail[u]$, %
    every $j \in [2]$, and every $C \in \Cat$, since $P_i^j$ is compatible with $I[\pdeg \mapsto \pdeg']$ for some $\pdeg' \in \pdegdom(\tail[u])$, %
    we have $\pse_i^j(z, C) = \inde(z, C)$.
    And for every %
    $j \in [2]$, $C \in \Cat$, and $z \in \subtree[v_i]$ we have 
    \[
        \pse^1_i(z, C) = \pse(z, C) = \pse^2_i(z, C),
    \]
    hence $\pse^1_i = \pse^2_i$.
    And finally, for every %
    $C \in \Cat$, since $\sheaf(v_1), \dots, \sheaf(v_t)$ are pairwise disjoint, we have 
    \[
        \psw^1_i(C) = \psw(C) \cap \sheaf(v_i) = \psw^2_i(C),
    \]
    i.e., $\psw^1_i = \psw^2_i$.
    Altogether, we get $(P^1_1, \dots, P^1_t) = (P^2_1, \dots, P^2_t)$ and therefore, $h_2$ is injective.
    \end{claimproof}
    The existence of injective mappings $h_1 \colon \Partial_u[I,a,b,\barc] \to \RHS$ and $h_2 \colon \RHS \to \Partial_u[I,a,b,\barc]$ implies that we have $|\Partial_u[I,a,b,\barc]| = |\RHS|$, i.e., \eqref{eq:join-proof-rewritten} holds, concluding the proof of \eqref{eq:p-join}.

    \medskip
    
    Now using that this property holds for every index $I \in \Index[u]$, we will apply the convolution result for DFT-transforms (cf.\ \cref{thm:fourier-convolution-result,thm:fourier-convolution-result-polynomials}) to prove the second claim of the lemma, namely that \eqref{eq:q-join} holds for all $I \in \Index[u]$.
    For this fix an arbitrary prime $p$ such that $\FF^*_p$ admits a $\pphi$-th root of unity.
    Before we proceed with the formal proof, let us briefly sketch the idea behind it.  
    If we fix a tuple $(\indf, \inde, \allow, \False)$ and only vary the mod-degree $\pdeg$, then $\cP_u[\indf, \inde, \allow, \cdot, \False]$ and $\cP_{v_i}(\indf, \inde, \allow, \cdot, \False)$ (modulo $p$) are functions from $(\tail[u] \times \Cat) \to \Fp$ to $\FF_{p}[\alpha, \beta, \gamma_1, \dots, \gamma_\ell]$, and $\cQ^p_u[\indf, \inde, \allow, \cdot, \False]$ and $\cQ^p_{v_i}(\indf, \inde, \allow, \cdot, \False)$ are their discrete Fourier transforms, respectively.
    With this, the convolution result (cf.\ \cref{thm:fourier-convolution-result,thm:fourier-convolution-result-polynomials}) directly implies the desired equality \eqref{eq:q-join}.
    Now we provide the formal details.

    In the remainder of the proof %
    $\omega$ denotes $\omega(p)$.
    Let $I = (\indf, \inde, \allow, \pdeg, \False) \in \Index[u]$ be arbitrary but fixed. 
    And let $D \defeq \tail[u] \times \Cat$. 
    Recall that we have $\tail[u] = \tail(v_i)$ for every $i \in [t]$.
    We define the mappings $h, h^1, \dots, h^t, g, g^1, \dots, g^t \colon (D \to \Fp) \to \FF_{p}[\alpha, \beta, \gamma_1, \dots, \gamma_\ell]$ via
    \begin{align*}
        h(y) \defeq \cP_u[\indf, \inde, \allow, y, \False] \bmod{p} && g(y) \defeq \cQ^p_u[\indf, \inde, \allow, y, \False] \bmod{p} \\
        h^i(y) \defeq \cP_{v_i}(\indf, \inde, \allow, y, \False) \bmod{p} && g^i(y) \defeq \cQ^p_{v_i}(\indf, \inde, \allow, y, \False) \bmod{p}
    \end{align*}
    for every $i \in [t]$ and every $y \colon D \to \Fp$.

    First, the equality \eqref{eq:p-join} implies that we have 
    \begin{equation}\label{eq:important-1}
        h(y) \equiv_p \sum\limits_{\substack{y^1, \dots, y^t \in \pdegdom(\tail[u]) \\ y^1 + \dots + y^t \equiv_{\pphi} y}} h^1(y^1) \cdot \ldots \cdot h^t(y^t) \equiv_p \sum\limits_{\substack{y^1, \dots, y^t \colon D \to \Fp \\ y^1 + \dots + y^t \equiv_{\pphi} y}} h^1(y^1) \cdot \ldots \cdot h^t(y^t).
    \end{equation}
    Second, by definitions of the polynomials $\cQ^p_u[\cdot]$ and $\cQ^p_{v_i}(\cdot)$ (cf.\ \cref{eq:def-q}) it holds that
    \[
        g(y) \equiv_p \sum_{q \in \pdegdom(\tail[u])} \omega^{q \cdot y} \cdot h(q) \equiv_p \sum_{q \colon D \to \Fp} \omega^{q \cdot y} \cdot h(q)
    \]
    and
    \[
        g^i(y) \equiv_p \sum_{q \in \pdegdom(\tail(v_i))} \omega^{q \cdot y} \cdot h^i(q)\equiv_p \sum_{q \colon D \to \Fp} \omega^{q \cdot y} \cdot h^i(q)
    \]
    for every $y \colon D \to \Fp$ and every $i \in [t]$, i.e., we have
    \begin{align}\label{eq:important-2}
        g \equiv_p \DFT_{D,\pphi,p}(h) && g^i \equiv_p \DFT_{D,\pphi,p}(h^i).
    \end{align}
    Then for every $y \colon D \to \FF_p$ we have:
    \begin{align*}
        h(y) \stackrel{\eqref{eq:important-1}}{\equiv_p} &\sum\limits_{\substack{y^1, \dots, y^t \colon D \to \Fp \\ y^1 + \dots + y^t \equiv_{\pphi} y}} h^1(y^1) \cdot \ldots \cdot h^t(y^t) \\
        \stackrel{\text{\cref{thm:fourier-convolution-result-polynomials}}}{\equiv_p} &\Bigl(\DFT_{D,\pphi,p}^{-1}\bigl(\DFT_{D,\pphi,p}(h^1) \cdot \ldots \cdot \DFT_{D,\pphi,p}(h^t)\bigr)\Bigr)(y) \stackrel{\eqref{eq:important-2}}{\equiv_p} \bigl(\DFT_{D,{\pphi},p}^{-1}(g^1 \cdot \ldots \cdot g^t)\bigr)(y),
    \end{align*}
    i.e., 
    \[
        h \equiv_p \DFT_{D,\pphi,p}^{-1}(g^1 \cdot \ldots \cdot g^t),
    \]
    and hence
    \begin{equation}\label{eq:g-equality-product-g-i-s}
        g \stackrel{\eqref{eq:important-2}}{\equiv_p} \DFT_{D,\pphi,p}(h) \equiv_p \DFT_{D,\pphi,p}(\DFT_{D,\pphi,p}^{-1}(g^1 \cdot \ldots \cdot g^t)) \stackrel{\text{\cref{thm:inverse-dft-polynomials}}}{\equiv_p} g^1 \cdot \ldots \cdot g^t.
    \end{equation}
    Finally, by definitions of $g, g^1, \dots, g^t$ we obtain
    \[
        \cQ^p_u[I] = \cQ^p_u[\indf, \inde, \allow, \pdeg, \False] \equiv_p g(\pdeg) \stackrel{\eqref{eq:g-equality-product-g-i-s}}{\equiv_p} (g^1 \cdot \ldots \cdot g^t)(\pdeg) \equiv_p \prod_{i=1}^t \cQ^p_{v_i}(\indf, \inde, \allow, \pdeg, \False) = \prod_{i=1}^t \cQ^p_{v_i}(I)
    \]
    concluding the proof of \eqref{eq:q-join}.
\end{proof}

Let us summarize what we achieved so far.
Suppose we have a prime number $p$ such that $\FF_p^*$ admits a $\pphi$-th root of unity and $p$ is sufficiently large, i.e., $p$ is larger than $|\Partial(\tss)|$ for every triad $\tss$ of $G$, but the length of $p$ is still polynomial in $n$---we provide a detailed discussion on the choice of the primes in the next section.
Now, because $p$ is sufficiently large, it suffices to carry out all computations in $\FF_p$ yielding the correct numbers of partial solutions for every relevant triad of $G$.
By \cref{lem:modelcheck:equivalence:G:models:phi,lem:root:inclusion:exclusion,lem:root-p-equal-q-lemma}, to check whether $G \vDash \phi$ holds, it suffices to compute the coefficients of the polynomials $\cQ^p_r(I)$ for the root $r$ and all indexes $I \in \Index(r)$ one after another.
After fixing one such index $I$, we recursively compute the polynomial $\cQ^p_r(I)$ by alternatingly applying \cref{lem:modelcheck:forget} and \cref{lem:modelcheck:join} until a leaf-node is reached where we rely on \cref{lem:modelcheck:leaf}.
At each level of the recursion we keep track of the current index as well as coefficients of the current polynomials. 
So the space complexity of this process is then polynomial in $n$.
Crucially, no branching occurs at join-computations.
So the running time is dominated by the branching factor of 
at most $2^{|\phi|} \cdot (\dphi+1)^{2\cdot|\phi|} \cdot \pphi^{|\phi|}$
at the forget-computations---here we use that $|\Cat| \leq |\phi|$ holds.
This branching factor upper-bounds the options for how the current index can be extended to the ``forgotten vertex'' (cf.\ \cref{lem:modelcheck:forget}).
It is straightforward to verify that the leaf procedure in \cref{lem:modelcheck:leaf} takes time $2^{|\phi|} \cdot |\phi|^{\cO(1)} \cdot n^{\cO(1)}$ while the non-recursive work at each other node is polynomial.
These ideas result in an algorithm with running time $((\dphi+1) \cdot \pphi)^{\cO(|\phi| \cdot \td)} \cdot n^{\cO(\ell)}$ using only polynomial space.
In the next section, we improve the space complexity to achieve only logarithmic dependence on $n$, in particular, we provide how the prime numbers $p$ are chosen.
In fact, we will use multiple primes such that the length of each of them is only logarithmic in $n$.

\section{Model Checking}
\label{sec:modelcheck}
Here, we show how given the equalities from the previous section, we can obtain the desired space-efficient algorithm checking whether the input formula $\phi \in \NEOtwo[\finrec]$ is satisfied by the graph $G$ when an elimination tree $T$ of depth $\td$ is provided with the input.
As sketched before, directly using those equalities, one could obtain a recursive algorithm with running time single-exponential in $\td$ and polynomial space complexity. 
In this section, we show how to improve the space complexity to have only logarithmic dependency on $n$ and eventually prove our main result.

\subsection{Getting to logarithmic space complexity}

To improve the space complexity to have only logarithmic dependency on $n$, we will rely on the Chinese remainder theorem.
It allows us to employ the evaluation of the polynomials from the previous subsection in a small set of short prime fields in order to reconstruct the coefficients of the desired polynomials for the root of $T$ in logarithmic in $n$ space.
This idea was first used on elimination trees by Pilipczuk and Wrochna~\cite{PilipczukW18} to solve the \textsc{Dominating Set} problem. 
In order to ensure that our equalities from the previous subsection hold, we will use the small primes $p$ such that, additionally, $\FF^*_p$ admits the $\pphi$-th root of unity. 
The idea of combining the Chinese remainder theorem with discrete Fourier transforms was used by Focke et al.~\cite{FockeMINSSW25} to obtain fast convolutions for $(\sigma, \rho)$-\textsc{Dominating Set} on graphs of bounded treewidth.
We remark that their algorithm has exponential space complexity while we aim for polynomial space.
Now we formalize this sketch.
We first set up the algebraic tools and after that, put everything together to obtain the desired algorithm.

First of all, the approach of Pilipczuk and Wrochna~\cite{PilipczukW18} only works for univariate polynomials so we start by providing a coefficient-preserving transformation of multivariate polynomials of bounded degree %
into univariate polynomials.
Bergougnoux et al.~\cite{BergougnouxCGKM25} carried out this transformation for polynomials in two variables, and we now generalize their idea to our polynomials in $\ell+2$ variables.
For this let 
\[
    \opa = 2 \cdot |\phi| \cdot \dphi \cdot n \cdot \td + 1,
\]
the choice of $\opa$ will become clear in the proof of \cref{lem:p-equal-evaluation}.
We also define 
\[
    \opa^* = \opa^{\ell+2} - 1.
\]
Now consider an arbitrary polynomial $\cP$ in $\ell+2$ formal variables $\alpha, \beta, \gamma_1, \dots, \gamma_\ell$ of form
\[
    \cP = \sum_{0 \leq a, b, c_1, \dots, c_\ell < \opa} q_{a, b, \barc} \cdot \alpha^a \cdot \beta^b \cdot \gamma^{\barc} \in \bZ[\alpha, \beta, \gamma_1, \dots, \gamma_\ell],
\]
i.e., the degree of every variable is upper-bounded by $\opa$ (in the above, we use $\barc = (c_1,\ldots,c_\ell)$).
Then we define the univariate polynomial $\cP^{\uparrow} \in \bZ[x]$ as
\[
    \cP^{\uparrow} = \sum_{0 \leq a, b, c_1, \dots, c_\ell \leq \opa} q_{a, b, \barc} \cdot x^a \cdot \Bigl(\prod_{i \in [\ell]} x^{c_i \cdot \opa^i}\Bigr) \cdot x^{b \cdot \opa^{\ell+1}} = \sum_{0 \leq a, b, c_1, \dots, c_\ell \leq \opa} q_{a, b, \barc} \cdot x^{a +\sum_{i \in [\ell]} c_i \cdot \opa^i + b \cdot \opa^{\ell+1}}.
\]

The following properties are crucial.
\begin{observation}\label{obs:p-prime-polynomial}
    The sum in the exponent of each monomial in $\cP^{\uparrow}$ is the representation of some number between $0$ and $\opa^*$ to the base $\opa$.
    Furthermore, every number has a unique base $\opa$ representation.
    Thus the following hold:
    \begin{itemize}
        \item The degree of $\cP^{\uparrow}$ is at most $\opa^*$.
        \item %
        For every exponent $0 \leq j \leq \opa^*$, the coefficient at $x^j$ in $\cP^{\uparrow}$ is precisely $q_{a, b, \barc}$ for the unique representation $j = a + \Bigl(\sum_{i \in [\ell]} c_i \cdot \opa^i\Bigr) + b \cdot \opa^{\ell+1}$ of $j$ to the base $\opa$---hence, the task of reconstructing the coefficients of $\cP$ is equivalent to the one of $\cP^{\uparrow}$.
        \item For any $s \in \bZ$ and for the evaluation of $\cP^{\uparrow}$ at $s$, it holds that
        \[
            \cP^{\uparrow}(s) = \cP\bigl(s, s^{\opa^{\ell+1}}, (s^{\opa^i})_{i \in [\ell]}\bigr).
        \]
    \end{itemize}
\end{observation}
For univariate polynomials we will rely on the following implication of the Chinese remainder theorem.
\begin{lemma}[\!\!\cite{ChiuDL01,PilipczukW18}]\label{thm:chinese-remainder}
    Let $n'$ be a non-negative integer and let $\cP(x) = \sum\limits_{i=0}^{n'} q_i \cdot x^i$ be a univariate polynomial of degree at most $n'$ with integer coefficients satisfying $0 \leq q_i \leq 2^{n'}$ for every $i \in [n']_0$. 
    Further, let $c$ be a non-negative integer. 
    And let $p_1, \dots, p_{n'}$ be prime numbers satisfying $n'+1 < p_1 < \dots < p_{n'} \leq c \cdot (n')^2$ and $\prod_{i=1}^{n'} p_i > 2^{n'}$ and such that for every $i \in [n']$, the value $p_i$ can be computed in time $(c \cdot n')^{\cO(1)}$ and space $\mathcal{O}(\log n' + \log c)$.
    Suppose that given an integer $i \in [n']$ %
    and a value $s \in \FF_{p_i}$, the evaluation $\cP(s) \pmod {p_i}$ can be computed in time~$\operatorname{time}$ and space $\operatorname{space}$. 
    Then the coefficients $q_1, q_2, \dots, q_{n'}$ of $\cP$ can be computed in time  $\mathcal{O}(\operatorname{time}) \cdot (c \cdot n')^{\cO(1)}$ and space $\mathcal{O}(\operatorname{space} + \log n' + \log c)$.
\end{lemma}
\begin{proof}
This lemma is a slightly adapted version of Theorem 5.3 by Pilipczuk and Wrochna~\cite{PilipczukW18}.
For the sake of completeness we provide the full proof of this adapted statement.
It suffices to describe how to compute a single coefficient $q_k$ as this can then be repeated for every $k \in [n']_0$ reusing the same space.
Due to $\prod_{i=1}^{n'} p_i > 2^{n'} \geq q_k$, by the Chinese remainder theorem the coefficient $q_k$ is uniquely determined by the rests $q_k \pmod{p_1}, q_k \pmod {p_2}, \dots, q_k \pmod{p_{n'}}$.
Moreover, there is an algorithmic version of this theorem \cite{ChiuDL01} stating that from these rests, the bits of $q_k$ can be recovered in time $(c \cdot n')^{\cO(1)}$ and space $\cO(\log c + \log n')$.
For clarification let us remark that this procedure does not store all required rests simultaneously but rather uses the compositionality of logspace algorithms and then outputs the bits of $q_k$ one after another.

So we remain with the task of computing $q_k \pmod{p_i}$ given some $i \in [n']$.
First, the prime $p_i$ can be computed in time $(c \cdot n')^{\cO(1)}$ and space $\cO(\log n' + \log c)$ by assumption.
The bit length of $p_i$ is also upper-bounded by $\cO(\log n' + \log c)$ due to $p_i \leq c \cdot (n')^2$.
So in the same time and space we can compute a generator $\alpha_i$ of the group $\FF^*_{p_i}$ by trying every number $t \in \FF_{p_i}^*$ and computing all of its exponents $t^j$ modulo $p_i$ with $j \in [p_i-1]$: the candidate $t$ is a generator if and only if the smallest exponent $j$ satisfying $t^j \equiv 1 \bmod{p_i}$ is equal to $p_i - 1$.
In their proof Pilipczuk and Wrochna~\cite{PilipczukW18} show that $n' < p_i - 1$ implies the equality
\[
    q_k \equiv_{p_i} - \sum_{j = 0}^{p_i - 2} \cP(\alpha_i^j) \cdot \alpha_i^{-jk}. 
\]
And this value can now be computed by evaluating $\cP$ at most $(c \cdot n')^{\cO(1)}$ times and carrying out at most $(c \cdot n')^{\cO(1)}$ arithmetic operations in $\FF_{p_i}$.
Altogether, the rest $q_k \pmod{p_i}$ is computed in time $\cO(\operatorname{time}) \cdot (c \cdot n')^{\cO(1)}$ and space $\cO(\operatorname{space} + \log n' + \log c)$ so the lemma follows.
\end{proof}

By \cref{lem:modelcheck:upperbounds:abc,lem:modelcheck:equivalence:G:models:phi}, to check if $G$ satisfies $\phi$, it suffices to iterate over all $\False \subseteq \EQ(\phi)$, all integers $j \in [2 \cdot |\phi| \cdot \dphi \cdot n \cdot \td]_0$, and all $\barc\in([n\cdot \td]_0)^\ell$ and check whether $\ip{\phi}^{(\False, \barc)} = 1$ holds and the monomial $\alpha^j \cdot \beta^j \cdot \gamma^{\barc}$ has a non-zero coefficient in $\cP^=_r(\emptyset, \emptyset, \emptyset, \emptyset, \False)$ %
(here $r$ denotes the root of the elimination tree~$T$).
To carry out this check in a space-efficient way we will make use of \cref{thm:chinese-remainder}.
We now define the set $p_1, \dots, p_m$ of primes that will be used in our algorithm.
The choices of $m$ and $\maxq$ in the definition below will become clear in the proof of \cref{thm:neo-theorem}.

\begin{definition}
    Let $\maxq = 2^{5 \cdot n \cdot \td \cdot |\phi| \cdot \log(\dphi+1)}$ and let $m = (5 \cdot n \cdot \td \cdot |\phi| \cdot \dphi)^{\ell+2}$.
    Consider the infinite set of the prime numbers having the rest $1$ modulo $\pphi$. We define $p_1 < p_2 < \dots < p_m$ as the $(m+1)$-st, $(m+2)$-nd, $\dots$, and the $(2m)$-th smallest number in this set, respectively.
\end{definition}

\begin{lemma}\label{lem:choice-of-the-primes}
    The following properties hold:
    \begin{enumerate}
    \begin{multicols}{3}
        \item $m \geq \opa^*$,
        \item $2^m \geq \maxq$,
        \item $\prod_{i \in [m]} p_i > 2^m$,
    \end{multicols}
        \item for $i \in [m]$, we have $m+1 < p_i \in \cO(2^{\pphi} \cdot m \cdot \log m)$,
        \item for $i \in [m]$, $\FF^*_{p_i}$ admits the $\pphi$-th root of unity, i.e., $\omega(p_i)$ is well-defined,
        \item for $i \in [m]$, the numbers $p_i$, $\omega(p_i)$, and the inverse of $\pphi^{|\Cat|}$ in $\FF_{p_i}$ can be computed in $\cO(\pphi + \log m)$ space and $(2^{\pphi} \cdot m \cdot |\phi|)^{\cO(1)}$ time.
    \end{enumerate}
\end{lemma}

\begin{proof}
    The choice of $m$, $\maxq$, and $\opa^*$ ensures that we have $m \geq \opa^*$ and $2^m \geq \maxq$.
    Further, each of $p_1, \dots, p_m$ is larger than $2$ so their product is larger than $2^m$.
    Since $p_1$ is larger than the $m$-th prime number, we have $p_1 > m + 1$.
    By the Prime Number Theorem for Arithmetic Progressions (Theorem 1.5 in \cite{bennett2018explicit}, Remark 4.20 in \cite{FockeMINSSW25}) it holds that $p_m \in \cO(2^{\pphi} \cdot m \cdot \log m)$.
    For $i \in [m]$, the number $\pphi$ divides the order $p_i-1$ of $\FF^*_{p_i}$ by definition so $\FF^*_{p_i}$ admits the $\pphi$-th root of unity.
    
    To compute $p_i$ we can iterate through positive integers $\pphi+1, 2 \cdot \pphi + 1,3 \cdot \pphi+ 1, \dots$ and for every candidate, check whether it is prime.
    For this, we iterate through all integers $a$ larger than one and smaller than the current candidate and check whether $a$ divides the candidate.
    Let us remark that we do not use the prominent polynomial-time algorithm for  \textsc{Prime}~\cite{agrawal2004} to ensure that the space complexity is only linear in the length of the candidate.
    Once we find the $(m+i)$-th candidate $p_i$ satisfying this property, we stop.
    The upper bound on the value of $p_i$ implies that at most $\cO(2^{\pphi} \cdot m \cdot \log m)$ iterations suffice and each iteration consists of $(2^{\pphi} \cdot m \cdot \log m)^{\cO(1)}$ arithmetic operations which can be carried out in space $\cO(\pphi + \log m)$.
    After computing $p_i$, we can iterate through candidates $1, 2, \dots$ from $\FF^*_{p_i}$ to check whether it is a $\pphi$-th root of unity: for this we compute (modulo $p_i$) all powers of this candidate up to $\pphi$-th and check whether $\pphi$-th power is the first power equal to $1$---similarly, this works in the claimed time and space.
    By the Fermat's Little theorem, for every non-zero number $a \in \FF_{p_i}$, we have $a^{p_i-1} \equiv_{p_i} 1$.
    Therefore, the inverse of $\pphi^{|\Cat|}$ in $\FF_{p_i}$ is equal to $(\pphi^{|\Cat|})^{p_i-2} \bmod{p_i}$.
    Thus, using exponentiation by squaring, we can compute this inverse in $\cO(\log p_i)$ arithmetic operations in $\FF_{p_i}$
    Altogether, all of those computations can be carried out in space linear in the length of $p_i$ and time polynomial in the sum of $p_i$ and $|\Cat| \leq |\phi|$ so the claim follows.
\end{proof}

So now to apply \cref{thm:chinese-remainder} with $n' = m$, it remains to describe how, given $i \in [m]$ and a number $s \in \FF_{p_i}$, to evaluate the polynomial $(\cP^=_r(\emptyset, \emptyset, \emptyset, \emptyset, \False))^{\uparrow}$ at $s$ modulo $p_i$.
We will use the simple fact that the evaluation of a sum (resp.\ product) of two polynomials is equal to the sum (resp.\ product) of their evaluations. 
We start by proving the following intermediate lemma about the evaluation of $\cQ^{p_i}_r(\emptyset, \emptyset, \emptyset, \emptyset, \False)$.

\begin{lemma}\label{lem:evaluation-of-q-root}
    Let $r$ be the root of the elimination tree $T$ and let $i \in [m]$.
    And let $\bars = (s_1, \dots, s_{\ell+2}) \in \FF_{p_i}^{\ell+2}$ be a vector of length $\ell+2$ with entries from $\FF_{p_i}$ such that for every $j \in [\ell+2]$ the value $s_j$ can be computed in time $\operatorname{time}$ and space $\operatorname{space}$.    
    Then given the prime number~$p_i$ and a set $\False \subseteq \EQ(\phi)$, we can evaluate (modulo $p_i$) the polynomial $\cQ^{p_i}_r(\emptyset, \emptyset, \emptyset, \emptyset, \False)$ at the vector $\bars$ in time $(2 \cdot (\dphi+1)^2 \cdot \pphi)^{|\phi| \cdot(\td+\cO(1))} \cdot 2^{\cO(\pphi)} \cdot n^{\cO(\ell)} \cdot \operatorname{time}$ and space $\cO(\td \cdot |\phi| \cdot (\log n + \log{|\phi|} + \log{\dphi} +{\pphi}) + \operatorname{space})$.
\end{lemma}
\begin{proof}
    We recall that by \cref{lem:choice-of-the-primes} we can compute $\omega(p_i)$ as well as the inverse of $\pphi^{|\Cat|}$ in $\FF_{p_i}$ in $\cO(\pphi + \log m)$ space and $(2^{\pphi} \cdot m \cdot |\phi|)^{\cO(1)} \leq (\dphi+1)^{\cO(|\phi|)} \cdot 2^{\cO(\pphi)} \cdot n^{\cO(\ell)}$ time.
    To compute the value $\cQ^{p_i}_r(\emptyset, \emptyset, \emptyset, \emptyset, \False)(\bars) \pmod{p_i}$, we alternatingly use \cref{lem:modelcheck:forget}~\eqref{eq:forget-node-transformed} and \cref{lem:modelcheck:join}~\eqref{eq:q-join} to recurse until we reach a polynomial of form $\cQ^{p_i}_{u}[I]$ for some leaf $u$ and some index $I \in \Index[u]$.
    Then we use \cref{lem:modelcheck:leaf} to evaluate this polynomial at $\bars$.
    At each level of recursion we store the information about the current recursive call as well as $\cO(1)$ partial sums and products in $\FF_{p_i}$.
    Now we bound the space and time complexity of this procedure.

    \subparagraph{Time complexity.}
    During the recursive evaluation at some node $u$ via \cref{lem:modelcheck:forget}~\eqref{eq:forget-node-transformed}, the restriction of the index on $\tail(u)$ remains unchanged and we only branch at most once on every possible restriction of the index on $\{u\}$.
    Also the value $\False$ remains fixed.
    During a recursive evaluation using  \cref{lem:modelcheck:join}~\eqref{eq:q-join}, the index remains the same in all recursive calls.
    Hence, for every node~$u$ and every index $I \in \Index[u]$ (resp.\ $I \in \Index(u)$), the evaluation of the polynomial $\cQ^{p_i}_u[I]$ (resp.\ $\cQ^{p_i}_u(I)$) at $\bars$ is computed at most once during this process. 
    
    We recall that $\td$ denotes the depth of $T$ so $\abs{\tail[u]}, \abs{\tail(u)} \leq \td$ holds for every node $u$ of $T$.
    Hence, for $\tail \in \{\tail[u], \tail(u)\}$ there are at most $2^{\abs{\phi} \cdot \td}$ vertex interpretations on $\tail$, at most $(\dphi+1)^{\abs{\phi} \cdot \td}$ neighborhood functions on $\tail$, and at most $\pphi^{\abs{\phi} \cdot \td}$ mod-degrees on $\tail$.
    So each of the sets $\Index[u]$ and $\Index(u)$ contains at most $(2 \cdot (\dphi+1)^2 \cdot \pphi)^{\abs{\phi} \cdot \td}$ indexes in which the last element of the tuple is equal to $\False$.
    The tree $T$ contains exactly $n$ nodes.
    So altogether there are at most $2 \cdot n \cdot (2 \cdot (\dphi+1)^2 \cdot \pphi)^{\abs{\phi} \cdot \td}$ recursive calls.

    Now we bound the running time required after each recursive call, we clarify this now.
    For \cref{lem:modelcheck:join}~\eqref{eq:q-join}, after computing the corresponding evaluation for one of the children of a node, this additional running time is determined by one multiplication in $\FF_{p_i}$ in order to update the current product.
    Next we argue that for the recursive equality \cref{lem:modelcheck:forget}~\eqref{eq:forget-node-transformed}, the additional running time after one recursive call
    is at most $(\ell+1) \cdot \operatorname{time}$ plus a value polynomial in $p_i + \pphi^{|\phi|}$. 
    The former addend occurs because for each of the $\ell+1$ formal variables other than $\beta$, it suffices to compute the corresponding coordinate of $\bars$ once (and then compute a certain power of it).
    We emphasize that we reuse the same space (upper-bounded by $\operatorname{space}$) for each coordinate of $\bars$ after each recursive call.
    In the latter addend, the dependence on $\pphi^{|\phi|}$ occurs due to considering every mod-degree $\hat{q}$ of a single vertex, this iteration also requires additional $|\phi| \cdot \log \pphi$ space to keep track of $\hat{q}$. %
    By the choice of $p_i$, the remaining arithmetic operations, for instance, computing the corresponding value of $\incexc$, updating the current sums, computing the powers of $\omega(p_i)$, and evaluating the powers of $\alpha$ and $\gamma$ at the corresponding elements of $\bars$, are also carried out in time polynomial in $p_i$ and space linear in $\log p_i$.
    We remark that a more careful analysis of this running time is possible but we omit it for the sake of readability.\footnote{For instance, arithmetic operations in $\FF_{p_i}$ can be carried out in time polynomial in $\log p_i$ rather than $p_i$, or even faster, depending on the machine model. However, the brute-force computation of, e.g., $p_i$ and $\omega(p_i)$, dominates these bounds.}
    And for the recursion using \cref{lem:modelcheck:join}~\eqref{eq:q-join} the additional running time after one recursive call is even smaller as only one additional multiplication occurs.
    By \cref{obs:prelim:treedepth:nb:edges}, the graph~$G$ has at most $n \cdot \td$ edges and every edge occurs in a sheaf of precisely one leaf of $T$.
    So in total, for the non-recursive leaf computations via \cref{lem:modelcheck:leaf} the running time is bounded by $(n \cdot \td \cdot 2^{|\phi|} \cdot |\phi|) \cdot p_i^{\cO(1)} \cdot \operatorname{time}$.
    Here, the factor $\operatorname{time}$, again, allows recomputing the coordinates of $\bars$ corresponding to formal variables $\beta$ and $\gamma$ on the fly.
    By \cref{lem:choice-of-the-primes}, we have $p_i \in \cO(2^{\pphi} \cdot m \cdot \log m)$.
    The definition of $m$ implies the claimed time bound. %

    \subparagraph{Space complexity.} 
    We already argued that we require only $\operatorname{space}$ to compute each coordinate of $\bars$ whenever it is needed.
    Now we bound the remaining space.
    At each of $\cO(\td)$ levels of recursion we store the information about the current recursive call (i.e., the index queried) as well as a constant number of partial sums and products in $\FF_{p_i}$, the latter were bounded earlier.
    Since during the recursive evaluation at some node $u$, the index restricted to $\tail(u)$ remains unchanged, at each level of recursion it suffices to store only the restriction of the index to $\{u\}$ instead of the whole index.
    Thus, using the bound on the the length $\log p_i$ of the numbers in $\FF_{p_i}$ we upper-bound the total space complexity by $\cO(\td \cdot |\phi| \cdot ( \log n + \log{|\phi|} + \log{\dphi} + {\pphi}) + \operatorname{space}))$.
\end{proof}
From this we can derive the following:
\begin{lemma}\label{lem:p-equal-evaluation}
    Let $r$ be the root of the elimination tree $T$ and let $i \in [m]$.
    First, the polynomial $(\cP^=_r(\emptyset, \emptyset, \emptyset, \emptyset, \False))^{\uparrow}$ is well-defined for every $\False \subseteq \EQ(\phi)$.
    Second, given the prime number $p_i$, a value $s \in \FF_{p_i}$, and a set $\False \subseteq \EQ(\phi)$, we can evaluate (modulo $p_i$) the polynomial $(\cP^=_r(\emptyset, \emptyset, \emptyset, \emptyset, \False))^{\uparrow}$ at the value $s$
    within time $(2 \cdot (\dphi+1)^2 \cdot \pphi)^{|\phi|\cdot(\td+\cO(1))} \cdot 2^{\cO(\pphi)} \cdot n^{\cO(\ell)}$ and space $\cO(\td \cdot |\phi| \cdot (\log n + \log{|\phi|} + \log{\dphi} + \pphi))$.
\end{lemma}
\begin{proof}
    By \cref{obs:prelim:treedepth:nb:edges}, the graph $G$ has at most $n \cdot \td$ edges.
    Then by \cref{lem:modelcheck:upperbounds:abc}, for every non-zero monomial in $\cP^=_r(\emptyset, \emptyset, \emptyset, \emptyset, \False)$, the degree of every variable is smaller than $\opa$. 
    So the polynomial $(\cP^=_r(\emptyset, \emptyset, \emptyset, \emptyset, \False))^{\uparrow}$ is indeed well-defined.
    Then by \cref{obs:p-prime-polynomial}, the evaluation of $(\cP^=_r(\emptyset, \emptyset, \emptyset, \emptyset, \False))^{\uparrow}$ at $s$ is equal to the evaluation of $\cP^=_r(\emptyset, \emptyset, \emptyset, \emptyset, \False)$ at the vector $\bars$ defined as $\bars \defeq (s, s^{\opa^{\ell+1}}, (s^{\opa^i})_{i \in [\ell]}) \in \FF_{p_i}^{\ell+2}$.
    So as the first step we explain how to compute any coordinate, say $s_j$ with $j \in [\ell+2]$, of the vector $\bars$.
    To carry this out efficiently %
    and 
    bound the space complexity, we rely on the iterative implementation of exponentiation by squaring.
    Then the computation of $s_j$ requires $\cO(\log(\opa^{\ell+1})) = \cO(\ell \cdot \log \opa)$ arithmetic operations in $\FF_{p_i}$ as well as arithmetic operations on numbers upper-bounded by $\opa^{\ell+1}$, the latter operations are to correctly update the current exponent.
    The definitions of $\ell \leq |\phi|$, $\opa$, and $p_i$ imply that this %
    running time can be upper-bounded by $\operatorname{time} \defeq (\pphi + \log n + |\phi| + \log \dphi)^{\cO(1)}$.
    As for the space complexity, it can be upper-bounded by $\cO(\ell \cdot \log \opa + \log(p_i))$.
    This is because the exponentiation by squaring keeps track of the current exponent as well as a constant number of numbers from $\FF_{p_i}$.
    Again, the definitions of $\ell \leq |\phi|$, $\opa$, and $p_i$ imply that this is upper-bounded by $\operatorname{space} \defeq \cO(\td \cdot |\phi| \cdot (\log n + \log{\dphi} + \log{|\phi|} + \pphi))$.
    
    Now by \cref{lem:root:inclusion:exclusion}
    the evaluation of $\cP^=_r(\emptyset, \emptyset, \emptyset, \emptyset, \False)$ at $\bars$ boils down to the evaluation of $\cP_r(\emptyset, \emptyset, \emptyset, \emptyset, \False')$ at $\bars$ for every $\False' \subseteq \False$ and summing up the results with the correct sign.
    So fix $\False' \subseteq \False$.
    By \cref{lem:root-p-equal-q-lemma}, it holds that $\cP_r(\emptyset, \emptyset, \emptyset, \emptyset, \False')(\bars) \equiv_{p_i} \cQ^{p_i}_r(\emptyset, \emptyset, \emptyset, \emptyset, \False')(\bars)$.
    By \cref{lem:evaluation-of-q-root} we can compute this evaluation in time $(2 \cdot (\dphi+1)^2 \cdot \pphi)^{|\phi|\cdot(\td+\cO(1))} \cdot 2^{\cO(\pphi)}  \cdot n^{\cO(\ell)} \cdot \operatorname{time}$ and space $\cO(\td \cdot |\phi| \cdot (\log{\dphi} + \log{\pphi} + \log n + \log{|\phi|}) + \operatorname{space})$.
    This process is then repeated at most $2^{|\phi|}$ times (for every $\False' \subseteq \False$) so the lemma follows.
\end{proof}

\subsection{Wrapping up the model checking algorithm}

Now we are ready to prove the main result of this section.

\begin{theorem}\label{thm:neo-theorem}
    There is an algorithm that given a formula $\phi \in \NEOtwo[\finrec]$ with $\ell$ size measurements, a graph $G$, and an elimination tree $T$ of $G$ of depth at most $\td$, runs in time $ ((\dphi+1) \cdot \pphi)^{\cO(|\phi|\cdot\td)} \cdot 2^{\cO(\pphi)} \cdot n^{\cO(\ell)}$ and space $\cO(\td \cdot |\phi| \cdot (\log n + \log{|\phi|} + \log{\dphi} + \pphi))$ and decides whether $G \vDash \phi$ holds.
\end{theorem}

\begin{proof}
    First, by \cref{obs:core} we may assume that $\phi$ is in $\core\NEOtwo[\finrec]$: otherwise, we can transform it with only a linear blowup in the length (we remark that this blowup is the reason for $\cO(|\phi|\cdot\td)$ instead of $|\phi|\cdot(\td+\cO(1))$ in the exponent of the running time), without increasing the values $\dphi$ and~$\pphi$, and by increasing the value $\ell$ by at most $2$.
    Let now $r$ denote the root of $T$.
    By \cref{lem:modelcheck:upperbounds:abc,lem:modelcheck:equivalence:G:models:phi}, to check if $G \vDash \phi$ holds, it suffices to iterate over all $\False \subseteq \EQ(\phi)$, all integers $j \in [2 \cdot |\phi| \cdot \dphi \cdot n \cdot \td]_0$, and all tuples $\barc = (c_1,\ldots, c_\ell)\in([n\cdot \td]_0)^\ell$ and check whether $\ip{\phi}^{(\False, \barc)} = 1$ holds and the coefficient in $\cP^=_r(\emptyset, \emptyset, \emptyset, \emptyset, \False)$ in front of $\alpha^j \cdot \beta^j \cdot \gamma^{\barc}$ is non-zero.
    We recall that $\ell \leq |\phi|$ holds.
    First of all, note that we can iterate over such candidates and check whether $\ip{\phi}^{(\False, \barc)} = 1$ holds in time 
    \begin{equation}\label{eq:run-time-trying-candidates}
        \cO(2^{|\phi|} \cdot |\phi| \cdot \dphi \cdot n \cdot \td \cdot (n \cdot \td + 1)^{\ell}) \cdot |\phi|^{\cO(1)} = 2^{|\phi|} \cdot |\phi|^{\cO(1)} \cdot \dphi \cdot n^{\cO(\ell)}
    \end{equation}
    and space $\cO(|\phi| \cdot \log n + \log \dphi)$.
    
    So now we focus on checking whether the corresponding coefficient is non-zero for fixed values of~$j$,~$\barc$, and~$\False$.
    By \cref{lem:p-equal-evaluation}, the polynomial $(\cP^=_r(\emptyset, \emptyset, \emptyset, \emptyset, \False))^{\uparrow}$ is well-defined. 
    By \cref{obs:p-prime-polynomial} its degree is upper-bounded by $\opa^*$ %
    and checking whether the desired coefficient is non-zero is equivalent to determining whether the coefficient at $x^{j + \Bigl(\sum_{i \in [\ell]} c_i \cdot (\opa+1)^i\Bigr) + j \cdot (\opa+1)^{\ell+1}}$ in $(\cP^=_r(\emptyset, \emptyset, \emptyset, \emptyset, \False))^{\uparrow}$ is non-zero.
    Note that to decide this, it suffices to compute the bits of this coefficient one by one and check if at least one is non-zero.

    By \cref{lem:bound-on-coefficients-of-p}, every coefficient of $\cP^=_r(\emptyset, \emptyset, \emptyset, \emptyset, \False)$, and equivalently, of $(\cP^=_r(\emptyset, \emptyset, \emptyset, \emptyset, \False))^{\uparrow}$, is non-negative and upper-bounded by $\maxq$.
    Furthermore, the degree of $(\cP^=_r(\emptyset, \emptyset, \emptyset, \emptyset, \False))^{\uparrow}$ is at most~$\opa^*$.
    We will apply \cref{thm:chinese-remainder} to the polynomial $(\cP^=_r(\emptyset, \emptyset, \emptyset, \emptyset, \False))^{\uparrow}$, the value $n' = m$, $c = 2^{\pphi}$, and the primes $p_1, \dots, p_m$.
    By \cref{lem:choice-of-the-primes} together with observations we just made, these values indeed satisfy the preconditions of the lemma.
    By \cref{lem:p-equal-evaluation}, for any $i \in [m]$ and any $s \in \FF_{p_i}$ we can compute the evaluation $(\cP^=_r(\emptyset, \emptyset, \emptyset, \emptyset, \False))^{\uparrow}$ at $s$ modulo $p_i$ in time $\operatorname{time} \defeq 
    ((\dphi+1)^2 \cdot \pphi)^{\cO(|\phi|\cdot\td)} \cdot 2^{\cO(\pphi)} \cdot n^{\cO(\ell)}$ (we discard $2$ in the base due to $\cO(|\phi|)$ in the exponent) and space $\operatorname{space} \defeq \cO(\td \cdot |\phi| \cdot (\log n + \log{|\phi|} +\log{\dphi} + \pphi))$.

    Recall that we have $p_i \in \cO(2^{\pphi} \cdot m^2)$ for each $i \in [m]$.
    \cref{thm:chinese-remainder} then implies that we can reconstruct the bits of the coefficient of $(\cP^=_r(\emptyset, \emptyset, \emptyset, \emptyset, \False))^{\uparrow}$ at the monomial $x^{j + \Bigl(\sum_{i \in [\ell]} c_i \cdot (\opa+1)^i\Bigr) + j \cdot (\opa+1)^{\ell+1}}$ by %
    increasing the running time by a factor of $(c \cdot n')^{\cO(1)}$ and using total space $\cO(\operatorname{space} + \log{n'} + \log c)$.
    Repeating this process for every candidate $j, \barc, \False$ increases the running time by a factor of \eqref{eq:run-time-trying-candidates} and keeps the same space complexity so the theorem follows.
\end{proof}

\section{Extensions of the Logic}
\label{sec:extensions}

\subsection{Clique constraints}

In this subsection, we prove \Cref{thm:plus-cliques} by describing how we extend our algorithm for $\NEOtwo[\finrec]$ (\Cref{thm:neo-theorem}) to handle the clique predicates of $\NEOtwo[\finrec]\plusk$.
Recall that for a vertex term $s$, $\clique(s)$ checks whether the subgraph induced by $s$ is a clique. 
The main tool will be the \emph{fast subset cover} by Björklund et al.~\cite{BjorklundHKK07} which we now present.
Let $U$ be a finite set and let $R$ be a ring. Let  $g_1, \dots, g_t \colon 2^U \to R$ be set functions, for some  positive integer~$t$. For every $i\in [t]$, the \emph{zeta-transform} $\zeta g_i \colon 2^U \to R$ and the  \emph{Möbius-transform} $\mu g_i \colon 2^U \to R$ of $g_i$ are defined as 
\begin{align*}
    (\zeta g_i)(Y) = \sum\limits_{X \subseteq Y} g_i(X) && (\mu g_i)(Y) = \sum\limits_{X \subseteq Y} (-1)^{|Y \setminus X|} g_i(X)
\end{align*}
for every $Y \subseteq U$.
The \emph{cover product} $g_1 \ast_c g_2 \ast_c \ldots \ast_c g_t \colon 2^U \to R$ of $g_1, \dots, g_t$ is defined as
\[
    (g_1 \ast_c g_2 \ast_c \ldots \ast_c g_t)(Y) = \sum\limits_{\substack{X_1, \dots, X_t \subseteq U \colon \\ X_1 \cup \dots \cup X_t = Y}} g_1(X_1) \cdot g_2(X_2) \cdot \ldots \cdot g_p(X_t).
\]
The fast subset cover result of Björklund et al.~\cite{BjorklundHKK07} provides the following equality. %
\begin{lemma}[Fast subset cover \!\!\cite{BjorklundHKK07}] \label{lem:cover-product}
    Let $U$ be a finite set, let $R$ be a ring, and let $g_1, \dots, g_t \colon 2^U \to R$ be set functions for a positive integer $t$.
    Then for every $Y \subseteq U$, it holds that
    \[
        (g_1 \ast_c g_2 \ast_c \ldots \ast_c g_t) (Y) = \mu\bigl((\zeta g_1) \cdot (\zeta g_2) \cdot \ldots \cdot (\zeta g_t)\bigr)(Y).
    \]
\end{lemma}
If one computes the cover product naively, there are up to $2^{t \cdot |U|}$ combinations of subsets $X_1, \dots, X_t$ of $Y$ to consider.
For the application below, one should think of $t$ as the number of children of some node in the elimination tree $T$.
This number can be of order $\Theta(n)$ so $2^{t \cdot |U|}$ iterations are prohibitively slow for our purposes.
Using the fast subset cover instead, one considers all subsets $X$ of $Y$---this is by definition of a Möbius-transform and there are at most $2^{|U|}$ such subsets---and then, for every $i \in [t]$, one considers all subsets $Z$ of $X$---this is by definition of a zeta-transform.
For different values of $i$, the subsets of $X$ are considered independently, one after the other, leading to at most $t \cdot 2^{|U|}$ iterations for a fixed subset $X$.
Hence, fast subset cover drastically improves the running time of this computation by eliminating a factor of $t$ in the exponent.
In our applications below, the size of the universe $U$ will be linear in the formula length and therefore, a factor of $2^{\cO(|U|)}$ in the running time will be acceptable.

Now we sketch how to use this technique to extend our approach developed so far to the clique constraint.
A crucial observation is that for every clique in the graph $G$, there exists a root-to-leaf path in the elimination tree $T$ on which every vertex of this clique lies.
Thus, a clique constraint can be \emph{falsified} in two ways.
Either two vertices of the non-clique are not in the ancestor-descendant relation---that is, they belong to the subtrees of two different children of some node---or they are in the ancestor-descendant relation but there is no edge between them.

Now consider a formula $\phi$ from $\NEOtwo \plusk$.
By \cref{obs:core}, we assume that $\phi$ is in $\core\NEOtwo \plusk$. In particular, for every clique predicate $\clique(t)$ of $\phi$, $t$ is a vertex set variable.
And let $\cls(\phi) \subseteq \var^V(\phi)$ ($\cls$ stands for ``clique set'') denote the set of all vertex set variables $X$ such that $\clique(X)$ occurs in $\phi$.
We first extend the definition of the index of a triad $\tss$ to a set $S \subseteq \cls(\phi)$.
For a partial solution $(\psf, \dots)$ of $\tss$, the compatibility with $S$ is defined in such a way that $S$ consists of all vertex set variables $X\in \cls(\phi)$ for which there is a vertex $v \in \psf(X) \cap \subtree$ and a vertex $w \in \psf(X)$ with $v \neq w$ and $vw \notin E(G)$.
So in simple words, the vertex $v$ of $\subtree$ ``falsifies'' the clique constraint $\clique(X)$.
Furthermore, to capture the possibility that two vertices land in different subtrees and thus, falsify some clique, we further extend the index by a set 
$Z \subseteq \cls(\phi)$ 
which reflects all vertex set variables $X \in \cls(\phi)$ satisfying $\psf(X) \cap \subtree \neq \emptyset$. 
Thus, every index in the following discussion has the form $I = (\indf, \inde, \allow, \pdeg, \False, S, Z)$.

With this, we can proceed as follows algorithmically.
At the root, in addition to trying out all subsets $\False \subseteq \EQ(\phi)$ and all (reasonable) vectors $\barc \in \bN^{\ell+2}$ (cf.\ \cref{lem:modelcheck:equivalence:G:models:phi}), we iterate over all possible subsets $S' \subseteq \cls(\phi)$ to fix the set of clique constraints to be falsified.
Then we check whether with this ``guess'', the formula evaluates to $1$; otherwise, we discard the candidate.
Also, at the root we iterate over the subsets $Z' \subseteq \cls(\phi)$ to fix which vertex set variables from $\cls(\phi)$ are evaluated to non-empty subsets of $V(G)$.
This increases the running time by a factor of $\cO(2^{|\phi|})$.

In the following, for every node type, we implicitly check that the inclusion $S \subseteq Z$ holds; otherwise, we immediately return $0$ as no compatible partial solution exists in this case by definition.
So we always assume $S \subseteq Z$ in the rest of the subsection.
At a leaf-node $u$ (cf.\ \cref{lem:modelcheck:leaf}), the set $\subtree[u]$ is empty so we additionally check whether at least one of the components $S$ or $Z$ of the index $I$ is non-empty.
If so, we simply return $0$ as no compatible partial solution exists, otherwise we proceed as before.

At a forget-node (cf.\ \cref{lem:modelcheck:forget}), when we branch over all restrictions $(\indf_u, \inde_u, \allow_u, \hat q)$ of the index on $\{u\}$, we additionally branch on the last two components of the index as follows.
We replace the sum over all $\cQ^p_{u}[\interp f_u^+, \bE_u^+, \allow_u^+, \pdeg \cup \hat q, \False]$ with the sum over all $\cQ^p_{u}[\interp f_u^+, \bE_u^+, \allow_u^+, \pdeg \cup \hat q, \False, \hat S, \hat Z]$ additionally satisfying the following requirements.
\begin{itemize}
    \item Let $S_u \subseteq \cls(\phi)$ be the set of clique constraints \emph{falsified} by the vertex $u$, that is, $S_u$ consists of all variables $X \in \cls(\phi)$ with $u \in \indf_u(X)$ and $\indf(X) \not\subseteq N_G(u)$. 
    First, we must have $S_u \subseteq S$: otherwise, the branch violates the requirement that precisely the clique constraints from the set~$S$ are falsified.
    Second, we require that $S\setminus S_u \subseteq \hat S \subseteq S$.
    On one hand, this guarantees that in $\subtree(u)$ we only falsify clique constraints from the set $S$.
    On the other hand, all clique constraints from $S$ not falsified by the vertex $u$ become falsified by $\subtree(u)$.
    \item Let $Z_u \subseteq \cls(\phi)$ consist of all variables $X \in \cls(\phi)$ such that $u\in \indf_u(X)$. 
    First, we must have $Z_u \subseteq Z$---this is to ensure that the vertex $u$ only belongs to the evaluations of vertex set variables from the set $Z$.
    Second, we require that $Z \setminus Z_u \subseteq \hat Z \subseteq Z$.
    Analogously to the first item, this ensures that vertices from $\subtree(u)$ are only evaluated to vertex variables from the set $Z$; and all vertex set variables from $Z$ not containing the vertex $u$ receive a non-empty evaluation in $\subtree(u)$.
\end{itemize}
Additionally trying out all such combinations of sets $\hat S$ and $\hat Z$ increases the running time of the forget computation by a factor of $2^{\cO(\abs{\phi})}$.

At a join-node (cf.\ \cref{lem:modelcheck:join}) $u$ with children $v_1, \dots, v_t$ we proceed as follows.
The procedure boils down to a careful application of fast subset cover technique by Björklund et al.~\cite{BjorklundHKK07} and its extension by Bergougnoux et al.~\cite{BergougnouxCGKM25}.
First, we branch on all subsets $S_2 \subseteq S$.
The subset $S_2$ is meant to contain precisely all variables $X \in S$ such that the variable $X$ is evaluated to a non-empty set in at least two of the subtrees $\subtree[v_1], \dots, \subtree[v_t]$.
Note that in this case, the constraint $\clique(X)$ is immediately falsified.

Next, we branch on all sets $\hat S$ satisfying $S \setminus S_2 \subseteq \hat S \subseteq S$.
The set $\hat S$ reflects which of the $\clique$-constraints are falsified by some vertex of some $\subtree[v_i]$ and its non-neighbor from $\broom[v_i]$. %
The requirement $S \setminus S_2 \subseteq \hat S \subseteq S$ reflects that the set $\hat S$ must contain all the constaints of $S$ ``not falsified'' by $S_2$ and is allowed to further falsify the $\clique$-constraints for the elements of $S_2$.
We rely on fast subset cover to ensure that the union of the penultimate components of the indexes in the recursive calls of the children $v_1, \dots, v_t$ is equal to $\hat S$.

In the remainder, we sketch how to handle the last component of the index correctly.
We refer to \cref{fig:clique-diagram} for an illustration of the following case distinction.
Consider a variable $X \in S \setminus S_2$.
Recall that we have $S \subseteq Z$ by assumption.
Due to the intended meaning of $S_2$, we are to ensure that the evaluation $\psf(X)$ is non-empty in \underline{exactly one} of the subtrees $\subtree[v_1], \dots, \subtree[v_t]$, i.e., the variable $X$ occurs in the last component of the index of the recursive call of exactly one child.
The same applies to every variable $X \in Z \setminus S$: because $X \notin S$, the constraint $\clique(X)$ is not allowed to be falsified in $\subtree(u)$ so $X$ must have a non-empty evaluation in $\subtree[v_i]$ for \underline{exactly one} child $v_i$ of $u$.
We now describe how to ensure this ``in-exactly-one-child'' property following the approach by Bergougnoux et al.~\cite{BergougnouxCGKM25}.

\begin{figure}[t]
    \includegraphics{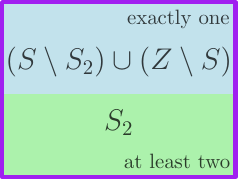}
    \centering
    \caption{
        Sketch for the case distinction in the join-procedure at a node $u$.
        The value $Z$ of the index determines which vertex set variables are evaluated to a non-empty set in $\subtree(u)$.
        The purple rectangle sketches the set $Z$. %
        The blue and green areas reflect that the corresponding variables from $Z$ have to occur in exactly one or at least two of $\subtree[v_1], \dots, \subtree[v_t]$, respectively.
        For the blue case, we equip the polynomials involved in the fast subset cover with a formal variable $\delta$ whose exponent is used to catch and discard the branches where some variable is assigned to more than one child.
        And for the green case, we create a branch where such a variable is assigned to at least one child (solved using fast subset cover) and a branch where it is assigned to exactly one child (as in the blue case), and then subtract the number of partial solutions in the former and latter branches to obtain the desired number of partial solutions.
    } 
    \label{fig:clique-diagram}
\end{figure}

So let $X$ be a variable from the set $Z_1 \defeq (S \setminus S_2) \cup (Z \setminus S)$.
We can rely on fast subset cover to ensure that $X$ occurs in the last coordinate of the index of \underline{at least one} child $v_i$.
To ensure that it also occurs in \underline{at most one} child, we additionally equip the polynomials of the children of $u$ with a new formal variable $\delta$.
The exponent of $\delta$ in the polynomial of the child $v_i$ is a number from $\{0,1\}$ and it reflects whether the evaluation $\psf(X)$ has a non-empty intersection with $\subtree[v_i]$, i.e., whether $X$ occurs in the last component of the index in the recursive call for $v_i$.
Thus, when we multiply the polynomials of the children of $u$, we can safely discard the partial products where the exponent of $\delta$ is larger than $1$ as this corresponds to the case that vertices of $\psf(X)$ occur in the subtrees of at least two children of $u$.
Crucially, the same variable $\delta$ can be used to keep track of all variables in $Z_1$ simultaneously.
More formally, the exponent of $\delta$ in the polynomial of $v_i$ is equal to the number of variables in $Z_1$ occurring in the last coordinate of the index for the recursive call of $v_i$. 
By discarding, during the join-computation, all partial products where the exponent of $\delta$ exceeds $|Z_1|$, we ensure that the exponent of $\delta$ remains bounded by $\cO(|\phi|)$.
This will later allow us to bound the space complexity as it limits the number of monomials to keep track of.

Finally, consider a variable $X \in S_2$.
The fact that the vertices of $\psf(X)$ occur in the subtrees of at least two children of $u$ corresponds to the fact that $X$ occurs in the last component used in the recursive call of at least two children of $u$.
For the sake of readability, in the rest of the paragraph we only speak about this last component of the index.
To satisfy the constraint ``$X$ occurs in the index of \underline{at least two} of %
$v_1, \dots, v_t$,
we can take the number of the partial solutions where $X$ occurs in the index of \underline{at least one} of %
$v_1, \dots, v_t$
and subtract the number of the partial solutions where $X$ occurs in the index of \underline{exactly one of} %
$v_1, \dots, v_t$.
To compute the number of partial solutions where the vertices of $X$ occur in \underline{at least one} of %
$v_1, \dots, v_t$,
we simply rely on the fast subset cover.
And to compute the number of partial solutions where the vertices of $X$ occur in \underline{exactly one} of 
$v_1, \dots, v_t$,
we rely on the fast subset cover with the above-described equipment with the formal variable $\delta$. %

Observe that in every branch of this join procedure, altogether, there are at most $\cO(|\phi|)$ variables for which we ensure that they occur exactly once in the last component of the indices of the children of $u$.
Hence, it suffices to only keep the partial products in which the degree of the formal variable $\delta$ is upper-bounded by $\cO(|\phi|)$.
So the usage of the variable $\delta$ increases the space complexity only by a factor of $\cO(\abs{\phi})$.
For technical details, we refer to the work of Bergougnoux et al.~\cite{BergougnouxCGKM25} where they employ similar ideas relying on fast subset cover with one new formal variable to solve \textsc{Independent Set} parameterized by shrubdepth.

Due to the application of fast subset cover together with branching on the values of $S_2$ and $\hat S$, for the join computation, we get a branching factor of $2^{\cO(\abs{\phi})}$ at every child of $u$ (compared to no branching before (cf.\ \cref{lem:modelcheck:join})).
We emphasize that the branches of any two different children are independent from each other (cf.\ \cref{lem:cover-product}).
For a detailed analysis of the running time of such a branching we, again, refer to the work by Bergougnoux et al.~\cite{BergougnouxCGKM25}.
With this we obtain the desired algorithm for $\NEOtwo[\finrec]\plusk$:
\thmMaink*
\subsection{Connectivity and acyclicity constraints} 

In this subsection, we prove \Cref{thm:plus-ac} by describing how to extend our algorithm (\Cref{thm:plus-cliques}) to additionally handle acyclicity and connectivity constraints.
Let us remark that with connectivity and acyclicity constraints involved, our algorithm will become randomized and it will no longer work in logarithmic in $n$ space.
The reason for these changes is the application of Isolation lemma.
As standard for algorithms of this flavor, we will require that the solution we seek is unique for technical reasons.
Isolation lemma guarantees that if we weight the solutions in a certain way, then (with high probability) there is a weight such that the solution with this weight is unique.
The generation of a weight function involves randomized sampling, and polynomial in $n$ space is required to store the weight function.
Except for that, the algorithm remains deterministic and still uses only logarithmic in $n$ space.

So let $\phi$ be a formula from $\NEO_2[\finrec]\ack$.
We may assume that $\phi$ is in $\core\NEO_2[\finrec]\ack$ by \cref{obs:core}.
We carry out some transformations of $\phi$ following the ideas of Bergougnoux et al.~\cite{BergougnouxDJ23}.
First, we bring the quantifier-free part of $\phi$ in the disjunctive normal form to obtain the equivalent formula. %
Now the formula is of the form $\exists \bar{X} \  \phi_1 \lor \phi_2 \lor \dots \lor \phi_t$ where $t \leq 2^{\cO(\abs{\phi})}$, $\exists \bar{X}$ is a sequence of existential quantifiers, and each $\phi_i$ is a conjunction of (possibly negated) set equalities, size measurements, connectivity and acyclicity constraints, and clique constraints.
For simplicity, we will use the notation $\psi_i \coloneqq \exists \bar{X} \ \phi_i$ for each $i \in [t]$.
We have $\abs{\psi_i} \in \cO(\abs{\phi})$, $p_{\psi_i} \leq p_{\phi}$, and $d_{\psi_i} \leq d_{\phi}$.
This is because every set equality, every size measurement, and every and connectivity or acyclicity constraint occurring in $\phi$ occurs at most once in $\psi_i$, possibly negated; and on the other hand, every such element occurring in $\psi_i$ has to occur in $\phi$ as well.
The task of deciding $G \vDash \phi$ is equivalent to checking if there exists an index $i \in [t]$ with $G \vDash \psi_i$.
So we fix a value $i \in [t]$ and in the remainder focus on checking $G \vDash \psi_i$. %

\paragraph{Connectivity.}
For now assume that our formula $\phi$ did not contain acyclicity constraints, we explain how to deal with them later.
We start by eliminating negated connectivity constraints of $\psi_i$ as follows.
Observe that a subgraph $(Z,W)$ of $G$ is disconnected if and only if the vertex set $Z$ can be partitioned into two non-empty sets with no edges from the edge set $W$ running between them.
This property is expressed by the following $\NEOtwo[\finrec]$ formula
\begin{align*}
    &\mathsf{dcon}(Z,W) \defeq  \exists Z_1 \exists Z_2 \   \neg(Z_1 = \emptyset) \land \neg(Z_2 = \emptyset) \land \mathsf{partition}(Z,Z_1,Z_2) \land N^{\bN^+}(Z_1,W) \cap Z_2 = \emptyset,\\
    & \text{where } \mathsf{partition}(Z,Z_1,Z_2) \defeq  Z_1 \cup Z_2 = Z \land Z_1 \cap Z_2 = \emptyset.
\end{align*}
Then, for a vertex set variable $X$, the constraint $\lnot\conn(X)$ is equivalent to $\mathsf{dcon}(X,\constant E)$; and for an edge set edge variable $Y$, the constraint $\lnot\conn(Y)$ is equivalent to $\mathsf{dcon}( N^{\bN^+}(\constant V, Y) ,Y)$ since the set $N^{\bN^+}(\constant V, Y)$ consists of the  endpoints of edges in $Y$.
We thus eliminate all negated connectivity constraints by replacing them by corresponding equivalent subformulas.
After that, we proceed as in \cref{obs:core} to bring the formula into the core logic.
It is straightforward to verify in the proof of \cref{obs:core} that after this transformation the formula is a conjunction and all connectivity constraints in the resulting formula remain non-negated.
The length of the formula increases by at most a constant multiplicative factor.
Also, as we assume $d_{\phi} \geq 2$ across the paper, this replacement does not increase any of the values $d_{\psi_i}$ and $p_{\psi_i}$.
For simplicity, we still denote the resulting equivalent formula by $\psi_i$.

Now we sketch how to equip the framework presented in the paper so far with the Cut\&Count technique by Cygan et al.~\cite{CyganNPPRW22} to ensure that, with high probability, for every constraint $\conn(X)$ in $\psi_i$, the vertex set $\conn(X)$ is connected. 

Before describing the application of Cut\&Count~\cite{CyganNPPRW22} to our setting, let us sketch the main idea behind it. %
Suppose we are looking for a vertex set $X$ with a certain property $\alpha$ such that, additionally, the subgraph of $G$ induced by $X$ is connected.
A partition $(X^L, X^R)$ of a vertex set~$X$ is called a \emph{consistent cut} if there is no edge in $G$ with one endpoint in $X^L$ and the other in $X^R$.
In other words, every connected component of $G[X]$ is contained in one of $X^L$ and $X^R$.
Hence, there are $2^k$ consistent cuts of $X$ where $k$ denotes the number of connected components of $G[X]$.
This number is divisible by $4$ if and only if $X$ induces a disconnected subgraph of $G$ (and it is exactly two otherwise).
Thus, if we may, for some reason, assume that there exists at most one vertex set $X$ with the property $\alpha$, then deciding the existence of $X$ becomes equivalent to counting, modulo 4, the pairs $(X^L, X^R)$ of disjoint vertex sets such that $X^L \cup X^R$ satisfies $\alpha$ and there are no edges between $X^L$ and $X^R$.
This replaces the time-consuming (and global) task of ensuring the connectivity to counting certain pairs of vertex sets with simpler local adjacency constraints.

Similarly, for an edge set $Y$, the endpoints of $Y$ comprise the set $V(Y) \defeq N^{\bN^+}(V(G), Y)$, and this set can be expressed in the logic. 
So to check connectivity of $Y$ we can similarly count the consistent cuts of the set $V(Y)$ of endpoints in the subgraph $(V(Y), Y)$.
For the sake of readability, we thus assume in the remainder that for every constraint $\conn(X)$ of $\psi_i$, $X$ is a vertex set variable; the edge set variables can be handled similarly.

This idea can naturally be generalized to multiple connectivity constraints: as long as we may assume that there exists a unique ``candidate solution'' and the task involves connectivity constraints on vertex sets $X_1, \dots, X_r$ of the candidate solution, we can eliminate all connectivity constraints by counting the combinations of consistent cuts of each of $X_1, \dots, X_r$. 
Then deciding the existence of a solution where each of these $r$ sets is connected is equivalent to checking whether the resulting number is not divisible by $2^{r+1}$ (i.e., whether for each of these sets the number of consistent cuts is not divisible by four).

However, if the solution is not unique, the number of consistent cuts of all solutions may still, in general, cancel out modulo $2^{r+1}$. 
The standard tool used to guarantee that, with high probability, the solution is unique (if it exists) is the classic Isolation lemma.
\begin{lemma}[Isolation lemma, \cite{MulmuleyVV87}]\label{thm:isolation-lemma}
    For a finite universe $U$, let $\mathcal{F} \subseteq 2^{U}$ be a non-empty set family over $U$. For each $a \in U$, choose a weight $w(a) \in [2 \cdot |U|]$ uniformly and independently at random. Then with probability at least $1/2$ there exists a unique set of minimum weight in $\mathcal{F}$.
\end{lemma}

Now we are ready to explain how to incorporate the required changes in our algorithm to handle the connectivity constraints.
Let $\widetilde{\psi_i}$ denote the formula obtained from $\psi_i$ by removing all connectivity constraints.
In the following arguments, we refer to the proof of \cref{thm:neo-theorem} and the partial solutions defined with respect to the formula $\widetilde{\psi_i}$.
Recall that any partial solution of $\widetilde{\psi_i}$ in \cref{sec:partial} has form $(\psf, \psg, \pse, \psw)$ where the elements of this tuple satisfy the conditions from \cref{def:modelcheck:partial}.
Thus we can view every partial solution as a subset of a universe $U$ of size at most $\abs{\psi_i} \cdot n \cdot \bigl(1 + \td + (d_{\psi_i}+1) + \td \cdot 2 \cdot d_{\psi_i}\bigr)$ which consists of the following elements:
\begin{enumerate}
    \item\label{item:weight-vertex-variables} at most $\abs{\psi_i} \cdot n$ pairs of the form $(X, v)$ where $X$ is a vertex set variable and $v$ is a vertex of $G$, this corresponds to a vertex interpretation $\psf$;
    \item\label{item:weight-edge-variables} at most $\abs{\psi_i} \cdot n \cdot \td$ pairs of the form $(Y, e)$ where $Y$ is an edge set variable and $e$ is an edge of $G$, this corresponds to an edge interpretation $\psg$ (cf.\ \cref {obs:prelim:treedepth:nb:edges});
    \item\label{item:weight-neighborhood-function} at most $\abs{\psi_i} \cdot n \cdot (d_{\psi_i}+1)$ triples of the form $(C, v, i)$ where $C \in \Cat$, $v$ is a vertex of $G$, and $i \in [d_{\psi_i}]_0$, this corresponds to the neighborhood function $\pse$ (here the set $\Cat$ is defined with respect to the formula $\widetilde{\psi_i}$);
    \item\label{item:weight-arcs} and at most $\abs{\psi_i} \cdot n \cdot \td \cdot 2 \cdot d_{\psi_i}$ pairs of the form $(C, \vec e)$ where $C \in \Cat$ and $\vec e$ is an arc of $\vec A(\vec G_{\widetilde{\psi_i}})$, this corresponds to the mapping $\psw$ (by \cref {obs:prelim:treedepth:nb:edges} and definition of $\vec A(\vec G_{\widetilde{\psi_i}})$, this digraph hast at most $n \cdot \td \cdot 2 \cdot d_{\psi_i}$ arcs). 
\end{enumerate}
So let $w$ be a weight function on the universe $U$ such that the weight of every element is sampled independently at random from the integer interval $[2 \cdot |U|]$.

Let~$\cS$ be the family of all partial solutions $(\psf, \psg, \pse, \psw)$ such that $(\psf, \psg, \pse, \psw)\in \Partial_r^=(\emptyset^4, \False, j, j, \barc)$ and $\ip{\widetilde{\psi_i}}^{(G,\psf,\psg)} = 1$ for some $j \in \bN$, $\barc \in \bN^\ell$, and $\False \subseteq \EQ(\widetilde{\psi_i})$.
The choice of the weight function together with \cref{thm:isolation-lemma} now implies the following.
If $G \models \psi_i$, then, with high probability, there exists a weight $W\in \bN$ such that the family $\cS$ contains a unique partial solution of weight $W$.
In this case, the number of consistent cuts of the partial solutions in $\cS$ of weight $W$ is equal to $2^q$ where $q$ denotes the number of connectivity constraints in $\psi_i$.
Otherwise, all such consistent cuts cancel out modulo~$2^{q+1}$.
So instead of just counting the partial solutions of $\widetilde{\psi_i}$, we will now count combinations of partial solutions of $\widetilde{\psi_i}$ and consistent cuts of all variables $X$ involved in the connectivity constraints of $\psi_i$ as follows.
For every such variable $X$, we extend the partial solution of $\widetilde{\psi_i}$ by a vertex interpretation of two new vertex set variables $X^L$ and~$X^R$.
We run our algorithm from the previous section with the following adaptations.\footnote{It is not difficult to express a consistent cut of a vertex set in $\NEO_2[\finrec]$. However, we would not be able to use the algorithm from the previous section as a black box as we have to weight the partial solutions of $\widetilde{\psi_i}$ while not weighting the sides of the consistent cuts. As we are opening the black box anyway, for the sake of readability, we decided to handle consistent cuts separately from partial solutions.}
\begin{enumerate}
    \item First, sample the weight function $w$ as described above and store it. 
    \item Extend all polynomials of the form $\cP$ and $\cQ$ with a new formal variable, say $\tau$, whose exponent keeps track of the weight of partial solutions of $\widetilde{\psi_i}$.
    Here, 
    similarly to the definition of $a(\cdot)$ (i.e., the exponent of the formal variable $\alpha$), the definition of the weight of a partial solution (cf.\ Items \ref{item:weight-vertex-variables} and \ref{item:weight-neighborhood-function}) takes into account only the vertices of the $\subtree$ of the current triad.
    This ensures that the exponent of $\tau$ is correctly updated at join computations.
    \item Vertex interpretations are extended by two new vertex set variables $X^L$ and $X^R$ for each connectivity constraint $\conn(X)$ of $\psi_i$. 
    \item At a leaf computation, when an edge $e$ is assigned to an edge set variable $Y$, we additionally multiply the polynomial by $\tau^{w(Y, e)}$. We proceed similarly for the arcs added to $\psw(C)$ with $C \in \Cat$.
    This ensures that the weight of each partial solution restricted to edges and arcs (cf.\ Items \ref{item:weight-edge-variables} and \ref{item:weight-arcs}) is correctly updated.
    Also, when the edge $e$ is processed, for every vertex set variable $X$ involved in a connectivity constraint, we verify that 
    if one endpoint of $e$ is assigned to $X^L$, then its other endpoint is not assigned to $X^R$.
    \item At a forget computation for a vertex $u$, similarly to the multiplication by $\gamma^{\barc}$ reflecting that $u$ now belongs to the $\subtree$, we multiply by $\tau^{t}$ where $t$ 
    \[
        t = \sum_{X \in \var^V(\psi_i) \mid u \in \indf_u(X)} w(X, u) + \sum_{C \in \Cat} w(C, u, \inde_u(u, C)) 
    \]
    is the total weight of $u$ in all vertex set variables it is evaluated into together with the weights of its values in $\bE_u$ (cf.\ Items \ref{item:weight-vertex-variables} and \ref{item:weight-neighborhood-function}). 
    Also, we only branch on the vertex interpretations on $\{u\}$ that additionally satisfy the following two conditions.
    \begin{itemize}
        \item If $u$ belongs to $\indf_u(X)$, then it occurs in exactly one of $\indf_u(X^L)$ and $\indf_u(X^R)$.
        \item If $u$ does not belong to $\indf_u(X)$, then it occurs in none of $\indf_u(X^L)$ and $\indf_u(X^R)$.
    \end{itemize}
    \item At the root, additionally to iterating through the values $j, \barc$, and $\False$, we iterate over all possible weights, i.e., exponents of $\tau$. 
    Observe that by definition, the weight of any partial solution is upper-bounded by %
    \[
        |U| \cdot 2 \cdot |U| 
        \in (n + \abs{\psi_i} + d_{\psi_i})^{\cO(1)}.
    \]
    For each of those combinations, we reconstruct the corresponding coefficient of the root polynomial and check whether it is divisible by $2^{q+1}$ where $q$ denotes the number of connectivity constraints. %
    If we find at least one combination where the coefficient is not divisible by $2^{q+1}$, then we output that the formula $\psi_i$ is satisfiable.
\end{enumerate}

This algorithm does not give false positives and gives false negatives with probability at most $1/2$ (cf.\ \cref{thm:isolation-lemma}).
Finally, we analyse how these adaptations influence the running time and the space complexity.
\begin{itemize}
    \item We reduce the model checking of $\phi$ to the problem of model checking at most $2^{\abs{\phi}}$ formulas $\psi_i$ such that the length of $\psi_i$ is upper-bounded by $\cO(\abs{\phi})$, $d_{\psi_i}$ is upper-bounded by $\dphi$, and $p_{\psi_i}$ is upper-bounded by $\pphi$.
    \item To generate and store the weight function in each branch $\psi_i$, we need space polynomial in $n$, $\dphi$, and $\abs{\phi}$.
    \item The largest exponent of any formal variable of the root polynomial is upper-bounded by a polynomial in $n$, $\dphi$, and $\abs{\phi}$.
    First, this implies that trying out all candidate weights at the root increases the running time by only a polynomial factor.
    Second, the bit length of each prime number we use (cf.\ \cref{sec:modelcheck}) remains logarithmic in $n$, $\dphi$, and $\abs{\phi}$.
\end{itemize}
Altogether we obtain the following intermediate result:
\begin{theorem}\label{thm:plus-connectivity}
    There is a Monte-Carlo algorithm that given a formula $\phi \in \NEOtwo[\finrec]\ack$ without acyclicity constraints, a graph $G$, and an elimination tree $T$ of $G$ of depth at most $\td$, runs in time $( (\dphi+1) \cdot \pphi)^{\cO(|\phi|\cdot\td)} \cdot 2^{\cO(\pphi)} \cdot n^{\cO(\ell)}$ and space polynomial in $n + \abs{\phi} + \dphi + \pphi$ and decides whether $G \vDash \phi$ holds.
    The algorithm cannot give false positives and may give false negatives with probability at most $1/2$. 
\end{theorem}

\paragraph{Acyclicity.}
Finally, we sketch how to tackle the acyclicity constraints following a similar but slightly more involved application of Cut\&Count~\cite{CyganNPPRW22}.
Recall that we are to check whether $G \vDash \psi_i$ holds.
By the replacement described earlier, we may assume that $\psi_i$ contains only non-negated connectivity constraints.
Similarly, we eliminate negated acyclicity constraints as follows.
Observe that a subgraph $(Z,W)$ of $G$ is acyclic if and only if there exists a non-empty subset $Z'\subseteq Z$ such that every vertex in the subgraph $(Z,W)[Z']$ has degree at least 2.
This property is expressed by the following $\NEOtwo[\finrec]$ formula:
\[
    \mathsf{cyc}(Z,W) \defeq  \exists Z' \ Z' \subseteq Z \land  Z' \subseteq N^{\bN \setminus \{0, 1\}}(Z', W).
\]
Then, for a vertex set variable $X$, the constraint $\lnot\acy(X)$ is equivalent to $\mathsf{cyc}(X,\constant E)$; and for an edge set variable $Y$, the constraint $\lnot\acy(Y)$ is equivalent to $\mathsf{acy}(\constant V ,Y)$.
We thus eliminate all negated acyclicity constraints by replacing them with corresponding equivalent formulas.
As before, we then proceed as in \cref{obs:core} to transform the formula in the core logic.
For simplicity, we still denote this equivalent formula by $\psi_i$.

Now all connectivity and acyclicity constraints are non-negated and $\psi_i$ is a conjunction. 
So we discard those constraints from $\psi_i$ to obtain a formula $\widetilde{\psi_i}$ and instead search for a model of $\widetilde{\psi_i}$ which additionally satisfies these constraints. 
In the remainder, we focus on acyclicity as the approach for connectivity was presented earlier. 

For this, we follow the idea of Bodlaender et al.~\cite{BodlaenderCKN15}.
At a high level, we slightly adapt the input graph in such a way that rather than seeking a general acyclic set (i.e., a forest) with certain properties, it suffices to look for a tree.
To this end, we artificially add a universal (i.e., adjacent to all vertices of $G$) vertex $v^*$ to the graph $G$ and obtain the graph $G^*$.
Let $E^*$ denote the set of edges incident to $v^*$.
We can easily obtain an elimination tree of $G^*$ of depth at most $\td+1$ by making~$v^*$ the new root with the root of $ T$ as its unique child. 
Let $T$ still denote the resulting elimination forest.
For a vertex set $X \subseteq V(G)$, by $E[X]$ we denote the set of edges induced by $X$ and by $E^*_X$ we denote the edges from $E^*$ with an endpoint in $X$.
Then the following lemma is easy to verify.
\begin{lemma}[\!\!\cite{BodlaenderCKN15} Subsection 3.4]
\label{lem:forest-to-a-tree}
    A vertex set $X \subseteq V(G)$ is acyclic in $G$ if and only if there exists a set $E_X$ satisfying $E[X] \subseteq E_X \subseteq E[X] \cup E^*_X$ such that the graph $(X \cup \{v^*\}, E_X)$ is a tree.
\end{lemma}
Indeed, such a set can be obtained by taking the subgraph induced by $X$ and for every connected component of this subgraph, adding one edge from $v^*$ to the set $E_X$.
For such a \emph{candidate set} $E_X$, the subgraph $(X \cup \{v^*\}, E_X)$ of $G^*$ is a tree if and only if 
\begin{enumerate} 
    \item it is connected and 
    \item\label{item:vts-edges-number-forest} the number of edges in $E_X$ is precisely $|X|$ (because a tree on $|X| + 1$ vertices has $|X|$ edges).
\end{enumerate}
Crucially, observe that if this subgraph is connected, then the number of edges in $E_X$ is at least $|X|$.
So we have to ensure connectivity of this subgraph and that the number of edges in $E_X$ is no larger than $|X|$.
It is not difficult to formulate an analogous criterion for acyclicity of an edge set so we focus on vertex set variables in the remainder.

Now we describe what changes need to be made to incorporate this idea in our model checking algorithm.
First, for every constraint $\acy(X)$ in $\psi_i$, we extend the edge interpretations in partial solutions by a new edge set variable $E_X$; and we extend the vertex interpretations by new vertex set variables $X^L$ and $X^R$. 
The latter variables are intended to reflect the consistent cuts of the vertex set $X \cup \{v^*\}$ in the graph $(X \cup \{v^*\}, E_X)$.
To handle the sides of consistent cuts as well as weight the partial solutions of $\widetilde{\psi_i}$, we proceed as described earlier.

The vertex $v^*$ and the edges in $E^*$ are ``transparent'' to the algorithm except for the treatment of the acyclicity constraints. 
By this we mean that the algorithm does not evaluate any vertex set variable other than the ones for the sides of consistent cuts to contain the vertex $v^*$.
Similarly, the edges in $E^*$ are only allowed to be evaluated in the new edge set variables of the form $E_X$.
Additionally, we introduce two new formal variables $\rho_1$ and $\rho_2$.
For a partial solution $(\psf, \psg \dots)$ of a triad $\tss$, the exponent of $\rho_1$ or $\rho_2$ is equal to $\sum |\psg(E_X)|$ and $\sum |\psf(X) \cap \subtree|$, respectively, %
where both summations are over all variables $X$ involved in acyclicity constraints of $\psi_i$.
We emphasize that, as usual, the exponent of the formal variable $\rho_2$ only considers the restriction of the evaluation to $\subtree$. %

For simplicity, in the following, $X$ always refers to a vertex set variable involved in an acyclicity constraint of $\psi_i$.
At a leaf computation, when considering an edge, say $vw$, such that both $v$ and $w$ belong to $\indf(X)$, the edge $vw$ is necessarily placed in $\psg(E_X)$.
More formally, we discard every partial evaluation $\psg_{vw}$ on $\{vw\}$ such that $vw \notin \psg_{vw}(E_X)$ (cf.\ \cref{lem:modelcheck:leaf}).
Furthermore, we only allow edges of such form as well as edges from $E^*$ with an endpoint in $\indf(X)$ to be in $\psg(E_X)$. 
We multiply the current polynomial by $\rho_1$ whenever an edge is added to some $E_X$.
Correctly updating the exponent of $\rho_2$ takes place, as usual, at forget procedures.

Then at the root, we only consider the monomials in which the exponents of $\rho_1$ and $\rho_2$ are equal (cf.\ Item~\ref{item:vts-edges-number-forest}).
Except for these changes, we proceed as described earlier for connectivity: in particular, we weight every variable~$E_X$ to ensure, with high probability, the uniqueness of the sought partial solution.
As before, we do not weight the sides of consistent cuts.
In the end, we check if there is a monomial of the root polynomial such that the corresponding coefficient is not divisible by $2^{q+1}$ where $q$ is the total number of connectivity constraints: those are both the original connectivity constraints together with the connectivity constraints introduced for acyclicity constraints.
Note that the degrees of the two new variables $\rho_1$ and $\rho_2$ are upper-bounded by $\cO(|\phi| \cdot n \cdot \td)$ similarly to the formal variables $\alpha$, $\beta$, and $\gamma$.
For this reason, the magnitudes of prime numbers involved in the computation remain asymptotically the same so that the running time is not affected.
As with connectivity constraints, the space complexity becomes polynomial (rather than logarithmic) in $n$ due to the need to store the weight function required for the application of the Isolation lemma.
With this, we obtain the final theorem.
\thmMainack*

\subsection{Directed graphs}

In this subsection, we show that the techniques developed in this work are not limited to undirected graphs.
In the following, we describe how to modify the neighborhood operators, the $\NEOtwo[\finrec]$ logic, and our model checking algorithm to deal with digraphs.
To adapt our logic, we replace
edge (set) variables with arc (set) variables, and we replace our neighborhood operator $N^\sigma$ with in- and out-neighborhood operators $N^\sigma_+$ and $N^\sigma_-$, respectively, which are defined on a digraph $\vec G$ as follows:
\begin{align*}
    N^\sigma_+(X,Y) &\defeq \{ v\in V(\vec G) \mid \abs{\{(u,v) \in Y \mid u \in X\}} \in \sigma\},\\
    N^\sigma_-(X,Y) &\defeq \{ v\in V(\vec G) \mid \abs{\{(v,u) \in Y \mid u \in X\}} \in \sigma\},
\end{align*}
where $X\subseteq V(\vec G)$ and $Y\subseteq A(\vec G)$.
Of course, to interpret a formula as well as in the notions associated with our algorithm, edge interpretations become arc interpretations, mapping arc set variables to arc sets.

Let $\phi$ be a formula of this variant of $\core\NEOtwo[\finrec]$ for digraphs where the core logic is defined analogously to the undirected setting.
Let us sketch how to adapt our algorithm to check whether $\vec G \models \phi$ for a given digraph $\vec G$. 
First, the arc set of the auxiliary digraph $\vec G_\phi$ is defined as $\{(u^1,v^i),(v^1,u^i) \mid (u,v)\in A(\vec G) \land i \in [\dphi]\}$.
Secondly, we define $\Cat$ as the set of all triples $(X,Y,\star)$ such that $N^\sigma_\star(X,Y)$ occurs in $\phi$ for some set $\sigma$; here $\star\in \{+,-\}$. 
For partial solutions, we need to replace Condition~\ref{item:partial:edge:consistent} of \Cref{def:modelcheck:partial} with the following one.
\begin{description}
    \item[(Arc-consistent)] For every $(X, Y, \star) \in \Cat$ and every arc $(u,v) \in \sheaf$ the following hold:
    \begin{itemize}
        \item When $\star=+$, there exists $j \in [\dphi]$ with $(u^1, v^j) \in \vec\cW(X, Y,+)$ if and only if $u \in \interp{\hat f}(X)$ and $(u,v) \in \psg(Y)$.
        \item When $\star=-$, there exists $j \in [\dphi]$ with $(v^1, u^j) \in \vec\cW(X, Y, -)$ if and only if $v \in \interp{\hat f}(X)$ and $(u,v) \in \psg(Y)$.
    \end{itemize}
\end{description}
These are the only major modifications to the definitions related to our model checking algorithm for $\NEOtwo[\finrec]$; for the others, we simply need to change all occurrences of $(X,Y)$ in $\Cat$ by $(X,Y,\star)$.
It remains to update the recursive formula for the leaves (cf.\ \Cref{lem:modelcheck:leaf}) of an elimination tree $T$ of $\vec G$.\footnote{An elimination tree of a digraph is an elimination tree of the underlying undirected graph. Likewise, arc directions are ignored when defining $\sheaf(u)$ of a leaf $u$.}
Indeed, as arc directions are involved, the leaf procedure becomes asymmetric.
In the following lemma, $\var^A(\phi)$ naturally denotes the arc set variables, and $\FalseFunc^A$ reflects the arc equalities falsified by a certain arc interpretation.

\begin{lemma}\label{lem:extension:digraph:leaf}
    Let $u$ be a leaf of $T$ and let $I=(\interp f, \bE, \allow, \pdeg, \False)\in \Index[u]$.
    Further, let $p$ be a prime number such that $\FF_p^*$ admits the $\pphi$-th root of unity.
    Then it holds that:
    \[
	\cQ_u^p[I] \equiv_p \prod_{a = (v,w) \in \sheaf(u)} \sum_{\psg_{a} \in \possible(a, \False)} \gamma^{\barc_{\psg_{a}}} \cdot \prod_{(X, Y, \star) \in \Cat} L((v, w), \psg_{a}, (X, Y, \star))
	\]
    where
    \[
        \possible(a, \False) = \{\psg_a \mid \psg_a \text{ is an arc interpretation on } \{a\} \text{ and } \FalseFunc^{A}(\psg_a) \subseteq \False\},
    \]
    $\barc_{\psg_{a}} = (c_{\psg_{a}}^1, \dots, c_{\psg_{a}}^\ell) \in \{0,1\}^\ell$ with $c_{\psg_{a}}^i = |\psg_{a}(X_i)|$ if $X_i \in \var^{A}(\phi)$ and $c_{\psg_{a}}^i = 0$ if $X_i \in \var^V(\phi)$ for all $i \in [\ell]$, and

   \begin{align*}
		L((v, w), \psg_{a}, (X, Y, \star)) = &
		\begin{cases}
			\cC(v, (X,Y,\star)) & \!\!\text{if } \star = -, w \in \indf(X) \text{\! and } a \in \psg_a(Y)\\
			\cC(w, (X,Y,\star)) & \!\!\text{if } \star = +, v \in \indf(X) \text{\! and } a \in \psg_a(Y)\\
			1 & \!\!\text{otherwise}
		\end{cases}
	\end{align*}
	where 
	\[ \cC(x, (X,Y, \star)) = \allow(x, (X, Y, \star)) \cdot \omega(p)^{\pdeg(x, (X, Y, \star))} \cdot \beta^{[\bE(x, (X, Y, \star)) < \dphi]}.\]
\end{lemma}
It is straightforward to verify that with these modifications, our model checking algorithm works for the digraph variant of $\NEOtwo[\finrec]$.

\section{Relations to other logics}
\label{sec:relation:other:logics}

As mentioned in the introduction, $\NEOtwo[\finrec]$ is essentially equivalent to the existential counting modal logic (\ECML for short) introduced by Pilipczuk~\cite{Pilipczuk11}, and the distance neighborhood logic (\DN) introduced by Bergougnoux, Dreier and Jaffke \cite{BergougnouxDJ23} is essentially the extension of  $\NEOone[\fincofin]$ with the ability to query neighborhood in power of the input graph. 
In the following, we elaborate on these relations.

\paragraph{Relation to \DN.} 
First, we need to define the distance $r$ variant of our neighborhood operator.
\begin{definition}\label{def:distance:neighborhood:operator}
	Given a graph $G$, $\sigma\subseteq \N$, $r\in \N$, $U\subseteq V(G)$ and $F\subseteq E(G)$, we define the \emph{$(r,\sigma,F)$-neighborhood of $U$ in $G$} as
	$$N^\sigma_r(U,F)\defeq \bigl\{v\in V(G) \mid \abs{\{ u  \in U \mid u \text{ is a distance at most $r$ from $v$ in } (V(G), F) \}} \in \sigma\bigr\}.$$
    Moreover, we define the $(r,\sigma)$-neighborhood of $U$, denoted by $N^\sigma_r(U)$, as the set $N^\sigma_r(U,E(G))$.
\end{definition}

Given a fragment of $\logicfont{L}$ of $\NEOtwo$, let us denote by $\D\logicfont{L}$ the extension of $\logicfont{L}$ where $N^\sigma_r(t_1)$ is a vertex term and, if $\logicfont{L}$ allows edge terms, then $N^\sigma_r(t_1,t_2)$ is also a vertex term---here $t_1$ is a vertex term, $t_2$ is an edge term, $\sigma\subseteq \N$, and $r\in \N$.

The \DN logic is defined as the restriction of $\D\NEOone[\fincofin]$ where all the sets $\sigma\subseteq \N$ occurring in a formula are singletons.
Further, \cite[Lemma 6.2]{BergougnouxDJ23} states that, for every vertex set variable $X$, $\sigma\in\fincofin$ and $r\in \N$, there exists a \DN vertex term $t^\sigma_r$
of length $\cO(d(\sigma))$ such that for every graph $G$ and every mapping $\beta \colon \{X\} \to 2^{V(G)}$, we have $\ip{t^\sigma_r}^{(G,\beta)} = N^\sigma_r(\beta(X))$.
Thus, every $\D\NEOone[\fincofin]$ formula $\phi$ can be expressed in \DN with a formula of length $\cO(\abs{\phi}\cdot d(\phi))$.

While $\D\NEOone[\fincofin]$ is slightly more powerful than \DN, it should be noted that the model checking algorithm of Bergougnoux et al. \cite{BergougnouxDJ23} for \DN can handle $\D\NEOone[\fincofin]$ formulas with no increase in the running time because the neighborhood operators of \DN are handled with \cite[Lemma~5.1]{BergougnouxDJ23} which works also for the neighborhood operators of $\NEOone[\fincofin]$.

\paragraph{Relation to \ECML.} Bergougnoux, Dreier and Jaffke proved in \cite{BergougnouxDJ23} that \DN is equivalent to a variant of \ECML called the \textit{existential distance modal logic}.
This latter is obtained from \ECML by removing edge-set quantification, ultimately periodic counting, adding the ability to query $d$-neighborhoods in the $r$-th power of the graph, and modifying how certain arithmetic part works.
In fact, we show that their arguments can be used to prove that $\NEOtwo[\finrec]$ and \ECML are essentially equivalent.
To define \ECML, we need to define a minimalist version of its inner logic called \CML:
\renewcommand{\blacksquare}{\widehat\square\,}
\begin{itemize}
    \item Every vertex or edge set variable $X$ is an unfinished \CML formula.
    \item $\constant{V}, \constant E$, $\emptyset$, and every unary or binary relational symbol is an unfinished \CML formula.
    \item If $\psi$ and $\chi$ are unfinished \CML formulas, then so are $\neg \psi$ and $\psi \land \chi$.
    \item If $\psi$ is an unfinished $\CML$ formula, then $\diamondsuit^\sigma \psi$ and $\square^\sigma \psi$ are also unfinished \CML formulas for any finitely recognizable set $\sigma$. 
    \item If $\psi$ is an unfinished \CML formula, then $\blacksquare \psi$ is a finished \CML formula.
\end{itemize}
As for \NEO logics, an interpretation of a finished \CML (and to be later defined \ECML) formula is a triple $(G,\interp f, \interp g)$ (see \Cref{def:logic:interpretation}) consisting of a graph, a vertex, and an edge interpretation.
However, an interpretation of an unfinished \CML formula $\psi$ is a tuple $(G,v,e,\interp f, \interp g)$ where, additionally, $v$ is a vertex of $G$---called an \emph{active} vertex---and $e$ is either an edge incident to $v$ or a dummy element $\ast$. During the evaluation of an \CML formula, the $\diamondsuit^\sigma$ and $\square^\sigma$ quantifiers changes the current active vertex $v$ into one of its neighbors $w$, and the edge $e$ is set to $vw$. 
The dummy value $\ast$ is used to initialize the evaluation of an unfinished formula $\psi$ when we evaluate a finished \CML formula $\blacksquare \psi$.

Formally, the semantics of \CML formulas is given as follows.
We have $\ip{\blacksquare \psi}^{(G,\interp f, \interp g)} = 1$ iff for all $v\in V(G)$, we have $\ip{\psi}^{(G,v,\ast,\interp f, \interp g)}=1$.
For a vertex set variable $X$, we have $\ip{X}^{(G,v,e,\interp f, \interp g)} = 1$ iff $v\in \interp f(X)$. For an edge set variable $Y$:
\begin{itemize}
    \item If $e \neq \ast$, we have $\ip{Y}^{(G,v,e,\interp f, \interp g)} = 1$ iff $e\in \interp g(Y)$. 
    \item If $e = \ast$, we have $\ip{Y}^{(G,v,\ast,\interp f, \interp g)}= 0$.
\end{itemize}
Unary and binary relational symbols, $\constant V$ and $\constant E$ are interpreted similarly to vertex set and edge set variables, respectively. 
We also have $\ip{\emptyset}^{(G,v,e,\interp f, \interp g)} = 0$.
The interpretations of $\neg \psi$ and $\psi \land \chi$ are defined naturally. 
Moreover, $\ip{\diamondsuit^\sigma \psi}^{(G,v,e,\interp f, \interp g)} = 1$ iff the number of neighbors $w$ of $v$ such that $\ip{\psi}^{(G,w,vw,\interp f, \interp g)} = 1$ belongs to the set $\sigma$.
Lastly, we have $\ip{\square^\sigma \psi}^{(G,v,e,\interp f, \interp g)} = \ip{\neg \diamondsuit^\sigma \neg \psi}^{(G,v,e,\interp f, \interp g)}$. 
In particular, observe that $\ip{\square^{\bN^+} \psi}^{(G,v,e,\interp f, \interp g)} = 1$ iff $v$ has no neighors $w$ such that $\ip{\psi}^{(G,w,vw,\interp f, \interp g)} = 0$.
An \ECML formula is a formula of the form 
$\exists \bar X (\xi \land \blacksquare \psi)$
where:
\begin{itemize}
    \item $\exists\bar X$ is a sequence of existentially quantified vertex and edge set variables.
    \item $\xi$ is an arbitrary quantifier-free arithmetic formula on some integers and the cardinalities of unary/binary relational symbols and some set variables that appear in $\bar X$.
    \item $\blacksquare \psi$ is a finished \CML formula such that all set variables of $\psi$ appear in $\bar X$.
\end{itemize}
Note that size measurements in $\NEOtwo[\finrec]$ cannot express arbitrary arithmetic formulas like \ECML, but, as we explain in the conclusions (cf.\ \cref{sec:conclusions}), it is easy to adapt our model checking algorithm to handle such arithmetic formulas.
On the other hand, \ECML cannot express arbitrary Boolean expressions of arithmetic and \CML formulas (though its model checking algorithm on graph classes of bounded treewidth \cite{Pilipczuk11} can easily be modified to handle it).
Consequently, these two logics are, strictly speaking, not equivalent. Nonetheless, we prove some sort of equivalence when we consider $\ECML^*$ which is the variant of \ECML where every formula is of the form $\exists \bar X \psi$ where $\psi$ is a Boolean combination of finished \CML formulas and size measurements $\abs{X}\prec m$ with $\prec \in \{\leq, \geq, =\}$ and $X$ a vertex or edge set variable.
We start with the simpler direction converting an $\NEOtwo[\finrec]$ formula into an equivalent $\ECML^*$ formula.
\begin{lemma}\label{logic:neotwo:ecml}
    For every $\NEOtwo[\finrec]$ formula $\phi$, there exists an equivalent $\ECML^*$ formula $\phi'$ such that $\abs{\phi'} = \cO(\abs{\phi})$.
\end{lemma}
\begin{proof}
    By \Cref{obs:core}, we can transform $\phi$ into a $\core\NEOtwo[\finrec]$ formula while increasing its length by a multiplicative factor of at most 10.
    Now $\phi$ is of the form $\exists \bar X \psi$
    where $\exists \bar X$ is a sequence of existentially quantified set variables and $\psi$ a quantifier-free Boolean combination of primitive $\NEOtwo$ formulas.
    The following list shows that every primitive $\NEOtwo$ formula that is not a size measurement is equivalent to a finished \CML formula.
    \begin{itemize}
        \item If $X$ is a vertex set variable, then $\constant{P} = X$ is equivalent to $\blacksquare \constant P \leftrightarrow X$.
        \item If $X$ is an edge set variable, then $\constant{P} = X$ is equivalent to $\blacksquare \square^{\bN^+} \constant P \leftrightarrow X$. Indeed, we have $\ip{\square^{\bN^+}\constant P \leftrightarrow X}^{(G,v,\ast,\interp f,\interp g)} =1$ iff the vertex $v$ has no neighbor $w$ such that the edge $vw$ is in the symmetric difference of the sets $\constant{P}(G)$ and $\interp{g}(X)$. 
        This holds for every vertex $v$ of $G$ if and only if for every edge $vw$ of $G$, either both $vw \in \constant{P}$ and $vw \in \interp g(X)$ hold, or none does.
        \item When $X,Y$ and $Z$ are vertex set variables, then $X = Y$, $\comp X = Y$, $X \cup Y = Z$, and $X \cap Y = Z$ are respectively equivalent to $\blacksquare X \leftrightarrow Y$, $\blacksquare \neg X \leftrightarrow Y$, $\blacksquare  (X \lor Y) \leftrightarrow Z$, and $\blacksquare (X \land Y) \leftrightarrow Z$.
        \item When $X,Y$ and $Z$ are edge set variables, then $X = Y$, $\comp X = Y$, $X \cup Y = Z$, and $X \cap Y = Z$ are respectively equivalent to $\blacksquare \square^{\bN^+} X \leftrightarrow Y$, $\blacksquare \square^{\bN^+} \neg X \leftrightarrow Y$, $\blacksquare \square^{\bN^+} (X \lor Y) \leftrightarrow Z$, and $\blacksquare \square^{\bN^+} (X \land Y) \leftrightarrow Z$. The arguments here are analogous to the ones for the second item.
        \item $N^\sigma(X,Y) = W$ is equivalent to $\blacksquare(\diamondsuit^\sigma X \land Y) \leftrightarrow W$.
    \end{itemize}
        We can convert $\phi$ into a $\ECML^*$ formula by substituting according to the above list.
\end{proof}

In the other direction, we can prove that every $\ECML^*$ formula is equivalent to a formula in $\NEOtwo[\finrec]$, 
however, the conversion involves an exponential blowup in the formula length.
Before we provide the formal statement, let us give some intuitions.
In the following examples, we use $X$ for vertex set variables and $Y,Z$ for edge set variables. 
A graph is a model for the formula $\phi = \exists X \exists Z (\blacksquare \diamondsuit^\sigma X \lor Z)$ iff there exist sets $X\subseteq V(G)$ and $Z\subseteq E(G)$ such that every vertex $v$ of $G$ has a number $\ell\in \sigma$ of neighbors $w$ such that $w\in X$ or $vw\in Z$.

To transform this formula in $\NEOtwo[\finrec]$, we use two edge set variables $Y_x$ and $Y_{x,\overline{x}}$ together with two $\NEOtwo[\finrec]$ formulas $\vartheta_x$ and $\vartheta_{x,\overline{x}}$ that are statisfied iff
$Y_{x}$ is the set of edges with both endpoints in $X$, and $Y_{x,\overline{x}}$ is the set of all edges between $X$ and $\overline{X}$.
Then, we can prove that $\phi$ is equivalent to 
\[
    \exists X \exists Z \exists Y_x \exists Y_{x,\overline{x}} \ \vartheta_x \land \vartheta_{x,\overline{x}} \land X \subseteq N^\sigma(\constant V, Z \cup Y_{x}) \land \overline{X} \subseteq N^\sigma(\constant V, Z \cup Y_{x,\overline{x}}).
\]
Intuitively, we can replace a vertex variable $X$ with $Y_x$ or $Y_{x,\overline{x}}$ when we know for sure that the predecessor of the active vertex belongs to respectively $X$ and $\overline{X}$. 
We can do something similar for a more complicated $\ECML^*$ formula like $\exists X_1 \exists Z_1 \dots \exists X_t \exists Z_t \; \blacksquare \diamondsuit^\sigma (X_1 \lor Z_1)\land \dots \land (X_t \lor Z_t)$ but %
we need to use
$2^t$ additional neighborhood operators. 
For example, with $t=2$, the subformula $\blacksquare \diamondsuit^\sigma (X_1 \lor Z_1)\land  (X_2 \lor Z_2)$ can be replaced with
\begin{align*}
    &X_1\cap X_2 \subseteq N^\sigma(\constant{V}, (Z_1 \cup Y_{x_1}) \cap (Z_2 \cup Y_{x_2})) 
     \land X_1 \cap \overline{X_2} \subseteq  N^\sigma(\constant{V}, (Z_1 \cup Y_{x_1}) \cap (Z_2 \cup Y_{x_2,\overline{x_2}})) \\
     & \land \overline{X_1} \cap {X_2} \subseteq N^\sigma(\constant{V}, (Z_1 \cup Y_{x_1,\overline{x_1}}) \cap (Z_2\cup Y_{x_2})) \land \overline{X_1} \cap \overline{X_2} \cap N^\sigma(\constant{V}, (Z_1 \cup Y_{x_1,\overline{x_1}}) \cap (Z_2 \cup Y_{x_2,\overline{x_2}})).
\end{align*}
The problem is that the interactions between vertex and edge terms in $\NEOtwo$ are quite restricted, namely, they only take place through neighborhood operators $N^\sigma(t_1,t_2)$.
When $t_1$ does not contain a neighborhood operator, $N^\sigma(t_1,t_2)$ is equivalent in $\ECML^\ast$ to $\diamondsuit^\sigma \phi_1 \land \phi_2$ where $\phi_1$ and $\phi_2$ are Boolean combinations of respectively (1)~only vertex sets terms and unary relational symbols,
and (2)~only edge set variables and binary relational symbols.
But in $\ECML^*$ (and \ECML), we can express an arbitrary Boolean combination of a mix of (1) and (2).
It would be interesting to know whether this feature of \ECML is useful to express natural problems in a significantly more compact way than in $\NEOtwo[\finrec]$.

\begin{lemma}\label{logic:ecml:neotwo}
    For every $\ECML^*$ formula $\phi$, there exists an equivalent $\NEOtwo[\finrec]$ formula $\phi'$ such that $\abs{\phi'} = \cO(2^{\abs{\phi}})$.
\end{lemma}
\begin{proof}
    To simplify our transformation, we consider the extension of $\ECML^*$ with the following rule:
    \begin{itemize}
        \item Every vertex term $t$ is an unfinished formula and $\ip{t}^{(\vec G,v,a,\interp f,\interp g)} = 1$ iff $v\in \ip{t}^{(\vec G,\interp f,\interp g)}$.
        \item Every edge term $t$ is an unfinished formula and $\ip{t}^{(\vec G,v,a,\interp f,\interp g)} = 1$ iff $a\in \ip{t}^{(\vec G,\interp f,\interp g)}$.
        \item If $t_1$ and $t_2$ are both vertex terms or both edge terms, then $t_1 = t_2$ is a finished formula with the same semantics as in \NEOtwo.
    \end{itemize} 
    Let $\phi= \exists \bar X\, \psi$ be an $\ECML^*$ formula where $\psi$ is a quantifier-free subformula. 
    First, we simplify $\phi$ by removing nested $\diamondsuit$ quantifier: a subformula $\diamondsuit^\sigma \xi$ contained in another subformula $\diamondsuit^{\rho} \xi'$. 
    We can get rid of one by considering the formula $\exists \bar X \exists W ( W = \diamondsuit^{\sigma} \xi) \land \psi'$ where $W$ is a new vertex variable and $\psi'$ is obtained from $\psi$ by replacing $\diamondsuit^\sigma \xi$ with $W$.
    We can do these replacements until $\phi$ is free of nested $\diamondsuit$ quantifiers. 
    After exhaustively applying these replacements, the resulting formula belongs to our extension of $\ECML^*$, is equivalent to $\phi$, and has length $\cO(\abs{\phi})$.
    
    For simplicity, we still denote this formula free of nested $\diamondsuit$ quantifiers by $\phi= \exists \bar X\, \psi$.
    Observe that for every subformula $\diamondsuit^\sigma \xi$ of $\phi$, the formula $\xi$ is a Boolean combination of 
    set variables or unary/binary relational symbols.
    Our next goal is to transform $\phi$ so that all these subformulas $\xi$ become Boolean combinations of edge terms only.
    For this, we fix a subformula $\zeta = \diamondsuit^\sigma \xi$.
    For every vertex term $t$, we consider a new edge set variable $Y_{t}$ and the following formula $\vartheta_{t}$
    \begin{equation*}
        \vartheta_{t} \defeq N^{\bN^+}(\overline{t}, Y_{t}) = \emptyset \land N^{{\bN^+}}(t, \overline{Y_{t}}) = \emptyset.
    \end{equation*}
    Observe that $\vartheta_{t}$ forces $Y_{t}$ to be precisely the set of all edges induced by $t$.
    We also consider a new edge set variable $Y_{t,\overline{t}}$ and the following formula $\vartheta_{t, \overline{t}}$ 
    \[
        \vartheta_{t,\overline{t}} \defeq \exists Y_{\overline{t}} \ \vartheta_{[\overline t]} \land Y_{t,\overline{t}} = \overline{Y_{t} \cup Y_{\overline{t}}}.
    \]
    Note that $\vartheta_{t,\overline{t}}$ makes sure that $Y_{t,\overline{t}}$ is precisely the set of edges between $t$ and $\overline{t}$. 
    Let $t_1,\dots,t_p$ be all the vertex terms that appear in $\zeta$. %
    Note that each $t_i$ is either a vertex set variable or a unary relational symbol.
    We define the formula $\phi_1$ as follows
    \[\phi_1 \defeq \exists \bar X \exists Y_{t_1}\exists Y_{t_1,\overline{t_1}} \dots \exists Y_{t_p}\exists Y_{t_p,\overline{t_p}}  \ \psi \land \bigwedge_{i\in [p]} \vartheta_{t_i} \land \vartheta_{t_i,\overline{t_i}}.\]
    Observe that any model $(G,\interp{f},\interp{g})$ of $\phi$ can be extended into a model of $\phi_1$ by setting $\interp{g}(Y_{t_i'})$ as the set of edges $uv$ with $u,v\in \ip{t_i'}^{(G,\interp{f}, \interp{g})}$, and setting $\interp{g}(Y_{t_i,\overline{t_i}})$ as the set of edges $uv$ with  $u\in\ip{t_i}^{(G,\interp{f}, \interp{g})}$ and $v\in\ip{\overline{t_i}}^{(G,\interp{f},\interp{g})}$ for every $i\in[p]$ and $t_i'\in\{t_i,\overline{t_i}\}$.
    On the other hand, every model of $\phi_1$ is also a model of $\phi$.
    Thus, $\phi_1$ is equivalent to $\phi$. Moreover, the length of $\phi_1$ is $\Oh(\abs{\phi})$ since we added $\cO(p) \leq \cO(\abs{\phi})$ constant-length subformulas to obtain $\phi_1$ from $\phi$.
    
    The following observation provides some conditional equivalences that we will use to replace $t_i$ with either $Y_{t_i}$ or $Y_{t_i,\overline{t_i}}$.
    We call an interpretation $(G,\interp{f},\interp{g})$ of $\phi_1$ \emph{good} if it satisfies all the formulas $\vartheta_{t_i}$ and $\vartheta_{t_i,\overline{t_i}}$ of $\phi_1$.
    \begin{observation}\label{obs:logic:ecml:equiv}
        For every good interpretation $(G,\interp f, \interp{g})$ of $\phi_1$, $i\in [p]$ and edge $vw$ of $G$, we have:
        \begin{itemize}
            \item If $v\in \ip{t_i}^{(G,\interp f, \interp{g})}$, then $\ip{t_i}^{(G,w,vw,\interp f, \interp{g})} = \ip{Y_{t_i}}^{(G,w,vw,\interp f, \interp{g})}$.
            \item If $v\notin \ip{t_i}^{(G,\interp f, \interp{g})}$, then $\ip{t_i}^{(G,w,vw,\interp f, \interp{g})} = \ip{Y_{t_i, \overline{t_i}}}^{(G,w,vw,\interp f, \interp{g})}$.
        \end{itemize}
    \end{observation}
    Let again $\zeta = \diamondsuit^\sigma \xi$ be a subformula of $\phi_1$ and let $t_1,\dots,t_q$
    the vertex terms appearing in $\zeta$.
    Let $S_\zeta$ be set of all tuples $(\widehat t_1, \dots, \widehat t_q)$ such that $\widehat t_i \in \{t_i, \overline{t_i}\}$ for all $i\in [q]$. 
    For each $(\widehat t_1,\dots,\widehat t_q)\in S_\zeta$, we define the formula $\xi_{(\widehat t_1,\dots,\widehat t_q)}$ obtained from $\xi$ by replacing each $t_i$ with (1)~$Y_{t_i}$ if $\widehat t_i = t_i$, or (2)~$Y_{t_i,\overline{t_i}}$ if $\widehat t_i = \overline{t_i}$.
    Now, we consider the \textit{nice version $\zeta'$ of $\zeta$} defined as follows
    \[\zeta' \defeq \bigvee_{(\widehat t_1,\dots,\widehat t_q)\in S_\zeta} (\widehat t_1 \cap \dots \cap \widehat t_q) \land  \diamondsuit^\sigma \xi_{(\widehat t_1,\dots,\widehat t_q)}. \]
    Crucially, note that $\xi_{(\widehat t_1,\dots,\widehat t_q)}$ contains no vertex terms.
    Let $\phi_1'$ be the formula obtained from $\phi_1$ by replacing $\zeta$ with $\zeta'$.    
    We claim that $\phi_1'$ is equivalent to $\phi_1$ (and thus also equivalent to $\phi$).
    By construction, it is sufficient to show that $\ip{\zeta}^{(G,v,\ast, \interp{f},\interp{g})} = \ip{\zeta'}^{(G,v,\ast, \interp{f},\interp{g})}$ for every vertex $v$ of $G$ and every good interpretation $(G,\interp{f},\interp{g})$ of $\phi_1$.
    Indeed, any model of $\phi_1$ or $\phi_1'$ must be good, and since there are no nested $\diamondsuit$, both $\zeta$ and $\zeta'$ will be interpreted with $\ast$ when we interpret $\phi_1$ or $\phi_1'$.
    Let $v$ be a vertex of $G$ and $(G,\interp{f},\interp{g})$ be a good interpretation of $\phi_1$.
    Notice that there exists an unique $(\ddot{ t_1},\dots,\ddot{ t_p}) \in S_\zeta$ with $v\in \ip{\ddot{ t_1} \cap \dots \cap \ddot{ t_p}}^{(G,\interp{f},\interp{g})}$.
    Hence, we have 
    \[
        \ip{\zeta}^{(G,v,\ast, \interp{f},\interp{g})} = \ip{\bigvee_{(\widehat t_1,\dots,\widehat t_p)\in S_\zeta} (\widehat t_1 \cap \dots \cap \widehat t_p) \land  \diamondsuit^\sigma \xi}^{(G,v,\ast, \interp{f},\interp{g})} = \ip{(\ddot{ t_1} \cap \dots \cap \ddot{ t_p}) \land  \diamondsuit^\sigma \xi}^{(G,v,\ast, \interp{f},\interp{g})}
    \]
    By definition of $\xi_{(\ddot{ t_1},\dots,\ddot{ t_p})}$ and \Cref{obs:logic:ecml:equiv}, for every edge $vw$ of $G$, we know that $\ip{\xi}^{(G,w,vw, \interp{f},\interp{g})} = \ip{\xi_{(\ddot{ t_1},\dots,\ddot{ t_p})}}^{(G,w,vw, \interp{f},\interp{g})} $.
    It follows that $\ip{\zeta}^{(G,v,\ast, \interp{f},\interp{g})} = \ip{\zeta'}^{(G,v,\ast, \interp{f},\interp{g})}$, and thus, $\phi_1'$ is equivalent to $\phi_1$.

    Let now $\phi_2$ be the formula obtained from $\phi_1$ by replacing every subformula $\diamondsuit^\sigma \xi$ with its nice version.
    Then $\phi_2$ is equivalent to $\phi_1$ (and $\phi$).
    Now observe that the length of $\phi_2$ is at most $\Oh(2^{\abs{\phi}})$ because the length of the nice version of a subformula $\diamondsuit^\sigma \xi$ has length $\Oh(2^{\abs{\xi}})$. 
    Moreover, for every subformula $\diamondsuit^\sigma \xi$ of $\phi_2$, the subformula $\xi$ is by construction a Boolean combination of edge terms.
    
    The final step is to transform $\phi_2$ into a $\NEOtwo[\finrec]$ formula.
    For every interpretation $(G,\interp{f},\interp{g})$ of $\phi_2$, vertex $v$ of $G$ and $e\in E(G)\cup\{\ast\}$, we have the following equivalences:
    \begin{itemize}
        \item $\ip{t_1\land t_2}^{(G, v, e, \interp f, \interp g)} = \ip{t_1 \cap t_2}^{(G, v, e, \interp f, \interp g)}$ where $t_1$ and $t_2$ are either both vertex terms or both edge terms. It follows from the fact that $\ip{t_1 \land t_2}^{(G, v, e, \interp f, \interp g)} = 1$ iff $v$ (resp. $e$) belongs to the vertex (resp. edge) set $\ip{t_1\cap t_2}^{(G,\interp f, \interp g)}$ when $t_1$ and $t_2$ are both vertex (resp. edge) terms.
        
        \item $\ip{\neg t}^{(G, v, e, \interp f, \interp g)} = \ip{\overline{t}}^{(G, v, e, \interp f, \interp g)}$ where $t$ is a vertex or an edge term.

        \item $\ip{\diamondsuit^\sigma t}^{(G, v, e, \interp f, \interp g)} = \ip{N^\sigma(\constant V, t)}^{(G, v, e, \interp f, \interp g)}$ where $t$ is edge term.
        
        \item $\ip{t}^{(G, v, \ast, \interp f, \interp{g})} = \ip{\emptyset}^{(G, v, \ast, \interp f, \interp{g})} = 0$ for every edge term $t$ because $\ast$ is not an edge.\footnote{In fact, an edge set variable or a binary relational symbol that is not inside a subformula of the form $\diamondsuit^\sigma \xi$ is useless in $\ECML^*$ (and also $\ECML$) as it will always be interpreted, with the dummy element $\ast$, to 0.}
    \end{itemize}
    We use these equivalences to iteratively transform $\phi_2$ into an equivalent $\NEOtwo[\finrec]$ formula $\phi'$.
    More precisely, we first transform every subformula $\diamondsuit^\sigma \delta$ of $\phi_2$
    into a formula $\diamondsuit^\sigma t$ where $t$ is an edge term. 
    For doing so, we use the first 2 equivalences. It works because $\delta$
    is a Boolean combination of edge terms.
    
    After these transformations, for every finished formula $\blacksquare \delta$ of $\psi_1$, $\delta$ is a Boolean combination of vertex terms and formulas of the form $\diamondsuit^\sigma t$ with $t$ an edge term. We can transform $\delta$ into a vertex term $t_\delta$ with all 4 equivalences.
    Finally, we replace $\blacksquare \delta$ by $\constant V = t_\delta$.
    This is correct since $\ip{\blacksquare t_\delta}^{(G, \interp f,\interp g)} =1$ iff every vertex of $G$ belongs to $\ip{t_\delta}^{(G, \interp f, \interp{g})}$.

    It is clear that the formula we obtain $\phi'$ is a $\NEOtwo[\finrec]$ formula. 
    Moreover, one easily checks that $\abs{\phi'} = \Oh(\abs{\phi_2}) = \Oh(2^{\abs{\phi}})$.
    Thanks to the equivalences, we know that $\phi'$ is equivalent to $\phi_2$ and by transitivity, we conclude that $\phi$ and $\phi'$ are equivalent.
\end{proof}

\section{Conclusions}
\label{sec:conclusions}
Before discussing open questions, let us provide some concluding remarks on our result.
The running time of our model checking algorithm is asymptotically tight under the Exponential Time Hypothesis (ETH) in the following sense.
First, the dependence on treedepth cannot be made subexponential because even some of the simplest problems like \textsc{Independent Set} expressible in our logic do not admit algorithms subexponential in the number of vertices (i.e., with running time $2^{o(n)}$) under ETH (see e.g.,~\cite{LokshtanovEtAl2011}).
Second, the exponent of $n$ depends on the number $\ell$ of size measurements in the input formula.
We now justify that the dependence on $\ell$ is unavoidable unless $\FPT = \W[1]$.
In the \textsc{Equitable Coloring} problem, given an $n$-vertex graph $G$ and an integer $k$, the question is whether the graph $G$ admits a proper $k$-coloring in which every color is used by at least $\lfloor n / k \rfloor$ and at most $\lceil n / k \rceil$ vertices, that is, the color classes have almost equal sizes.
It has been observed in~\cite{GomesGS23} that the reduction by Fellows et al.~\cite{FellowsFLRSST11} implies that the problem is \W[1]-hard parameterized by $\td(G) + k$, i.e., the treedepth of the input graph and the number of colors.
Crucially, we can express the condition underlying the problem definition in $\NEOtwo[\finrec]$ as
\begin{align*}
    \phi \defeq &\ \exists X_1 \dots \exists X_k \ N^{\bN^+}(X_1) \cap X_1 = \emptyset \land \dots \land N^{\bN^+}(X_k) \cap X_k = \emptyset \\
    &\land \abs{X_1} \geq \lfloor n / k \rfloor \land \abs{X_1} \leq \lceil n / k \rceil \land \dots \land \abs{X_k} \geq \lfloor n / k \rfloor \land \abs{X_k} \leq \lceil n / k \rceil.
\end{align*}
In particular, the length of this formula is bounded by $\cO(k)$, we have $d_{\phi} = 1$, $\pphi = 1$, and the number of size measurements is $\ell = 2 k$. %
So if it were possible to get rid of the factor $n^{\cO(\ell)}$ in the running time of our model checking algorithm, one could solve \textsc{Equitable Coloring} parameterized by $k + \td(G)$ in fixed-parameter tractable time as follows.
First, we compute an optimal elimination forest of $G$ in time $2^{\Oh(\td(G)^2)} \cdot n^{\Oh(1)}$~\cite{NadaraPS22}.
Then we run the model checking algorithm in time $2^{\cO(k \cdot \td(G))} \cdot n^{\cO(1)}$.
This yields the desired FPT algorithm, contradicting the assumption $\FPT \neq \W[1]$.
Consequently, shaving off the factor $n^{\cO(\ell)}$ from the running time is unlikely to be possible.

Similarly to \cite{Pilipczuk11}, %
it is easy to extend our approach to capture the size measurements of form $h(|X_1|, \dots, |X_z|) \prec m$ where $h$ is an efficiently computable function such that the images of $h$ are not too long numbers.
For instance, polynomial functions $h$ can be treated this way.
For this, one would branch at the root for all possible cardinalities $|X_1|, \dots, |X_z|$, evaluate $h$, compare with $m$, and then either discard the candidate or proceed with the same recursion as before.

The first open question is whether we can derandomize and/or improve the space complexity of our model checking algorithm for $\NEOtwo[\finrec]\ack$ (i.e., when connectivity and acyclicity constraints are allowed) to become logarithmic in $n$.

Also, it can be observed that if $\dphi = 1$, then our algorithm for $\NEOtwo[\finrec]$ actually counts the models of $\phi$.
So we are able to count, for example, dominating sets or independent sets: in fact, for these problems, for each integer $t$, we can count the number of solutions of size $t$.
However, in general, the same model of the formula yields multiple partial solutions at the root (due to the component $\psw$ of the partial solution).
So is it possible to extend our algorithm for model counting?

Further, our logic cannot capture, for example, \textsc{Triangle Packing}. %
However, it is not difficult to use the techniques presented in this paper to solve this problem.
So it is also reasonable to ask if there is a natural generalization of our logic (ideally not using too many ad-hoc operators) which places all problems under the same umbrella.

Regarding our extension to the directed $\NEOtwo[\finrec]$, it is natural to ask whether our model checking algorithm can be extended to handle strong connectivity and directed acyclicity constraints. To the best of our knowledge Cut\&Count is not known to be applicable for this.

\paragraph{AI-disclosure.} The authors only used AI tools for spell checking.

\bibliographystyle{plain}
\bibliography{ref}

\newpage

\appendix

\section{Discrete Fourier Transforms for Polynomials}
\label{sec:appendix:dft}
In this appendix, we carefully check that \Cref{thm:inverse-dft,thm:fourier-convolution-result} hold for polynomials as well.
We recall that $D$ is a finite set, $r$ is a positive integer, and $p$ is a prime such that $\FF_p$ admits the $r$.th root of unity $\omega$.
For $y, q \colon D \to \bZ_r$ we define  $y+q \colon D \to \bZ_r$ as $(y+q)(d) \equiv_r y(d) + q(d)$ for every $d \in D$. 
For $h, g \colon (D \to \bZ_r) \to \FF_{p}$ we define also $h \cdot g \colon (D \to \bZ_r) \to \FF_{p}$ via $(f \cdot g)(y) \equiv_p h(y) \cdot g(y)$ for every $y \in D \to \bZ_r$.
In the following, we use $\DFT$ and $\DFT^{-1}$ as shortcuts for $\DFT_{D, r, p}$ and $\DFT_{D, r, p}^{-1}$, respectively.

\begin{lemma}
    \label{lem:dft:poly:sum}
    For $k \in \bN$, $h^1, \dots, h^k \colon (D \to \bZ_r) \to \FF_{p}[\gamma_1, \dots, \gamma_\ell]$, 
    we have 
    \begin{equation}\label{eq:dft-linear}
        \DFT(h^1 + \dots + h^k) \equiv_p \DFT(h^1) + \dots + \DFT(h^k)
    \end{equation}
    and
    \begin{equation}\label{eq:inv-dft-linear}
        \DFT^{-1}(h^1 + \dots + h^k) \equiv_p \DFT^{-1}(h^1) + \dots + \DFT^{-1}(h^k).
    \end{equation}
\end{lemma}

\begin{proof}
    Let $y \colon D \to \bZ_r$ be arbitrary:
    First, it holds that 
    \begin{align*}
        &\DFT(h^1 + \dots + h^k)(y) \equiv_p \sum_{q \colon D \to \bZ_r} \omega^{q \cdot y} (h^1 + \dots + h^k)(q) \\
        &\equiv_p \sum_{q \colon D \to \bZ_r} \omega^{q \cdot y} (h^1(q) + \dots + h^k(q)) \equiv_p \Bigl(\sum_{q \colon D \to \bZ_r} \omega^{q \cdot y} (h^1(q)\Bigr) + \dots + \Bigl(\sum_{q \colon D \to \bZ_r} \omega^{q \cdot y} (h^k(q)\Bigr) \\
        &\equiv_p (\DFT(h^1))(y) + \dots + (\DFT(h^k))(y) \equiv_p (\DFT(h^1) + \dots + \DFT(h^k))(y). \\
    \end{align*}
    And similarly
    \begin{align*}
        &\DFT^{-1}(h^1 + \dots + h^k)(y) \equiv_p \frac{1}{r^{|D|}} \sum_{q \colon D \to \bZ_r} \omega^{-q \cdot y} (h^1 + \dots + h^k)(q) \\ &\equiv_p \frac{1}{r^{|D|}} \sum_{q \colon D \to \bZ_r} \omega^{-q \cdot y} (h^1(q) + \dots + h^k(q)) \equiv_p \frac{1}{r^{|D|}} \sum_{q \colon D \to \bZ_r} \omega^{-q \cdot y} (h^1(q) + \dots + h^k(q)) \\ &\equiv_p 
        \Bigl(\frac{1}{r^{|D|}} \sum_{q \colon D \to \bZ_r} \omega^{-q \cdot y} h^1(q)\Bigr) + \dots + \Bigl(\frac{1}{r^{|D|}} \sum_{q \colon D \to \bZ_r} \omega^{-q \cdot y} h^k(q)\Bigr) \\
        & \equiv_p (\DFT^{-1}(h^1))(y) + \dots + (\DFT^{-1}(h^k))(y) \equiv_p
        (\DFT^{-1}(h^1) + \dots + \DFT^{-1}(h^k))(y). \\
    \end{align*}
\end{proof}

For a vector $\bara = (a^1, \dots, a^\ell) \in \bN^\ell$ and $h \colon (D \to \bZ_r) \to \FF_{p}$ we define the mapping $h \cdot \prod_{i=1}^\ell \gamma_i^{a^i} \colon (D \to \bZ_r) \to \FF_{p}[\gamma_1, \dots, \gamma_\ell]$ as 
\[
    \Bigl(h \cdot \prod_{i=1}^\ell \gamma_i^{a^i}\Bigr)(y) \equiv_p h(y) \cdot \prod_{i=1}^\ell \gamma_i^{a^i}
\]
for all $y \colon D \to \bZ_r$.

\begin{lemma}
    Let $h \colon (D \to \bZ_r) \to \FF_{p}$, $\bara = (a^1, \dots, a^\ell) \in \bN^\ell$. %
    Then it holds that 
    \begin{equation}\label{eq:dft-variables}
        \DFT\Bigl(h \cdot \prod_{i=1}^\ell \gamma_i^{a^i} \Bigr) \equiv_p \DFT(h) \cdot \prod_{i=1}^\ell \gamma_i^{a^i}
    \end{equation}
    and
    \begin{equation}\label{eq:inv-dft-variables}
        \DFT^{-1}\Bigl(h \cdot \prod_{i=1}^\ell \gamma_i^{a^i}\Bigr) \equiv_p \DFT^{-1}(h) \cdot \prod_{i=1}^\ell \gamma_i^{a^i}
    \end{equation}
\end{lemma} 

\begin{proof}
    Let $y \colon D \to \bZ_r$ be arbitrary:
    First, it holds that:
    \begin{align*}
        &\DFT\Bigl(h \cdot \prod_{i=1}^\ell \gamma_i^{a^i}\Bigr)(y) \equiv_p \sum_{q \colon D \to \bZ_r} \omega^{q \cdot y} \Bigl(h \cdot \prod_{i=1}^\ell \gamma_i^{a^i}\Bigr)(q) \equiv_p \sum_{q \colon D \to \bZ_r} \omega^{q \cdot y} h(q) \cdot \prod_{i=1}^\ell \gamma_i^{a^i} \\
        &\equiv_p \Bigl(\sum_{q \colon D \to \bZ_r} \omega^{q \cdot y} h(q)\Bigr) \cdot \prod_{i=1}^\ell \gamma_i^{a^i} \equiv_p \DFT(h)(y) \cdot \prod_{i=1}^\ell \gamma_i^{a^i}. \\
    \end{align*}
    And similarly, 
    \begin{align*}
        &\DFT^{-1}\Bigl(h \cdot \prod_{i=1}^\ell \gamma_i^{a^i} \Bigr)(y) \equiv_p \frac{1}{r^{|D|}} \sum_{q \colon D \to \bZ_r} \omega^{-q \cdot y} \Bigl(h \cdot \prod_{i=1}^\ell \gamma_i^{a^i} \Bigr)(q) \equiv_p \frac{1}{r^{|D|}} \sum_{q \colon D \to \bZ_r} \omega^{-q \cdot y} h(q) \cdot \prod_{i=1}^\ell \gamma_i^{a^i} \\
        &\equiv_p \Bigl(\frac{1}{r^{|D|}} \sum_{q \colon D \to \bZ_r} \omega^{-q \cdot y} h(q)\Bigr) \cdot \prod_{i=1}^\ell \gamma_i^{a^i} \equiv_p \DFT^{-1}(h)(y) \cdot \prod_{i=1}^\ell \gamma_i^{a^i}. \\
    \end{align*}
\end{proof}

Further, for all $\interp{a} = (a^1, \dots, a^\ell) \in \bN^\ell$ we define the functions $h_{\bara} \colon (D \to \bZ_r) \to \FF_{p}$ in such a way that 
\[
    h(y) \equiv_p \sum_{\bara \in \bN^\ell} h_{\bara}(y) \prod_{i=1}^\ell \gamma_i^{a^i}.
\]
holds for every $y \colon D \to \bZ_r$---since every polynomial is uniquely determined by the coefficients at its monomials, these functions are well-defined.
By definition we then have
\begin{equation}\label{eq:function-into-addends}
    h \equiv_p \sum_{\bara \in \bN^\ell} h_{\bara} \cdot \prod_{i=1}^\ell \gamma_i^{a^i}
\end{equation}

Now we are ready to prove the two main results about DFTs for polynomials from $\FF_{p}[\gamma_1, \dots, \gamma_\ell]$ by relying on their analogues for $\FF_{p}$:

\begin{lemma}\label{thm:inverse-dft-polynomials}
    For every $h \colon (D \to \bZ_r) \to \FF_{p}[\gamma_1, \dots, \gamma_\ell]$ it holds that 
    \begin{equation}\label{eq:inverse-dft-polynomials}
        \DFT^{-1}(\DFT(h)) \equiv_p h \text{ and } \DFT(\DFT^{-1}(h)) \equiv_p h. 
    \end{equation}
\end{lemma}

\begin{proof}
    It holds that:
    \begin{align*}
        &\DFT^{-1}(\DFT(h)) \stackrel{\eqref{eq:function-into-addends}}{\equiv_p} \DFT^{-1}(\DFT\Bigl(\sum_{\bara \in \bN^\ell} h_{\bara} \cdot \prod_{i=1}^\ell \gamma_i^{a^i}\Bigr)) \stackrel{\eqref{eq:dft-linear}}{\equiv_p} \DFT^{-1}\Bigl(\sum_{\bara \in \bN^\ell} \DFT\Bigl(h_{\bara} \cdot \prod_{i=1}^\ell \gamma_i^{a^i}\Bigr)\Bigr) \\
        \stackrel{\eqref{eq:dft-variables}}{\equiv_p}&\DFT^{-1}\Bigl(\sum_{\bara \in \bN^\ell} \DFT(h_{\bara}) \cdot \prod_{i=1}^\ell \gamma_i^{a^i}\Bigr) \stackrel{\eqref{eq:inv-dft-linear}}{\equiv_p} \sum_{\bara \in \bN^\ell} \DFT^{-1}\Bigl(\DFT(h_{\bara}) \cdot \prod_{i=1}^\ell \gamma_i^{a^i}\Bigr) \\
        \stackrel{\eqref{eq:inv-dft-variables}}{=} &\sum_{\bara \in \bN^\ell} \DFT^{-1}(\DFT(h_{\bara})) \cdot \prod_{i=1}^\ell \gamma_i^{a^i} \stackrel{\text{\cref{thm:inverse-dft}}}{\equiv_p} \sum_{\bara \in \bN^\ell} h_{\bara} \cdot \prod_{i=1}^\ell \gamma_i^{a^i} \stackrel{\eqref{eq:function-into-addends}}{\equiv_p} h.\\
    \end{align*}
    And similarly,
    \begin{align*}
        &\DFT(\DFT^{-1}(h)) \stackrel{\eqref{eq:function-into-addends}}{\equiv_p} \DFT(\DFT^{-1}\Bigl(\sum_{\bara \in \bN^\ell} h_{\bara} \cdot \prod_{i=1}^\ell \gamma_i^{a^i}\Bigr)) \stackrel{\eqref{eq:inv-dft-linear}}{\equiv_p} \DFT\Bigl(\sum_{\bara \in \bN^\ell} \DFT^{-1}\Bigl(h_{\bara} \cdot \prod_{i=1}^\ell \gamma_i^{a^i}\Bigr)\Bigr) \\
        \stackrel{\eqref{eq:inv-dft-variables}}{\equiv_p}&\DFT\Bigl(\sum_{\bara \in \bN^\ell} \DFT^{-1}(h_{\bara}) \cdot \prod_{i=1}^\ell \gamma_i^{a^i}\Bigr) \stackrel{\eqref{eq:dft-linear}}{\equiv_p} \sum_{\bara \in \bN^\ell} \DFT\Bigl(\DFT^{-1}(h_{\bara}) \cdot \prod_{i=1}^\ell \gamma_i^{a^i}\Bigr)  \\
        \stackrel{\eqref{eq:dft-variables}}{\equiv_p}&\sum_{\bara \in \bN^\ell} \DFT(\DFT^{-1}(h_{\bara})) \cdot \prod_{i=1}^\ell \gamma_i^{a^i} \stackrel{\text{\cref{thm:inverse-dft}}}{\equiv_p} \sum_{\bara \in \bN^\ell} h_{\bara} \cdot \prod_{i=1}^\ell \gamma_i^{a^i} \stackrel{\eqref{eq:function-into-addends}}{\equiv_p} h.\\
    \end{align*}
\end{proof}
And it also holds that
\begin{lemma}\label{thm:fourier-convolution-result-polynomials}
    Let $t \in \bN$, let $h^1, \dots, h^t \colon (D \to \bZ_r) \to \FF_{p}[\gamma_1, \dots, \gamma_\ell]$, and let $y \in D \to \bZ_r$. Then it holds that 
    \[
        (\DFT^{-1} (\DFT(h^1) \cdot \ldots \cdot \DFT(h^t)))(y) \equiv_p \sum_{\substack{y^1, \dots, y^t \colon D \to \bZ_r \\ y^1 + \dots + y^t \equiv_r y}} h^1(y^1) \cdot \ldots \cdot h^t(y^t).
    \]
\end{lemma}

\begin{proof}
    \begin{align*}
        &(\DFT^{-1} (\DFT(h^1) \cdot \ldots \cdot \DFT(h^t)))(y) \equiv_p (\DFT^{-1} \Bigl(\prod_{i=1}^t \DFT(h^i)\Bigr))(y)  \\
        \stackrel{\eqref{eq:function-into-addends}}{\equiv_p}&(\DFT^{-1} \Bigl(\prod_{i=1}^t \DFT\Bigl(\sum_{\bara_i \in \bN^\ell} (h^i_{\bara_i} \cdot \prod_{j=1}^\ell \gamma_j^{a_i^j})\Bigr)\Bigr))(y) \stackrel{\eqref{eq:dft-linear}}{\equiv_p} (\DFT^{-1} \Bigl(\prod_{i=1}^t \sum_{\bara \in \bN^\ell} \DFT\Bigl(h^i_{\bara_i} \cdot \prod_{j=1}^\ell \gamma_j^{a^j_i}\Bigr)\Bigr))(y) \\
        \stackrel{\eqref{eq:dft-variables}}{\equiv_p} &(\DFT^{-1} \Bigl(\prod_{i=1}^t \sum_{\bara_i \in \bN^\ell} (\DFT(h^i_{\bara_i}) \cdot \prod_{j=1}^\ell \gamma_j^{a^j_i})\Bigr))(y) \equiv_p (\DFT^{-1} \Bigl(\sum_{\bara_1, \dots, \bara_t \in \bN^\ell}\prod_{i=1}^t \Bigl(\DFT(h^i_{\bara_i}) \cdot \prod_{j=1}^\ell \gamma_j^{a^j_i}\Bigr)\Bigr))(y)  \\
        \equiv_p &(\DFT^{-1} \Bigl(\sum_{\bara_1, \dots, \bara_t \in \bN^\ell} \Bigl(\prod_{j=1}^\ell \gamma_j^{\sum_{i=1}^t a^j_i}\Bigr) \Bigl(\prod_{i=1}^t \DFT(h^i_{\bara_i})\Bigr)\Bigr))(y)  \\
    \end{align*}
    \begin{align*}
        \equiv_p &(\DFT^{-1} \Bigl(\sum_{\bara \in \bN^\ell} \sum_{\substack{\bara_1, \dots, \bara_t \in \bN^\ell \colon \\ \bara_1 + \dots + \bara_t = \bara}} \Bigl(\prod_{j=1}^\ell \gamma_j^{\sum_{i=1}^t a^j_i}\Bigr) \Bigl(\prod_{i=1}^t \DFT(h^i_{\bara_i})\Bigr)\Bigr))(y) \\
        \equiv_p &(\DFT^{-1} \Bigl(\sum_{\bara \in \bN^\ell} \sum_{\substack{\bara_1, \dots, \bara_t \in \bN^\ell \colon \\ \bara_1 + \dots + \bara_t = \bara}} \Bigl(\prod_{j=1}^\ell \gamma_j^{a^j}\Bigr) \Bigl(\prod_{i=1}^t \DFT(h^i_{\bara_i})\Bigr)\Bigr))(y) \\
        \equiv_p &(\DFT^{-1} \Bigl(\sum_{\bara \in \bN^\ell} \Bigl(\sum_{\substack{\bara_1, \dots, \bara_t \in \bN^\ell \colon \\ \bara_1 + \dots + \bara_t = \bara}} \prod_{i=1}^t \DFT(h^i_{\bara_i})\Bigr) \cdot \prod_{j=1}^\ell \gamma_j^{a^j}\Bigr))(y) \\
        \stackrel{\eqref{eq:inv-dft-linear}}{\equiv_p} &(\sum_{\bara \in \bN^\ell} \DFT^{-1} \Bigl(\Bigl(\sum_{\substack{\bara_1, \dots, \bara_t \in \bN^\ell \colon \\ \bara_1 + \dots + \bara_t = \bara}} \prod_{i=1}^t \DFT(h^i_{\bara_i})\Bigr) \cdot \prod_{j=1}^\ell \gamma_j^{a^j}\Bigr))(y) \\
        \stackrel{\eqref{eq:inv-dft-variables}}{\equiv_p} &(\sum_{\bara \in \bN^\ell} \Bigl(\DFT^{-1}\Bigl(\sum_{\substack{\bara_1, \dots, \bara_t \in \bN^\ell \colon \\ \bara_1 + \dots + \bara_t = \bara}} \prod_{i=1}^t \DFT(h^i_{\bara_i})\Bigr) \cdot \prod_{j=1}^\ell \gamma_j^{a^j}\Bigr))(y) \\
        \stackrel{\eqref{eq:inv-dft-linear}}{\equiv_p} &(\sum_{\bara \in \bN^\ell} \Bigl(\sum_{\substack{\bara_1, \dots, \bara_t \in \bN^\ell \colon \\ \bara_1 + \dots + \bara_t = \bara}} \DFT^{-1}\Bigl(\prod_{i=1}^t \DFT(h^i_{\bara_i})\Bigr)\Bigr) \cdot \prod_{j=1}^\ell \gamma_j^{a^j})(y) \\
        \equiv_p &\sum_{\bara \in \bN^\ell} \Bigl(\sum_{\substack{\bara_1, \dots, \bara_t \in \bN^\ell \colon \\ \bara_1 + \dots + \bara_t = \bara}} \Bigl(\Bigl(\DFT^{-1}\Bigl(\prod_{i=1}^t \DFT(h^i_{\bara_i})\Bigr)\Bigr)(y)\Bigr) \cdot \prod_{j=1}^\ell \gamma_j^{a^j}\Bigr) \\ 
        \stackrel{\text{\cref{thm:fourier-convolution-result}}}{\equiv_p}  &\sum_{\bara \in \bN^\ell} \Bigl(\sum_{\substack{\bara_1, \dots, \bara_t \in \bN^\ell \colon \\ \bara_1 + \dots + \bara_t = \bara}} \Bigl(\sum_{\substack{y_1, \dots, y_t \colon D \to \bZ_r \colon \\ y_1 + \dots y_t \equiv_r y}} \prod_{i=1}^t h^i_{\bara_i}(y_i)\Bigr)\Bigr) \cdot \prod_{j=1}^\ell \gamma_j^{a^j} \\ 
        \equiv_p &\sum_{\bara \in \bN^\ell} \sum_{\substack{\bara_1, \dots, \bara_t \in \bN^\ell \colon \\ \bara_1 + \dots + \bara_t = \bara}} \sum_{\substack{y_1, \dots, y_t \colon D \to \bZ_r \colon \\ y_1 + \dots y_t \equiv_r y}} \prod_{i=1}^t \Bigl(h^i_{\bara_i}(y_i) \cdot \Bigl(\prod_{j=1}^\ell \gamma_j^{a^j_i}\Bigr) \Bigr) \\
        \equiv_p &\sum_{\substack{y_1, \dots, y_t \colon D \to \bZ_r \colon \\ y_1 + \dots y_t \equiv_r y}} \sum_{\bara_1, \dots, \bara_t \in \bN^\ell}  \prod_{i=1}^t \Bigl(h^i_{\bara_i}(y_i) \cdot \Bigl(\prod_{j=1}^\ell \gamma_j^{a^j_i}\Bigr) \Bigr) \\
        \equiv_p &\sum_{\substack{y_1, \dots, y_t \colon D \to \bZ_r \colon \\ y_1 + \dots y_t \equiv_r y}} \prod_{i=1}^t \sum_{\bara_i \in \bN^\ell} \Bigl( h^i_{\bara_i}(y_i) \prod_{j=1}^\ell \gamma_j^{a^j_i} \Bigr) \\
        \stackrel{\eqref{eq:function-into-addends}}{\equiv_p} &\sum_{\substack{y_1, \dots, y_t \colon D \to \bZ_r \colon \\ y_1 + \dots y_t \equiv_r y}} \prod_{i=1}^t h^i(y_i) \equiv_p  \sum_{\substack{y_1, \dots, y_t \colon D \to \bZ_r \colon \\ y_1 + \dots y_t \equiv_r y}} h^1(y_1) \cdot \ldots \cdot h^t(y_t) \\
    \end{align*}
\end{proof}

\end{document}